\documentclass[10pt,a4paper]{article}

\newif\ifignore 
\ignorefalse

\newcommand{\auxproof}[1]{%
\ifignore\mbox{}\newline
\textbf{PROOF:} \dotfill\newline
{\it #1}\mbox{}\newline
\textbf{ENDPROOF}\dotfill\newline
\fi}

\usepackage[english]{babel}
\usepackage{amsmath}
\usepackage{amssymb}
\usepackage{amstext}
\usepackage{amsthm}
\usepackage{latexsym}
\usepackage{xspace}
\usepackage{stmaryrd}
\usepackage{url,hyperref}
\usepackage{fancybox}
\usepackage{nicefrac}
\usepackage{multind}
\usepackage{multicol}
\usepackage{xy}
\usepackage{mathrsfs}
\xyoption{all} 
\xyoption{2cell}
\UseTwocells

\usepackage{tikz}
\usetikzlibrary{circuits.ee.IEC}

\newdimen\proofrulebreadth \proofrulebreadth=.05em
\newdimen\proofdotseparation \proofdotseparation=1.25ex
\newdimen\proofrulebaseline \proofrulebaseline=2ex
\newcount\proofdotnumber \proofdotnumber=3
\let\then\relax
\def\hfi{\hskip0pt plus.0001fil}
\mathchardef\squigto="3A3B
\newif\ifinsideprooftree\insideprooftreefalse
\newif\ifonleftofproofrule\onleftofproofrulefalse
\newif\ifproofdots\proofdotsfalse
\newif\ifdoubleproof\doubleprooffalse
\let\wereinproofbit\relax
\newdimen\shortenproofleft
\newdimen\shortenproofright
\newdimen\proofbelowshift
\newbox\proofabove
\newbox\proofbelow
\newbox\proofrulename
\def\shiftproofbelow{\let\next\relax\afterassignment\setshiftproofbelow\dimen0 }
\def\shiftproofbelowneg{\def\next{\multiply\dimen0 by-1 }%
\afterassignment\setshiftproofbelow\dimen0 }
\def\setshiftproofbelow{\next\proofbelowshift=\dimen0 }
\def\setproofrulebreadth{\proofrulebreadth}

\def\prooftree{
\ifnum  \lastpenalty=1
\then   \unpenalty
\else   \onleftofproofrulefalse
\fi
\ifonleftofproofrule
\else   \ifinsideprooftree
        \then   \hskip.5em plus1fil
        \fi
\fi
\bgroup
\setbox\proofbelow=\hbox{}\setbox\proofrulename=\hbox{}%
\let\justifies\proofover\let\leadsto\proofoverdots\let\Justifies\proofoverdbl
\let\using\proofusing\let\[\prooftree
\ifinsideprooftree\let\]\endprooftree\fi
\proofdotsfalse\doubleprooffalse
\let\thickness\setproofrulebreadth
\let\shiftright\shiftproofbelow \let\shift\shiftproofbelow
\let\shiftleft\shiftproofbelowneg
\let\ifwasinsideprooftree\ifinsideprooftree
\insideprooftreetrue
\setbox\proofabove=\hbox\bgroup$\displaystyle 
\let\wereinproofbit\prooftree
\shortenproofleft=0pt \shortenproofright=0pt \proofbelowshift=0pt
\onleftofproofruletrue\penalty1
}

\def\eproofbit{
\ifx    \wereinproofbit\prooftree
\then   \ifcase \lastpenalty
        \then   \shortenproofright=0pt  
        \or     \unpenalty\hfil         
        \or     \unpenalty\unskip       
        \else   \shortenproofright=0pt  
        \fi
\fi
\global\dimen0=\shortenproofleft
\global\dimen1=\shortenproofright
\global\dimen2=\proofrulebreadth
\global\dimen3=\proofbelowshift
\global\dimen4=\proofdotseparation
\global\count255=\proofdotnumber
$\egroup  
\shortenproofleft=\dimen0
\shortenproofright=\dimen1
\proofrulebreadth=\dimen2
\proofbelowshift=\dimen3
\proofdotseparation=\dimen4
\proofdotnumber=\count255
}

\def\proofover{
\eproofbit 
\setbox\proofbelow=\hbox\bgroup 
\let\wereinproofbit\proofover
$\displaystyle
}%
\def\proofoverdbl{
\eproofbit 
\doubleprooftrue
\setbox\proofbelow=\hbox\bgroup 
\let\wereinproofbit\proofoverdbl
$\displaystyle
}%
\def\proofoverdots{
\eproofbit 
\proofdotstrue
\setbox\proofbelow=\hbox\bgroup 
\let\wereinproofbit\proofoverdots
$\displaystyle
}%
\def\proofusing{
\eproofbit 
\setbox\proofrulename=\hbox\bgroup 
\let\wereinproofbit\proofusing
\kern0.3em$
}

\def\endprooftree{
\eproofbit 
  \dimen5 =0pt
\dimen0=\wd\proofabove \advance\dimen0-\shortenproofleft
\advance\dimen0-\shortenproofright
\dimen1=.5\dimen0 \advance\dimen1-.5\wd\proofbelow
\dimen4=\dimen1
\advance\dimen1\proofbelowshift \advance\dimen4-\proofbelowshift
\ifdim  \dimen1<0pt
\then   \advance\shortenproofleft\dimen1
        \advance\dimen0-\dimen1
        \dimen1=0pt
        \ifdim  \shortenproofleft<0pt
        \then   \setbox\proofabove=\hbox{%
                        \kern-\shortenproofleft\unhbox\proofabove}%
                \shortenproofleft=0pt
        \fi
\fi
\ifdim  \dimen4<0pt
\then   \advance\shortenproofright\dimen4
        \advance\dimen0-\dimen4
        \dimen4=0pt
\fi
\ifdim  \shortenproofright<\wd\proofrulename
\then   \shortenproofright=\wd\proofrulename
\fi
\dimen2=\shortenproofleft \advance\dimen2 by\dimen1
\dimen3=\shortenproofright\advance\dimen3 by\dimen4
\ifproofdots
\then
        \dimen6=\shortenproofleft \advance\dimen6 .5\dimen0
        \setbox1=\vbox to\proofdotseparation{\vss\hbox{$\cdot$}\vss}%
        \setbox0=\hbox{%
                \advance\dimen6-.5\wd1
                \kern\dimen6
                $\vcenter to\proofdotnumber\proofdotseparation
                        {\leaders\box1\vfill}$%
                \unhbox\proofrulename}%
\else   \dimen6=\fontdimen22\the\textfont2 
        \dimen7=\dimen6
        \advance\dimen6by.5\proofrulebreadth
        \advance\dimen7by-.5\proofrulebreadth
        \setbox0=\hbox{%
                \kern\shortenproofleft
                \ifdoubleproof
                \then   \hbox to\dimen0{%
                        $\mathsurround0pt\mathord=\mkern-6mu%
                        \cleaders\hbox{$\mkern-2mu=\mkern-2mu$}\hfill
                        \mkern-6mu\mathord=$}%
                \else   \vrule height\dimen6 depth-\dimen7 width\dimen0
                \fi
                \unhbox\proofrulename}%
        \ht0=\dimen6 \dp0=-\dimen7
\fi
\let\doll\relax
\ifwasinsideprooftree
\then   \let\VBOX\vbox
\else   \ifmmode\else$\let\doll=$\fi
        \let\VBOX\vcenter
\fi
\VBOX   {\baselineskip\proofrulebaseline \lineskip.2ex
        \expandafter\lineskiplimit\ifproofdots0ex\else-0.6ex\fi
        \hbox   spread\dimen5   {\hfi\unhbox\proofabove\hfi}%
        \hbox{\box0}%
        \hbox   {\kern\dimen2 \box\proofbelow}}\doll%
\global\dimen2=\dimen2
\global\dimen3=\dimen3
\egroup 
\ifonleftofproofrule
\then   \shortenproofleft=\dimen2
\fi
\shortenproofright=\dimen3
\onleftofproofrulefalse
\ifinsideprooftree
\then   \hskip.5em plus 1fil \penalty2
\fi
}

\newcommand\after{\mathbin{\circ}}

\newsavebox\sbpafterd
\savebox\sbpafterd{\begin{tikzpicture}[baseline=-2.5pt]
            \filldraw[fill=white,draw=white] circle (2.1pt);
            \filldraw[fill=white,draw=black, line width=0.2pt] circle (1.8pt);
            \filldraw[fill=black, line width=0] circle (0.2pt);
                                \end{tikzpicture}}

\newsavebox\sbpaftert
\savebox\sbpaftert{\begin{tikzpicture}[baseline=-2.5pt]
            \filldraw[fill=white,draw=white] circle (2.1pt);
            \filldraw[fill=white,draw=black, line width=0.2pt] circle (1.8pt);
            \filldraw[fill=black, line width=0] circle (0.2pt);
                                \end{tikzpicture}}
\newsavebox\sbpafters
\savebox\sbpafters{\begin{tikzpicture}[baseline=-1.75pt]
            \filldraw[fill=white,draw=white] circle (1.7pt);
            \filldraw[fill=white,draw=black,line width=0.15pt] circle (1.4pt);
            \filldraw[fill=black, line width=0] circle (0.15pt);
                                \end{tikzpicture}}

\newsavebox\sbpafterss
\savebox\sbpafterss{\begin{tikzpicture}[baseline=-1.25pt]
            \filldraw[fill=white,draw=white] circle (1.4pt);
            \filldraw[fill=white,draw=black,line width=0.1pt] circle (1.1pt);
            \filldraw[fill=black, line width=0] circle (0.15pt);
                                \end{tikzpicture}}

\newcommand{\pafter}{\begingroup
        \mathchoice{\mathbin{\usebox\sbpafterd}}
        {\mathbin{\usebox\sbpaftert}}
        {\mathbin{\usebox\sbpafters}}
        {\mathbin{\usebox\sbpafterss}}\endgroup}

\newsavebox\sbtafterd
\savebox\sbtafterd{\begin{tikzpicture}[baseline=-2.5pt]
                             \filldraw[fill=black,draw=white] circle (1.4pt);
                                \end{tikzpicture}}

\newsavebox\sbtaftert
\savebox\sbtaftert{\begin{tikzpicture}[baseline=-2.5pt]
                             \filldraw[fill=black,draw=white] circle (1.4pt);
                                \end{tikzpicture}}
\newsavebox\sbtafters
\savebox\sbtafters{\begin{tikzpicture}[baseline=-1.75pt]
                             \filldraw[fill=black,draw=white] circle (1.2pt);
                                \end{tikzpicture}}
\newsavebox\sbtafterss
\savebox\sbtafterss{\begin{tikzpicture}[baseline=-1.25pt]
                             \filldraw[fill=black,draw=white] circle (1.0pt);
                                \end{tikzpicture}}

\newcommand{\tafter}{
        \begingroup\mathchoice{\mathbin{\usebox\sbtafterd}}
        {\mathbin{\usebox\sbtaftert}}
        {\mathbin{\usebox\sbtafters}}
        {\mathbin{\usebox\sbtafterss}}\endgroup}

\newsavebox\sbtto
\savebox\sbtto{\begin{tikzpicture}[baseline=-2.5pt]
            \filldraw[fill=black,draw=white] circle (1.4pt);
                \end{tikzpicture}}

\newcommand\tto{\mathrel{\ooalign{$\to$\cr
            \hfil$\usebox\sbtto$\hfil\cr}}}

\newsavebox\sbpto
\savebox\sbpto{\begin{tikzpicture}[baseline=-2.5pt]
            \filldraw[fill=white,draw=white] circle (1.4pt);
            \filldraw[fill=white,draw=black,line width=0.2pt] circle (1.2pt);
            \filldraw[fill=black, line width=0] circle (0.2pt);
                \end{tikzpicture}}
\newcommand\pto{\mathrel{\ooalign{$\to$\cr
            \hfil$\usebox\sbpto$\hfil\cr}}}

\newsavebox\sbtplusd
\savebox\sbtplusd{\begin{tikzpicture}[baseline=-2.5pt]
                 \filldraw[fill=black,draw=white] circle (1.4pt);
                    \end{tikzpicture}}

\newsavebox\sbtplust
\savebox\sbtplust{\begin{tikzpicture}[baseline=-2.5pt]
                 \filldraw[fill=black,draw=white] circle (1.4pt);
                    \end{tikzpicture}}

\newsavebox\sbtpluss
\savebox\sbtpluss{\begin{tikzpicture}[baseline=-1.75pt]
                 \filldraw[fill=black,draw=white] circle (1.2pt);
                    \end{tikzpicture}}

\newsavebox\sbtplusss
\savebox\sbtplusss{\begin{tikzpicture}[baseline=-1.25pt]
                 \filldraw[fill=black,draw=white] circle (1.0pt);
                    \end{tikzpicture}}

\newcommand{\tplus}{
        \mathchoice%
    {\mathbin{\ooalign{$\displaystyle+$\cr\hfil$\usebox\sbtplusd$\hfil\cr}}}
    {\mathbin{\ooalign{$\textstyle+$\cr\hfil$\usebox\sbtplust$\hfil\cr}}}
    {\mathbin{\ooalign{$\scriptstyle+$\cr\hfil$\usebox\sbtpluss$\hfil\cr}}}
    {\mathbin{\ooalign{$\scriptscriptstyle+$\cr\hfil$\usebox\sbtplusss$\hfil\cr}}}}

\newsavebox\sbpplusd
\savebox\sbpplusd{\begin{tikzpicture}[baseline=-2.5pt]
            \filldraw[fill=white,draw=white] circle (1.4pt);
            \filldraw[fill=white,draw=black, line width=0.2pt] circle (1.2pt);
            \filldraw[fill=black, line width=0] circle (0.2pt);
                                \end{tikzpicture}}

\newsavebox\sbpplust
\savebox\sbpplust{\begin{tikzpicture}[baseline=-2.5pt]
            \filldraw[fill=white,draw=white] circle (1.4pt);
            \filldraw[fill=white,draw=black, line width=0.2pt] circle (1.2pt);
            \filldraw[fill=black, line width=0] circle (0.2pt);
                                \end{tikzpicture}}

\newsavebox\sbppluss
\savebox\sbppluss{\begin{tikzpicture}[baseline=-1.75pt]
            \filldraw[fill=white,draw=white] circle (1.3pt);
            \filldraw[fill=white,draw=black,line width=0.15pt] circle (1.1pt);
            \filldraw[fill=black, line width=0] circle (0.15pt);
                                \end{tikzpicture}}

\newsavebox\sbpplusss
\savebox\sbpplusss{\begin{tikzpicture}[baseline=-1.25pt]
            \filldraw[fill=white,draw=white] circle (1.2pt);
            \filldraw[fill=white,draw=black,line width=0.1pt] circle (1.0pt);
            \filldraw[fill=black, line width=0] circle (0.15pt);
                                \end{tikzpicture}}

\newcommand{\pplus}{
        \mathchoice%
    {\mathbin{\ooalign{$\textstyle+$\cr\hfil$\usebox\sbpplusd$\hfil\cr}}}
    {\mathbin{\ooalign{$\textstyle+$\cr\hfil$\usebox\sbpplust$\hfil\cr}}}
    {\mathbin{\ooalign{$\scriptstyle+$\cr\hfil$\usebox\sbppluss$\hfil\cr}}}
    {\mathbin{\ooalign{$\scriptscriptstyle+$\cr\hfil$\usebox\sbpplusss$\hfil\cr}}}}

\DeclareSymbolFont{T1op}{T1}{cmr}{m}{n}
\SetSymbolFont{T1op}{bold}{T1}{cmr}{bx}{n}
\DeclareMathSymbol{\mathguilsinglleft}{\mathopen}{T1op}{'016}
\DeclareMathSymbol{\mathguilsinglright}{\mathclose}{T1op}{'017}
\newcommand{\klin}[1]{\mathguilsinglleft#1\mathguilsinglright}

\DeclareMathSymbol{\lasyguilsinglleft}{\mathopen}{lasy}{40}
\DeclareMathSymbol{\lasyguilsinglright}{\mathclose}{lasy}{41}

\DeclareFontFamily{U}{mathb}{\hyphenchar\font45}
\DeclareFontShape{U}{mathb}{m}{n}{
      <5> <6> <7> <8> <9> <10> gen * mathb
      <10.95> mathb10 <12> <14.4> <17.28> <20.74> <24.88> mathb12
      }{}
\DeclareSymbolFont{mathb}{U}{mathb}{m}{n}
\DeclareFontSubstitution{U}{mathb}{m}{n}
\DeclareFontFamily{U}{mathx}{\hyphenchar\font45}
\DeclareFontShape{U}{mathx}{m}{n}{
      <5> <6> <7> <8> <9> <10>
      <10.95> <12> <14.4> <17.28> <20.74> <24.88>
      mathx10
      }{}
\DeclareSymbolFont{mathx}{U}{mathx}{m}{n}
\DeclareFontSubstitution{U}{mathx}{m}{n}
\DeclareMathDelimiter{\lgroupabx}{4}{mathb}{"70}{mathx}{"76}
\DeclareMathDelimiter{\rgroupabx}{5}{mathb}{"71}{mathx}{"77}

\newcounter{main}
\newtheorem{theorem}[main]{Theorem}
\newtheorem{lemma}[main]{Lemma}
\newtheorem{proposition}[main]{Proposition}
\newtheorem{corollary}[main]{Corollary}
\newtheorem{definition}[main]{Definition}
\newtheorem{exercise}[main]{Exercise}

\theoremstyle{definition}
\newtheorem{remark}[main]{Remark}
\newtheorem{discussion}[main]{Discussion}
\newtheorem{postulate}[main]{Postulate}
\newtheorem{example}[main]{Example}
\newtheorem{examples}[main]{Examples}

\renewenvironment{proof}[1][Proof]%
   { \begin{trivlist}%
     \item[\hskip \labelsep {\bfseries #1}]%
   }%
   { \end{trivlist}%
   }

\newcommand{\set}[2]{\{#1\;|\;#2\}}
\newcommand{\setin}[3]{\{#1\in#2\;|\;#3\}}

\newcommand{\all}[2]{\forall{#1}.\,#2}
\newcommand{\allin}[3]{\forall{#1\in#2}.\,#3}

\newcommand{\exin}[3]{\exists{#1\in#2}.\,#3}

\newcommand{\tuple}[1]{\langle#1\rangle}
\newcommand{\dtuple}[1]{\tuple{\!\tuple{#1}\!}}

\newcommand{\downset}{\mathop{\downarrow}\!}

\newcommand{\pideal}[2]{\tuple{#1}_{#2}} 
\newcommand{\N}{\mathbb{N}}
\newcommand{\NNO}{\N}
\newcommand{\C}{\mathbb{C}}
\newcommand{\R}{\mathbb{R}}
\newcommand{\Z}{\mathbb{Z}}
\newcommand{\intd}{{\kern.2em}\mathrm{d}{\kern.03em}}
\newcommand{\ket}[1]{\ensuremath{|{\kern.1em}#1{\kern.1em}\rangle}}
\newcommand{\bigket}[1]{\ensuremath{\big|{\kern.1em}#1{\kern.1em}\big\rangle}}
\newcommand{\bra}[1]{\langle\,#1\,|}
\newcommand{\ie}{\textit{i.e.}\xspace}

\newcommand{\bang}{\mathord{!}}
\newcommand{\andthen}{\mathrel{\&}}

\newcommand{\Dst}{\mathcal{D}}
\newcommand{\sDst}{\Dst_{\leq 1}}
\newcommand{\Mlt}{\mathcal{M}}
\newcommand{\Giry}{\mathcal{G}}

\newcommand{\Pow}{\mathcal{P}}

\newcommand{\idmap}[1][]{\ensuremath{\mathrm{id}_{#1}}}

\newcommand{\op}[1]{#1^{\mathrm{op}}}
\newcommand{\orthogonal}{\mathrel{\bot}}
\newcommand{\bigovee}{\mathop{\vphantom{\sum}\mathchoice%
        {\vcenter{\hbox{\huge $\ovee$}}}%
        {\vcenter{\hbox{\Large $\ovee$}}}%
        {\ovee}{\ovee}}\displaylimits}
\newcommand{\supp}{\textsl{supp}}

\newcommand{\Sub}{\ensuremath{\textsl{Sub}}}
\newcommand{\Quot}{\ensuremath{\textsl{Quot}}\xspace}

\newcommand{\Hom}{\textsl{Hom}}

\newcommand{\dis}{\ensuremath{\textsl{dis}}}
\newcommand{\tr}{\ensuremath{\textsl{tr}}}

\newcommand{\indic}[1]{\mathbf{1}_{#1}}
\newcommand{\img}{\mathrm{im}}
\newcommand{\kerbot}{\ensuremath{\ker^{\bot}}}
\newcommand{\imgbot}{\ensuremath{\img^{\bot}}}
\newcommand{\iep}{\mathrm{ie}}

\newcommand{\coker}{\ensuremath{\mathrm{coker}}}
\newcommand{\sotimes}{\mathrel{\raisebox{.05pc}{$\scriptstyle \otimes$}}}
\newcommand{\inprod}[2]{\ensuremath{\langle #1\,|\,#2 \rangle}}
\newcommand{\floor}[1]{\ensuremath{\lfloor#1\rfloor}}
\newcommand{\ceil}[1]{\ensuremath{{\lceil#1\rceil}}}
\newcommand{\cmpr}[2]{\ensuremath{\{#1|{\kern.2ex}#2\}}}
\newcommand{\instr}{{\textrm{instr}}}
\newcommand{\dc}{{\textrm{dc}}}    
\newcommand{\asrt}{{\textrm{asrt}}}
\newcommand{\defeq}{\smash{\stackrel{\mathrm{def}}{=}}}
\newcommand{\tbox}[1]{\ensuremath{#1^{\ast}}}
\newcommand{\pbox}[1]{\ensuremath{#1^{\Box}}}
\newcommand{\pdiam}[1]{\ensuremath{#1^{\Diamond}}}
\newcommand{\tstat}[1]{\ensuremath{#1_{\ast}}}
\newcommand{\one}{\ensuremath{\mathbf{1}}}
\newcommand{\zero}{\ensuremath{\mathbf{0}}}

\newcommand{\cat}[1]{\ensuremath{\mathbf{#1}}\xspace}
\newcommand{\Cat}[1]{\ensuremath{\mathbf{#1}}\xspace}
\newcommand{\Sets}{\Cat{Sets}}
\newcommand{\Ab}{\Cat{Ab}}
\newcommand{\OAb}{\Cat{OAb}}
\newcommand{\OUG}{\Cat{OUG}}
\newcommand{\OUS}{\Cat{OUS}}
\newcommand{\Top}{\Cat{Top}}
\newcommand{\Meas}{\Cat{Meas}}
\newcommand{\CH}{\Cat{CH}}
\newcommand{\CRng}{\Cat{CRng}}

\newcommand{\PCM}{\Cat{PCM}}
\newcommand{\PCMod}{\Cat{PCMod}}
\newcommand{\EA}{\Cat{EA}}
\newcommand{\BA}{\Cat{BA}}

\newcommand{\EMod}{\Cat{EMod}}

\newcommand{\Conv}{\Cat{Conv}}
\newcommand{\Hilb}{\Cat{Hilb}}

\newcommand{\Vect}{\Cat{Vect}}

\newcommand{\vNA}{\Cat{vNA}}

\newcommand{\CvNA}{\Cat{CvNA}}

\newcommand{\OMLatGal}{\Cat{OMLatGal}}
\newcommand{\LSub}{\ensuremath{\textsl{LSub}}\xspace}
\newcommand{\CLSub}{\ensuremath{\textsl{CLSub}}\xspace}
\newcommand{\Pred}{\ensuremath{\textsl{Pred}}\xspace}
\newcommand{\PPred}{\ensuremath{\Pred_{\Box}}\xspace}
\newcommand{\Stat}{\ensuremath{\textsl{Stat}}\xspace}
\newcommand{\SStat}{\ensuremath{\textsl{SStat}}\xspace}
\newcommand{\Par}{\ensuremath{\textsl{Par}}\xspace}
\newcommand{\Tot}{\ensuremath{\textsl{Tot}}\xspace}
\newcommand{\Causal}{\ensuremath{\textsl{Caus}}\xspace}
\newcommand{\End}{\ensuremath{\textsl{End}}\xspace}
\newcommand{\sEnd}{\ensuremath{\End_{\leq\idmap}}\xspace}
\newcommand{\CP}{\ensuremath{\textsl{CP}}\xspace}
\newcommand{\Cont}{\ensuremath{\textsl{C}}\xspace}
\newcommand{\ShaPred}{\ensuremath{\textsl{ShaPred}}\xspace}

\newcommand{\scalar}{\mathrel{\bullet}}

\newcommand{\B}{\mathcal{B}}

\newsavebox\sbground
\savebox\sbground{\begin{tikzpicture}[circuit ee IEC,yscale=0.5,xscale=0.4]
                \draw (0,-2ex) to (0,0) node[ground,rotate=90,xshift=.65ex] {};
                \end{tikzpicture}}
\newcommand\ground{\mathbin{\text{\raisebox{-0.2ex}{\usebox\sbground}}}}

\newcommand{\Kl}{\mathcal{K}{\kern-.2ex}\ell}
\newcommand{\EM}{\mathcal{E}{\kern-.2ex}\mathcal{M}}
\newcommand{\DM}{\ensuremath{\mathcal{D}{\kern-.85ex}\mathcal{M}}}
\newcommand{\QEDbox}{\square} 
\newcommand{\QED}{\hspace*{\fill}$\QEDbox$}

\newcommand{\conglongrightarrow}{\mathrel{\smash{\stackrel{
           \raisebox{.5ex}{$\scriptstyle\cong$}}{
           \raisebox{0ex}[0ex][0ex]{$\longrightarrow$}}}}}

\newcommand{\spmat}[1]{\ensuremath{\smash{\left(%
\begin{smallmatrix}#1\end{smallmatrix}%
\right)}}}

\newcommand{\pullback}[1][dr]{\save*!/#1-1.2pc/#1:(-1,1)@^{|-}\restore}
\makeatother
 \newdir{ >}{{}*!/-7.5pt/@{>}}
 \newdir{|>}{!/4.5pt/@{|}*:(1,-.2)@^{>}*:(1,+.2)@_{>}}
 \newdir{ |>}{{}*!/-3pt/@{|}*!/-7.5pt/:(1,-.2)@^{>}*!/-7.5pt/:(1,+.2)@_{>}}
\newcommand{\xyline}[2][]{\ensuremath{\smash{\xymatrix@1#1{#2}}}}
\newcommand{\xyinline}[2][]{\ensuremath{\smash{\xymatrix@1#1{#2}}}^{\rule[8.5pt]{0pt}{0pt}}}

\newcommand{\IV}{\raisebox{-.2em}{$\xy
(0,2)*{\cdot};
(1,2)*{\cdot};
(2,2)*{\cdot};
(0,0)*{\cdot};
(1,0)*{\cdot};
{\ar@{-}(0,2); (0,0)};
{\ar@{-}(1,2); (1,0)};
{\ar@{-}(2,2); (1,0)};
\endxy$}}

\newcommand{\XI}{\raisebox{-.2em}{$\xy
(0,2)*{\cdot};
(1,2)*{\cdot};
(2,2)*{\cdot};
(0,0)*{\cdot};
(1,0)*{\cdot};
{\ar@{-}(0,2); (1,0)};
{\ar@{-}(1,2); (0,0)};
{\ar@{-}(2,2); (1,0)};
\endxy$}}

\makeindex{N} 
\makeindex{S} 

\begin{document}

\title{An Introduction to Effectus Theory}

\author{Kenta Cho, Bart Jacobs, Bas Westerbaan, Abraham Westerbaan \\[1em]
\small Institute for Computing and Information Sciences (iCIS), \\[-.5em]
\small Radboud University Nijmegen, The Netherlands. \\[-.5em]
\small Email: \texttt{\{K.Cho,bart,bwesterb,awesterb\}@cs.ru.nl}}

\date{\small\today}

\maketitle

\begin{abstract}
Effectus theory is a new branch of categorical logic that aims to
capture the essentials of quantum logic, with probabilistic and
Boolean logic as special cases.  Predicates in effectus theory are not
subobjects having a Heyting algebra structure, like in topos theory,
but `characteristic' functions, forming effect algebras. Such effect
algebras are algebraic models of quantitative logic, in which double
negation holds.  Effects in quantum theory and fuzzy predicates in
probability theory form examples of effect algebras.

This text is an account of the basics of effectus theory. It includes
the fundamental duality between states and effects, with the
associated Born rule for validity of an effect (predicate) in a
particular state.  A basic result says that effectuses can be
described equivalently in both `total' and `partial' form. So-called
`commutative' and `Boolean' effectuses are distinguished, for
probabilistic and classical models. It is shown how these Boolean
effectuses are essentially extensive categories. A large part of the
theory is devoted to the logical notions of comprehension and
quotient, which are described abstractly as right adjoint to truth,
and as left adjoint to falisity, respectively. It is illustrated how
comprehension and quotients are closely related to measurement.  The
paper closes with a section on `non-commutative' effectus theory,
where the appropriate formalisation is not entirely clear yet.
\end{abstract}

\clearpage

\tableofcontents

\clearpage
\section{Introduction}\label{sec:intro}

Effectus theory is a new branch of categorical logic that started
with~\cite{Jacobs15a}. The theory is already used in several other
publications, such
as~\cite{JacobsWW15,Cho15,ChoJWW15,Adams14,AdamsJ15}, and is being
actively developed. Therefore we think that it is appropriate to
provide a systematic, accessible introduction to this new area.

The aim of effectus theory is to axiomatise the categorical
characteristics of quantum computation and logic. This axiomatisation
should include probabilistic and classical computation and logic as
special cases.  The relevant logical structure is described in terms
of effect algebras. Such algebras have been introduced in theoretical
physics to describe generalised (quantum) probability theories,
see~\cite{FoulisB94}, and also~\cite{ChovanecK95,GiuntiniG94}. Effect
algebras generalise both classical logic, in the form of Boolean
algebras, and quantitative logic, with the unit interval $[0,1]$ of
real numbers as main example. The double negation law holds in such
effect algebras. One of the basic insights of~\cite{Jacobs15a} is that
certain elementary categorical structure --- involving coproducts and
pullbacks --- provides predicates with effect algebra structure.

In the usual set-theoretic world one can define predicates on a set
$X$ as subsets $P\subseteq X$ or equivalently as characteristic
functions $p \colon X \rightarrow 2$, where $2 = 1+1 = \{0,1\}$ is the
two-element set. The `subset' or `spatial' approach is the basis for
much of categorical logic where predicates are described as subobjects
$P\rightarrowtail X$. In topos theory, there is a bijective
correspondence between such subobjects of $X$ and characteristic maps
$X \rightarrow \Omega$ to a distinguished `classifier' object
$\Omega$. These predicates on an object $X$ in a topos form a Heyting
algebra, a basic formalisation of intuitionistic logic, in which the
double negation law fails. Topos theory has developed into a rich area
of mathematics, combining geometry and logic, see for instance the
handbooks~\cite{MaclaneM92} or~\cite{Johnstone02}.

Effectus theory breaks with this spatial approach and uses maps of the
form $X \rightarrow 1+1$ as predicates. Negation then corresponds to
swapping of outcomes in $1+1$, so that the double negation law
immediately holds. As mentioned, predicates in an effectus are not
Heyting algebras, but effect algebras: their logic is completely
different.

The current text provides an introduction to this new area. It starts
with some basic results from the original source~\cite{Jacobs15a}.
But it differs from~\cite{Jacobs15a} in that it combines the total and
partial perspectives on effectuses right from the beginning, working
towards the equivalence of these two approaches, as developed
in~\cite{Cho15}. Similarly, this article provides a systematic
description of quotients and comprehension in effectuses, which are
neatly described as a chain of four adjunctions in a row, which
quotients as left adjoint to falsity, and comprehension as right
adjoint to truth.  This extends~\cite{ChoJWW15} where quotients and
comprehension are described only concretely, in specific examples.

This paper contains several definitions and results that have not been
published before, such as:
\begin{itemize}
\item the construction of effectuses from `grounded' biproduct
  categories in Section~\ref{sec:biprod};

\item the definition of commutative (probabilistic) effectuses and
  Boolean (classical) effectuses in Section~\ref{sec:commbool}, giving
  a picture:
$$\xymatrix@C-1.5pc{
\ovalbox{\begin{minipage}{5em}\begin{center} Boolean \\ effectuses
   \end{center}\end{minipage}}
& \subseteq &
\ovalbox{\begin{minipage}{6em}\begin{center} commutative \\ effectuses 
   \end{center}\end{minipage}}
& \subseteq &
\ovalbox{\begin{minipage}{9em}\begin{center} (non-commutative) \\ effectuses
   \end{center}\end{minipage}}
}$$

\noindent The adjective `non-commutative' means `arbitrary'; it is
only used to distinguish the general case from the two special
subcases.

\item the equivalence between a Boolean effectus with comprehension
  and the well-known notion of extensive category; this gives many
  examples of (Boolean) effectuses, including categories of (compact
  Hausdorff) topological spaces, and of measurable spaces; also, the
  opposite $\op{\CRng}$ of the category of commutative rings is an
  example, see~\cite{ChoJWW15,Jacobs15a};

\item the first steps in non-commutative --- properly quantum ---
  effectus theory.
\end{itemize}

A recurring topic in effectus theory is the association of an `assert'
partial map $\asrt_{p} \colon X \rightarrow X+1$ with a predicate
$p\colon X \rightarrow 1+1$. This assert map is the `action' of the
predicate and is typical for a quantum setting, where an observation,
determined by a predicate, may have a side-effect and disturb the
system under observation. In a commutative or Boolean effectus one
does have such action maps $\asrt_p$, but they are side-effect
free. This is formalised by saying that these assert maps are below
the identity, in a suitable partial order on partial maps. Indeed, the
identity map does nothing, so maps below the identity also do not
change the state. This predicate-action correspondence is one of the
main topics.

Four running examples are used throughout the text, namely:
\begin{enumerate}
\item the category $\Sets$ of sets and functions, which is a Boolean
effectus;

\item the Kleisli category $\Kl(\Dst)$ of the distribution monad
  $\Dst$ on $\Sets$, which is a commutative effectus, modeling
  discrete probabilistic computation;

\item the opposite $\op{\OUG}$ of the category $\OUG$ of order unit
  groups, which is used as a relatively simple toy example to
  illustrate some basic notions;

\item the opposite $\op{\vNA}$ of the category $\vNA$ of von Neumann
  algebras, which is the prime example of a non-commutative effectus.
\end{enumerate}

\noindent The last two categories occur in opposite form because they
incorporate Heisenberg's view where quantum computations are predicate
transformers which operate in opposite direction, transforming
predicates on the post-state to predicates on the pre-state.

The insight that these diverse examples are all instance of the notion
of effectus has significant value. It allows us to grasp in these
examples what, for instance, the comprehension or quotient structure
of an effectus amounts to. These descriptions can be non-trivial and
hard to understand without the abstract categorical viewpoint. This
applies in particular to von Neumann algebras, where the `opposite'
interpretation can be confusing, and where notions like support
projection exist already for a long time
(see \textit{e.g.}~\cite[Dfn.~1.10.3]{Sakai71}),
but without their proper universal properties.

This text uses von Neumann algebras (or $W^*$-algebras) instead of
$C^*$-alge\-bras. The key aspect of von Neumann algebras that we use
is that that predicates/effects form a directed complete poset. This
allows us to define for an arbitrary predicate/effect $p$ the least
sharp predicate $\ceil{p}$ above $p$. This construction is crucial for
images, comprehension and for quotients. Additionally, (normal) states
separate predicates in von Neumann algebras, which is used in the
quotient construction. These two properties, directed completeness and
separation characterise von Neumann algebras among $C^*$-algebras,
following Kadison~\cite[Dfn.~1 and Thm.~1]{Kadison56} (see
also~\cite{Rennela14}).

Before starting, we would like to put effectus theory in broader
perspective.  The Oxford school in categorical quantum theory started
with~\cite{AbramskyC04}, see the recent books~\cite{HeunenV15}
and~\cite{CoeckeK16}. This approach takes the tensor $\otimes$ for
parallel composition as the main ingredient of quantum theory ---
following Schr\"odinger's emphasis on composite systems, and not von
Neumann's focus on measurement, see~\cite[\S\S2.2]{CoeckeL13}. In the
associated graphical calculus this tensor is represented simply via
paralell wires. This gives a powerful formalism in which many
protocols can be described. Coproducts (or biproducts) played a role
originally, but have become less prominent, probably because they do
not work well in a projective (\textit{i.e.}~`scalar-free')
setting~\cite{CoeckeD11} and make the graphical calculus more
complicated. The main examples for the Oxford school are so-called
dagger-compact closed monoidal categories, such as Hilbert spaces,
sets and relations, or finite-dimensional $C^*$-algebras and
completely positive maps.

In contrast, in effectus theory the coproduct $+$ is taken as
primitive constructor (and not the tensor $\otimes$), leading to a
strong emphasis on logic, validity and measurement, with the category
of von Neumann algebras as leading quantum example. A connection with
the Oxford school exists via the $\CP^*$-categories (with biproducts)
of~\cite{CoeckeHK14}. Section~\ref{sec:biprod} below shows that the
causal maps in a `grounded' biproduct category --- including
$\CP^*$-categories --- form an effectus.  Tensors can be added to an
effectus, see Section~\ref{sec:monoidal}, and lead to a much richer
theory, more closely connected to the Oxford school. But tensors are
not essential for the definition of an effectus theory, and for the
associated logic of effect algebras.

Other connections are emerging in recent work~\cite{Tull16} on
operational probabilistic theories~\cite{ChiribellaAP10} and
effectuses.  This makes it possible to transfer methods and results in
existing quantum theory and effectus theory back and forth.

One may ask: is effectus theory the new topos theory? It is definitely
too early to say. But what hopefully becomes clear from this
introduction is that effectus theory is a rich novel direction in
categorical logic that promises to capture essential aspects of the
quantum world and to provide a unifying framework that includes
the probabilistic and Boolean worlds as special cases.

This text gives a survey of first results in the theory of effectuses.
Some parts focus more on explaining ideas and examples, and some parts
serve as reference for basic facts --- with long series of results. We
briefly discuss the contents of the various sections. First, a
separate section is devoted to notation, where distinct notation is
introduced and explained for total and partial maps. Then,
Section~\ref{sec:effectus} introduces the central notion of effectus,
as a category with coproducts $(+, 0)$ and a final object $1$
satisfying certain elementary axioms. This same section describes
categorical consequences of these axioms, in particular about partial
projections and pairing, and introduces the leading examples of
effectuses.  Section~\ref{sec:partialpred} continues the exploration
and focuses on the partial monoid structure $(\ovee, 0)$ on partial
maps $X \rightarrow Y+1$ in an effectus. Predicates $X \rightarrow
1+1$ are a special case and have additional structure: it is shown
that they form effect algebras with scalar multiplication --- making
them effect modules. Subsequently, Section~\ref{sec:triangle} shows
that states $1 \rightarrow X$ of an object $X$ in an effectus form
convex sets, and that the basic adjunction between effect modules and
convex sets gives what is called a state-and-effect triangle of the
form:
$$\vcenter{\xymatrix{
\op{(\EMod_{M})}\ar@/^1.5ex/[rr]^-{\Hom(-,M)} & \top & 
   \Conv_{M}\ar@/^1.5ex/[ll]^-{\Hom(-,M)} \\
& \cat{B}\ar[ul]^{\Hom(-,1+1)=\Pred\quad}\ar[ur]_{\;\Stat=\Hom(1,-)} &
}}$$ 

\noindent where $\cat{B}$ is an effectus and $M = \Pred(1)$ is the
effect monoid of scalars in $\cat{B}$, see
Theorem~\ref{thm:effectusTriangle} for details. In this situation we
can describe the Born rule for validity $\omega\models p$ of a
predicate $p\colon X \rightarrow 1+1$ and a state $\omega \colon 1
\rightarrow X$ simply via composition, as scalar $p \after \omega
\colon 1\rightarrow 1+1$.

Section~\ref{sec:biprod} describes a new construction --- inspired by
$\CP^*$-categories~\cite{CoeckeHK14} --- to obtain an effectus as the
subcategory of `causal' maps in a biproduct category with special
ground maps $\ground$. Many examples of effectuses can be obtained
(and understood) in this way, where, for instance, causal maps
correspond to the unital ones. Subsequently, Section~\ref{sec:kerimg}
is more auxiliary: it introduces kernels and images in an effectus and
collects many basic results about them that will be used later on. A
basic result in the theory of effectuses, namely the equivalence of
the descriptions in terms of total and partial maps, is the topic of
Section~\ref{sec:partialtotal}, following~\cite{Cho15}. It plays a
central role in the theory, and once we have seen it, we freely switch
between the two descriptions, by using phrases like an `effectus in
total form', or an `effectus in partial form'.

Special axioms in an effectus are identified in
Section~\ref{sec:commbool} that capture probabilistic models and
Boolean models as special cases. Section~\ref{sec:monoidal} describes
tensors $\otimes$ in an effectus, for parallel composition. These
tensors come equipped with projections $X \leftarrow X\otimes Y
\rightarrow Y$, for weakening/discarding. But these tensors do
\emph{not} have copiers/diagonals $X \rightarrow X\otimes X$, since
contraction/duplication/cloning does not exist in the quantum world.
We show that if copiers exist, then the effectus becomes commutative.

The subsequent two Sections~\ref{sec:cmpr} and~\ref{sec:quot}
introduce the important notions of comprehension $\cmpr{X}{p}$ and
quotients $X/p$ for a predicate $p$ on an object $X$ in an
effectus. These notions are nicely captured categorically as right
adjoint to truth and as left adjoints to falisity, in a chain of
comprehesions:
$$\mbox{quotient } \dashv \mbox{ falsity } \dashv \mbox{ forget }
   \dashv \mbox{ truth } \dashv \mbox{ comprehension}$$

\noindent Such chains exist in all our leading examples, see
also~\cite{ChoJWW15}. The combination of comprehension and quotients
is studied in Section~\ref{sec:cmprquot}.

Inbetween, Section~\ref{sec:extcat} returns to Boolean effectuses and
gives a new characterisation: it shows that Boolean effectuses with
comprehension are essentially the same as extensive categories. The
latter are well-known kinds of categories in which finite coproducts
$(+,0)$ are well-behaved.

Our final Section~\ref{sec:noncomm} is of a different nature because
it describes some unclarities in the theory of effectuses which
require further research. What is unclear is the precise formalisation
of the general non-commutative case that captures the essentials of
quantum theory. This involves the intruiging relations between the
predicate-action correspondence for measurement on the one hand, and
quotients and comprehension on the other. We show that the `assert'
action maps are uniquely determined by postulates in von Neumann
algebras, but we cannot, at this stage, prove uniqueness at a general,
axiomatic level. More directions for future research are collected in
Section~\ref{sec:conclusion}.

\section{Notation}\label{sec:notation}

In the theory of effectuses, both total and partial maps play an
important role. When reasoning in an effectus one often switches
between total and partial maps. Since this may lead to confusion, we
address this topic explicitly in the current section, before
introducing the notion of effectus itself in the next section.

There are two equivalent ways of defining effectuses, see
Section~\ref{sec:partialtotal} for details.
\begin{enumerate}
\item The `total' approach starts from a category of total maps, and
  introduces partial maps as special, Kleisli maps for the lift monad
  $(-)+1$ in a larger enveloping Kleisli category. This is how
  effectuses were originally introduced in~\cite{Jacobs15a}.

\item The `partial' approach starts from a category of partial maps,
  and describes the total maps via a smaller `wide' subcategory with
  the same objects where the total maps are singled out via a special
  property. This approach comes from~\cite{Cho15}.
\end{enumerate}

\noindent In both cases we have a faithful functor between two
categories with the same objects:
$$\xymatrix@C+1pc{
\left({\begin{array}{c} \text{total} \\[-.2em] \text{maps} \end{array}}\right)
  \ar@{^(->}[r] &
\left({\begin{array}{c} \text{partial} \\[-.2em] \text{maps} \end{array}}\right)
}$$

\noindent These categories each have their own form of composition. In
order to disambiguate them we shall use different notation, for the
time being, namely $\tafter$\index{N}{$\tafter$, composition} and
$\pafter$.\index{N}{$\pafter$, composition in the Kleisli category of
  the lift monad} In fact, we shall also use different notation for
morphisms, namely $\tto$\index{N}{$\tto$, total map} and
$\pto$,\index{N}{$\pto$, partial map} and also for sums of maps,
namely $\tplus$ and $\pplus$. These notational differences are
relevant as long as we have not established the precise relationship
between total and partial maps, see
Theorem~\ref{thm:partialtotal}. Once we know the relationship, we
re-evaluate the situation in Discussion~\ref{dis:partialtotal}.

Only at that stage do we see clearly how to switch back-and-forth, and
can decide whether to take total or partial maps as first class
citizens.

\subsection*{Total maps}

Let $\cat{B}$ be a category with finite coproducts $(+, 0)$ and a
final object $1\in\cat{B}$. We consider the maps in $\cat{B}$ as
total, and we write the coprojections $\kappa_i$\index{N}{$\kappa_i$,
  $i$-th coprojection} and cotupling $[-,-]$\index{N}{$[f,g]$, cotuple
  of maps $f,g$} as maps:
$$\xymatrix@C-.7pc{
X_{1}\ar[r]|-{\tafter}^-{\kappa_1} & X_{1}+X_{2} & 
   X_{2}\ar[l]|-{\tafter}_-{\kappa_2} 
\qquad\mbox{and}\qquad
X_{1}+X_{2}\ar[rr]|-{\tafter}^-{[f_{1},f_{2}]} & & Y 
\quad\mbox{for}\quad 
X_{i}\ar[r]|-{\tafter}^-{f_i} & Y
}$$

\noindent The sum of two maps $f\colon X \tto A$, $g\colon Y \tto B$
is written as $f\tplus g = [\kappa_{1} \tafter f, \kappa_{2} \tafter
  g] \colon X+Y \tto A+B$.\index{N}{$\tplus$, sum of total maps} The
$n$-fold coproduct $X + \cdots + X$ of the same object $X$ is usually
written as $n\cdot X$, and called a copower. We use the notation
$\nabla = [\idmap, \idmap] \colon X+X \tto X$\index{N}{$\nabla =
  [\idmap,\idmap]$, codiagonal} for the codiagonal.

For an arbitrary object $X\in\cat{B}$ we write
$\bang_{X}$,\index{N}{$\bang$, unique map from an initial, or to a
  final, object} or simply $\bang$, both for the unique map from $0$
to $X$, and for the unique map from $X$ to $1$, as in:
$$\xymatrix{
0\ar|-{\tafter}[r]^-{\bang} & X
& &
X\ar|-{\tafter}[r]^-{\bang} & 1
}$$

\subsection*{Partial maps}

A partial map\index{S}{partial!-- map} is a map of the form $X
\rightarrow Y+1$. As is well-known, the assignment $X \mapsto X+1$
forms a monad on the category $\cat{B}$. It is often called the lift
or the maybe monad. The unit of this monad is the first coprojection
$\kappa_{1} \colon X \tto X+1$, and the multiplication is the cotuple
$[\idmap, \kappa_{2}] \colon (X+1)+1 \tto X+1$. We shall write
$\Par(\cat{B})$\index{N}{cat@$\Par(\cat{B})$ Kleisli category of the
  lift monad on $\cat{B}$} for the Kleisli category of this lift
monad.  Its objects are the objects of $\cat{B}$, and its maps from
$X$ to $Y$, written as $X \pto Y$, are morphisms $X \tto Y+1$ in
$\cat{B}$. Thus, (ordinary) maps in $\Par(\cat{B})$ are partial maps
in $\cat{B}$. Composition of $f\colon X \pto Y$ and $g\colon Y \pto Z$
in $\Par(\cat{B})$ is written as $g \pafter f \colon X \pto Z$, and is
described in $\Par(\cat{B})$ and in $\cat{B}$ as:
$$\xymatrix@C-.7pc{
g\pafter f = \Big(X\ar[r]|-{\pafter}^-{f} & 
   Y\ar[r]|-{\pafter}^-{g} & Z\Big) = 
\Big(X\ar[r]|-{\tafter}^-{f} & 
   Y+1\ar[rr]|-{\tafter}^-{[g,\kappa_{2}]} & & Z+1\Big).
}$$

\noindent Each map $h\colon X \tto Y$ in $\cat{B}$ gives rise to a map
$\klin{h} \colon X\pto Y$ in $\Par(\cat{B})$,\index{N}{$\klin{h} =
  \kappa_{1} \after h$, embedding into Kleisli category of the lift
  monad} namely:
$$\xymatrix{
\klin{h} = \Big(X\ar[r]|-{\tafter}^-{h} & Y\ar[r]|-{\tafter}^-{\kappa_{1}}
  & Y+1\Big).
}$$

\noindent We call such maps in $\Par(\cat{B})$ of the form $\klin{h}$
\emph{total},\index{S}{total map!-- in a category of partial maps} and
write them with this (single) guillemet notation $\klin{-}$. Later on,
in Lemma~\ref{lem:zero}, we shall see an alternative characterisation
of total maps, which will be much more useful.

The identity on $X$ in $\Par(\cat{B})$ is then the (total) map
$\klin{\idmap[X]} = \kappa_{1} \colon X \pto X$. There is an
identity-on-objects functor $\cat{B} \rightarrow \Par(\cat{B})$, which
we simply write as $\klin{-}$. It preserves composition, since:
$$\begin{array}{rcccccl}
\klin{k} \pafter \klin{h}
& = &
[\kappa_{1} \tafter k, \kappa_{2}] \tafter \kappa_{1} \tafter h
& = &
\kappa_{1} \tafter k \tafter h
& = &
\klin{k \tafter h}.
\end{array}$$

\noindent It is also not hard to see that:
$$\begin{array}{rclcrcl}
g \pafter \klin{h}
& = &
g \tafter h
& \qquad\mbox{and}\qquad &
\klin{k} \pafter f
& = &
(k\tplus\idmap) \tafter f.
\end{array}$$

The initial object $0\in\cat{B}$ is a \emph{zero object} in
$\Par(\cat{B})$, because there are unique maps $0 \pto X$ and $X\pto
0$ in $\Par(\cat{B})$, namely the maps $\bang\colon 0 \tto X+1$ and
$\bang\colon X \tto 0+1 \cong 1$ in $\cat{B}$. Hence, for each pair of
objects $X, Y$ there is a special \emph{zero
  map}\index{N}{$\zero$,!zero map} between them in $\Par(\cat{B})$,
defined as:
$$\xymatrix@C-.5pc{
{\Big(X}\ar[r]|-{\pafter}^-{\zero} & Y\Big) = 
   \Big(X\ar[r]|-{\tafter}^-{\bang} & 1\ar[r]|-{\tafter}^-{\kappa_{2}} & Y+1\Big).
}$$

\noindent Notice that $\zero \pafter g = \zero = f \pafter \zero$ for
all (partial) maps $f,g$. Notice that we write this zero map $\zero$
in bold face, in order to distinguish it from the zero object $0$. In
general, we shall write the top and bottom element of a poset as such
bold face $\one$ and $\zero$. Frequently, the zero map is indeed the
bottom element in a poset structure on homsets.

The presence of this zero map $\zero$ allows us to define kernel and
cokernel maps in $\Par(\cat{B})$, like in Abelian categories, namely
as (co)equaliser of a map $X \pto Y$ and the zero map $\zero\colon X
\pto Y$. This will appear later.

The coproducts $(+, 0)$ in $\cat{B}$ also form coproducts in
$\Par(\cat{B})$. The coprojections in $\Par(\cat{B})$ are:
$$\xymatrix@C+3pc{
X_{1}\ar[r]|-{\tafter}^-{\klin{\kappa_{1}} = \kappa_{1} \tafter \kappa_{1}} 
   & X_{1}+X_{2} &
   X_{2}\ar[l]|-{\tafter}_-{\klin{\kappa_{2}} = \kappa_{1} \tafter \kappa_{2}}
}$$

\noindent The cotuple $[f,g]$ in $\Par(\cat{B})$ is the same as in
$\cat{B}$. For two maps $f\colon X\pto A$ and $g\colon Y\pto B$ there
is thus a sum of maps $f \pplus g \colon X+Y\pto
A+B$\index{N}{$\pplus$ sum of partial maps} in $\Par(\cat{B})$, which is
defined in $\cat{B}$ as:
$$\xymatrix@C+2pc{
f \pplus g = \Big(X+Y\ar[rr]|-{\pafter}^-{[(\kappa_{1}\tplus \idmap) \tafter f,
   (\kappa_{2}\tplus \idmap) \tafter g]} & & (A+B)+1\Big).
}$$

\noindent It is not hard to see that:
$$\begin{array}{rclcrcl}
\big[\klin{f}, \klin{g}\big]
& = &
\klin{[f,g]}
& \qquad\mbox{and}\qquad &
\klin{h} \pplus \klin{k}
& = &
\klin{h\tplus k}.
\end{array}$$

\auxproof{
$$\begin{array}{rcl}
\klin{h} \pplus \klin{k}
& = &
[(\kappa_{1}+\idmap) \after \kappa_{1} \after h
   (\kappa_{2}+\idmap) \after \kappa_{1} \after k] \\
& = &
[\kappa_{1} \after \kappa_{1} \after h, 
   \kappa_{1} \after \kappa_{2} \after k] \\
& = &
\kappa_{1} \after [\kappa_{1} \after h, \kappa_{2} \after k] \\
& = &
\kappa_{1} \after (h+k) \\
& = &
\klin{h+k}.
\end{array}$$
}

The coproduct $X_{1}+X_{2}$ in $\Par(\cat{B})$ comes with `partial
projections'\index{S}{partial!-- projection} $\rhd_{i} \colon X_{1}+X_{2} \pto
X_{i}$,\index{N}{$\rhd_{i}$, $i$-th partial projection} described in
$\cat{B}$ as:
\begin{equation}
\label{diag:partprojtot}
\xymatrix{
X_{1}+1 & & X_{1}+X_{2}\ar[ll]|-{\tafter}_-{\rhd_{1} = \idmap\tplus \bang}
   \ar[rr]|-{\tafter}^-{\rhd_{2} = [\kappa_{2}\tafter\bang, \kappa_{1}]} & & X_{2}+1
}
\end{equation}

\noindent Equivalently they can be described in $\Par(\cat{B})$ via
the zero map $\zero$ as cotuples:
\begin{equation}
\label{diag:partprojpart}
\xymatrix{
X_{1} & & X_{1}+X_{2}\ar[ll]|-{\pafter}_-{\rhd_{1} = [\idmap,\zero]}
   \ar[rr]|-{\pafter}^-{\rhd_{2} = [\zero,\idmap]} & & X_{2}
}
\end{equation}

\noindent These partial projections are natural in $\Par(\cat{B})$, in
the sense that:
\begin{equation}
\label{eqn:partprojnat}
\begin{array}{rcl}
\rhd_{i} \pafter (f_{1}\pplus f_{2})
& = &
f_{i} \pafter \rhd_{i}.
\end{array}
\end{equation}

\noindent It is easy to define these projections $\rhd_{i} \colon
X_{1} + \cdots + X_{n} \pto X_{i}$ for $n$-ary coproducts.

The coproduct $+$ in $\Par(\cat{B})$ forms what is sometimes called a
`split product' or `butterfly' product (see
\textit{e.g.}~\cite[2.1.7]{Grandis12}): the following triangles
commute in $\Par(\cat{B})$, where the two diagonals are zero maps:
\begin{equation}
\label{diag:butterfly}
\vcenter{\xymatrix@R-1pc@C+1pc{
X\ar@{=}[dd]\ar[rd]|-{\pafter}^-{\klin{\kappa_{1}}} & & 
   Y\ar[ld]|-{\pafter}_{\klin{\kappa_{2}}}\ar@{=}[dd]
\\
& X+Y\ar[ld]|-{\pafter}^{\rhd_1}\ar[rd]|-{\pafter}_{\rhd_2} &
\\
X & & Y
}}
\end{equation}

\auxproof{
$$\begin{array}{rcl}
\rhd_{1} \pafter (f_{1}+f_{2})
& = &
[\idmap+\bang, \kappa_{2}] \after [(\kappa_{1}+\idmap) \after f_{1}, 
   (\kappa_{2}+\idmap) \after f_{2}] \\
& = &
[[\kappa_{1}, \kappa_{2}] \after f_{1}, 
   [\kappa_{2}\after\bang, \kappa_{2}] \after f_{2}] \\
& = &
[f_{1}, \kappa_{2} \after [\bang,\idmap] \after f_{2}] \\
& = &
[f_{1}, \kappa_{2} \after \bang] \\
& = &
[f_{1}, \kappa_{2}] \after (\idmap+\bang) \\
& = &
f_{1} \pafter \rhd_{1}
\\
\rhd_{2} \pafter (f_{1}+f_{2})
& = &
[[\kappa_{2}\after\bang, \kappa_{1}],\kappa_{2}] \after 
   [(\kappa_{1}+\idmap) \after f_{1}, (\kappa_{2}+\idmap) \after f_{2}] \\
& = &
[[\kappa_{2}\after\bang, \kappa_{2}] \after f_{1}, 
   [\kappa_{1}, \kappa_{2}] \after f_{2}] \\
& = &
[\kappa_{2} \after \bang, f_{2}] \\
& = &
[f_{2}, \kappa_{2}] \after [\kappa_{2}\after\bang, \kappa_{1}] \\
& = &
f_{2} \pafter \rhd_{2}
\\
J(g_{1})+J(g_{2})
& = &
[(\kappa_{1}+\idmap) \after \kappa_{1} \after g_{1}, 
   (\kappa_{2}+\idmap) \after \kappa_{1} \after g_{2}] \\
& = &
[\kappa_{1} \after \kappa_{1} \after g_{1},
   \kappa_{1} \after \kappa_{2} \after g_{2}] \\
& = &
\kappa_{1} \after \kappa_{1} \after g_{1}, \kappa_{2} \after g_{2}] \\
& = &
\kappa_{1} \after (g_{1}+g_{2}) \\
& = &
J(g_{1}+g_{2}).
\end{array}$$

We also check the butterfly equations:
$$\begin{array}{rcl}
\rhd_{1} \pafter J(\kappa_{1})
& = &
(\idmap+\bang) \after \kappa_{1} \\
& = &
\kappa_{1} \\
\rhd_{1} \pafter J(\kappa_{2})
& = &
(\idmap+\bang) \after \kappa_{2} \\
& = &
\kappa_{2} \after \bang \\
& = &
0 \\
\rhd_{2} \pafter J(\kappa_{1})
& = &
[\kappa_{2} \after \bang, \kappa_{1}] \after \kappa_{1} \\
& = &
\kappa_{2} \after \bang \\
& = &
0 \\
\rhd_{2} \pafter J(\kappa_{2})
& = &
[\kappa_{2} \after \bang, \kappa_{1}] \after \kappa_{2} \\
& = &
\kappa_{1}.
\end{array}$$
}

\noindent In particular, the coprojections $\klin{\kappa_i}$ are split
monos, and the projections $\rhd_{i}$ are split epis in
$\Par(\cat{B})$.

In this diagram~\eqref{diag:butterfly} we write $\klin{\kappa_1}
\colon X \pto X+Y$ for the coprojection in a category of partial maps
$\Par(\cat{B})$. Similarly we have written $\klin{\idmap}$ for the
identity in $\Par(\cat{B})$. From now on we shall omit these
guillemets $\klin{-}$ for coprojections and identities when the
context is clear. Thus we simply write $\kappa_{1} \colon X \pto X+Y$
and $\idmap \colon X \pto X$ for a coprojection or identity map in
$\Par(\cat{B})$.

\subsection*{Predicates and predicate transformers}

In the context of effectus theory a \emph{predicate} on an object $X$
is a total map of the form $X \tto 1+1$, or equivalently, a partial
map $X \pto 1$. The predicates $\one$ for true and $\zero$ for false
are given by:
$$\xymatrix@C-.7pc{ \one = \Big(X\ar[r]|-{\tafter}^-{\bang} &
  1\ar[r]|-{\tafter}^-{\kappa_1} & 1+1\Big) & & \zero =
  \Big(X\ar[r]|-{\tafter}^-{\bang} & 1\ar[r]|-{\tafter}^-{\kappa_2} &
  1+1\Big).  }$$\index{N}{$\one$,!truth
  predicate}\index{N}{$\zero$,!falsity predicate}

\noindent The zero predicate on $X$ is thus the zero map $X \pto 1$.
The negation, called \emph{orthosupplement} in this setting, of a
predicate $p\colon X \tto 1+1$ is obtained by swapping the outcomes:
\begin{equation}
\label{diag:orthosupplement}
\vcenter{\xymatrix@C-.5pc{
p^{\bot} = \Big(X\ar[r]|-{\tafter}^-{p} & 
   1+1\ar[rr]|-{\tafter}^-{[\kappa_{2},\kappa_{1}]}_-{\cong}
   & & 1+1\Big).
}}\index{N}{$p^{\bot}$, orthosupplement of predicate $p$}
\end{equation}

\noindent Clearly, $p^{\bot\bot} = p$ and $\one^{\bot} = \zero$. One
may expect that these predicates $X \tto 1+1$ form a Boolean algebra,
but this is not true in general.  When $\cat{B}$ is an effectus, the
predicates on $X$ form an \emph{effect algebra}. This is one of the
fundamental results of effectus theory, see
Section~\ref{sec:partialpred} for details.  The double negation law
$p^{\bot\bot} = p$ always holds in an effect algebra.

We use two kinds of re-indexing\index{S}{re-indexing} or
substitution\index{S}{substitution} operations for a map, each
transferring predicates on the codomain to predicates on the
domain. For \emph{total} maps $f\colon Y \tto X$ and \emph{partial}
maps $g \colon Y \pto X$ we write:
\begin{equation}
\label{diag:subst}
\vcenter{\xymatrix@C-.9pc{
\tbox{f}(p) = \Big(Y\ar[r]|-{\tafter}^-{f} & X\ar[r]|-{\tafter}^-{p} & 1+1\Big)
& 
\pbox{g}(p) = \Big(Y\ar[r]|-{\tafter}^-{g} & 
   X+1\ar[rr]|-{\tafter}^-{[p, \kappa_{1}]} & & 1+1\Big).
}}\index{N}{subst@$\tbox{f}$, substitution predicate transformer for 
total map $f$}\index{N}{subst@$\pbox{g}$, modal predicate
    transformer for partial map $g$}
\end{equation}

\noindent The idea is that $\tbox{f}(p)$ is true on an input to $f$
iff $p$ is true on its output. Similarly, $\pbox{g}(p)$ is true on an
input to $g$ iff either $g$ does not terminate, or $g$ terminates and
$p$ is true on its output (\textit{i.e.}~if $g$ terminates then $p$ is
true on the output). The `box' notation $\pbox{g}$ suggests this modal
reading, see also the De Morgan dual $\pdiam{g}$ in
point~\eqref{exc:subst:partial} below.

The next result collects some basic observations. The proofs
are easy and left to the reader.

\begin{exercise}
\label{exc:subst}
The two substitution (or modal) operations $\tbox{(-)}$ and
$\pbox{(-)}$ for total and partial maps satisfy:
\begin{enumerate}
\item \label{exc:subst:totalfun} $\tbox{\idmap}(p) = p$ and
  $\tbox{(f_{2} \tafter f_{1})}(p) =
  \tbox{f_{1}}\big(\tbox{f_{2}}(p)\big)$;

\item \label{exc:subst:partialfun} $\pbox{\klin{\idmap}}(p) = p$
  and $\pbox{(g_{2} \pafter g_{1})}(p) =
  \pbox{g_{1}}\big(\pbox{g_{2}}(p)\big)$;

\item \label{exc:subst:totalpartial} $\pbox{\klin{f}}(p) = \tbox{f}(p)$;

\item \label{exc:subst:partial} $\pbox{g}(p) = (p^{\bot} \pafter
  g)^{\bot} = \pdiam{g}(p^{\bot})^{\bot}$, if we define $\pdiam{g}(q)
  = q \pafter f$;\index{N}{subst@$\pdiam{g}$, modal predicate
    transformer for partial map $g$}

\item \label{exc:subst:totalpres} $\tbox{f}(p^{\bot}) =
  \tbox{f}(p)^{\bot}$ and $\tbox{f}(\one) = \one$, and thus also
  $\tbox{f}(\zero) = \zero$;

\item \label{exc:subst:partialtruth} $\pbox{g}(\one) = \one$. \QED
\end{enumerate}
\end{exercise}

\auxproof{
$$\begin{array}{rcl}
\pbox{(g_{2} \pafter g_{1})}(p)
& = &
[p, \kappa_{1}] \after [g_{2}, \kappa_{2}] \after g_{1} \\
& = &
[[p, \kappa_{1}] \after g_{2}, \kappa_{1}] \after g_{1} \\
& = &
[\pbox{g_{2}}(p), \kappa_{1}] \after g_{1} \\
& = &
\pbox{g_{1}}\big(\pbox{g_{2}}(p)\big) 
\\
\pbox{\klin{f}}(p)
& = &
[p, \kappa_{1}] \after \kappa_{1} \after f \\
& = &
p \after f \\
& = &
\tbox{f}(p)
\\
(p^{\bot} \pafter g)^{\bot}
& = &
[\kappa_{2},\kappa_{1}] \after 
   [[\kappa_{2},\kappa_{1}] \after p, \kappa_{2}] \after g \\
& = &
[p, \kappa_{1}] \after g \\
& = &
\pbox{g}(p)
\\
\pbox{g}(1)
& = &
[\kappa_{1} \after \bang, \kappa_{1}] \after g \\
& = &
\kappa_{1} \after [\bang, \idmap] \after g \\
& = &
\kappa_{1} \after \bang \after g \\
& = &
\kappa_{1} \after \bang \\
& = &
1.
\end{array}$$
}

\section{Effectuses}\label{sec:effectus}

In this section we give the definition of effectus, we provide several
leading examples, and we prove some basic categorical results about
effectuses. The definition below is slightly simpler than the original
one in~\cite{Jacobs15a}. We shall show that the requirements used
below are equivalent to the original ones, see
Lemma~\ref{lem:effectuspb}~\eqref{lem:effectuspb:tot}. There is no
clear intuition behind the nature of the axioms below: they ensure
that coproducts are well-behaved and provide the logical structure
that will be described in Section~\ref{sec:partialpred}.

\begin{definition}
\label{def:effectus}
An \emph{effectus}\index{S}{effectus} is a category $\cat{B}$ with
finite coproducts $(+, 0)$ and a final object $1$, such that both:
\begin{enumerate}
\item \label{def:effectus:pb} diagrams of the following form are
  pullbacks in $\cat{B}$:
\begin{equation}
\label{diag:effectuspb}
\vcenter{\xymatrix{
X+Y\ar[r]|-{\tafter}^-{\idmap\tplus \bang}\ar[d]|-{\tafter}_{\bang\tplus \idmap} & 
   X + 1\ar[d]|-{\tafter}^{\bang\tplus\idmap}
& &
X\ar[r]|-{\tafter}^-{\bang}\ar[d]|-{\tafter}_{\kappa_1} & 1\ar[d]|-{\tafter}^{\kappa_1}
\\
1+ Y\ar[r]|-{\tafter}_-{\idmap\tplus \bang} & 1 + 1
& &
X+Y\ar[r]|-{\tafter}_-{\bang\tplus \bang} & 1 + 1
}}
\end{equation}

\item \label{def:effectus:jm} the following two maps are jointly monic
  in $\cat{B}$:
\begin{equation}
\label{diag:effectusjointlymonic}
\vcenter{\xymatrix@C+2pc{
(1+1)+1\ar@/^1ex/[rr]|-{\tafter}^-{\IV = 
   [[\kappa_{1}, \kappa_{2}], \kappa_{2}]}
   \ar@/_1ex/[rr]|-{\tafter}_-{\XI = 
   [[\kappa_{2}, \kappa_{1}], \kappa_{2}]} & & 1+1
}}\index{N}{$\IV$}\index{N}{$\XI$}
\end{equation}

\noindent Joint monicity means that if $f,g$ satisfy $\IV \tafter f =
\IV \tafter g$ and $\XI \tafter f = \XI \tafter g$, then $f=g$. 
\end{enumerate}

A map of the form $X \tto 1+1$, or equivalently $X \pto 1$, will be
called a \emph{predicate}\index{S}{predicate!-- in an effectus} on
$X\in\cat{B}$. We write $\Pred(X)$\index{N}{$\Pred(X)$, collection of
  predicates on an object $X$} for the collection of predicates on
$X$.  Predicates of the special form $1 \tto 1+1$ are also called
\emph{scalars}.

A \emph{state}\index{S}{state!-- in an effectus} is a total map of the
form $1 \tto X$, and a \emph{substate}\index{S}{substate in an
  effectus} is a partial map $1 \pto X$. We write
$\Stat(X)$\index{N}{$\Stat(X)$, collection of states of an object $X$}
and $\SStat(X)$\index{N}{$\SStat(X)$, collection of substates of an
  object $X$} for the collections of such states and substates of $X$.
%
%
\end{definition}

The diagram on the left in~\eqref{diag:effectuspb} provides a form of
partial pairing for compatible partial maps, see
Lemma~\ref{lem:pairing} below; the diagram on the right ensures that
coprojections are disjoint, see Proposition~\ref{prop:effectuscoproj}
below. The joint monicity
requirement~\eqref{diag:effectusjointlymonic} is equivalent to joint
monicity of partial projections, see Lemma~\ref{lem:effectusjm} below.
It captures a form of cancellation.  The symbols $\IV$ and $\XI$
should suggest what these two maps do, as functions from a 3-element
set to a 2-element set, read downwards.

The diagram on the right in~\eqref{diag:effectuspb} only mentions the
first coprojection $\kappa_{1} \colon X \tto X+Y$. A corresponding
pullback exists for the second coprojection $\kappa_{2} \colon Y \tto
X+Y$ since it is related to $\kappa_1$ via the swap isomorphism $X+Y
\cong Y+X$. Similarly, the results below only mention the first
(co)projection, but hold for the second one as well.

An important property of effectuses is that predicates in an effectus
have the logical structure of an effect module. This will be
elaborated in the next few sections.

The next few results describe some categorical consequences of the
structure in an effectus. They focus especially on pullbacks. We
recall the following general result, known as the Pullback
Lemma:\index{S}{pullback lemma} consider a situation with two
commuting squares $A$ and $B$,
$$\xymatrix@R-2pc@C-2pc{
\cdot\ar[rr]\ar[dd] & & \cdot\ar[rr]\ar[dd] & & \cdot\ar[dd]
\\
& A & & B &
\\
\cdot\ar[rr] & & \cdot\ar[rr] & & \cdot
}$$

\noindent Then: if $B$ is a pullback, then $A$ is a pullback if
and only if the outer rectangle is a pullback.

\begin{lemma}
\label{lem:effectuspb}
Let $\cat{B}$ be a category with finite coproducts and a final object,
in which the squares~\eqref{diag:effectuspb} in the definition of an
effectus are pullbacks.
\begin{enumerate}
\item \label{lem:effectuspb:tot} Rectangles of the following forms are
  also pullbacks in $\cat{B}$.
\begin{equation}
\label{diag:effecutsderivedpb}
\vcenter{\xymatrix{
X+Y\ar[r]|-{\tafter}^-{\idmap\tplus g}\ar[d]|-{\tafter}_{f\tplus \idmap} & 
   X + B\ar[d]|-{\tafter}^{f\tplus\idmap}
& &
X\ar[r]|-{\tafter}^-{f}\ar[d]|-{\tafter}_{\kappa_1} & A\ar[d]|-{\tafter}^{\kappa_1}
\\
A + Y\ar[r]|-{\tafter}_-{\idmap\tplus g} & A + B
& &
X+Y\ar[r]|-{\tafter}_-{f\tplus g} & A + B
}}
\end{equation}

\item \label{lem:effectuspb:par} In the category $\Par(\cat{B})$ of
  partial maps the following squares are pullbacks, where $f\colon X
  \tto A$ and $g\colon Y \tto B$ are total maps.
\begin{equation}
\label{diag:effectuspartprojpb}
\vcenter{\xymatrix{
X+Y\ar[r]|-{\pafter}^-{\idmap\pplus \klin{g}}
   \ar[d]|-{\pafter}_{\klin{f}\pplus \idmap} & 
   X + B\ar[d]|-{\pafter}^{\klin{f}\pplus\idmap}
& &
X+Y\ar[r]|-{\pafter}^-{\klin{f}\pplus\idmap}\ar[d]|-{\pafter}_{\rhd_{1}} & 
   A+Y\ar[d]|-{\pafter}^{\rhd_{1}}
\\
A + Y\ar[r]|-{\pafter}_-{\idmap\pplus\klin{g}} & A + B
& &
X\ar[r]|-{\pafter}_-{\klin{f}} & A
}}
\end{equation}

\noindent Recall that we write $\rhd_{i}$ for the partial projection
maps from~\eqref{diag:partprojpart}.
\end{enumerate}
\end{lemma}

The rectangles in~\eqref{diag:effectuspb} used to define effectuses
are instances of the more general
formulations~\eqref{diag:effecutsderivedpb} provided in this lemma.

\begin{proof}
We start with the diagram on the left
in~\eqref{diag:effecutsderivedpb}.  Suppose we have $h\colon C \tto
A+Y$ and $k\colon C \tto X+B$ satisfying $(\idmap\tplus g) \tafter h =
(f\tplus\idmap) \tafter k$. Consider the following diagram.
$$\xymatrix{
C\ar@{..>}[dr]|-{\tafter}^-{\ell}
   \ar@/_2ex/[ddr]|-{\tafter}_{h}\ar@/^2ex/[drr]|-{\tafter}^{k}
   \ar@/_4ex/[dddr]|-{\tafter}_{(\bang\tplus\idmap)\tafter h}
   \ar@/^4ex/[drrr]|-{\tafter}^{(\idmap\tplus \bang)\tafter k}
\\
& X+Y\ar[r]|-{\tafter}^-{\idmap\tplus g}\ar[d]|-{\tafter}^{f\tplus \idmap} & 
   X + B\ar[d]|-{\tafter}^{f\tplus\idmap}\ar[r]|-{\tafter}^-{\idmap\tplus \bang} &
   X + 1\ar[d]|-{\tafter}^{f\tplus\idmap} &
\\
& A + Y\ar[r]|-{\tafter}_-{\idmap\tplus g}\ar[d]|-{\tafter}^{\bang\tplus \idmap} & 
   A + B\ar[r]|-{\tafter}_-{\idmap\tplus \bang}\ar[d]|-{\tafter}^{\bang\tplus\idmap} &
   A + 1\ar[d]|-{\tafter}^{\bang\tplus\idmap}
\\
& 1 + Y\ar[r]|-{\tafter}_-{\idmap\tplus g} & 
   1+B\ar[r]|-{\tafter}_-{\idmap\tplus \bang} & 1+1
}$$

\noindent The outer diagram commutes since:
$$\begin{array}{rcl}
(\idmap\tplus\bang) \tafter (\idmap\tplus g) \tafter (\bang\tplus\idmap) \tafter h
& = &
(\idmap\tplus\bang) \tafter (\bang\tplus\idmap) \tafter (\idmap\tplus g) \tafter h \\
& = &
(\idmap\tplus\bang) \tafter (\bang\tplus\idmap) \tafter (f\tplus\idmap) \tafter k \\
& = &
(\bang\tplus\idmap) \tafter (f\tplus\idmap) \tafter (\idmap\tplus\bang) \tafter k.
\end{array}$$

\noindent The outer rectangle, consisting of all four squares, is a
pullback since it is an instance of the pullback on the left
in~\eqref{diag:effectuspb}. Hence there is a unique map $\ell \colon C
\tto A+B$ with $(\bang\tplus\idmap) \tafter \ell = (\bang\tplus\idmap) \tafter
h$ and $(\idmap\tplus\bang) \tafter \ell = (\idmap\tplus\bang) \tafter k$.  We
obtain $(f\tplus\idmap) \tafter \ell = h$ by uniqueness of mediating
maps, since the lower two squares form a pullback, again
by~\eqref{diag:effectuspb}. Similarly $(\idmap\tplus g) \tafter \ell =
k$ holds because the two square on the right together form a pullback.

In a similar way we show that the diagram on the right
in~\eqref{diag:effecutsderivedpb} is a pullback: let $h\colon C \tto X+Y$
and $k\colon C \rightarrow A$ satisfy $(f\tplus g) \tafter h = \kappa_{1}
\tafter k$. Then we have a situation:
$$\xymatrix{
C\ar@{..>}[dr]|-{\tafter}^-{\ell}
   \ar@/_2ex/[ddr]|-{\tafter}_{h}\ar@/^2ex/[drr]|-{\tafter}^{k}
   \ar@/^4ex/[drrr]|-{\tafter}^{\bang}
\\
& X\ar[r]|-{\tafter}^-{f}\ar[d]|-{\tafter}^{\kappa_{1}} & 
   A\ar[d]|-{\tafter}^{\kappa_1}\ar[r]|-{\tafter}^-{\bang} &
   1\ar[d]|-{\tafter}^{\kappa_1} &
\\
& X + Y\ar[r]|-{\tafter}_-{f\tplus g} &
   A + B\ar[r]|-{\tafter}_-{\bang\tplus \bang} &
   1 + 1
}$$

\noindent Since the two squares together and the square on the right
form a pullback, as on the right in~\eqref{diag:effectuspb}. Hence we
obtain a unique map $\ell \colon C \tto X$ with $\kappa_{1} \tafter
\ell = h$. We get $f \tafter \ell = k$ by uniqueness of mediating
maps, using the pullback on the right. 

We turn to the diagram~\eqref{diag:effectuspartprojpb} in the category
$\Par(\cat{B})$. We write $\alpha$ for the standard associativity
isomorphism $U+(V+W)\tto (U+V)+W$ in $\cat{B}$, given explicitly by
$\alpha = [\kappa_{1} \tafter \kappa_{1}, \kappa_{2}\tplus\idmap]$ and
$\alpha^{-1} = [\idmap\tplus\kappa_{1}, \kappa_{2} \tafter
  \kappa_{2}]$. For total maps $f\colon X \tto A$ and $g\colon Y \tto
B$ we have commuting diagrams in $\cat{B}$:
$$\vcenter{\xymatrix@C-1pc{
X+(Y+1)\ar[d]|-{\tafter}_{f\tplus\idmap}\ar[r]|-{\tafter}^-{\alpha}_-{\cong} &
   (X+Y)+1\ar[d]|-{\tafter}^-{[\klin{f}\pplus\idmap, \kappa_{2}]}
\\
A+(Y+1)\ar[r]|-{\tafter}^-{\alpha}_-{\cong} & (A+Y)+1
}}
\qquad
\vcenter{\xymatrix@R-1pc@C+.5pc{
X+(Y+1)\ar[r]|-{\tafter}^-{\idmap\tplus (g\tplus\idmap)}
   \ar[d]|-{\tafter}_{\alpha}^{\cong} &
   (X+B)+1\ar[d]|-{\tafter}_{\alpha}^{\cong}
\\
(X+Y)+1\ar[r]|-{\tafter}_-{[\idmap\pplus\klin{g}, \kappa_{2}]} & (X+B)+1
}}$$

\auxproof{
\noindent This works since:
$$\begin{array}{rclcrcl}
[\klin{f}\pplus\idmap, \kappa_{2}]
& = &
(f\tplus\idmap)\tplus\idmap
& \qquad\mbox{and}\qquad &
[\idmap\pplus\klin{g}, \kappa_{2}]
& = &
(\idmap\tplus g)\tplus\idmap.
\end{array}$$

$$\begin{array}{rcl}
[\klin{f}\pplus\idmap, \kappa_{2}]
& = &
[[(\kappa_{1}\tplus\idmap) \tafter \kappa_{1} \tafter f,
   (\kappa_{2}\tplus\idmap) \tafter \kappa_{1}], \kappa_{2}] \\
& = &
[[\kappa_{1} \tafter \kappa_{1} \tafter f, \kappa_{1} \tafter \kappa_{2}], 
   \kappa_{2}] \\
& = &
[\kappa_{1} \tafter (f\tplus\idmap), \kappa_{2}] \\
& = &
(f\tplus\idmap)\tplus\idmap
\\
{[\idmap\pplus\klin{g}, \kappa_{2}]}
& = &
[[(\kappa_{1}\tplus\idmap) \after \kappa_{1}, 
   (\kappa_{2}\tplus\idmap) \tafter \kappa_{1} \tafter g], \kappa_{2}] \\
& = &
[[\kappa_{1} \tafter \kappa_{1}, \kappa_{1} \tafter \kappa_{2} \tafter g],
   \kappa_{2}] \\
& = &
[\kappa_{1} \tafter (\idmap\tplus g), \kappa_{2}] \\
& = &
(\idmap\tplus g)\tplus\idmap.
\end{array}$$
}

\noindent As a result, the rectangle on the left below is a pullback
in $\cat{B}$, since it is related via the associativity isomorphisms
$\alpha$ to the rectangle on the right, which is a pullback
by~\eqref{diag:effecutsderivedpb}.
$$\xymatrix@C+.5pc{
(X+Y)+1\ar[r]|-{\tafter}^-{[\idmap\pplus\klin{g}, \kappa_{2}]}
   \ar[d]|-{\tafter}_-{[\klin{f}\pplus\idmap, \kappa_{2}]} & 
   (X+B)+1\ar[d]|-{\tafter}^-{[\klin{f}\pplus\idmap, \kappa_{2}]}
& 
X+(Y+1)\ar[r]|-{\tafter}^-{\idmap\tplus (g\tplus\idmap)}
   \ar[d]|-{\tafter}_{f\tplus\idmap} &
   X+(B+1)\ar[d]|-{\tafter}^{f\tplus\idmap}
\\
(A+Y)+1\ar[r]|-{\tafter}^-{[\idmap\pplus\klin{g}, \kappa_{2}]} & (A+B)+1
& 
A+(Y+1)\ar[r]|-{\tafter}^-{\idmap\tplus (g\tplus\idmap)} &
   A+(B+1)
}$$

\noindent The fact that the above rectangle on the left is a pullback
in $\cat{B}$ implies that the rectangle on the left
in~\eqref{diag:effectuspartprojpb} is a pullback in $\Par(\cat{B})$.

We use basically the same trick for the diagram on the right
in~\eqref{diag:effectuspartprojpb}. Let $h\colon C \pto X$ and
$k\colon C \pto A+Y$ be (partial) maps with $\klin{g} \pafter h =
\rhd_{1} \pafter k$.  When translated to the category $\cat{B}$ the
latter equation becomes $(g\tplus\idmap) \tafter h = [\rhd_{1},
  \kappa_{2}] \tafter k$.
$$\xymatrix{
C\ar@/_2ex/[ddr]|-{\tafter}_-{h}\ar@{..>}[dr]|-{\tafter}^(0.5){\ell}
   \ar@/^2ex/[drr]|-{\tafter}^{\raisebox{.2em}{$\scriptstyle\alpha^{-1} \tafter k$}}
\\
& X+(Y+1)\ar[r]|-{\tafter}^-{f\tplus\idmap}
   \ar[d]|-{\tafter}_{\idmap\tplus \bang}\pullback & 
   A+(Y+1)\ar[d]|-{\tafter}^{\idmap\tplus \bang = [\rhd_{1},\kappa_{2}] \tafter \alpha} 
\\
& X + 1\ar[r]|-{\tafter}_-{f\tplus\idmap} & A+1
}$$

\auxproof{
\noindent The outer diagram commutes, since:
$$\begin{array}{rcl}
[\rhd_{1},\kappa_{2}] \tafter \alpha
& = &
[\idmap\tplus \bang,\kappa_{2}] \tafter 
   [\kappa_{1} \tafter \kappa_{1}, \kappa_{2}\tplus\idmap] \\
& = &
[\kappa_{1} , [\kappa_{2}\tafter \bang, \kappa_{2}]] \\
& = &
[\kappa_{1} , \kappa_{2} \tafter [\bang,\idmap]] \\
& = &
\idmap\tplus \bang.
\end{array}$$
}

\noindent The rectangle is a pullback, as instance of the square on
the left in~\eqref{diag:effecutsderivedpb}. Then $\ell' = \alpha
\tafter \ell \colon C \pto X+Y$ is the required mediating map in
$\Par(\cat{B})$. \QED

\auxproof{
Since:
$$\begin{array}{rcl}
\rhd_{1} \pafter \ell'
& = &
[\rhd_{1}, \kappa_{2}] \tafter \ell' \\
& = &
[\rhd_{1}, \kappa_{2}] \tafter \alpha \tafter \ell \\
& = &
[\idmap\tplus \bang, \kappa_{2}] \tafter 
   [\kappa_{1} \tafter \kappa_{1}, \kappa_{2}\tplus\idmap] \tafter \ell \\
& = &
[\kappa_{1}, [\kappa_{2}\tafter \bang, \kappa_{2}]] \tafter \ell \\
& = &
[\kappa_{1}, \kappa_{2}\tafter [\bang, \kappa_{2}]] \tafter \ell \\
& = &
[\kappa_{1}, \kappa_{2}\tafter \bang] \tafter \ell \\
& = &
(\idmap\tplus \bang) \tafter \ell \\
& = &
h
\\
\klin{g\tplus\idmap} \pafter \ell'
& = &
((g\tplus\idmap)\tplus\idmap) \tafter \ell' \\
& = &
((g\tplus\idmap)\tplus\idmap) \tafter \alpha \tafter \ell \\
& = &
\alpha \tafter (g\tplus\idmap) \tafter \ell \\
& = &
\alpha \tafter \alpha^{-1} \tafter k \\
& = &
k.
\end{array}$$
}
\end{proof}

\begin{proposition}
\label{prop:effectuscoproj}
In an effectus coprojections $\kappa_{i} \colon X_{i} \tto
X_{1}+X_{2}$ are monic and disjoint, and the initial object $0$ is
strict.\index{S}{strict initial object}
\end{proposition}

An easy consequence of this result is that for an effectus $\cat{B}$
the canonical functor $\klin{-} \colon \cat{B} \rightarrow
\Par(\cat{B})$ from $\cat{B}$ to the Kleisli category $\Par(\cat{B})$
of the lift monad $(-)+1$ is faithful. This is the case because
$\klin{f} = \kappa_{1} \tafter f$, and $\kappa_{1}$ is monic.

\begin{proof}
The first coprojection $\kappa_{1} \colon A \tto A+B$ is monic since,
using the pullback on the right in~\eqref{diag:effecutsderivedpb}, we
get a diagram as on the left below. Its two rectangles are pullbacks,
and so their combination too.
$$\vcenter
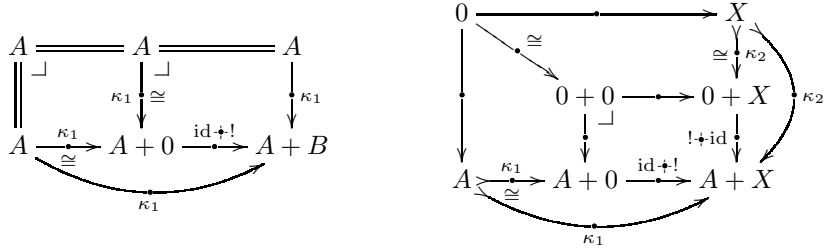
{\xymatrix{
A\pullback\ar@{=}[r]\ar@{=}[d] & 
   A\pullback\ar@{=}[r]\ar[d]|-{\tafter}_{\kappa_1}^{\cong} & 
   A\ar[d]|-{\tafter}^{\kappa_{1}}
\\
A\ar[r]|-{\tafter}^-{\kappa_1}_-{\cong}\ar@/_4ex/[rr]|-{\tafter}_-{\kappa_{1}} & 
   A+0\ar[r]|-{\tafter}^-{\idmap\tplus \bang} & A+B
}}
\qquad\qquad
\vcenter{\xymatrix@C+0pc@R-.5pc{
0\ar[dr]|-{\tafter}^{\cong}\ar[dd]|-{\tafter}\ar[rr]|-{\tafter} & & 
   X\ar@{ >->}[d]|-{\tafter}^{\kappa_{2}}_{\cong}
   \ar@{ >->}@/^5ex/[dd]|-{\tafter}^{\kappa_2} 
\\
&  0+0\pullback\ar[r]|-{\tafter}\ar[d]|-{\tafter} & 
   0+X\ar[d]|-{\tafter}_{\bang\tplus\idmap}
\\
A\ar@{ >->}[r]|-{\tafter}^-{\kappa_1}_-{\cong}
   \ar@{ >->}@/_4ex/[rr]|-{\tafter}_-{\kappa_{1}} & 
   A+0\ar[r]|-{\tafter}^-{\idmap\tplus\bang} & A+X
}}$$

\noindent The above diagram on the right shows that the intersection
(pullback) of $\kappa_{1}, \kappa_{2}$ is the initial object $0$: the
small rectangle is a pullback since it is an instance of the pullback
on the left in~\eqref{diag:effecutsderivedpb}. Hence the outer
rectangle is a pullback via the isomorphisms in the diagram.

Strictness of the initial object $0$ means that each map $f\colon X
\tto 0$ is an isomorphism. For such a map, we have to prove that the
composite $X \tto 0 \tto X$ is the identity. Consider the diagram:
$$\xymatrix{
X\ar@/^2ex/[drr]|-{\tafter}^{f}\ar@{ >->}@/_2ex/[ddr]|-{\tafter}_{\kappa_2}
   \ar@{..>}[dr]|-{\tafter}
 \\
& 0\pullback\ar@{=}[r]\ar@{ >->}[d]|-{\tafter}_-{\kappa_{1}} & 
   0\ar@{ >->}[d]|-{\tafter}^{\kappa_{1} = \kappa_{2}} 
\\
& 0+X\ar[r]|-{\tafter}_-{\idmap\tplus f}\ar[d]|-{\tafter}_{[\bang,\idmap]} & 0+0
\\
& X
}$$

\noindent The rectangle is a pullback, as on the right
in~\eqref{diag:effecutsderivedpb}. By initiality of $0$, we have
$\kappa_{1} = \kappa_{2} \colon 0 \tto 0+0$, so that the outer diagram
commutes. Then we get the dashed map, as indicated, which must be
$f\colon X \rightarrow 0$. But now we are done:
$$\begin{array}{rcccccl}
\bang \tafter f
& = &
[\bang, \idmap] \tafter \kappa_{1} \tafter f
& = &
[\bang, \idmap] \tafter \kappa_{2}
& = &
\idmap. 
\end{array}\eqno{\QEDbox}$$
\end{proof}

The next result shows that the joint monicity requirement in the
definition of an effectus has many equivalent formulations.

\begin{lemma}
\label{lem:effectusjm}
Let $\cat{B}$ be a category with finite coproducts and a final object,
in which the squares in~\eqref{diag:effectuspb} are pullbacks. Then
the following statements are equivalent.
\begin{enumerate}
\item \label{lem:effectusjm:twoOne} The category $\cat{B}$ is an
  effectus, that is, the two maps $\IV, \XI \colon (1+1)+1 \tto 1+1$
  in~\eqref{diag:effectusjointlymonic} are jointly monic in $\cat{B}$.

\item \label{lem:effectusjm:twoOneProj} The two partial projection
  maps $\rhd_{1}, \rhd_{2} \colon 1+1 \pto 1$ are jointly monic in
  $\Par(\cat{B})$.

\item \label{lem:effectusjm:twoX} For each object $X$, the two partial
  projection maps $\rhd_{1}, \rhd_{2} \colon X+X \pto X$ are jointly
  monic in $\Par(\cat{B})$.

\item \label{lem:effectusjm:two} For each pair of objects $X_{1},
  X_{2}$, the two projections $\rhd_{i} \colon X_{1}+X_{2} \pto X_{i}$
      are jointly monic in $\Par(\cat{B})$.

\item \label{lem:effectusjm:manyX} For each $n$-tuple of objects
  $X_{1}, \ldots, X_{n}$, with $n\geq 1$, the $n$ projections
  $\rhd_{i} \colon X_{1}+\cdots+X_{n} \pto X_{i}$ are jointly monic in
  $\Par(\cat{B})$.
\end{enumerate}
\end{lemma}

\begin{proof}
The implication $\eqref{lem:effectusjm:twoOne} \Rightarrow
\eqref{lem:effectusjm:twoOneProj}$ is a simple reformulation,
using that $\IV = [\rhd_{1}, \kappa_{2}]$ and $\XI = [\rhd_{2}, \kappa_{2}]$.

\auxproof{
$$\begin{array}{rcl}
[\rhd_{1}, \kappa_{2}]
& = &
[\idmap\pplus \bang, \kappa_{2}] \\
& = &
[\idmap\pplus\idmap, \kappa_{2}] \\
& = &
[\idmap, \kappa_{2}] \\
& = &
\IV
\\
{[\rhd_{2}, \kappa_{2}]}
& = &
[[\kappa_{2} \tafter \bang, \kappa_{1}], \kappa_{2}] \\
& = &
[[\kappa_{2}, \kappa_{1}], \kappa_{2}] \\
& = &
\XI.
\end{array}$$
}

For the implication $\eqref{lem:effectusjm:twoOneProj} \Rightarrow
\eqref{lem:effectusjm:twoX}$, let $f, g\colon Y \pto X+X$ satisfy
$\rhd_{i} \pafter f = \rhd_{i} \pafter g \colon Y \pto X$ for $i =
1,2$. Consider the following diagram in $\Par(\cat{B})$.
$$\xymatrix@C+1pc{
Y\ar@/^1ex/[dr]|-{\pafter}^-{f}\ar@/_1ex/[dr]|-{\pafter}_-{g}
\\
& X+X\pullback\ar[r]|-{\pafter}_-{\klin{\bang}\pplus\idmap}
   \ar[d]|-{\pafter}^{\idmap\pplus\klin{\bang}}
   \ar@/^4ex/[rr]|-{\pafter}^-{\rhd_2}\ar@/_5ex/[dd]|-{\pafter}_-{\rhd_1} & 
   1+X\pullback\ar[r]|-{\pafter}_-{\rhd_{2}}\ar[d]|-{\pafter}^{\idmap\pplus\klin{\bang}}
   & X\ar[d]|-{\pafter}^{\klin{\bang}}
\\
& X+1\pullback\ar[r]|-{\pafter}_-{\klin{\bang}\pplus\idmap}\ar[d]|-{\pafter}^{\rhd_{1}}
   & 1+1\ar[r]|-{\pafter}_-{\rhd_{2}}
   \ar[d]|-{\pafter}^{\rhd_{1}} & 1
\\
& X\ar[r]|-{\pafter}_-{\klin{\bang}} & 1
}$$

\noindent All rectangles are pullbacks in $\Par(\cat{B})$ by
Lemma~\ref{lem:effectuspb}~\eqref{lem:effectuspb:par}. The definitions
$f' = (\klin{\bang}\pplus\klin{\bang}) \pafter f$ and $g' =
(\klin{\bang}\pplus\klin{\bang})  \pafter g$ yield equal maps $Y \pto 1+1$ by
point~\eqref{lem:effectusjm:twoOneProj}, since $\rhd_{i} \pafter f' =
\klin{\bang} \pafter \rhd_{i} \pafter f = \klin{\bang} \pafter \rhd_{i}
\pafter g = \rhd_{i} \pafter g'$ by~\eqref{eqn:partprojnat}. But then:
$$\left\{\begin{array}{rcll}
(\klin{\bang}\pplus\idmap) \pafter f
& = &
(\klin{\bang}\pplus\idmap) \pafter g \qquad &
   \mbox{by the upper right pullback} \\
(\idmap\pplus\klin{\bang}) \pafter f
& = &
(\idmap\pplus\klin{\bang}) \pafter g &
   \mbox{by the lower left pullback.}
\end{array}\right.$$

\noindent Hence $f=g$ by the upper left pullback.

For the implication $\eqref{lem:effectusjm:twoX} \Rightarrow
\eqref{lem:effectusjm:two}$ let $f,g \colon Z \pto X_{1}+X_{2}$
satisfy $\rhd_{i} \pafter f = \rhd_{i} \pafter g$. The trick is to
consider $(\kappa_{1}\pplus\kappa_{2}) \pafter f,
(\kappa_{1}\pplus\kappa_{2}) \pafter g \colon Z \pto
(X_{1}+X_{2})+(X_{1}+X_{2})$ instead. They satisfy, by
naturality~\eqref{eqn:partprojnat} of the partial projections:
$$\begin{array}{rcl}
\rhd_{i} \pafter (\kappa_{1}\pplus\kappa_{2}) \pafter f
& = &
\kappa_{i} \pafter \rhd_{i} \pafter f \\
& = &
\kappa_{i} \pafter \rhd_{i} \pafter g \\
& = &
\rhd_{i} \pafter (\kappa_{1}\pplus\kappa_{2}) \pafter g.
\end{array}$$

\noindent We obtain $(\kappa_{1}\pplus\kappa_{2}) \pafter f
= (\kappa_{1}\pplus\kappa_{2}) \pafter g$ from
point~\eqref{lem:effectusjm:twoX}. But then we are done:
$$\begin{array}{rcccccl}
f
& = &
\nabla \pafter (\kappa_{1}\pplus\kappa_{2}) \pafter f
& = &
\nabla \pafter (\kappa_{1}\pplus\kappa_{2}) \pafter g
& = &
g.
\end{array}$$

\auxproof{
When we reduce the equality that we use here to $\cat{B}$ we get,
as expected:
$$\begin{array}{rcl}
\nabla \pafter (\kappa_{1}\pplus\kappa_{2})
& = &
[[\klin{\idmap}, \klin{\idmap}], \kappa_{2}] \tafter
   [(\kappa_{1}+\idmap) \tafter \kappa_{1}, 
    (\kappa_{2}+\idmap) \tafter \kappa_{2}] \\
& = &
[[\kappa_{1}, \kappa_{1}], \kappa_{2}] \tafter
   [\kappa_{1} \tafter \kappa_{1} \tafter \kappa_{1},
    \kappa_{1} \tafter \kappa_{2} \tafter \kappa_{2}] \\
& = &
[\kappa_{1} \tafter \kappa_{1}, \kappa_{1} \tafter \kappa_{2}] \\
& = &
\kappa_{1} \\
& = &
\klin{\idmap}.
\end{array}$$
}

The implication $\eqref{lem:effectusjm:two} \Rightarrow
\eqref{lem:effectusjm:manyX}$ is obtained via induction. Finally, the
implication $\eqref{lem:effectusjm:manyX} \Rightarrow
\eqref{lem:effectusjm:twoOne}$ is trivial. \QED

\auxproof{
The implication $\eqref{lem:effectusjm:twoX} \Rightarrow
\eqref{lem:effectusjm:manyX}$ is done by induction on $n$. The case $n =
1$, where $[\rhd_{1}, \kappa_{2}] = \idmap = [\rhd_{2}, \kappa_{2}]
\colon X+1 \rightarrow X+1$, holds, and the case $n=2$ holds by
induction. 

Now let $f, g \colon Y \rightarrow (n+1)\cdot X + 1 = (n\cdot X + X) +
1$ satisfy $[\rhd_{i}, \kappa_{2}] \after f = [\rhd_{i}, \kappa_{2}]
\after g$, for $i\leq n+1$. Take $f' = [\rhd_{1},\kappa_{2}] \after
f\colon Y \rightarrow (n+1)\cdot X +1 \rightarrow n\cdot X + 1$.  and
similarly, $g' = [\rhd_{1},\kappa_{2}] \after g$. Then for each $i\leq
n$, $[\rhd_{i},\kappa_{2}] \after f' = [\rhd_{i},\kappa_{2}] \after
g'$. Hence $f' = g'$. But since $[\rhd_{2},\kappa_{2}] \after f =
[\rhd_{n+1},\kappa_{2}] \after f = [\rhd_{n+1},\kappa_{2}] \after g =
[\rhd_{2},\kappa_{2}] \after g$, we get $f = g$
by~\eqref{lem:effectusjm:twoX}.
}
\end{proof}

In an effectus these partial projections $\rhd_i$ are projections for
a `partial pairing' operation $\dtuple{-,-}$ that will be described
next.  In fact one can understand the pullback on the left
in~\eqref{diag:effectuspb} in the definition of effectus as precisely
providing such a pairing for maps which are suitably orthogonal to
each other. Later on, in Section~\ref{sec:kerimg} we can make this
requirement more precise in terms of kernels that are each other's
orthosupplement.

\begin{lemma}
\label{lem:pairing}
In an effectus, each pair of partial maps $f\colon Z \pto X$ and
$g\colon Z \pto Y$ with $\one \pafter f = (\one \pafter g)^{\bot}$,
determines a unique total map $\dtuple{f,g} \colon Z \tto
X+Y$\index{N}{$\dtuple{f,g}$, partial pairing of $f,g$}
\index{S}{partial!-- pairing} with:
$$\begin{array}{rclcrcl}
\rhd_{1} \tafter \dtuple{f,g}
& = &
f
& \qquad\mbox{and}\qquad &
\rhd_{2} \tafter \dtuple{f,g}
& = &
g.
\end{array}$$

\noindent Uniqueness gives equations:
\begin{equation}
\label{eqn:dtupleuniqueness}
\begin{array}{rclcrcl}
\dtuple{f,g} \tafter h
& = &
\dtuple{f \tafter h, g\tafter h}
& \qquad\mbox{and}\qquad &
\dtuple{\rhd_{1} \tafter k, \rhd_{2}\tafter k} 
& = &
k
\end{array}
\end{equation}

\noindent for each $h\colon W \rightarrow Z$ and $k\colon Z \tto X+Y$.
Also, for total maps $h,k$ we have:
\begin{equation}
\label{eqn:dtuplenat}
\begin{array}{rcl}
(h\tplus k) \tafter \dtuple{f,g}
& = &
\dtuple{\klin{h} \pafter f, \klin{k} \pafter g}.
\end{array}
\end{equation}
\end{lemma}

Note that the partial composite $\one \pafter f$ can be unfolded to:
$$\begin{array}{rcccccl}
\one \pafter f
& = &
[\one, \kappa_{2}] \tafter f
& = &
[\kappa_{1} \tafter \bang, \kappa_{2}] \tafter f
& = &
(\bang\tplus\idmap) \tafter f.
\end{array}$$

\noindent This map $\one \pafter (-)$ plays an important role,
see also Lemma~\ref{lem:zero} below. It is the orthosupplement of
the kernel operation, see Section~\ref{sec:kerimg} below.

\begin{proof}
The assumption $\one \pafter f = (\one \pafter g)^{\bot}$ says that
the outer diagram commutes in:
$$\xymatrix{
Z\ar@/_2ex/[ddr]|-{\tafter}_{[\kappa_{2}, \kappa_{1}] \tafter g}
   \ar@/^2ex/[drr]|-{\tafter}^{f}
   \ar@{..>}[dr]|-{\tafter}^(0.7){\dtuple{f,g}}
\\
& X+Y\ar[r]|-{\tafter}^-{\idmap\tplus \bang}
   \ar[d]|-{\tafter}^{\bang\tplus \idmap}\pullback & 
   X + 1\ar[d]|-{\tafter}^{\bang\tplus\idmap}
\\
& 1+ Y\ar[r]|-{\tafter}_-{\idmap\tplus \bang} & 1 + 1
}$$

\auxproof{
$$\begin{array}{rcl}
(\idmap\tplus \bang) \tafter [\kappa_{2}, \kappa_{1}] \tafter g
& = &
[\kappa_{2}, \kappa_{1}] \tafter (\bang\tplus\idmap) \tafter g \\
& = &
[\kappa_{2}, \kappa_{1}] \tafter \kerbot(g) \\
& = &
[\kappa_{2}, \kappa_{1}] \tafter \ker(f) \\
& = &
\kerbot(f) \\
& = &
(\bang\tplus\idmap) \tafter f.
\end{array}$$
}

\noindent The (inner) rectangle is a pullback
by~\eqref{diag:effectuspb}, yielding the pairing $\dtuple{f,g}$ as the
unique total map with $\rhd_{1} \tafter \dtuple{f,g} = (\idmap\tplus\bang)
\tafter \dtuple{f,g} = f$ and:
$$\begin{array}{rcccccl}
\rhd_{2} \tafter \dtuple{f,g}
& = &
[\kappa_{2}, \kappa_{1}] \tafter (\bang\tplus\idmap) \tafter \dtuple{f,g}
& = &
[\kappa_{2}, \kappa_{1}] \tafter [\kappa_{2}, \kappa_{1}] \tafter g
& = &
g.
\end{array}$$

\noindent Clearly, $\tuple{f,g}$ is the unique (mediating) total map
$Z \tto X+Y$ satisfying these equations. The
equations~\eqref{eqn:dtupleuniqueness} follow directly from this
uniqueness property. Finally, the equation~\eqref{eqn:dtuplenat} is
obtained from uniqueness of mediating maps in pullback like the one
above:
$$\begin{array}[b]{rcl}
(\idmap\tplus\bang) \tafter (h\tplus k) \tafter \dtuple{f,g}
& = &
(h\tplus\idmap) \tafter (\idmap\tplus\bang) \tafter \dtuple{f,g} \\
& = &
(h\tplus\idmap) \tafter f \\
& = &
\klin{h} \pafter f
\\
(\bang\tplus\idmap) \tafter (h\tplus k) \tafter \dtuple{f,g}
& = &
(\idmap\tplus k) \tafter (\bang\tplus\idmap) \tafter \dtuple{f,g} \\
& = &
(\idmap\tplus k) \tafter [\kappa_{2}, \kappa_{1}] \tafter g \\
& = &
[\kappa_{2}, \kappa_{1}] \tafter (k\tplus \idmap) \tafter g \\
& = &
[\kappa_{2}, \kappa_{1}] \tafter (\klin{k} \pafter g).
\end{array}\eqno{\QEDbox}$$

\auxproof{
$$\begin{array}{rcl}
\rhd_{1} \tafter \dtuple{f,g}
& = &
(\idmap\tplus \bang) \tafter \dtuple{f,g} \\
& = &
f
\\
\rhd_{1} \tafter \dtuple{f,g}
& = &
[\kappa_{2} \tafter \bang, \kappa_{1}] \tafter \dtuple{f,g} \\
& = &
[\kappa_{2}, \kappa_{1}] \after (\bang\tplus\idmap) \tafter \dtuple{f,g} \\
& = &
[\kappa_{2}, \kappa_{1}] \after [\kappa_{2}, \kappa_{1}] \after g \\
& = &
g.
\end{array}$$

\noindent Uniqueness of $\dtuple{f,g}$ follows from
Lemma~\ref{lem:effectusjm}.
}
\end{proof}

This partial pairing $\dtuple{-,-}$ will be generalised later, see
Lemma~\ref{lem:sumpairing}, and the end of
Discussion~\ref{dis:partialtotal}. It plays an important role in the
sequel, for instance in decomposition, see
Lemma~\ref{lem:quot}~\eqref{lem:quot:dc}, \eqref{lem:quot:mapdc}.

The following result is extremely simple but at the same time
extremely useful.

\begin{lemma}
\label{lem:zero}
Let $f\colon X \pto Y$ be a partial map in an effectus. Then:
$$\begin{array}{rclcrcl}
\one \pafter f = \zero
& \Longleftrightarrow &
f = \zero
& \qquad\mbox{and}\qquad &
\one \pafter f = \one
& \Longleftrightarrow &
f \mbox{ is total.}
\end{array}$$
\end{lemma}

\begin{proof}
The implication $(\Leftarrow)$ on the left is trivial, and for the
implication $(\Leftarrow)$ on the right, let $f$ be total, say $f =
\kappa_{1} \tafter g$ for $g\colon X \tto Y$. Then:
$$\begin{array}{rcccccccl}
\one \pafter f
& = &
[\kappa_{1} \tafter \bang, \kappa_{2}] \tafter \kappa_{1} \tafter g
& = &
\kappa_{1} \tafter \bang \tafter g
& = &
\kappa_{1} \tafter \bang
& = &
\one.
\end{array}$$

\noindent The two implications $(\Rightarrow)$ follow from the
following two diagrams.
$$\xymatrix{
X\ar@/_2ex/[ddr]|-{\tafter}_{f}\ar@/^2ex/[drr]|-{\tafter}^-{\bang}
   \ar@{..>}[dr]|-{\tafter} & &
&
X\ar@/_2ex/[ddr]|-{\tafter}_{f}\ar@/^2ex/[drr]|-{\tafter}^-{\bang}
   \ar@{..>}[dr]|-{\tafter}
\\
& 1\ar@{ >->}[d]|-{\tafter}_{\kappa_2}\ar@{=}[r]\pullback &
  1\ar@{ >->}[d]|-{\tafter}^{\kappa_2}
&
& Y\ar@{ >->}[d]|-{\tafter}_{\kappa_1}\ar[r]|-{\tafter}^-{\bang}\pullback &
  1\ar@{ >->}[d]|-{\tafter}^{\kappa_1}
\\
& Y+1\ar[r]|-{\tafter}_-{\bang\tplus\idmap} & 1+1
&
& Y+1\ar[r]|-{\tafter}_-{\bang\tplus\idmap} & 1+1
}$$

\noindent Both rectangles are pullbacks, as instances of the diagram
on the right in~\eqref{diag:effecutsderivedpb} --- once for the second
coprojection $\kappa_2$. \QED
\end{proof}





\subsection{Examples of effectuses}\label{subsec:effectusEx}

Below we describe the running examples of effectuses that will be used
throughout this text. We will not go into all details and proofs for
these examples, but just sketch the essentials.

\begin{example}
\label{ex:effectusSets}
The category $\Sets$\index{N}{cat@$\Sets$, category of sets} of sets and
functions has coproducts via disjoint union: $X+Y =
\set{\tuple{x,1}}{x\in X} \cup \set{\tuple{y,2}}{y\in Y}$. The numbers
$1$ and $2$ in this set are used as distinct labels, to make sure that
we have a disjoint union. The coprojections $\kappa_{1} \colon X \to
X+Y$ and $\kappa_{2} \colon Y \to X+Y$ are given by $\kappa_{1}x =
\tuple{x,1}$ and $\kappa_{2}y = \tuple{x,2}$. For functions $f\colon X
\to Z$ and $g\colon Y \to Z$ there is a cotuple map $[f,g] \colon X+Y
\to Z$ given by $[f,g](\kappa_{1}x) = f(x)$ and $[f,g](\kappa_{2}y) =
g(y)$.

The emtpy set is the initial object $0\in\Sets$, precisely because for
any set $X$, there is precisely one function $0 \to X$, namely the
empty function. Any singleton set is final in $\Sets$. We typically
write $1 = \{*\}$ for (a choice of) the final object.

Predicates on a set $X$ are `characteristic' functions $X \to 1+1 = 2
\cong \{0,1\}$, which correspond to subsets of $X$. There is thus an
isomorphism $\Pred(X) \cong \Pow(X)$, where $\Pow(X)$ is the
powerset. In particular, the set $\Pred(1)$ of scalars is the
two-element set of Booleans $\Pow(1) \cong 2$. A state is a function
$1 \rightarrow X$, and thus corresponds to an element of the set $X$.

Kleisli maps $f\colon X \to Y+1$ of the lift monad correspond to
partial functions from $X$ to $Y$, written according to our convention
as maps $f\colon X \pto Y$.  We say that $f$ is undefined at $x\in X$
if $f(x) = *\in 1$ --- or more formally, if $f(x) = \kappa_{2}*$. It
is easy to see that Kleisli composition corresponds to the usual
composition of partial functions.

Later, in Section~\ref{sec:extcat}, we shall see that the category
$\Sets$ is an instance of an extensive category, and that such
extensive categories are instances of effectuses, that can be
characterised in a certain way. In fact, every topos is an extensive
category, and thereby an effectus.
\end{example}

\begin{example}
\label{ex:effectusKlD}
In order to capture (discrete) probabilistic models we use two monads
on $\Sets$, namely the \emph{distribution}\index{S}{distribution!--
  monad}\index{S}{monad!distribution --} monad $\Dst$ and the
\emph{subdistribution} monad $\sDst$.

For an arbitrary set $X$ we write $\Dst(X)$\index{N}{fun@$\Dst$,
  distribution monad} for the set of formal convex
combinations\index{S}{convex!-- combination} of elements in $X$. There
are two equivalent ways of describing such convex combinations.
\begin{itemize}
\item We can use expressions $r_{1}\ket{x_1} + \cdots +
  r_{n}\ket{x_n}$\index{N}{$\ket{x}$, element $x$ as part of a formal
    sum} where $x_{i}\in X$ and $r_{i} \in [0,1]$ satisfy
  $\sum_{i}r_{i} = 1$. The `ket' notation $\ket{x}$ is syntactic
  sugar, used to distinguish an element $x\in X$ from its occurrence
  in such sums. The expression $\sum_{i}r_{i}\ket{x_i} =
  r_{1}\ket{x_1} + \cdots + r_{n}\ket{x_n} \in \Dst(X)$ may be
  understood as: the probability of element $x_i$ is $r_i$.


\item We also use functions $\omega \colon X \to [0,1]$ with finite
  support $\supp(\omega) = \setin{x}{X}{\omega(x) \neq 0}$, and with
  property $\sum_{x} \omega(x) = 1$. We can write this $\omega$ as
  formal convex sum $\sum_{x} \omega(x)\ket{x}$. Conversely, each
  formal convex sum $\sum_{i}r_{i}\ket{x_i}$ yields a function $X \to
  [0,1]$ with $x_{i} \mapsto r_{i}$, and $x \mapsto 0$ for all other
  $x\in X$.
\end{itemize}

\noindent We shall freely switch between these two ways of describing
formal convex combinations. Sometimes we use the phrase `discrete
probability distribution'\index{S}{distribution!discrete probability
  --} for such a combination.

The subdistribution\index{S}{subdistribution
  monad}\index{S}{monad!subdistribution --} monad
$\sDst$\index{N}{fun@$\sDst$, subdistribution monad} has
`subconvex'\index{S}{subconvex combination} combinations
$\sum_{i}r_{i}\ket{x_i} \in \sDst(X)$, where $r_{i}\in [0,1]$ satisfy
$\sum_{i}r_{i} \leq 1$.

We write $\Kl(\Dst)$\index{N}{cat@$\Kl(-)$, Kleisli category} for the
Kleisli category of the monad $\Dst$. Objects of this category are
sets, and morphisms $X \to Y$ are functions $X \to \Dst(Y)$. If $X,Y$
are finite sets, such maps $X \rightarrow \Dst(Y)$ are precisely the
stochastic matrices. Such functions can be understood as Markov
chains, describing probabilistic computations.  For functions $f\colon
X \to \Dst(Y)$ and $g\colon Y \to \Dst(Z)$ we define $g \tafter f
\colon X \to \Dst(Z)$ essentially via matrix multiplication:
$$\begin{array}{rcl}
(g \tafter f)(x)
& = &
{\displaystyle\sum_{z}} \big(\sum_{y} g(y)(z)\cdot f(x)(y)\big)\bigket{z}.
\end{array}$$

\noindent Notice that we have mixed the above two descriptions: we
have used the ket-notation for the distribution $(g \tafter f)(x) \in
\Dst(Z)$, but we have used the distributions $f(x)\in\Dst(Y)$ and
$g(y)\in\Dst(Z)$ as functions with finite support, namely as functions
$f(x) \colon Y \to [0,1]$ and $g(y) \colon Z \to
    [0,1]$. This allows us to multiply (and add) the probabilities
    $g(y)(z) \in [0,1]$ and $f(x)(y) \in [0,1]$. It can be checked
    that the above definition yields a distribution again, and that
    the operation $\tafter$ is associative.

We also need an identity map $X \to X$ in $\Kl(\Dst)$. It is
the function $X \to \Dst(X)$ that sends $x\in X$ to the
`Dirac' distribution $1\ket{x} \in \Dst(X)$. As function $X
\to [0,1]$ this Dirac distribution sends $x$ to $1$, and any
other element $x'\neq x$ to $0$. In this way $\Kl(\Dst)$ becomes a
category. The empty set $0$ is initial in $\Kl(\Dst)$, and the
singleton set $1$ is final, because $\Dst(1) \cong 1$.

The Kleisli category $\Kl(\Dst)$ inherits finite coproducts from the
underlying category $\Sets$. On objects it is the disjoint union
$X+Y$, with coprojection $X \to \Dst(X+Y)$ given by $x \mapsto
1\ket{\kappa_{i}x}$.  The cotuple is as in $\Sets$. We leave it to the
reader to check that the requirements of an effectus hold in
$\Kl(\Dst)$. Later, in Example~\ref{ex:biproductSetsKlD}, we shall see
that this example effectus $\Kl(\Dst)$ is an instance of a more
general construction.
 
Predicates on a set $X$ are functions $X \to \Dst(1+1) = \Dst(2) \cong
[0,1]$. Hence predicates in this effectus are
fuzzy:\index{S}{predicate!fuzzy --}\index{S}{fuzzy predicate}
$\Pred(X) \cong [0,1]^{X}$. The scalars are the probabilities from
$[0,1]$. A state $1 \tto X$ corresponds to a distribution $\omega
\in\Dst(X)$.

Interestingly, the category $\Par(\Kl(\Dst))$ of partial maps in the
Kleisli category of the distribution monad $\Dst$ is the Kleisli
category $\Kl(\sDst)$ of the subdistribution monad. The reason is that
subdistribution in $\sDst(Y)$ can be identified with a a distribution
in $\Dst(Y+1)$. Indeed, given a subdistribution
$\sum_{i}r_{i}\ket{y_i}$, we take $r = 1 - (\sum_{i}r_{i})$ and get a
proper distribution $\sum_{i}r_{i}\ket{y_i} + r\ket{\!*\!}$, where $1
= \{*\}$. Thus we often identify a partial map $X \pto Y$ in
$\Kl(\Dst)$ with a function $X \to \sDst(Y)$.

We have used the distribution monad $\Dst$ on $\Sets$, for
\emph{discrete} probability. For \emph{continuous} probability one can
use the Giry monad\index{S}{Giry monad}\index{S}{monad!Giry --}
$\Giry$\index{N}{fun@$\Giry$, Giry monad} on the category of
measurable spaces, see~\cite{Jacobs15a,Jacobs13} for details. The
Kleisli category $\Kl(\Giry)$ of this monad $\Giry$ is also an
effectus.
\end{example}

The next example of an effectus involves order unit groups. It does
not capture a form of computation, like $\Sets$ or $\Kl(\Dst)$. Still
this example is interesting because it shares some basic structure
with the quantum model given by von Neumann algebras --- see the
subsequent Example~\ref{ex:effectusNA} --- but it is mathematically
much more elementary. Hence we use it as an intermediate stepping
stone towards von Neumann algebras.

\begin{example}
\label{ex:effectusOUG}
We write $\Ab$\index{N}{cat@$\Ab$, category of Abelian groups} for the
category of Abelian groups,\index{S}{Abelian!--
  group}\index{S}{group!Abelian --} with group homomorphisms between
them. An Abelian group $G$ is called \emph{ordered} if there is a
partial order $\leq$ on $G$ that satisfies: $x \leq x'$ implies $x+y
\leq x'+y$, for all $x,x',y\in G$.  One can then prove, for instance
$x\leq x'$ iff $0 \leq x' - x$ iff $x + y = x'$ for some $y\geq 0$.

A homomorphism $f\colon G \rightarrow H$ of ordered Abelian groups is
a group homomorphism that is monotone: $x \leq x' \Rightarrow f(x)
\leq f(x')$. Such a group homomorphism is monotone iff it is
positive.\index{S}{positive!-- map} The latter means: $x \geq 0
\Rightarrow f(x) \geq 0$. We write $\OAb$\index{N}{cat@$\OAb$,
  category of ordered Abelian groups} for the category of ordered
Abelian groups with monotone/positive group homomorphisms between
them. There is an obvious functor $\OAb \rightarrow \Ab$ which forgets
the order.

An ordered Abelian group $G$ is called an \emph{order unit
  group}\index{S}{order unit!-- group}\index{S}{group!order unit --}
if there is a positive unit $1\in G$ such that for each $x\in G$ there
is an $n\in\NNO$ with $-n\cdot 1 \leq x \leq n\cdot 1$. A homomorphism
$f\colon G \rightarrow H$ of order unit groups is a positive group
homomorphisms which is `unital',\index{S}{unital map!-- in an order
  unit group} that is, it preserves the unit: $f(1) = 1$. We often
write `PU'\index{S}{PU, positive unital} for positive unital. We
obtain another category $\OUG$\index{N}{cat@$\OUG$, category of order
  unit groups} of order unit groups and PU group homomorphisms between
them. We have a functor $\OUG \rightarrow \OAb$ that forgets the unit.

It is easy to see that order unit groups are closed under finite
products: the cartesian product, commonly written as direct sum
$G_{1}\oplus G_{2}$, is again an ordered Abelian group, with
componentwise operations and order, and with the pair $(1,1)$ as
unit. The trivial singleton group $\{0\}$ with $1 = 0$ is final in the
category $\OUG$. The group of integers $\Z$ is initial in $\OUG$,
since for each order unit group $G$ there is precisely one map
$f\colon \Z \rightarrow G$, namely $f(k) = k\cdot 1$.

We now claim that the \emph{opposite} $\op{\OUG}$ of the category of
order unit groups is an effectus. The direct sums $(\oplus, \{0\})$
form coproducts in $\op{\OUG}$, and the integers $\Z$ form the final
object. The effectus requirements in Definition~\ref{def:effectus} can
be verified by hand, but they also follow from a general construction
in Example~\ref{ex:biproductOUGvNA}. This category is in opposite form
since it incorporates Heisenberg's view where computations are
predicate transformers, acting in opposite direction, transforming
`postcondition' predicates on the codomain of the computation to
`precondition' predicates on the domain.

We elaborate on predicates and on partial maps since they can be
described in alternative ways that also apply in other settings.
\begin{enumerate}
\item A predicate on an order unit group $G$ is formally a map
  $p\colon G \tto 1+1$ in the effectus $\op{\OUG}$. In the category
  $\OUG$ this amounts to a map $p\colon \Z\oplus\Z \rightarrow G$.  We
  claim that such predicates correspond to \emph{effects}, that is, to
  elements in the unit interval $[0,1]_{G} = \setin{x}{G}{0 \leq x
    \leq 1}$.\index{D}{$[0,1]_{G}$, unit interval in an order unit
    group $G$} Indeed, the element $p(1,0)$ is positive, and below the
  unit since $p(1,0) + p(0,1) = p(1,1) = 1$, where $p(0,1)$ is
  positive too. In the other direction, an effect $e\in [0,1]_{G}$
  gives rise to a map $p_{e} \colon \Z\oplus\Z \rightarrow G$ by
  $p_{e}(k,m) = k\cdot e + m\cdot e^{\bot}$, where $e^{\bot} = 1 - e$.

\item A partial map $f\colon G \pto H$ in $\op{\OUG}$ is a map
  $f\colon G \tto H+1$, which corresponds to a positive unital group
  homomorphism $f\colon H\oplus \Z \rightarrow G$. We next claim that
  these homomorphisms $f$ correspond to \emph{subunital} positive
  group homomorphisms $g\colon H \rightarrow G$, where
  subunital\index{S}{subunital map!-- in an order unit group} means
  that $g(1) \leq 1$. This correspondence works as follows: given $f$,
  take $\overline{f}(y) = f(y, 0)$; and given $g$, take
  $\overline{g}(y, k) = g(y) + k\cdot g(1)^{\bot}$, where $g(1)^{\bot}
  = 1 - g(1)$. We use the abbreviation `PsU'\index{S}{PsU, positive
    subunital} for `positive subunital'.
\end{enumerate}

\noindent Summarising, there are bijective correspondences for
predicates and partial maps:
\begin{equation}
\label{corr:OUGpredpartmap}
\hspace*{-1em}\begin{prooftree}
\begin{prooftree}
\xymatrix{G\ar[r]|-{\tafter} & 1+1}\rlap{\hspace*{1.1em}in $\op{\OUG}$}
\Justifies
\xymatrix{\Z\oplus \Z \ar[r]|-{\tafter} & G}\rlap{\hspace*{1em}in $\OUG$}
\end{prooftree}
\Justifies
e\in [0,1]_{G}
\end{prooftree}
\hspace*{9em}
\begin{prooftree}
\begin{prooftree}
\xymatrix{G\ar[r]|-{\pafter} & H}\rlap{\hspace*{1.8em}in $\Par(\op{\OUG})$}
\Justifies
\xymatrix{G\ar[r]|-{\tafter} & H+1}\rlap{\hspace*{1.1em}in $\op{\OUG}$}
\end{prooftree}
\Justifies
\begin{prooftree}
\xymatrix{H\oplus \Z \ar[r]|-{\tafter} & G}\rlap{\hspace*{1.1em}in $\OUG$}
\Justifies
\xymatrix{H \ar[r]_-{\text{PsU}} & G}
\end{prooftree}
\end{prooftree}\hspace*{6em}
\end{equation}

\noindent A state of an order unit group $G$ is a homomorphism $G
\rightarrow \Z$.

Order unit groups provide a simple example of an effectus, in which
part of the structure of more complex models --- like von Neumann
algebras, see below --- already exists. Similar structures are order
unit \emph{spaces}.\index{S}{order unit!-- space} They are ordered
vector spaces over the reals, which are at the same time an order unit
group. The category $\OUS$\index{N}{cat@$\OUS$, category of order unit
  spaces} of order unit spaces has positive unital linear maps as
homormorphisms. It plays a promiment role in generalised probabilistic
theories, see \textit{e.g.}~\cite{Wilce12}. One can prove that the
category $\op{\OUS}$ is an effectus too. The states of an order unit
space $V$ are the homomorphisms $V \rightarrow \R$. They carry much
more interesting structure than the states of an order unit group. For
instance, this set of states $\Stat(V)$ forms a convex compact
Hausdorff space. More information about order unit groups can be found
in~\cite{Goodearl86}, and about order unit spaces
in~\cite{Nagel74,JacobsM12b}.
\end{example}

\begin{example}
\label{ex:effectusNA}
The mathematically most sophisticated example of an effectus is the
opposite $\op{\vNA}$ of the category $\vNA$\index{N}{cat@$\vNA$,
  category of von Neumann algebras} of von Neumann algebra,
  which provides a model for quantum computation. 
For example, the qubit is represented by the von Neumann algebra
of $2\times 2$ complex matrices (see 
point~\ref{ex:effectusNA-qubit} below).

For completeness, we give the definition of von Neumann algebra (in
point~\ref{ex:effectusNA-defNA}), and the morphisms of~$\vNA$ (in
point~\ref{ex:effectusNA-defMorph}), but we refer the reader
to~\cite{kadison97,paulsen2002,Sakai71} for the intricate details.

\begin{enumerate}
\item
The structure that invoked the study of von Neumann algebra is the set
$\B(\mathscr{H})$ of \emph{bounded} operators
$T\colon\mathscr{H}\to\mathscr{H}$ on a Hilbert space~$\mathscr{H}$.
Recall that an operator~$T$ on~$\mathscr{H}$ is bounded provided that:
$$\begin{array}{rcccl}
\|T\|
& = &
\inf\setin{b}{[0,\infty)}{\allin{x}{\mathscr{H}}{\|T(x)\|\leq b\cdot \|x\|}}
& < &
\infty.
\end{array}$$

\noindent Not only is~$\B(\mathscr{H})$ a complex vector
space (with coordinatewise operations), but it also carries a
multiplication (namely composition of operators), an involution
(viz.~taking adjoint), and a unit (viz.~the identity operator), which
interact appropriately, and $\B(\mathscr{H})$ is therefor a
\emph{unital $*$-algebra}.  The adjoint of an operator $S\in
\B(\mathscr{H})$ is the the unique operator~$S^{*}\in
\B(\mathscr{H})$ with $\inprod{S(x)}{y} = \inprod{x}{S^*(y)}$
for all~$x,y\in\mathscr{H}$.

We assume that every $*$-algebra has a unit (`is unital'). A
$*$-subalgebra of a $*$-algebra~$\mathscr{A}$ is 
a complex linear subspace $\mathscr{S}$
of~$\mathscr{A}$ which is closed under multiplication and involution, $(-)^*$,
and contains the unit $1$ 
of~$\mathscr{A}$. An isomorphism between $*$-algebras
is a linear bijection which preserves multiplication, involution and
unit.

\item Let~$\mathscr{H}$ be a Hilbert space. A \emph{unital
  $C^*$-algebra of operators} $\mathscr{A}$ 
  on~$\mathscr{H}$ is a $*$-subalgebra
  of $\B(\mathscr{H})$ which is \emph{norm closed}, that is,
  if $(T_\alpha)_\alpha$ is a net in~$\mathscr{A}$ and $S\in \mathscr{A}$ 
	  such that
  $\|T_\alpha-S\|\rightarrow 0$, then $S\in \mathscr{A}$. A \emph{von Neumann
    algebra of operators} $\mathscr{A}$ on~$\mathscr{H}$ is a $*$-subalgebra
  $\B(\mathscr{H})$, which is \emph{weakly closed}, that is,
  if $(T_\alpha)_\alpha$ is a net in~$\mathscr{A}$ and $S\in \mathscr{A}$ 
	  such that
  $\inprod{x}{(S-T_\alpha )(x)}\rightarrow 0$ for
  all~$x\in\mathscr{H}$, then $S\in \mathscr{A}$.

\item \label{ex:effectusNA-defNA} A \emph{unital $C^*$-algebra}
  \index{S}{C@$C^*$-algebra} $\mathscr{A}$ 
  is a $*$-algebra which is isomorphic to
  a unital $C^*$-algebra of operators on some Hilbert space.  (This is
  not the usual definition, but it is equivalent,
  see~\cite[Thm.~4.5.6]{kadison97}). A \emph{von Neumann
    algebra}\index{S}{von Neumann algebra} 
    (also  \emph{$W^*$-algebra}\index{S}{$W^*$-algebra}) 
    $\mathscr{A}$ is a $*$-algebra which
  is isomorphic to a von Neumann algebra of operators on some Hilbert
  space (cf.~text above Example 5.1.6 of~\cite{kadison97}).  Every von
  Neumann algebra is also a $C^*$-algebra, but the converse is false.
  In this text we will deal mainly with von Neumann
  algebras; we added the definition of $C^*$-algebra here only 
  for comparison.

\item \label{ex:effectusNA-product} For every Hilbert
  space~$\mathscr{H}$, the $*$-algebra of operators
  $\B(\mathscr{H})$ is a von Neumann algebra.  In particular,
  $\{0\}$, $\C$, $M_2$, $M_3$, \dots are all von Neumann algebras,
  where $M_n$ is the $*$-algebra of $n\times n$ complex matrices. The
  cartesian product, $\mathscr{A}\oplus \mathscr{B}$, of two von
  Neumann algebras is again a von Neumann algebra with coordinatewise
  operations. In particular, the complex plane $\mathbb{C}^2 =
  \mathbb{C}\oplus\mathbb{C}$ is a von Neumann algebra. Given a von
  Neumann algebra~$\mathscr{A}$ the $*$-algebra~$M_n(\mathscr{A})$ of
  $n\times n$-matrices with entries drawn from~$\mathscr{A}$ is a von
  Neumann algebra.

\item \label{ex:effectusNA-positive} An element~$a$ of a von Neumann
  algebra~$\mathscr{A}$ 
  is called \emph{positive} \index{S}{positive!-- element}
  if $a=b^* \cdot b$ for some~$b\in \mathscr{A}$, and \emph{self-adjoint}
  \index{S}{self-adjoint} if $a^*=a$.  The set of positive elements
  of~$\mathscr{A}$ is denoted by~$\mathscr{A}_+$ 
  and the set of self-adjoint elements
  of~$\mathscr{A}$ is denoted by~$\mathscr{A}_{\mathrm{sa}}$.
   All positive elements of~$\mathscr{A}$
  are self-adjoint.

A von Neumann algebra~$\mathscr{A}$ is partially 
ordered by: $a\leq b$ iff $b-a$
is positive.  The self-adjoint elements of a von Neumann algebra~$\mathscr{A}$
form an order unit space,~$\mathscr{A}_{\mathrm{sa}}$. 
For every positive element
$a\in \mathscr{A}$ of a von Neumann algebra there is a unique positive element
$b\in \mathscr{A}$ with $b^2 = a$, which is denoted by~$\sqrt{a}$.

In a von Neumann algebra every bounded directed net of self-adjoint
elements~$D$ has a supremum~$\bigvee D$ in $\mathscr{A}_{\mathrm{sa}}$.  Both
`bounded' and `directed' are necessary, and the statement is false
for $C^*$-algebras.

\item
A linear map $f\colon \mathscr{A}\to \mathscr{B}$ between 
von Neumann algebras
is called 
\begin{enumerate}
\item \emph{unital}\index{S}{unital map!-- between von Neumann
  algebras} if $f(1)=1$;

\item \emph{subunital}\index{S}{subunital map!-- between von Neumann
  algebras} if $f(1)\leq 1$;

\item \emph{positive}\index{S}{positive map between von Neumann
  algebras} if $f(a)\in \mathscr{B}_+$ for all~$a\in \mathscr{A}_+$;

\item \emph{involutive}\index{S}{involutive map} if $f(a^*)=f(a)^*$
  for all~$a\in \mathscr{A}$;

\item \emph{multiplicative}\index{S}{multiplicative map} if $f(a\cdot
  b)=f(a)\cdot f(b)$ for all~$a,b\in \mathscr{A}$;

\item \emph{completely positive}\index{S}{positive!completely
  --}\index{S}{completely positive map} (see~\cite{paulsen2002}) if
  for every $n\in \mathbb{N}$ the map
$$\xymatrix@C+1pc{
M_n(\mathscr{A}) \ar[r]^-{M_n(f)} & M_n(\mathscr{B})
}$$

\noindent given by $(M_n(f)(A))_{ij}=f(A_{ij})$ for all~$A\in
M_n(\mathscr{A})$ is positive;

\item \emph{normal}\index{S}{normal map} if $f(\bigvee D)=\bigvee
  f(D)$ for every bounded directed set $D\subseteq
  \mathscr{A}_{\mathrm{sa}}$ (cf.~\cite{Sakai71}).
\end{enumerate}

\item\label{ex:effectusNA-defMorph} We denote the category of von
  Neumann algebras and completely positive normal unital maps
  by~$\vNA$.  It has an initial object~$\mathbb{C}$, final
  object~$\{0\}$, and products given by $\oplus$ (see
  point~\ref{ex:effectusNA-product}). The opposite category
  $\op{\vNA}$ then has finite coproducts and a final object.  In fact,
  $\op{\vNA}$ is an effectus for almost the same reasons as
  $\op{\OUS}$ and $\op{\OUG}$ are an effectuses (see
  Example~\ref{ex:effectusOUG}).

As is the case for~$\op{\OUS}$, the predicates on a von Neumann
algebra $\mathscr{A}$ correspond to elements $a\in \mathscr{A}$ with $0\leq a
\leq 1$, which are called \emph{effects}.  In fact, the map
$\Pred(\mathscr{A})\to [0,1]_\mathscr{A}$ 
given by $p\mapsto p(1,0)$ is an isomorphism of
effect algebras, see the correspondences~\eqref{corr:OUGpredpartmap}.
Also, the partial maps in $\Par(\op{\vNA})$ correspond to the
completely positive normal \emph{subunital} maps between von Neumann
algebras.

\item 
A von Neumann algebra~$\mathscr{A}$ is commutative
if $a\cdot b = b\cdot a$ for all $a,b\in \mathscr{A}$.
We write $\CvNA \hookrightarrow \vNA$\index{N}{cat@$\CvNA$,
  category of commutative von Neumann
  algebras}\index{S}{commutative!-- von Neumann algebra} 
\index{S}{von Neumann algebra!commutative --} for the full subcategory
of \emph{commutative} von Neumann algebras.  We recall that a positive
unital map $\mathscr{A} \rightarrow \mathscr{B}$ is automatically
completely positive when either $\mathscr{A}$ or $\mathscr{B}$ is
commutative.  These commutative von Neumann algebras capture
probabilistic models, as special case of the general quantum models in
$\vNA$. Also the category $\op{\CvNA}$ is an effectus.

\item \label{ex:effectusNA-qubit} For quantum computation the most
  important von Neumann algebra is perhaps the $*$-algebra~$M_2$ of
  $2\times 2$ matrices, because it models a
  \emph{qubit};\index{S}{qubit} it is closely related to the Hilbert
  space~$\C^2$, since $M_{2} = \B(\C^{2})$.

A \emph{pure state}\index{S}{pure state!-- of a
  qubit}\index{S}{state!pure --!--- of a qubit} of a qubit is a vector
in~$\C^2$ of length~$1$, and can thus be written as
$\alpha\left|0\right> + \beta\left|1\right>$, where $\ket{0} = (1,0)$
and $\ket{1} = (0,1)$ are base vectors in $\C^2$, and
$\alpha,\beta\in\C$ are scalars with $|\alpha|^2 + |\beta|^2 =1$.  One
says that any pure state of a qubit is in a superposition of~$\ket{0}$
and~$\ket{1}$.  Two pure states $(\alpha_1,\beta_1)$ and
$(\alpha_2,\beta_2)$ of a qubit are equivalent if
$(z\cdot\alpha_1,z\cdot \beta_1) = (\alpha_2,\beta_2)$ for some~$z\in
\C$ (with $|z|=1$).

A \emph{sharp predicate}\index{S}{sharp!-- predicate!-- on a qubit} on
a qubit is a linear subspace~$C$ of~$\C^2$ which in turn corresponds
to a projection~$P\in M_2$ (namely the orthogonal projection
onto~$C$).  For instance, the predicate ``the qubit is in state
$\ket{0}$'' is represented by the linear subspace $\{(\alpha,0)\colon
\alpha\in \C\}$ (the $x$-axis), and by the
projection~$\spmat{1&0\\0&0}$.

\item Let~$C$ be a linear subspace of~$\C^2$ (i.e.~a sharp predicate
  of a qubit) with orthogonal projection~$P\in M_2$.  Then~$P^\bot =
  I-P$ is again a projection, onto the \emph{orthocomplement} $C^\bot
  = \setin{y}{\C^2}{\allin{x}{C}{\inprod{x}{y}=0}}$ of~$C$.

Let $x\in \C^2$ be a pure state of a qubit (so $\|x\|=1$).  Then:
$$\begin{array}{rcl}
x \in C 
& \iff &
\text{predicate $C$  certainly holds in state~$x$} 
\\
x \in C^\bot 
& \iff &
\text{predicate $C$ certainly does not hold in state~$x$.}
\end{array}$$

\noindent However, if $x\notin C\cup C^\bot$, then the matter
whether~$C$ holds in~$x$ is undecided.  The matter can, however, be
pressed using measurement; it forces the qubit to the state
$P(x)\cdot\|P(x)\|^{-1}$ with probability~$\|P(x)\|^2$, and to the
state $P^\perp(x)\cdot\|P^\perp(x)\|^{-1}$ with
probability~$\|P^\bot (x)\|^2$. 

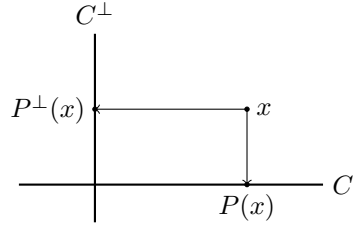
\begin{center}
\begin{tikzpicture}
\draw[line width=.8pt] (-1,0) -- (3,0);
\draw (3,0) node[anchor=west]{$C$};

\draw[line width=.8pt] (0,-.5) -- (0,2);
\draw (0,2) node[anchor=south]{$C^\perp$};

\draw[fill=black] (2,1) circle (.8pt);
\draw (2,1) node[anchor=west]{$x$};

\draw[->] (2,1) -- (2,0);
\draw[fill=black] (2,0) circle (.8pt);
\draw (2,0) node[anchor=north]{$P(x)$};

\draw[->] (2,1) -- (0,1);
\draw[fill=black] (0,1) circle (.8pt);
\draw (0,1) node[anchor=east]{$P^\perp (x)$};
\end{tikzpicture}
\end{center}
Thus, using measurement one finds that
the predicate~$C$ holds in state~$x$ with probability $\|P(x)\|^2$.
After measurement, the state of the qubit has become a (probabilistic)
mixture of
$P(x)\cdot\|P(x)\|^{-1}\in C $ and
$P^\perp(x)\cdot\|P^\perp(x)\|^{-1} \in C^\perp$.

\item Any pure state~$x$ of a qubit gives rise to a
  state~$\omega_x\colon M_2\to \C$ of $\op{\vNA}$ given by
  $\omega_x(A)=\inprod{Ax}{x}$ for all~$A\in M_2$. The states of the
  form~$\omega_x$ are precisely the extreme points of the convex set
  $\Stat(M_2)$ of all states on~$M_2$. The states which are not
  extreme correspond to probabilistic mixtures of pure states.  For
  example, if one measures whether the pure state~$\ket{{+}} =
  \frac{1}{\sqrt{2}}(\ket{0} + \ket{1})$ is in the state $\ket{0}$,
  then the resulting state is
  $\frac{1}{2}(\omega_{\ket{0}}+\omega_{\ket{1}})$, which not extreme,
  and corresponds to a qubit which is either in state $\ket{0}$ or
  $\ket{1}$ with equal probability.

Any state of~$M_2$ can be written as a convex combination
of extreme states. But this is not true in general:
a von Neumann algebra might have no extreme states
at all.

\item
Any sharp predicate~$C$ of the qubit
gives rise to a predicate
on~$M_2$ in~$\op{\vNA}$.
To see this recall that
the predicates on~$M_2$ in~$\op{\vNA}$
correspond to the elements $P\in M_2$ with~$0\leq P \leq 1$.
Now, the sharp predicate~$C$ corresponds to the 
orthogonal projection~$P_C$ onto~$C$.

The predicates on~$M_2$ of the form~$P_C$
are precisely those predicates~$P$ with:
for every $Q\in M_2$ with $0\leq Q \leq 1$,
if~$Q\leq P$ and~$Q \leq P^\perp$,
then~$Q=0$. 

\item
Let us think about the conjunction of 
two sharp predicates~$C$ and~$D$ 
on the qubit. Surely,
it should be~$C\cap D$.
Is measuring whether $C\cap D$
holds simply a matter of first measuring~$C$
and then measuring~$D$?

Let~$P$ be the orthogonal projection on~$C$,
and let~$Q$ be the orthogonal projection onto~$D$.
First note that if~$x\in\mathbb{C}^2$
is a state of the qubit,
then the probability to find that~$C$ holds
upon measurent, and then that~$D$ holds up measurement
is $\|Q (P (x))\|^2$,
and the resulting state is 
$Q (P (x))\cdot\|Q (P (x))\|^{-1}$.
By the way, $\|Q (P (x))\|^2 = \inprod{P (Q (P (x)))}{x}
= \varphi_x(PQP)$,
so first measuring~$C$ and then~$D$
corresponds in this formalism to the operator $PQP \in M_2$
(and not to $QP$ as one might have guessed). 
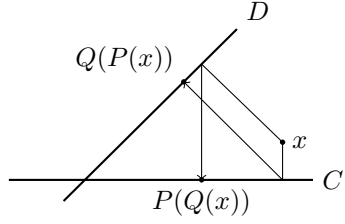
\begin{center}
\begin{tikzpicture}
\draw[line width=.8pt] (-1,0) -- (3,0);
\draw (3,0) node[anchor=west]{$C$};

\draw[line width=.8pt] (-.28,-.28) -- (2,2);
\draw (2,2) node[anchor=south west]{$D$};

\draw[fill=black] (2.6,.5) circle (.8pt);
\draw (2.6,.5) node[anchor=west]{$x$};

\draw (2.6,.5) -- (2.6,0);
\draw[->] (2.6,0) -- (1.3,1.3);
\draw[fill=black] (1.3,1.3) circle (.8pt);
\draw (1.3,1.3) node[anchor=south east]{$Q(P(x))$};

\draw (2.6,.5) -- (1.535,1.535);
\draw [->] (1.535,1.535) -- (1.535,0);
\draw[fill=black] (1.535,0) circle (.8pt);
\draw (1.535,0) node[anchor=north]{$P(Q(x))$};
\end{tikzpicture}
\end{center}
Now, 
(by choosing $C:=\{(\alpha,0)\colon \alpha\in \mathbb{C}\}$
and $D:=\{(\alpha,\alpha)\colon \alpha\in \mathbb{C}\}$)
one easily sees
the following peculiar behaviour typical for quantum systems:
\begin{enumerate}
\item
Measurement disturbes the system:
while $P(x)\cdot \|P(x)\|^{-1}$
is in~$C$,
the state
$Q(P(x))\cdot \|Q(P(x))\|^{-1}$
need not be in~$C$;
\item
The order of measurement matters for the probability
of the outcome:
we might have $\|Q(P(x))\|^2\neq \|P(Q(x))\|^2$;
\item
The composition of two sharp measurements
need not be sharp:
$PQP$ need not be a projection.
\end{enumerate}
Thus, 
measuring whether $C\cap D$ holds
is not always simply a matter of measuring~$C$
and then measuring~$D$.

The operator $PQP$ is called the sequential product of~$P$ and~$Q$ and
is denoted by~$P \andthen Q$.  More generally, if $\mathscr{A}$ is a
von Neumann algebra, and $p,q\in [0,1]_\mathscr{A}$, then the
sequential product is given by $p \andthen q = \sqrt{p}q\sqrt{p}$, and
it represents first measuring~$p$ and then measuring~$q$.

We will see in Theorem~\ref{thm:uniqueness-asrt}
that~$p\andthen q$
can be defined in~$\op{\vNA}$
using only the language
of category theory.

\item We have the following situation for a qubit:
$$\begin{array}{rclcrcl}
\{\text{pure states}\}
& \subseteq &
\C^2
& \qquad &
\{\text{sharp predicates}\} 
& \subseteq &
M_2.
\end{array}$$

\noindent This situation is typical: any Hilbert space~$\mathscr{H}$,
representing a purely quantum system, gives rise to the von Neumann
algebra $\B(\mathscr{H})$ of bounded operators
on~$\mathscr{H}$; the elements of norm~$1$ in~$\mathscr{H}$ are the
pure states of this quantum system, while the projections
in~$\B(\mathscr{H})$ are the sharp predicates, corresponding
to \emph{closed} linear subspaces of~$\mathscr{H}$.  For instance,
$M_3$ represents a qutrit, and $\B(\ell_2)$ represents a
quantum integer.

\item
Let~$\mathscr{A}$ be a von Neumann algebra. An important technical
fact is that the states on~$\mathscr{A}$ in~$\op{\vNA}$ (usually
called \emph{normal} states in the literature) are \emph{separating},
that is, for all~$a,b\in \mathscr{A}$ we have $a\leq b$ if and only if
for every $\omega\colon\mathscr{A}\to\mathbb{C}$ in~$\vNA$ we have
$\omega(a)\leq\omega(b)$.

This entails
that the \emph{ultraweak topology}\index{S}{ultraweak!-- topology}---the 
least topology
on~$\mathscr{A}$
such that all states on~$\mathscr{A}$ in~$\op{\vNA}$
are continuous---is Hausdorff.
A net $(a_i)_{i\in D}$
of elements of~$\mathscr{A}$
converges \emph{ultraweakly}\index{S}{ultraweak!-- convergence}
(i.e.~with respect to the ultraweak topology)
to an element $a\in\mathscr{A}$
iff $(\omega(a_i))_{i\in D}$
converges to~$\omega(a)$
for every state~$\omega$ on~$\mathscr{A}$ in~$\op{\vNA}$.

It is not hard to see that a positive linear map
$f\colon \mathscr{A}\to\mathscr{B}$
between von Neumann algebras
is ultraweakly continuous\index{S}{ultraweak!-- continuity}
if and only if~$f$ is normal.
In this sense
the ultraweak topology
plays the same role in
the theory of von Neumann algebras
as the Scott topology plays in domain theory.

There is much more
to be said about the ultraweak topology in particular,
and von Neumann algebras in general,
but since we draw
most (new) results
about von Neumann algebras from
more specialized publications, \cite{WesterbaanW15,Rennela14,Cho14},
we refrain from diverging
any farther into the theory of operator algebras.
\end{enumerate}
\end{example}

\section{The structure of partial maps and predicates}\label{sec:partialpred}

This section concentrates on the structure that exists on partial maps
$X \pto Y$ in an effectus. It turns out that they come equipped with a
partial binary sum operation, written as $\ovee$, which has the zero
map as unit and is commutative and associative (in an appropriately
partial sense). Abstractly, the relevant observation is that the
category $\Par(\cat{B})$ of partial maps of an effectus $\cat{B}$ is
enriched over the category $\PCM$\index{N}{cat@$\PCM$, category of partial
  commutative monoids} of partial commutative monoids.  We first
describe what such PCMs are.

The section continues with predicates, as special case. The main
result of this section states that these predicates in an effectus
form an effect module, in a functorial way.

\begin{definition}
\label{def:pcm}
A \emph{partial commutative monoid},\index{S}{monoid!partial
  commutative --}\index{partial!-- commutative monoid} abbreviated as
PCM,\index{S}{PCM|see{partial commutative monoid}} is given by a set
  $M$ with a zero element $\zero\in M$ and a partial binary operation
  $\ovee$\index{N}{$\ovee$, partial sum} on $M$, satisfying the
  following requirements --- where we call $x,y$ and orthogonal and
  write $x \orthogonal y$\index{N}{$\orthogonal$, orthogonality} when
  $x\ovee y$ is defined.
\begin{enumerate}
\item $x \orthogonal \zero$ and $x \ovee \zero = x$, for each $x\in M$;

\item if $x \orthogonal y$ then $y \orthogonal x$ and $x\ovee y = y
  \ovee x$, for all $x,y\in M$;

\item if $x \orthogonal y$ and $(x \ovee y) \orthogonal z$, then $y
  \orthogonal z$ and $x \orthogonal (x \ovee z)$, and moreover
  $(x\ovee y) \ovee z = x \ovee (y\ovee z)$, for all $x,y,z\in M$.
\end{enumerate}

\noindent A homomorphism $f\colon M \rightarrow N$ between two PCMs is
a function that preserves $\zero$ and $\ovee$, in the sense that:
\begin{enumerate}
\item $f(\zero) = \zero$;

\item if $x \orthogonal y$ in $M$, then $f(x) \orthogonal f(y)$ in $N$
  and $f(x \ovee y) = f(x) \ovee f(y)$.
\end{enumerate}

\noindent We write $\PCM$ for the resulting category.
\end{definition}

This notion of PCM shows up in the following way. Sums $\ovee$ are
defined for parallel partial maps, via bounds. These definitions work
well and satisfy appropriate properties via the categorical structure
in effectuses, as identified in the previous section.

\begin{proposition}[From~\cite{Cho15}]
\label{prop:effectusPCM}
Let $\cat{B}$ be an effectus.
\begin{enumerate}
\item \label{prop:effectusPCM:hom} Each homset $\Par(\cat{B})\big(X,
  Y\big)$ of partial maps $X \pto Y$ is a PCM in the following way.
\begin{itemize}
\item $f \orthogonal g$ iff there is a `bound'\index{S}{bound}
  $b\colon X \pto Y+Y$ with $\rhd_{1} \pafter b = f$ and $\rhd_{2}
  \pafter b = g$. By Lemma~\ref{lem:effectusjm} such a bound $b$, if
  it exists, is necessarily unique.

\item if $f \orthogonal g$ via bound $b$, then we define $f \ovee g =
  \nabla \pafter b = (\nabla\tplus\idmap) \tafter b \colon X \pto Y$.
\end{itemize}

\item \label{prop:effectusPCM:pres} This PCM-structure is preserved
  both by pre- and by post-composition in $\Par(\cat{B})$.
\end{enumerate}
\end{proposition}

The latter statement can be stated more abstractly as: $\Par(\cat{B})$
is enriched over $\PCM$. This works since the category $\PCM$ is
symmetric monoidal, see~\cite[Thm.~8]{JacobsM12a}. Another way is to
say that $\Par(\cat{B})$ is a finitely partially additive category
(abbreviated as `FinPAC'), using the terminology of~\cite{ArbibM1980}
(and~\cite{Cho15}). The original definition involves countable sums,
whereas we only use finite sums, see Subsection~\ref{subsec:FinPAC}.

More results about these partial homsets follow, in
Proposition~\ref{prop:effectusPCMorder}.

\begin{proof}
For an arbitrary map $f\colon X \pto Y$ we can take $b =
\kappa_{1} \pafter f \colon X \pto Y+Y$ as bound for $f
\orthogonal \zero$, showing $f \ovee \zero = f$. If $f \orthogonal g$ via $b
\colon X \pto Y+Y$, then swapping the outcomes of $b$ yields a bound
for $g \orthogonal f$, proving $f \ovee g = g \ovee f$. The more
difficult case is associativity.

So assume $f \orthogonal g$ via bound $a$, and $(f \ovee g)
\orthogonal h$ via bound $b$. We then have:
$$\left\{\begin{array}{rcl}
\rhd_{1} \pafter a
& = &
f
\\
\rhd_{2} \pafter a
& = &
g
\\
\nabla \pafter a
& = &
f \ovee g
\end{array}\right.
\qquad\mbox{and}\qquad
\left\{\begin{array}{rcl}
\rhd_{1} \pafter b
& = &
f \ovee g
\\
\rhd_{2} \pafter b
& = &
h
\\
\nabla \pafter b
& = &
(f \ovee g) \ovee h.
\end{array}\right.$$

\noindent Consider the diagram below, where the rectangle is a
pullback in $\Par(\cat{B})$, as instance of the diagram on the right
in~\eqref{diag:effectuspartprojpb}. Here we use that the codiagonal
$\nabla = [\klin{\idmap}, \klin{\idmap}] = \klin{[\idmap,\idmap]}$ is
total in $\Par(\cat{B})$.
$$\xymatrix{
X\ar@/^2ex/[drr]|-{\pafter}^-{b}\ar@/_2ex/[ddr]|-{\pafter}_-{a}
    \ar@{..>}[dr]|-{\pafter}^-{c}
\\
& (Y+Y)+Y\ar[r]|-{\pafter}_-{\nabla\pplus\idmap}\ar[d]|-{\pafter}_{\rhd_{1}}
   \pullback & 
   Y+Y\ar[d]|-{\pafter}^{\rhd_1}
\\
& Y+Y\ar[r]|-{\pafter}_-{\nabla} & Y
}$$

\noindent We now take $c' = (\rhd_{2}\pplus\idmap) \pafter c \colon X
\pto Y+Y$. This $c'$ is a bound for $g$ and $h$, giving
$g\orthogonal h$ since:
$$\begin{array}{rclcrcl}
\rhd_{1} \pafter c'
& = &
\rhd_{1} \pafter (\rhd_{2}\pplus\idmap) \pafter c
& \qquad &
\rhd_{2} \pafter c'
& = &
\rhd_{2} \pafter c 
\\
& = &
\rhd_{2} \pafter \rhd_{1} \pafter c
& &
& = &
\rhd_{2} \pafter (\nabla\pplus\idmap) \pafter c 
\\
& = &
\rhd_{2} \pafter a
& &
& = &
\rhd_{2} \pafter b 
\\
& = &
g
& &
& = &
h.
\end{array}$$

\noindent Next, the map $c'' = [\idmap, \kappa_{2}] \pafter c \colon X
\pto Y+Y$ is a bound for $f$ and $g\ovee h$ since:
$$\begin{array}{rclcrcl}
\rhd_{1} \pafter c''
& = &
[\rhd_{1} \pafter \idmap, \rhd_{1} \pafter \kappa_{2}] \pafter c
& \qquad &
\rhd_{2} \pafter c''
& = &
[\rhd_{2} \pafter \idmap, \rhd_{2} \pafter \kappa_{2}] \pafter c 
\\
& = &
[\rhd_{1}, \zero] \pafter c
& &
& = &
[\rhd_{2}, \idmap] \pafter c
\\
& = &
\rhd_{1} \pafter [\idmap, \zero] \pafter c
& &
& = &
\nabla \pafter (\rhd_{2}\pplus\idmap) \pafter c
\\
& = &
\rhd_{1} \pafter \rhd_{1} \pafter c
& &
& = &
\nabla \pafter c'
\\
& = &
f
& &
& = &
g \ovee h.
\end{array}$$

\noindent We now obtain the associativity of $\ovee$:
$$\begin{array}{rcl}
f\ovee (g\ovee h)
\hspace*{\arraycolsep} = \hspace*{\arraycolsep}
\nabla \pafter c''
\hspace*{\arraycolsep} = \hspace*{\arraycolsep}
\nabla \pafter [\idmap, \kappa_{2}] \pafter c 
& = &
[\nabla, \idmap] \pafter c \\
& = &
\nabla \pafter (\nabla\pplus\idmap) \pafter c \\
& = &
\nabla \pafter b \\
& = &
f \ovee (g\ovee h).
\end{array}$$

Finally we check that $\ovee$ is preserved by pre- and post-composition.
Let $f \orthogonal g$ via bound $b \colon X \pto Y+Y$.
\begin{itemize}
\item For $h\colon Z \pto X$ it is easy to see that $b \pafter h$ is
a bound for $f \pafter h$ and $g \pafter h$, giving $(f \pafter h) \ovee
(g \pafter h) = (f \ovee g) \pafter h$.

\auxproof{
Clearly $\rhd_{1} \pafter b \pafter h = f \pafter h$ and $\rhd_{2} \pafter
b \pafter h = g \pafter h$, and thus:
$$\begin{array}{rcccl}
(f \pafter h) \ovee (g \pafter h)
& = &
\nabla \pafter b \pafter h
& = &
(f \ovee g) \pafter h.
\end{array}$$
}

\item For $k\colon Y \pto Z$ we claim that $(k\pplus k) \pafter b$ is
  a bound for $k \pafter f$ and $k \pafter g$. This follows from
  naturality~\eqref{eqn:partprojnat} of the partial projections
  $\rhd_{i}$. In this way we get $(k \pafter f) \ovee (k \pafter g) =
  k \pafter (f \ovee g)$. \QED

\auxproof{
$$\begin{array}{rcccccl}
(k \pafter f) \ovee (k \pafter g)
& = &
\nabla \pafter (k \pplus k) \pafter b 
& = &
k \pafter \nabla \pafter b
& = &
k \pafter (f \ovee g).
\end{array}$$
}
\end{itemize}
\end{proof}

Here is a concrete illustration: the two partial projections
$\rhd_{1}, \rhd_{2} \colon X+X \pto X$ in an effectus are orthogonal,
via the identity map $X+X \rightarrow X+X$. This is obvious. Hence we
have as sum $\rhd_{1} \ovee \rhd_{2} = \nabla$.

We briefly show what this PCM-structure is in our leading examples.

\begin{example}
\label{ex:PCM}
In the category $\Par(\Sets)$ of partial maps of the effectus $\Sets$
two partial maps $f,g \colon X \tto Y+1$ are orthogonal if they have
disjoint domains of definition. That is:
$$\begin{array}{rcl}
f \orthogonal g
& \Longleftrightarrow &
\allin{x}{X}{f(x) = * \vee g(x) = *}.
\end{array}$$

\noindent In that case we have:
$$\begin{array}{rcl}
(f \ovee g)(x)
& = &
\left\{\begin{array}{ll}
f(x) \quad & \mbox{if } g(x) = * \\
g(x) \quad & \mbox{if } f(x) = * \\
\end{array}\right.
\end{array}$$

In the effectus $\Kl(\Dst)$ for discrete probability, orthogonality of
two partial maps $f, g \colon X \rightarrow \sDst(Y)$ is given by:
$$\begin{array}{rcl}
f \orthogonal g
& \Longleftrightarrow &
\allin{x}{X}{\sum_{y\in Y}f(x)(y) + g(x)(y) \leq 1}.
\end{array}$$

\noindent In that case we have $(f \ovee g)(x)(y) = f(x)(y) + g(x)(y)$.

In the effectus $\op{\OUG}$ of order unit groups two subunital
positive maps $f,g \colon G \rightarrow H$ are orthogonal if $f(1) +
g(1) \leq 1$. In that case we simply have $(f\ovee g)(x) = f(x) +
g(x)$. The same description applies in the effectus $\op{\vNA}$ of von
Neumann algebras.
\end{example}

The PCM-structure on partial maps $X \pto Y$ from
Proposition~\ref{prop:effectusPCM} exists in particular for predicates
--- which are maps of the form $X \pto 1$. In the language of total
maps, a bound for $p,q\colon X \tto 1+1$ is a map $b\colon X \tto
(1+1)+1$ with $\IV \tafter b = p$ and $\XI \tafter b = q$, where $\IV,
\XI \colon (1+1)+1 \rightarrow 1+1$ are the jointly monic
maps~\eqref{diag:effectusjointlymonic} in the definition of an
effectus. The actual sum is then given by $p \ovee q =
(\nabla\tplus\idmap) \tafter b \colon X \tto 1+1$.

In this special case of predicates there is more than PCM-structure:
predicates on $X$ form an effect algebra and also an effect
module. This will be shown next. We begin with the definition.

\begin{definition}
\label{def:EA}
An \emph{effect algebra}\index{S}{effect!-- algebra} is a partial
commutative monoid $(M, \ovee, \zero)$ with an orthosupplement
operation $(-)^{\bot} \colon M \rightarrow M$ such that for each $x\in
M$:
\begin{enumerate}
\item \label{def:EA:bot} $x^{\bot}$\index{N}{$x^{\bot}$,
  orthosupplement in an effect algebra} is the unique element with $x
  \ovee x^{\bot} = \one$, where $\one = \zero^{\bot}$;

\item \label{def:EA:ortho} $x \orthogonal \one$ implies $x = \zero$.
\end{enumerate}

\noindent A homomorphism of effect algebras $f \colon M \rightarrow N$
is a homomorphism of PCMs satisfying $f(1) = 1$. We write
$\EA$\index{N}{cat@$\EA$, category of effect algebras} for the resulting
category.
\end{definition}

The unit interval $[0,1]$ is an example of an effect algebra, with $r
\orthogonal s$ iff $r+s \leq 1$, and in that case $r \ovee s = r+s$.
The orthosupplement is given by $r^{\bot} = 1 - r$. Every Boolean
algebra\index{S}{Boolean!-- algebra} is also an effect algebra with $x
\orthogonal y$ iff $x \wedge y = 0$, and in that case $x \ovee y = x
\vee y$. The orthosupplement is negation $x^{\bot} = \neg x$. In
particular, the set $2 = \{0,1\}$ is an effect algebra. It is initial
in the category $\EA$ of effect algebras.

In~\cite{JacobsM12a} it is shown that the category $\EA$ of effect
algebras is symmetric monoidal, where the unit for the tensor
$\otimes$ is the initial effect algebra $2 = \{0,1\}$.

Each effect algebra carries an order, defined by $x \leq y$ iff $x
\ovee z = y$ for some $z$. We collect some basic results. Proofs
can be found in the literature on effect algebras, see
\textit{e.g.}~\cite{DvurecenskijP00}.

\begin{exercise}
\label{exc:EA}
In an effect algebra one has:
\begin{enumerate}
\item \label{exc:EA:Inv} orthosupplement is an involution: $x^{\bot\bot} = x$;

\item \label{exc:EA:Canc} cancellation:\index{S}{cancellation} $x\ovee
  y = x\ovee y'$ implies $y=y'$;

\item \label{exc:EA:ZeroSum} positivity:\index{S}{positive!-- sum}
  $x\ovee y = \zero$ implies $x=y=\zero$;

\item \label{exc:EA:PO} $\leq$ is a partial order with $\one$ as top and
  $\zero$ as bottom element;

\item \label{exc:EA:OrthoAntiMonotone} $x\leq y$ implies $y^{\bot} \leq x^{\bot}$;

\item \label{exc:EA:Orthogonal} $x\ovee y$ is defined iff $x\orthogonal
  y$ iff $x\leq y^{\bot}$ iff $y\leq x^{\bot}$;

\item \label{exc:EA:SumMonotone} $x\leq y$ and $y \orthogonal z$ implies 
$x \orthogonal z$ and $x\ovee z \leq y\ovee z$;

\item \label{exc:EA:Downset} the downset $\downset y = \set{x}{x \leq
  y}$ is again an effect algebra, with $y$ as top, orthosupplement
  $x^{\bot_y} = (y^{\bot}\ovee x)^{\bot}$, and sum $x\ovee_{y} x'$
  which is defined iff $x \orthogonal x'$ and $x\ovee x' \leq y$, and
  in that case $x \ovee_{y} x' = x \ovee x'$. \QED
\end{enumerate}
\end{exercise}

The uniqueness of the orthosupplement $x^{\bot}$ is an important
property in the proof of this exericse. For instance it gives us
$f(x^{\bot}) = f(x)^{\bot}$ for a map of effect algebras, since:
$$\begin{array}{rcccccccl}
f(x^{\bot}) \ovee f(x)
& = &
f(x^{\bot} \ovee x)
& = &
f(\one)
& = &
\one
& = &
f(x)^{\bot} \ovee f(x).
\end{array}$$

A \emph{test} in an effect algebra is a finite set of elements $x_{1},
\ldots, x_{n}$ which are pairwise orthogonal and satisfy $x_{1} \ovee
\ldots \ovee x_{n} = \one$. Each effect algebra thus gives rise to a
functor (presheaf) $\NNO \rightarrow \Sets$, sending a number $n$ to
the collection of $n$-element tests. This structure is investigated
in~\cite{StatonU15}.

Recall that we write $\Pred(X)$ for the collection of predicates $X
\tto 1+1$ and $\tbox{f} = (-) \tafter f$ for the substitution map
$\Pred(X) \rightarrow \Pred(Y)$, where $f \colon Y \tto X$.

\begin{proposition}
\label{prop:effectusEA}
Let $\cat{B}$ be an effectus. Then:
\begin{enumerate}
\item \label{prop:effectusEA:EA} $\Pred(X)$ is an effect algebra, for
  each object $X\in\cat{B}$;

\item \label{prop:effectusEA:fun} sending $X\mapsto\Pred(X)$ and
  $f\mapsto \tbox{f}$ gives a functor $\Pred \colon \cat{B}
  \rightarrow \op{\EA}$.
\end{enumerate}
\end{proposition}

\begin{proof}
From Proposition~\ref{prop:effectusPCM}~\eqref{prop:effectusPCM:hom}
we know that $\Pred(X)$ is a PCM. We show that it is an effect algebra
with the orthosupplement $p^{\bot} = [\kappa_{2}, \kappa_{1}] \tafter
p$ from~\eqref{diag:orthosupplement}.

The three equations below show that $b = \kappa_{1} \tafter p \colon X
\rightarrow (1+1)+1$ is a bound for $p$ and $p^{\bot}$, yielding
$p\ovee p^{\bot} = \one$.
$$\begin{array}{rcl}
\IV \tafter b
& = &
[\idmap, \kappa_{2}] \tafter \kappa_{1} \tafter p 
\hspace*{\arraycolsep}=\hspace*{\arraycolsep}
p 
\\
\XI \tafter b
& = &
[[\kappa_{2}, \kappa_{1}], \kappa_{2}] \tafter \kappa_{1} \tafter p 
\hspace*{\arraycolsep}=\hspace*{\arraycolsep}
[\kappa_{2}, \kappa_{1}] \tafter p 
\hspace*{\arraycolsep}=\hspace*{\arraycolsep}
p^{\bot} 
\\
(\nabla\tplus\idmap) \tafter b 
& = &
(\nabla\tplus\idmap) \tafter \kappa_{1} \tafter p 
\hspace*{\arraycolsep}=\hspace*{\arraycolsep}
\kappa_{1} \tafter \nabla \tafter p
\hspace*{\arraycolsep}=\hspace*{\arraycolsep}
\kappa_{1}
\hspace*{\arraycolsep}=\hspace*{\arraycolsep}
\one.
\end{array}$$

\noindent Next we show that $p^{\bot}$ is the only predicate with
$p\ovee p^{\bot} = \one$. So assume also for $q\colon X \rightarrow 1+1$
we have $p\ovee q = \one$, say via bound $b$. Then $\IV \tafter b = p$,
$\XI \tafter b = q$ and $(\nabla\tplus\idmap) \tafter b = \one =
\kappa_{1} \tafter \bang$. We use the pullback on the right
in~\eqref{diag:effecutsderivedpb}:
$$\xymatrix{
X\ar@/_2ex/[ddr]|-{\tafter}_{b}\ar@/^2ex/[drr]|-{\tafter}^-{\bang}
   \ar@{..>}[dr]|-{\tafter}^(0.6){c} 
\\
& 1+1\ar[d]|-{\tafter}_{\kappa_1}\ar[r]|-{\tafter}^-{\nabla}\pullback & 
   1\ar[d]|-{\tafter}^{\kappa_1}
\\
& (1+1)+1\ar[r]|-{\tafter}_-{\nabla\tplus\idmap} & 1+1
}$$

\noindent The fact that the bound $b$ is of the form $\kappa_{1}\tafter
c$ is enough to obtain $q = p^{\bot}$ in:
$$\begin{array}{rcl}
p^{\bot}
\hspace*{\arraycolsep} = \hspace*{\arraycolsep}
[\kappa_{2}, \kappa_{1}] \tafter \IV \tafter b
& = &
[\kappa_{2}, \kappa_{1}] \tafter [\idmap,\kappa_{2}] \tafter \kappa_{1} 
   \tafter c \\
& = &
[[\kappa_{2}, \kappa_{1}], \kappa_{2}] \tafter \kappa_{1} \tafter c 
\hspace*{\arraycolsep} = \hspace*{\arraycolsep}
\XI \tafter b
\hspace*{\arraycolsep} = \hspace*{\arraycolsep}
q.
\end{array}$$

\noindent Next, suppose $\one \orthogonal p$; we must prove $p=\zero
\colon X \tto 1+1$. We may assume a bound $b\colon X \tto (1+1)+1$
with $\rhd_{1} \pafter b = \one = \kappa_{1} \tafter \bang$ and
$\rhd_{2} \pafter b = p$. We use the pullback in $\cat{B}$ on the
right in~\eqref{diag:effecutsderivedpb} to obtain the unique map $X
\tto 1$ as mediating map in:
$$\xymatrix{
X\ar[dd]|-{\tafter}_{b}\ar@/^2ex/[drr]|-{\tafter}^-{\bang}
   \ar@{..>}[dr]|-{\tafter}_-{\bang}
\\
& 1\ar@{=}[r]\ar@{ >->}[d]|-{\tafter}_{\kappa_1}\pullback & 
   1\ar@{ >->}[d]|-{\tafter}_{\kappa_1}
\\
(1+1)+1\ar[r]|-{\tafter}^-{\alpha^{-1}}_-{\cong}
   \ar@/_4ex/[rr]|-{\tafter}_-{\IV = [\rhd_{1}, \kappa_{2}]} & 
   1+(1+1)\ar[r]|-{\tafter}^-{\idmap\tplus\nabla} & 1+1
}$$

\noindent As before, $\alpha$ is the associativity isomorphism. Hence
$b = \alpha \tafter \kappa_{1} \tafter \bang = \kappa_{1} \tafter
\kappa_{1} \tafter \bang = \klin{\one}$. Then, as required:
$$\begin{array}{rcccccccccl}
p
& = &
\rhd_{2} \pafter b
& = &
[\zero, \idmap] \pafter \klin{\one}
& = &
[\kappa_{2}, \idmap] \tafter \kappa_{1} \tafter \bang
& = &
\kappa_{2} \tafter \bang
& = &
\zero.
\end{array}$$

For the second point of the proposition we only have to prove that
substitution $\tbox{f}$ preserves $\ovee$ and $\one$. Preservation of
$\ovee$ follows from
Proposition~\ref{prop:effectusPCM}~\eqref{prop:effectusPCM:pres}, and
preservation of $\one$ holds by
Exercise~\ref{exc:subst}~\eqref{exc:subst:totalpres}. \QED
\end{proof}

In Example~\ref{ex:EMod} below we shall describe this effect algebra
structure concretely for our running examples. But we first we show
that the result can be strengthened further: predicates in an effectus
are not only effect \emph{algebras}, but also effect
\emph{modules}. The latter are effect algebras with a scalar
multiplication, a bit like in vector spaces. In examples the scalars
are often probabilities from $[0,1]$ or Booleans from $\{0,1\}$. But
more abstractly they are characterised as effect monoids.

\begin{definition}
\label{def:EMon}
An \emph{effect monoid}\index{S}{effect!--
  monoid}\index{S}{monoid!effect --} is an effect algebra which is at
the same time a monoid, in a coherent way: it is given by an effect
algebra $M$ with (total) associative multiplication operation $\cdot
\colon M\times M \rightarrow M$ which preserves $\zero,\ovee$ in each
coordinate separately, like in the second part of
Definition~\ref{def:pcm}, and satisfies $\one\cdot x = x = x\cdot
\one$.

An effect monoid is commutative if its multiplication is commutative.
\end{definition}

More abstractly, using that the category $\EA$ of effect algebras is
symmetric monoidal~\cite{JacobsM12a}, an effect monoid $(M, \cdot,
\one)$ is a monoid in the monoidal category $\EA$ of effect algebras
of the form:
$$\xymatrix{
2\ar[r] & M & M\otimes M\ar[l]
}$$

\noindent satisfying the monoid requirements. Since the tensor unit
$2=\{0,1\}$ is initial in the category $\EA$, the map on the left is
uniquely determined, and the multiplication map on the right is the
only structure.

The unit interval $[0,1]$ is the prime example of an effect monoid,
via its standard multiplication. It is clearly commutative. This
structure can be extended pointwise to fuzzy predicates $[0,1]^{X}$,
for a set $X$, and to continuous predicates $\Cont(Y, [0,1])$ for a
topological space $Y$.

The Booleans $\{0,1\}$ also form an effect monoid, via conjunction
(multiplication). More generally, each Boolean algebra is an effect
monoid, with conjunction as multiplication.

Recall that any order unit group gives rise to an effect algebra by
considering its unit interval.  If the order unit group carries an
associative bilinear positive multiplication for which the order unit
is neutral, then its unit interval is an effect monoid. We should warn
that this does not apply to $C^*$-algebras: their self-adjoint
elements form an order unit group, but their multiplication is not
positive and does not restrict to self-adjoint elements --- except
when the $C^*$-algebra is commutative. Indeed, we shall see in
Section~\ref{sec:commbool} that predicates in a `commutative' effectus
form commutative effect monoids, see
esp.\ Lemma~\ref{lem:commeffectus}~\eqref{lem:commeffectus:EAmap}.

Most obvious examples of an effect monoid are commutative. But here is
an example of a non-commutative effect monoid.  Consider~$\R^5$ with
standard basis~$e_1, \ldots, e_5$ ordered lexicographically such
that~$e_1 \gg e_2 \gg \ldots \gg e_5$, where~$v \gg w$ denotes~$v \geq
\lambda w$ for any~$\lambda > 0$.  There is a unique associative
bilinear positive product $*$ fixed by~$e_1 * e_j = e_j * e_1 = e_j$,
$e_2 * e_2 = e_4$, $e_3 * e_2 = e_5$ and in the remaining cases~$e_i *
e_j= 0$.  For details, see~\cite{Westerbaan13}.

Scalars in an effectus are predicates $1 \tto 1+1$ on the final object
$1$.  They form an effect algebra by
Proposition~\ref{prop:effectusEA}~\eqref{prop:effectusEA:EA}.  When we
view scalars as partial maps $1 \pto 1$, we directly see that they
also carry a monoid structure, namely Kleisli/partial composition
$\pafter$. The latter preserves the sums $\ovee, 0$ in each coordinate
by Proposition~\ref{prop:effectusPCM}~\eqref{prop:effectusPCM:pres}.
Since the scalar $\one$ is the first coprojection $\kappa_{1} \colon 1
\tto 1+1$, it is the identity map on $1$ in the category of partial
maps, and thus the unit for $\pafter$.  We now summarise the
situation.

\begin{definition}
\label{def:effectusScalars}
For an arbitrary effectus $\cat{B}$ with final object $1$, we write
$\Pred(1) = \Par(\cat{B})(1,1)$ for the effect monoid of scalars with
partial composition $\pafter$.
\end{definition}

Once we know what scalars are, we can define associated modules
having scalar multiplication.

\begin{definition}
\label{def:EMod}
Let $M$ be an effect monoid. 
\begin{enumerate}
\item An \emph{effect module over $M$}\index{S}{effect!-- module} is
  an effect algebra $E$ with a scalar 
multiplication\index{S}{scalar multiplication!-- in an effect module} 
(action) $M\otimes E \rightarrow E$ in $\EA$. In such a module we have
elements $r\cdot e\in E$, for $r\in M$ and $e\in E$, satisfying, apart
from preservation of $0,\ovee$ in each coordinate, $\one\cdot e = e$
and $r\cdot (s\cdot e) = (r\cdot s) \cdot e$.

\item A map of effect modules is a map of effect algebras $f\colon E
  \rightarrow D$ satisfying $f(r\cdot e) = r\cdot f(e)$ for each $r\in
  M$, $e\in E$.  We write $\EMod_{M}$\index{N}{cat@$\EMod_{M}$, category
    of effect modules over an effect monoid $M$} for the category of
  effect modules over $M$, with effect module maps as morphisms
  between them.
\end{enumerate}
\end{definition}

We can also define effect modules more abstractly: an effect monoid
$M$ is a monoid $M\in\EA$ in the symmetric monoidal category of effect
algebras.  Hence the functor $M \otimes (-)$ is a monad on $\EA$. The
category $\EMod_M$ of effect modules over $M$ is the resulting
category of Eilenberg-Moore algebras of this monad, \ie~the category
of actions of the monoid $M$. There is an obvious forgetful functor
$\EMod_{M} \rightarrow \EA$.

Effect modules over the Booleans $\{0,1\}$ are just effect
algebras. Almost always in examples we encounter effect modules over
the probabilities $[0,1]$. That's why we often simply write $\EMod$
for $\EMod_{[0,1]}$. In the context of effect modules, the elements of
the underlying effect monoid are often called scalars.

For each set $X$, the collection $[0,1]^{X}$ of `fuzzy predicates' on
$X$ is an effect algebra, by Proposition~\ref{prop:effectusEA}, where
$p \orthogonal q$ iff $p(x)+q(x)\leq 1$ for all $x\in X$, and in that
case $(p\ovee q)(x) = p(x) + q(x)$. The orthosupplement $p^\bot$ is
defined by $p^{\bot}(x) = 1- p(x) = p(x)^{\bot}$. This effect algebra
$[0,1]^{X}$ is an effect module over $[0,1]$, via scalar
multiplication $r\cdot p$ defined as $(r\cdot p)(x) = r\cdot
p(x)$. This mapping $X \mapsto [0,1]^{X}$ yields a functor $\Sets
\rightarrow \op{\EMod}$.

We can now strengthen Proposition~\ref{prop:effectusEA} in the
following way.

\begin{theorem}[From~\cite{Jacobs15a}]
\label{thm:effectusEMod}
For each effectus $\cat{B}$ the mapping $X \mapsto \Pred(X)$ yields a
\emph{predicate}\index{S}{predicate!-- functor} functor:
$$\xymatrix@C+1pc{
\cat{B}\ar[r]^-{\Pred} & {\op{\Big(\EMod_{M}\Big)}}
   \qquad\mbox{where}\qquad M = \Pred(1).
}$$

\noindent This functor $\Pred$ preserves finite coproducts and the
final object.
\end{theorem}

This theorem says that predicates in an effectus form effect modules
over the effect monoid $\Pred(1)$ of scalars in the effectus.

\begin{proof}
For a predicate $p\colon X \pto 1$ on $X$ and a scalar $s\colon 1\pto
1$ we define scalar multiplication simply as partial composition $s
\pafter p$. The PCM-structure $\ovee, 0$ is preserved in each
coordinate by
Proposition~\ref{prop:effectusPCM}~\eqref{prop:effectusPCM:pres}.  We
have $r \pafter (s \pafter p) = (r\pafter s) \pafter p$ by
associativity of partial composition, and $\one \pafter p = p$ because
$\one = \kappa_{1} \colon 1 \tto 1+1$ is the identity on $1$ for
partial composition $\pafter$.

For a (total) map $f\colon Y \tto X$ in $\cat{B}$ we see that
substitution $\tbox{f}$ preserves scalar multiplication:
$$\begin{array}{rcccccccl}
\tbox{f}(s \pafter p)
& = &
(s \pafter p) \tafter f
& = &
(s \pafter p) \pafter \klin{f}
& = &
s \pafter (p \pafter \klin{f})
& = &
s \pafter \tbox{f}(p).
\end{array}$$

\noindent We still have to prove that the functor $\Pred \colon
\cat{B} \rightarrow \op{(\EMod_{M})}$ preserves finite coproducts and
the final object. This means that it sends coproducts in $\cat{B}$ to
products in $\EMod_M$, and the final object to the initial one in
$\EMod_{M}$.
\begin{itemize}
\item There is precisely one predicate $0 \rightarrow 1+1$, so
$\Pred(0)$ is a singleton, which is final in $\EMod_{M}$.

\item The scalars $M = \Pred(1)$ are initial in $\EMod_{M}$.

\item For objects $X,Y\in\cat{B}$, there is an isomorphism of effect
  modules:
$$\xymatrix{
\Pred(X+Y)\ar@/^2ex/[rr]^-{\tuple{\tbox{\kappa_{1}}, \tbox{\kappa_{2}}}}
   & \cong &
\Pred(X)\times\rlap{$\,\Pred(Y)$}\ar@/^2ex/[ll]^-{[-,-]}
}\qquad\eqno{\QEDbox}$$
\end{itemize}
\end{proof}

With this final object $M$ and these finite coproducts the category
$\op{(\EMod_{M})}$ is an effectus. Its scalars are precisely the
elements of the effect monoid $M$.

\begin{example}
\label{ex:EMod}
We briefly review what this theorem means concretely for the
running examples from Subsection~\ref{subsec:effectusEx}.
\begin{enumerate}
\item The scalars in the effects $\Sets$ of sets and functions are the
  Booleans $2 = \{0,1\}$, see Example~\ref{ex:effectusSets}. Since
  effect modules over $2$ are just effect algebras, the predicate
  functor takes the form $\Pred \colon \Sets \rightarrow \op{\EA}$. It
  sends a set $X$ to the effect algebra $\Pow(X)$ of subsets of
  $X$. As is well-known, these predicates $\Pow(X)$ form a Boolean
  algebra, but that is a consequence of the fact that $\Sets$ is a
  special, `Boolean' effectus, see
  Proposition~\ref{prop:booleffectus}. For a function $f \colon Y \to
  X$, the associated substitution functor $\tbox{f} \colon \Pow(X)
  \rightarrow \Pow(Y)$ is inverse image:
$$\begin{array}{rcl}
\tbox{f}(U)
& = &
\setin{y}{Y}{f(y)\in U}.
\end{array}$$

\item We recall from Example~\ref{ex:effectusKlD} that scalars in the
  Kleisli category $\Kl(\Dst)$ are the probabilities $[0,1]$, and
  predicates on a set $X$ are functions $X \rightarrow [0,1]$. Indeed,
  as we saw before Theorem~\ref{thm:effectusEMod}, the set of fuzzy
  predicates $[0,1]^{X}$ is an effect module over $[0,1]$. For a map
  $f\colon Y \to X$ in $\Kl(\Dst)$, that is, for a function $f\colon
  Y \rightarrow \Dst(X)$, the associated substitution map $\tbox{f}
  \colon [0,1]^{X} \rightarrow [0,1]^{Y}$ is defined as:
$$\begin{array}{rcl}
\tbox{f}(p)(y)
& = &
\sum_{x\in X} f(y)(x) \cdot p(x).
\end{array}$$

\noindent This map $\tbox{f}$ is indeed a homomorphism of effect
modules. Hence we have a predicate functor $\Kl(\Dst) \rightarrow
\op{\EMod}$.

\item The scalars in the effectus $\op{\OUG}$ of order unit groups are
  the Booleans $2 = \{0,1\}$, and the predicates on an order unit
  group $G$ are the effects in the unit interval, $[0,1]_{G} =
  \setin{x}{G}{0 \leq x \leq 1}$, see Example~\ref{ex:effectusOUG}.
  For a map $f\colon G \tto H$ in $\op{\OUG}$, that is, for a
  homomorphism of order unit groups $f\colon H \rightarrow G$ the
  substitution function $\tbox{f} \colon [0,1]_{H} \rightarrow
  [0,1]_{G}$ is given simply by function application:
$$\begin{array}{rcl}
\tbox{f}(x)
& = &
f(x).
\end{array}$$

\noindent This is well-defined since $f$ is positive, so that $f(x)
\geq 0$, and unital so that $f(x) \leq f(1) = 1$. In this case we have
a predicate functor $\op{\OUG} \rightarrow \op{\EA}$.

For the special case of order unit \emph{spaces}, the scalars in the
relevant effectus $\op{\OUS}$ are the probabilities $[0,1]$. In this
case the predicate functor is of the form $\Pred \colon \op{\OUS}
\rightarrow \op{\EMod}$. It is full and faithful.

\item The situation is similar for the effectus $\op{\vNA}$ of von
  Neumann algebras, see Example~\ref{ex:effectusNA}: the scalars are
  the probabilities $[0,1]$, and the predicates on a von Neumann
  algebra $\mathscr{A}$ are the effects in $[0,1]_{\mathscr{A}} =
  \setin{a}{\mathscr{A}}{0 \leq a \leq 1}$. This is an effect module since
  $r\cdot a \in [0,1]_{\mathscr{A}}$ for $r\in [0,1]$ and $a\in
        [0,1]_{\mathscr{A}}$. Substitution $\tbox{f}$ works, like
        for order unit groups, via function application:
$$\begin{array}{rcl}
\tbox{f}(a)
& = &
f(a).
\end{array}$$

\noindent We thus obtain a predicate functor $\op{\vNA} \rightarrow
\op{\EMod}$. It is also full and faithful, see~\cite{FurberJ15}.
\end{enumerate}
\end{example}

In (the proof of) Theorem~\ref{thm:effectusEMod} we have seen scalar
multiplication $s \pafter p$ via partial \emph{post}-composition, for
a scalar $s \colon 1 \pto 1$ and a predicate $p\colon X \pto 1$. The
same trick can be used for substates $\omega \colon 1 \pto X$ via
partial \emph{pre}-composition. Substates thus form a PCM with scalar
multiplication (which is a map of PCMs in each coordinate). We call
these structure partial commutative modules (PC-modules, for
short). They are organised in a category $\PCMod$ in the obvious
manner. We have to keep in mind that writing partial composition
$\pafter$ in the usual order gives `right' modules, with the scalar
written on the right.

\begin{lemma}
\label{lem:substatePCMod}
The sets $\SStat(X)$ of substates $\omega \colon 1 \pto X$ in an
effectus are partial commutative modules over the effect monoid of
scalars $\Pred(1)$, via $\omega \pafter s$.
\end{lemma}

\begin{proof}
Obviously, this definition $\omega \pafter s$ determines a right
action of the monoid $\Pred(1)$ of scalars on the set $\SStat(X)$ of
substates. It preserves the PCM structure in each coordinate by
Proposition~\ref{prop:effectusPCM}~\eqref{prop:effectusPCM:pres}. \QED
\end{proof}

Later, in Section~\ref{sec:monoidal} we shall see that in the presence
of tensor products $\otimes$ all partial homsets, and not just the
ones of substates, become PC-modules.

This scalar multiplication on 
substates\index{S}{scalar multiplication!-- on substates}
makes it possible to define when a substate is pure. In quantum theory
a state is called pure is it is not a mixture (convex combination) of
other states. In the current setting this takes the following form.

\begin{definition}
\label{def:puresstate}
A non-zero substate $\omega\colon 1 \pto X$ in an effectus is called
\emph{pure}\index{S}{pure state}\index{S}{state!pure --} if for each
pair of orthogonal substates $\omega_{1}, \omega_{2} \colon 1 \pto X$
with $\omega = \omega_{1} \ovee \omega_{2}$ there is a scalar $s$
with:
$$\begin{array}{rclcrcl}
\omega_{1}
& = &
\omega \pafter s
& \qquad\mbox{and}\qquad &
\omega_{2}
& = &
\omega \pafter s^{\bot}.
\end{array}$$
\end{definition}

For instance, subdistributions of the form $r\ket{x} \in \sDst(X)$,
for $r > 0$, are pure: if $r\ket{x} = \omega_{1} \ovee \omega_{2}$,
for $\omega_{1},\omega_{2}\in\sDst(X)$, then $\omega_{1}(y) +
\omega_{2}(y) = (\omega_{1} \ovee \omega_{2})(y) r\ket{x}(y) = 0$ for
$y\neq x$.  Hence $\omega_{1}(y) = \omega_{2}(y) = 0$. This means that
$\omega_{1} = r_{1}\ket{x}$ and $\omega_{2} = r_{2}\ket{x}$ for
certain $r_{1},r_{2} \in [0,1]$ with $r_{1} + r_{2} = r$. Hence we
take $s = \frac{r_{1}}{r} \in [0,1]$. Clearly:
$$\begin{array}{rcccccl}
s\cdot \omega
& = &
\frac{r_{1}}{r}\cdot r \ket{x}
& = &
r_{1}\ket{x}
& = &
\omega_{1}.
\end{array}$$

And:
$$\begin{array}{rcccccccl}
s^{\bot}\cdot \omega
& = &
(1 - \frac{r_{1}}{r})\cdot r \ket{x}
& = &
(r - r_{1})\ket{x}
& = &
r_{2}\ket{x}
& = &
\omega_{2}.
\end{array}$$


\section{State and effect triangles}\label{sec:triangle}

In quantum theory there is a basic duality between states and effects,
see \textit{e.g.}~\cite{HeinosaariZ12}. This duality can be formalised
in categorical terms as an adjunction $\op{\EMod} \leftrightarrows
\Conv$ between the opposite of the category of effect modules and the
category of convex sets. This adjunction will be described in more
detail below. 

The duality between states and effects is related to the different
approaches introduced by two of the founders of quantum theory, namely
Erwin Schr\"odinger and Werner Heisenberg. Schr\"odinger's approach is
state-based and works in a forward direction, whereas Heisenberg's
describes how quantum operations work on effects, in a backward
direction. It turns out that the difference between these two
approaches is closely related to a well-known distinction in the
semantics of computer programs, namely between state transformer
semantics and predicate transformer semantics,
see~\cite{Jacobs15b}. The situation can be described in terms of a
triangle:
\begin{equation}
\label{diag:computationtriangle}
\qquad\vcenter{\xymatrix@C-1.2pc@R-1.5pc{
{\ovalbox{\textbf{Heisenberg}}} & & {\ovalbox{\textbf{Schr\"odinger}}} \\
\llap{$\op{\Cat{Log}}=\;$}{\left(\begin{array}{c} \text{predicate} \\[-.3em]
      \text{transformers} \end{array}\right)}\ar@/^1em/[rr] 
& \top &
{\left(\begin{array}{c} \text{state} \\[-.3em]
      \text{transformers} \end{array}\right)}\ar@/^1em/[ll]  \\
\\
& {\Big(\text{computations}\Big)}\ar[uul]^(0.45){\Pred}\ar[uur]_(0.45){\Stat} &
}}
\end{equation}

\noindent The main result of this section shows that each effectus
gives rise to a such a `state and effect' triangle.

We start with a closer inspection of the structure of states in an
effectus. Recall that a state is a map of the form $\omega \colon 1
\tto X$.  We will show that states are closed under convex
combinations $\sum_{i}r_{i}\omega_{i}$, where the $r_{i}$ are scalars
$1 \tto 1+1$.  We recall from Definition~\ref{def:effectusScalars}
that these scalars form an effect monoid. Hence we need to understand
convexity with respect to such effect monoids, generalising the usual
form of convexity with respect to the effect monoid $[0,1]$ of
probabilities.

\begin{definition}
\label{def:convexset}
Let $M$ be an effect monoid.
\begin{enumerate}
\item A \emph{convex set}\index{S}{convex!-- set} over $M$ is a set
  $X$ with sums of convex combinations with scalars from $M$. More
  precisely, for each $n$-tuple $r_{1}, \ldots, r_{n} \in M$ with
  $\bigovee_{i}r_{i} = 1$ and $n$-tuple $x_{1}, \ldots, x_{n} \in X$
  there is an element $\sum_{i}r_{i}x_{i} \in X$. These convex sums
  must statisfy the following two properties:
$$\begin{array}{rclcrcl}
1x
& = &
x
& \qquad\mbox{and}\qquad &
\sum_{i} r_{i}\big(\sum_{j}s_{ij}x_{ij}\big)
& = &
\sum_{ij} (r_{i}\cdot s_{ij})x_{ij}
\end{array}$$

\item A function $f\colon X \rightarrow Y$ between two convex sets is
  called \emph{affine}\index{S}{affine function} if it preserves sums
  of convex combinations: $f(\sum_{i}r_{i}x_{i}) =
  \sum_{i}r_{i}f(x_{i})$. Convex sets and affine functions between
  them form a category $\Conv_{M}$.\index{N}{cat@$\Conv_{M}$, category
    of convex sets over an effect monoid $M$}
\end{enumerate}
\end{definition}

Convex sets over $M$ can be described more abstractly as an
Eilenberg-Moore algebra of a distribution monad $\Dst_{M}$ defined in
terms of formal convex combinations with scalars from $M$,
see~\cite{Jacobs15a}. This explains the form of the above two
equations.

Convex sets over the unit interval $[0,1]$ have sums of `usual' convex
combinations, with scalars from $[0,1]$. That's why we often simply
write $\Conv$ for the category $\Conv_{[0,1]}$ --- just like $\EMod$
is $\EMod_{[0,1]}$.  Convex sets over the Booleans $\{0,1\}$ are
ordinary sets: in an $n$-tuple $r_{1}, \ldots, r_{n} \in \{0,1\}$ with
$\bigovee_{i}r_{i} = 1$ there is precisely one $i$ with $r_{i} = 1$,
and $r_{j} = 0$ for $j\neq i$. In that case we can define
$\sum_{j}r_{j}x_{j} = x_{i}$.  This works for any set $X$. Hence
$\Conv_{2} \cong \Sets$.

The next result gives a categorical formalisation of the duality
between states and effects in quantum physics. It goes back
to~\cite{Jacobs10a}.

\begin{proposition}
\label{prop:EmodConvAdj}
Let $M$ be an effect monoid. By ``homming into $M$'' one obtains an
adjunction:
$$\xymatrix{
{\op{\big(\EMod_{M}\big)}}\ar@/^1.5ex/[rr]^-{\Hom(-,M)} 
   & \top & \Conv_{M}\qquad\ar@/^1.5ex/[ll]^-{\Hom(-,M)}
}$$
\end{proposition}

\begin{proof}
Given a convex set $X\in\Conv_{M}$, the homset $\Conv(X,M)$ of affine
maps is an effect module, with $f\orthogonal g$ iff $\allin{x}{X}{f(x)
  \orthogonal g(x)}$ in $M$. In that case one defines $(f\ovee g)(x) =
f(x)\ovee g(x)$.  It is easy to see that this is again an affine
function. Similarly, the pointwise scalar product $(r\cdot f)(x) =
r\cdot f(x)$ yields an affine function. This mapping $X \mapsto
\Conv(X,M)$ gives a contravariant functor since for $h\colon
X\rightarrow X'$ in $\Conv_{M}$ pre-composition with $h$ yields a map
$(-) \after h \colon \Conv(X', M) \rightarrow \Conv(X, M)$ of effect
modules.

In the other direction, for an effect module $E\in\EMod_{M}$, the
homset $\EMod(E, M)$ of effect module maps yields a convex set: a
convex sum $f = \sum_{j}r_{j}f_{j}$, where $f_{j} \colon E \rightarrow
M$ in $\EMod_{M}$ and $r_{j}\in M$, can be defined as $f(y) =
\bigovee_{j}\, r_{j}\cdot f_{j}(y)$. This $f$ forms a map of effect
modules. Again, functoriality is obtained via pre-composition.

The dual adjunction between $\EMod_M$ and $\Conv_M$ involves a
bijective correspondence that is obtained by swapping arguments. \QED
\end{proof}

\begin{lemma}
\label{lem:effectusConv}
Let $\cat{B}$ be an effectus, with effect monoid of scalars $\Pred(1)$.
\begin{enumerate}
\item For each object $X$ the set of states $\Stat(X) = \cat{B}(1, X)$
  is a convex set over the effect monoid $\Pred(1)$ of scalars in
  $\cat{B}$.

\item Each (total) map $f\colon X \tto Y$ in $\cat{B}$ gives rise to
  an affine function $\tstat{f} = f \tafter (-) \colon \Stat(X)
  \rightarrow \Stat(Y)$.\index{N}{subst@$\tstat{f}$, state transformer
    associated with total map $f$}
\end{enumerate}

\noindent Thus we have a state\index{S}{state!-- functor} functor:
$$\xymatrix@C+1pc{
\cat{B}\ar[r]^-{\Stat} & \Conv_{M} \qquad\mbox{where}\qquad M=\Pred(1).
}$$
\end{lemma}

\begin{proof}
We first have to prove that convex sums exist in the set of states
$\Stat(X)$. So let $r_{1}, \ldots, r_{n} \in \Pred(1)$ be scalars with
$\bigovee_{i} r_{i} = \one$. This means that there is a bound $b\colon
1 \pto n\cdot 1 = 1 + \cdots + 1$ with $\rhd_{i} \pafter b = r_{i}$
and $\nabla \pafter b = \one \colon 1 \pto 1$, where $\nabla \colon
n\cdot 1 \pto 1$ is the $n$-ary codiagonal. The equation $\nabla
\pafter b = \one$ translates to $(\nabla\tplus\idmap) \tafter b =
\kappa_{1} \tafter \bang \colon 1 \tto 1+1$. Hence $b$ is a total map
$1 \tto n\cdot 1$, of the form $b = \kappa_{1} \tafter s$ in:
$$\xymatrix{
1\ar@/_2ex/[ddr]|-{\tafter}_{b}\ar@/^2ex/[drr]|-{\tafter}^-{\bang}
   \ar@{..>}[dr]|-{\tafter}^-{s}
\\
& n\cdot 1\ar@{ >->}[d]|-{\tafter}_{\kappa_1}
   \ar[r]|-{\tafter}_-{\nabla}\pullback &
  1\ar@{ >->}[d]|-{\tafter}^{\kappa_1}
\\
& n\cdot 1+1\ar[r]|-{\tafter}_-{\nabla\tplus\idmap} & 1+1
}$$

\noindent For an $n$-tuple of states $\omega_{i} \colon 1 \tto X$ we
now define $[\omega_{1}, \ldots, \omega_{n}] \tafter s$ to be the
convex sum $\sum_{i}r_{i}\omega_{i}$.

Post-composition with $f\colon X \rightarrow Y$ is clearly affine,
since:
$$\begin{array}[b]{rcl}
\tstat{f}\big(\sum_{i}r_{i}\omega_{i}\big)
& = &
f \tafter [\omega_{1}, \ldots, \omega_{n}] \tafter s \\
& = &
[f \tafter \omega_{1}, \ldots, f \tafter \omega_{n}] \tafter s
\hspace*{\arraycolsep}=\hspace*{\arraycolsep}
\sum_{i}r_{i}\tstat{f}(\omega_{i}).
\end{array}\eqno{\QEDbox}$$
\end{proof}

When we combine this result with Theorem~\ref{thm:effectusEMod} and
Proposition~\ref{prop:EmodConvAdj} we obtain the following result.

\begin{theorem}
\label{thm:effectusTriangle}
Let $\cat{B}$ be an effectus, with $M = \Pred(1) = \Stat(1+1)$ its
effect monoid of scalars. There is a `state and effect'\index{S}{state
  and effect triangle} triangle of the form:
\begin{equation}
\label{diag:effectustriangle}
\vcenter{\xymatrix{
\op{(\EMod_{M})}\ar@/^1.5ex/[rr]^-{\Hom(-,M)} & \top & 
   \Conv_{M}\ar@/^1.5ex/[ll]^-{\Hom(-,M)} \\
& \cat{B}\ar[ul]^{\Hom(-,1+1)=\Pred\quad}\ar[ur]_{\;\Stat=\Hom(1,-)} &
}} 
\end{equation}

\noindent For a predicate (effect) $p\colon X \tto 1+1$ and a state
$\omega \colon 1 \tto X$ we define the validity $\omega \models p$ as
the scalar obtained by composition:
\begin{equation}
\label{eqn:Born}
\begin{array}{rcl}
(\omega \models p)
& = &
p \after \omega \colon 1 \tto 1+1.
\end{array}\index{N}{$\models$, validity given by the Born rule}
\end{equation}

\noindent We call this abstract definition the \emph{Born
  rule},\index{S}{Born rule} see
Example~\ref{ex:triangle}~\eqref{ex:triangle:vNA} below.

This validity $\models$ satisfies the following Galois
correspondence:
\begin{equation}
\label{eqn:validityGalois}
\begin{array}{rcccl}
\tstat{f}(\omega) \models q
& = &
q \tafter f \tafter \omega
& = &
\omega \models \tbox{f}(q),
\end{array}
\end{equation}

\noindent where $f\colon X \tto Y$, $\omega\colon 1 \tto X$, and $q
\colon Y \tto 1+1$.

This validity relation $\models$ gives rise to two natural
transformations in:
\begin{equation}
\label{diag:effectusValidity}
\vcenter{\xymatrix{
\op{(\EMod_{M})} & & \Conv_{M} \\
& \cat{B}\urtwocell^{\Stat\quad}_{\hspace*{5em}\Hom(\Pred(-), M)}{}
  \ultwocell_{\quad\Pred}^{\Hom(\Stat(-),M)\hspace*{5em}}{}
}}
\end{equation}
\end{theorem}

\begin{proof}
Most of this is a summary of what we have seen before. We concentrate
on the natural transformations in the last part of the theorem. The
natural transformation $\Hom(\Stat(-), M) \Rightarrow \Pred$ consists
of functions $\Pred(X) \rightarrow M^{\Stat(X)}$ in $\EMod_{M}$, given
by $p \mapsto (\omega \mapsto \omega\models p)$. Similarly, the
natural transformation $\Stat \Rightarrow \Hom(\Pred(-), M)$ consists
of affine maps $\Stat(X) \rightarrow M^{\Pred(X)}$ given by $\omega
\mapsto (p \mapsto \omega\models p)$. Naturality of these functions is
given by the equations~\eqref{eqn:validityGalois}. \QED
\end{proof}

Notice that we do not require that the
triangle~\eqref{diag:effectustriangle} commutes (up to isomorphism),
that is, that the natural transformations
in~\eqref{diag:effectusValidity} are isomorphisms. In some examples
they are, in some examples they are not.

\begin{example}
\label{ex:triangle}
We shall review the running examples in the remainder of this section.
\begin{enumerate}
\item \label{ex:triangle:Sets} For the effectus $\Sets$ the
  triangle~\eqref{diag:effectustriangle} takes the following form.
$$\xymatrix{
\op{\EA}\ar@/^1.5ex/[rr]^-{\Hom(-,2)} & \top & 
   \Sets\ar@/^1.5ex/[ll]^-{\Hom(-,2)} \\
& \Sets\ar[ul]^{2^{(-)} = \Pred\;}\ar[ur]_{\;\Stat = \idmap} &
}$$

\noindent Here we use that convex sets over the effect monoid $2 =
\{0,1\}$ of scalars in $\Sets$ are just sets. The predicate functor
$2^{(-)}$ is powerset, see Example~\ref{ex:effectusSets}. For a state
$x\in X$ and a predicate $p\in 2^{X}$ the Born rule~\ref{eqn:Born}
gives an outcome in $\{0,1\}$ determined by membership:
$$\begin{array}{rcl}
x \models p
& = &
p(x).
\end{array}$$

\item \label{ex:triangle:KlD} For the effectus $\Kl(\Dst)$ we have a
  triangle:
$$\xymatrix{
\op{\EMod}\ar@/^1.5ex/[rr]^-{\Hom(-,[0,1])} & \top & 
   \Conv\ar@/^1.5ex/[ll]^-{\Hom(-,[0,1])} \\
& \Kl(\Dst)\ar[ul]^{[0,1]^{(-)} = \Pred\;}\ar[ur]_{\;\Stat} &
}$$

\noindent For a Kleisli map $f\colon X \rightarrow \Dst(Y)$ the
associated state transformer function $\tstat{f} \colon \Dst(X)
\rightarrow \Dst(Y)$ is Kleisli extension, given by:
$$\begin{array}{rcl}
\tstat{f}(\omega)(y)
& = &
\sum_{x} \omega(x) \cdot f(x)(y).
\end{array}$$

\noindent This is `baby' integration. For the Giry monad $\Giry$ it is
proper integration, see~\cite{Jacobs13}. For a state
$\omega\in\Dst(X)$ and a predicate $p\in [0,1]^{X}$ the Born rule
gives the expected value:
$$\begin{array}{rcl}
\omega \models p
& = &
\sum_{x} \omega(x) \cdot p(x) \;\in\; [0,1].
\end{array}$$

\noindent The (Born) validity rules for discrete and continuous
probability have been introduced in~\cite{Kozen81,Kozen85}, in the
context of semantics of probabilistic programs. For continuous
probability, the correspondence~\eqref{eqn:validityGalois} occurs
in~\cite{Jacobs13} for the Giry monad as an equation:
$$\begin{array}{rcccccl}
\tstat{f}(\omega) \models q
& = &
\displaystyle \int q \intd \tstat{f}(\omega)
& = &
\displaystyle \int \tbox{f}(q) \intd \omega
& = &
\omega \models \tbox{f}(q).
\end{array}$$

\noindent Here, $f\colon X \rightarrow \Giry(Y)$ is a measurable
function, $q\colon Y \rightarrow [0,1]$ is a (measurable) predicate on
$Y$ and a maesure/state $\omega\in\Giry(X)$. The operation
$\tstat{f}$ is Kleisli extension and $\tbox{f}$ is substitution.

\item \label{ex:triangle:OUG} The state and effect triangle for the
  effectus $\op{\OUG}$ of order unit groups is:
$$\xymatrix{
\op{\EA}\ar@/^1.5ex/[rr]^-{\Hom(-,2)} & \top & 
   \Sets\ar@/^1.5ex/[ll]^-{\Hom(-,2)} \\
& \op{\OUG}\ar[ul]^{[0,1]_{(-)} = \Pred\quad}\;\ar[ur]_{\quad\Stat = \Hom(-,\Z)} &
}$$

\noindent For a predicate $e\in [0,1]_{G}$ and a state $\omega \colon G
\rightarrow \Z$ validity is:
$$\begin{array}{rcl}
\omega \models e
& = &
\omega(e) \;\in\; \{0,1\}.
\end{array}$$

\item \label{ex:triangle:vNA} The effectus $\op{\vNA}$ of von Neumann
  algebras yields:
$$\xymatrix{
\op{\EMod}\ar@/^1.5ex/[rr]^-{\Hom(-,[0,1])} & \top & 
   \Conv\ar@/^1.5ex/[ll]^-{\Hom(-,[0,1])} \\
& \op{\vNA}\ar[ul]^{[0,1]_{(-)} = \Pred\quad}\ar[ur]_{\quad\Stat=\Hom(-,\C)} &
}$$

\noindent For a predicate $e\in [0,1]_{\mathscr{A}}$ and a state $\omega
\colon \mathscr{A} \rightarrow \C$ one interpretes validity again
as expected probability:
$$\begin{array}{rcl}
\omega \models e
& = &
\omega(e) \;\in\; [0,1].
\end{array}$$

\noindent An interesting special case is $\mathscr{A} =
\B(\mathscr{H})$,\index{N}{$\B(\mathscr{H})$, space of operators on a
  Hilbert space $\mathscr{H}$} where $\mathscr{H}$ is a
 Hilbert space, with associated von Neumann algebra
$\B(\mathscr{H})$ of bounded operators $\mathscr{H}
\rightarrow \mathscr{H}$. It is well-known that (normal) states $\omega \colon
\B(\mathscr{H}) \rightarrow \C$ correspond to density matrices $\rho
\colon \mathscr{H} \rightarrow \mathscr{H}$, via $\omega =
\tr(-\rho)$.\index{N}{$\tr$, trace operation} The Born rule
formula~\eqref{eqn:Born}, for an effect $E \in [0,1]_{\B(\mathscr{H})}
= \set{F\colon \mathscr{H} \rightarrow \mathscr{H}}{0 \leq F \leq
  \idmap}$ becomes the Born rule:
$$\begin{array}{rcl}
\tr(-\rho) \models E
& = &
\tr(E\rho).
\end{array}$$
\end{enumerate}
\end{example}

\begin{remark}
\label{rem:normalisation}
In Theorem~\ref{thm:effectusEMod} we have seen that the predicate
functor $\Pred \colon \cat{B} \rightarrow \op{(\EMod_{M})}$ of an
effectus $\cat{B}$ preserves finite coproducts. The same preservation
property does not hold in general for the states functor $\Stat \colon
\cat{B} \rightarrow \Conv_{M}$ from Lemma~\ref{lem:effectusConv}.  The
matter is investigated in~\cite{JacobsWW15}, for the case where the
effect monoid $M$ is $[0,1]$.

It turns out that preservation of finite coproducts by the states
functor is closely related to \emph{normalisation} of substates.  This
works as follows. 

We say that an effectus $\cat{B}$ satisfies
normalisation\index{S}{normalisation} if for each non-zero substate
$\omega \colon 1 \pto X$ there is a unique state $\rho \colon 1 \tto
X$ with:
\begin{equation}
\label{eqn:normalisation}
\begin{array}{rcl}
\omega
& = &
\klin{\rho} \pafter \one \pafter \omega.
\end{array}
\end{equation}

\noindent This says that the substate $\omega$ is scalar multiplcation
$\klin{\rho}\pafter s$ in the sense of Lemma~\ref{lem:substatePCMod},
where the scalar $s = \one \pafter \omega \colon 1 \pto 1$ is
determined by $\omega$ itself.

We briefly show that the effectuses $\Kl(\Dst)$ and $\op{\vNA}$
satisfy normalisation.
\begin{enumerate}
\item A substate in $\Kl(\Dst)$ is a subdistribution
  $\omega\in\sDst(X)$. Let us assume that it is non-zero, so that the
  associated scalar $s = \one \pafter \omega = \sum_{x}\omega(x) \in
  [0,1]$ is non-zero. We take $\rho = \frac{\omega}{s} = \sum_{x}
  \frac{\omega(x)}{s}\ket{x}$. This $\rho$ is a proper state
  (distribution) since:
$$\begin{array}{rcccl}
\sum_{x}\rho(x)
& = &
\sum_{x}\frac{\omega(x)}{s}
& = &
1.
\end{array}$$

\noindent Moreover, we have $s\cdot \rho = \omega$ by construction.

\item A state in the effectus $\op{\vNA}$ is a positive subunital map
  $\omega \colon \mathscr{A} \rightarrow \C$. If it is non-zero, then
  $s = \omega(1) \in [0,1]$ is a non-zero scalar, so we can define
  $\rho\colon \mathscr{A} \rightarrow \C$ as $\rho(a) =
  \frac{\omega(a)}{s}$. By construction $\rho$ is unital, and
  satisfies $\omega = s\cdot \rho$.
\end{enumerate}
\end{remark}

\section{Effectuses from biproduct categories}\label{sec:biprod}


This section describes a construction that turns a biproduct category
with a suitable `ground' map into an effectus. The construction is
inspired by causal maps in the context of $\CP^*$-categories,
see~\cite{CoeckeHK14}.

A \emph{biproduct category}\index{S}{category!biproduct
  --}\index{S}{biproduct!-- category} is a category with finite
biproducts $(0, \oplus)$.\index{S}{biproduct} This means first of all
that the object $0$ is both initial and final, and thus gives rise to
zero maps $\zero \colon X \rightarrow 0 \rightarrow Y$ between
arbitrary objects $X,Y$. Next, for each pair of objects $X_{1},
X_{2}$, the object $X_{1}\oplus X_{2}$ is both a product, and a
coproduct, with coprojections and projections:
$$\vcenter{\xymatrix{
X_{i}\ar[r]^-{\kappa_i} & X_{1} \oplus X_{2}\ar[r]^-{\pi_j} & X_{j}
}}
\quad
\mbox{satisfying}
\quad
\begin{array}{rcl}
\pi_{j} \after \kappa_{i}
& = &
\left\{\begin{array}{ll}
\idmap \quad & \mbox{if } i = j \\
\zero & \mbox{if } i \neq j
\end{array}\right.
\end{array}$$

\noindent It follows immediately that $\kappa_{1} = \tuple{\pi_{1}
  \after \kappa_{1}, \pi_{2} \after \kappa_{1}} = \tuple{\idmap,
  \zero}$ and $\pi_{1} = [\pi_{1} \after \kappa_{1}, \pi_{1} \after
  \kappa_{2}] = [\idmap, \zero]$, and similarly for $\kappa_{2}$ and
  $\pi_2$. Moreover, each homset of maps $X \rightarrow Y$ is a
  commutative monoid, with sum of maps $f,g\colon X \rightarrow Y$
  given by:
$$\xymatrix@C+1pc{
f+g = \Big(X\ar[r]^-{\Delta = \tuple{\idmap,\idmap}} & 
   X\oplus X\ar[r]^-{f\oplus g} & 
   Y\oplus Y\ar[r]^-{\nabla = [\idmap,\idmap]} & Y\Big)
}$$

\noindent The zero element for this addition $+$ is the zero map
$\zero \colon X \rightarrow Y$.

\begin{definition}
\label{def:biprodground}
We call a category $\cat{A}$ a \emph{grounded biproduct
  category}\index{S}{grounded biproduct
  category}\index{S}{biproduct!grounded --
  category}\index{S}{category!grounded biproduct --} if $\cat{A}$ has
finite biproducts $(\oplus, 0)$ and has a special object $I$ with for
each $X\in\cat{A}$ a `ground' map\index{S}{ground map} $\ground_{X}
\colon X \rightarrow I$\index{N}{$\ground$, ground map} satisfying the
four requirements below. We omit the subscript $X$ in $\ground_{X}$
when it is clear from the context.
\begin{enumerate}
\item \label{def:biprodground:I} the ground on $I$ is the identity:
  $\ground_{I} = \idmap[I] \colon I \rightarrow I$;

\item \label{def:biprodground:coprod} coprojections commute with
  ground: $\big(X_{i} \xrightarrow{\kappa_i} X_{1}\oplus X_{2}
  \xrightarrow{\ground} I\big) = \big(X_{i} \xrightarrow{\ground}
  I\big)$.

\item \label{def:biprodground:zm} ground maps are `zero-monic', that
  is $\ground \after f = \zero$ implies $f=\zero$

\item \label{def:biprodground:canc} `subcausal cancellation' holds: if
  $f + g = f + h = \ground \colon X \rightarrow I$, then $g=h$.
\end{enumerate}

\noindent A map $f\colon X \rightarrow Y$ in $\cat{A}$ is called
\emph{causal}\index{S}{causal map} if $\ground_{Y} \after f =
\ground_{X}$.  We write $\Causal(\cat{A}) \hookrightarrow
\cat{A}$\index{N}{cat@$\Causal(\cat{A})$, category of causal maps} for
the subcategory with causal maps.
\end{definition}

As we shall see in Examples~\ref{ex:biproductSetsKlD}
and~\ref{ex:biproductOUGvNA} below, these ground maps describe a unit
elements, possibly in opposite direction. Causal maps preserve this
units, and may thus be called unital --- or co-unital. The idea of
defining them causal maps in this way occurs in~\cite{CoeckeL13},
building on the causality axiom in~\cite{ChiribellaAP11}. The above
second point says that coprojections are causal. Third point, about
cancellation, fails in the $\CP^*$-category obtained from the category
of relations, where union $\cup$ is used as sum $+$.

We can now state and prove the main result, which was obtained jointly
with Aleks Kissinger.

\begin{theorem}
\label{thm:effectusFromGround}
The category $\Causal(\cat{A})$ of causal maps, for a grounded
biproduct category $\cat{A}$, is an effectus.
\end{theorem}

\begin{proof}
The category $\Causal(\cat{A})$ has coproduct $\oplus$ since the
coprojections $\kappa_{i}$ are causal by definition, and the cotuple
$[f,g]$ is causal if $f\colon X\rightarrow Z$, $g\colon Y\rightarrow
Z$ are causal by requirement~\eqref{def:biprodground:coprod} in
Definition~\ref{def:biprodground}:
$$\begin{array}{rcccccccl}
\ground_{Z} \after [f,g]
& = &
[\ground_{Z} \after f,\ground_{Z} \after g]
& = &
[\ground_{X}, \ground_{Y}]
& \smash{\stackrel{\eqref{def:biprodground:coprod}}{=}} &
[\ground_{X+Y} \after \kappa_{1}, \ground_{X+Y} \after \kappa_{2}]
& = &
\ground_{X+Y}.
\end{array}$$

\noindent The zero object $0$ in $\cat{A}$ is initial in
$\Causal(\cat{A})$ since $\bang_{X} \colon 0 \rightarrow X$ is
causal. We have $\ground_{X} \after \bang_{X} = \bang_{I} =
\ground_{0}$. The object $I\in\cat{A}$ is final in $\Causal(\cat{A})$,
since: the ground map $\ground_{X} \colon X \rightarrow I$ is causal
because $\ground_{I} \after \ground_{X} = \idmap[I] \after \ground_{X}
= \ground_{X}$. Further, any causal map $f\colon X \rightarrow I$
satisfies $f = \idmap[I] \after f = \ground_{I} \after f =
\ground_{X}$.

We have to prove that the two diagrams from
Definition~\ref{def:effectus}~\eqref{def:effectus:pb} are pullbacks in
$\Causal(\cat{A})$.
$$\vcenter{\xymatrix@R-.5pc{
X\oplus Y\ar[r]^-{\idmap\oplus \ground}\ar[d]_{\ground\oplus\idmap} & 
  X\oplus I\ar[d]^{\ground\oplus\idmap}
& &
X\ar[r]^-{\ground}\ar[d]_{\kappa_1} & I\ar[d]^{\kappa_1}  
\\
I\oplus Y\ar[r]_-{\idmap\oplus\ground} & I\oplus I
& &
X\oplus Y\ar[r]_-{\ground\oplus\ground} & I\oplus I
}}$$

\noindent For the diagram on the left let $A\in\cat{A}$ be an object
with causal maps $f \colon A \rightarrow X\oplus I$ and $g \colon A
\rightarrow I\oplus Y$ satisfying $(\ground\oplus\idmap) \after f =
(\idmap\oplus\ground) \after g$. There is an obvious, unique mediating
map $\tuple{\pi_{1} \after f, \pi_{2} \after g} \colon A \rightarrow
X\oplus Y$. We only have to prove that it is causal. First notice
that:
$$\begin{array}{rcccccccl}
\ground_{Y} \after \pi_{1} \after f
& = &
\pi_{1} \after (\ground\oplus\idmap) \after f
& = &
\pi_{1} \after (\idmap\oplus\ground) \after g
& = &
\pi_{1} \after g
& = &
\ground_{I} \after \pi_{1} \after g.
\end{array}$$

\noindent Now we have that:
$$\begin{array}{rcl}
\ground_{X\oplus Y}
\hspace*{\arraycolsep}=\hspace*{\arraycolsep}
\ground_{X\oplus Y} \after [\kappa_{1}, \kappa_{2}]
& = &
[\ground_{X\oplus Y} \after \kappa_{1}, \ground_{X\oplus Y} \after \kappa_{2}] \\
& = &
[\ground_{X}, \ground_{Y}]
\hspace*{\arraycolsep}=\hspace*{\arraycolsep}
\nabla \after (\ground_{X}\oplus\ground_{Y}).
\end{array}\eqno{(*)}$$

\noindent Hence:
$$\begin{array}{rcl}
\ground \after \tuple{\pi_{1} \after f, \pi_{2} \after g}
& \smash{\stackrel{(*)}{=}} &
\nabla \after (\ground\oplus\ground) \after 
   \tuple{\pi_{1} \after f, \pi_{2} \after g} \\
& = &
\nabla \after \tuple{\ground \after \pi_{1} \after f, 
   \ground \after \pi_{2} \after g} \\
& = &
\nabla \after \tuple{\ground \after \pi_{1} \after g, 
   \ground \after \pi_{2} \after g}
   \qquad\qquad \mbox{as shown above} \\
& = &
\nabla \after (\ground\oplus\ground) \after 
   \tuple{\pi_{1} \after g, \pi_{2} \after g} \\
& \smash{\stackrel{(*)}{=}} &
\ground \after g \\
& = &
\ground.
\end{array}$$

\noindent The $\tuple{\pi_{1} \after f, \pi_{2} \after g}$ is causal,
and the diagram on the left is a pullback.

For the above diagram on the right, let $f\colon A \rightarrow X\oplus
Y$ be a causal map satisfying $(\ground\oplus\ground) \after f =
\kappa_{1} \after \ground$. The obvious mediating map is $\pi_{1}
\after f \colon A \rightarrow X$. It is causal since $\ground \after
\pi_{1} \after f = \pi_{1} \after (\ground\oplus\ground) \after f =
\pi_{1} \after \kappa_{1} \after \ground = \ground$. We obtain
$\pi_{2} \after f = \zero$ via
Definition~\ref{def:biprodground}~\eqref{def:biprodground:zm} from:
$$\begin{array}{rccccccccl}
\ground \after \pi_{2} \after f
& = &
\pi_{2} \after (\ground\oplus\ground) \after f
& = &
\pi_{2} \after \kappa_{1} \after \ground
& = &
\zero \after \ground
& = &
\zero.
\end{array}$$

\noindent We now get:
$$\begin{array}{rcccccccl}
f
& = & 
\tuple{\pi_{1} \after f, \pi_{2} \after f}
& = &
\tuple{\pi_{1} \after f, \zero}
& = &
\tuple{\idmap, \zero} \after \pi_{1} \after f
& = &
\kappa_{1} \after \pi_{1} \after f.
\end{array}$$

\noindent Hence $\pi_{1} \after f \colon A \rightarrow X$ is the 
required unique mediating (causal) map. 

\auxproof{
This map is indeed grounded:
$$\begin{array}{rcccccl}
\ground \after \pi_{1} \after f
& = &
\pi_{1} \after (\ground\oplus\ground) \after f
& = &
\pi_{1} \after \kappa_{1} \after \ground
& = &
\ground.
\end{array}$$

\noindent Moreover, it is unique if $g\colon A \rightarrow X$
satisfies $\kappa_{1} \after g = f$, then $\pi_{1} \after f = \pi_{1}
\after \kappa_{1} \after g = g$.
}

Finally we have to prove that the following two maps $\IV$ and $\XI$ are
jointly monic, in:
$$\vcenter{\xymatrix@C-.5pc{
(I\oplus I)\oplus I\ar@<+.5ex>[r]^-{\IV}\ar@<-.5ex>[r]_-{\XI} & I\oplus I
}}
\;\mbox{ where }\;
\left\{\begin{array}{l}
\IV 
=
[\idmap, \kappa_{2}]
=
\tuple{\pi_{1} \after \pi_{1}, (\pi_{2} \after \pi_{1}) + \pi_{2}} \\
\XI 
=
[[\kappa_{2}, \kappa_{1}], \kappa_{2}]
=
\tuple{\pi_{2} \after \pi_{1}, (\pi_{1} \after \pi_{1}) + \pi_{2}}.
\end{array}\right.$$

\noindent Let $f,g\colon X \rightarrow (I\oplus I) \oplus I$ be causal
maps satisfing $\IV \after f = \IV \after g$ and $\XI \after f = \XI
\after g$. Then $\pi_{1} \after f = \pi_{1} \after g$ since:
$$\begin{array}{rcccccl}
\pi_{1} \after \pi_{1} \after f
& = &
\pi_{1} \after \IV \after f
& = &
\pi_{1} \after \IV \after g
& = &
\pi_{1} \after \pi_{1} \after g
\\
\pi_{2} \after \pi_{1} \after f
& = &
\pi_{1} \after \XI \after f
& = &
\pi_{1} \after \XI \after g
& = &
\pi_{2} \after \pi_{1} \after g.
\end{array}$$

\noindent Hence it suffices to show $\pi_{2} \after f = \pi_{2} \after
g \colon X \rightarrow I$. For this we use subcausal cancellation from
Definition~\ref{def:biprodground}~\eqref{def:biprodground:canc}:
$$\begin{array}{rcl}
(\nabla \after \pi_{1} \after f) + (\pi_{2} \after f)
& = &
\nabla \after \tuple{\nabla \after \pi_{1} \after f,\pi_{2} \after f} \\
& = &
\nabla \after (\nabla\oplus\idmap[I]) \after f \\
& = &
\ground_{I} \after \nabla \after (\nabla\oplus\idmap[I]) \after f \qquad
   \mbox{since $\ground_{I}=\idmap[I]$}\\
& = &
\ground_{X},
\end{array}$$

\noindent where the last equality holds since $f$ is causal, and
causal maps are closed under composition, cotuple and coproduct.
Similarly we have $(\nabla \after \pi_{1} \after g) + (\pi_{2} \after
g)=\ground_{X}$.  By $\pi_{1} \after f = \pi_{1} \after g$ and the
cancellation in
Definition~\ref{def:biprodground}~\eqref{def:biprodground:canc} we
obtain $\pi_{2} \after f = \pi_{2} \after g$. Hence $f = g$, showing
that $\IV$ and $\XI$ are jointly monic. \QED
\end{proof}

In the remainder of this section we show how each of our four running
example effectuses can be understood as category of causal maps in a
grounded biproduct category.

\begin{example}
\label{ex:biproductSetsKlD}
Let $S$ be a semiring which is positive and cancellative, that is, it
satisfies $x+y = 0 \Rightarrow x=y=0$ and $x+y = x+z \Rightarrow y=z$.
We write $\Mlt_{S} \colon \Sets \rightarrow
\Sets$\index{N}{fun@$\Mlt$, multiset monad} for the multiset monad
with scalars from $S$. Thus, elements of $\Mlt_{S}(X)$ are finite
formal sums $\sum_{i}s_{i}\ket{x_i}$ with $s_{i}\in S$ and $x_{i}\in
X$. It is easy to see that $\Mlt_{S}$ is a monad, with unit $\eta(x) =
1\ket{x}$. It is an `additive' monad (see~\cite{CoumansJ13}), since
$\Mlt_{S}(0) \cong 1$ and $\Mlt_{S}(X+Y) \cong
\Mlt_{S}(X)\times\Mlt_{S}(Y)$. The Kleisli category $\Kl(\Mlt_{S})$
then has finite biproducts $(+, 0)$.

We claim that $\Kl(\Mlt_{S})$ is a grounded biproduct category. We
take $I = 1$ and use that $\Mlt_{S}(1) \cong S$. We thus take as map
$\ground \colon X \rightarrow 1$ in $\Kl(\Mlt_{S})$ the function $X
\rightarrow S$ given by $\ground(x) = 1 \in S$ for all $x\in X$.  We
note that each map $f$ in $\Kl(\Mlt_{S})$, of the form $f = \eta
\after g$ for $g\colon X \rightarrow Y$ in $\Sets$, is causal, since
in $\Kl(\Mlt_{S})$:
$$\begin{array}{rcl}
\big(\ground \after f\big)(x)
& = &
\big(\mu \after \Mlt_{S}(\ground) \after \eta \after g\big)(x) \\
& = &
\big(\mu \after \eta \after \ground \after g\big)(x) 
\hspace*{\arraycolsep} = \hspace*{\arraycolsep}
\ground(g(x)) 
\hspace*{\arraycolsep} = \hspace*{\arraycolsep}
1
\hspace*{\arraycolsep} = \hspace*{\arraycolsep}
\ground(x).
\end{array}$$

\noindent We now briefly check that the ground maps $\ground$ in
$\Kl(\Mlt_{S})$ satisfy the four requirements from
Definition~\ref{def:biprodground}.
\begin{enumerate}
\item The ground map $\ground \colon 1 \rightarrow 1$ in $\Kl(\Mlt_{S})$
  is $\ground(*) = 1$, which is the unit $\eta$ of the monad $\Mlt_{S}$,
  and thus the identity in $\Kl(\Mlt_{S})$.

\item The coprojection $\kappa_{1} \colon X \rightarrow X+Y$ in
  $\Kl(\Mlt_{S})$ is $\eta \after k_1$, where, for the moment, we
  write $k_1$ for the coprojection in $\Sets$. Hence it is causal, as
  noted above.

\item If $\ground \after f = 0$ in $\Kl(\Mlt_{S})$, for a Kleisli map
  $f\colon X \rightarrow Y$, then $\sum_{y} f(x)(y)\cdot \ground(y) =
  0$ for each $x\in X$.  Since $S$ is positive, we get $f(x)(y) = 0$,
  for each $y\in Y$. But then $f(x) = 0 \in \Mlt_{S}(X)$, and thus $f
  = \zero$.

\item If $f + g = f + h = \ground$, for $f, g, h\colon 1\rightarrow
  \Mlt_{S}(X)$, then we may identify $f,g,h$ with multisets in
  $\Mlt_{S}(X)$ that satisfy $f(x) + g(x) = f(x) + h(x) = 1$, for each
  $x\in X$. But then $g(x) = h(x)$ by cancellation in $S$, for each
  $x\in X$, and thus $g = h$.
\end{enumerate}

\noindent We now consider two special choices for the semiring $S$.
\begin{itemize}
\item By taking $S = \NNO$ we obtain the category $\Sets$ as the
  effectus of causal maps in $\Kl(\Mlt_{\NNO})$. Indeed, maps $f\colon
  X \rightarrow \Mlt_{\NNO}(Y)$ with $\sum_{y} f(x)(y) = 1$, for each
  $x\in X$ are determined by a unique $y\in Y$ with $f(x)(y) =
  1$. Hence $f$ corresponds to a function $X \rightarrow Y$.

\item Next we take $S = \R_{\geq 0}$, the semiring of non-negative
  real numbers. We now obtain the Kleisli category $\Kl(\Dst)$ as
  effectus of causal maps in the grounded biproduct category
  $\Kl(\Mlt_{\R_{\geq 0}})$, since maps $f\colon X \rightarrow
  \Mlt_{\R_{\geq 0}}(Y)$ with $\sum_{y} f(x)(y) = 1$, for each $x$,
  are precisely the maps $X \rightarrow \Dst(Y)$.
\end{itemize}
\end{example}

\begin{example}
\label{ex:biproductOUGvNA}
It is well-known that the category $\Ab$ of Abelian groups has finite
biproducts $(\oplus, \{0\})$, given by cartesian products. These
biproducts restrict to the category $\OAb$ of ordered Abelian groups
with positive/monotone group homomorphisms (described in
Example~\ref{ex:effectusOUG}). Let's use the \textit{ad hoc} notation
$\underline{\OUG}$ for the category with order unit groups as objects,
but with positive group homomorphisms as maps. Hence we have a full
and faithful functor $\underline{\OUG} \rightarrow \OAb$ since we do
not require that units are preserved. It is not hard to see that
$\underline{\OUG}$ also has biproducts.

We claim that $\op{\underline{\OUG}}$ is a grounded biproduct
category. It has biproducts since they are invariant under taking the
opposite. We take $I = \Z$. For an order unit group $G$ we have to
define a ground map $G \rightarrow \Z$ in $\op{\underline{\OUG}}$.  We
define this function $\ground \colon \Z \rightarrow G$ simply as
$\ground(k) = k\cdot 1 \in G$. We check the four requirements from
Definition~\ref{def:biprodground}.
\begin{enumerate}
\item The ground map $\ground \colon \Z \rightarrow \Z$ is given by
  $\ground(k) = k\cdot 1 = k$ and is thus the identity.

\item We have a commuting diagram in $\underline{\OUG}$:
$$\xymatrix@C+1pc@R-.5pc{
\Z\ar[r]^-{\ground}\ar[dr]_{\ground} & G_{1}\oplus G_{2}\ar[d]^{\pi_i}
\\
& G_{i}
}$$

\noindent since $\pi_{i}(\ground(k)) = \pi_{i}(k\cdot (1,1)) =
\pi_{1}(k\cdot 1, k\cdot 1) = k\cdot 1 = \ground(k)$.

\item If $f \colon G \rightarrow H$ in $\underline{\OUG}$ satisfies $f
  \after \ground = \zero$, then $f(1) = 0$. From this we obtain $f(x)
  = 0$ for an arbitrary $x\in G$ in the following way. Since $G$ is an
  order unit group there is an $n\in\NNO$ with $-n\cdot 1 \leq x \leq
  n\cdot 1$.  But then because $f$ is monotone:
$$\begin{array}{rcccccccccccl}
0
& = &
-n \cdot f(1)
& = &
f(-n\cdot 1)
& \leq &
f(x)
& \leq &
f(n\cdot 1)
& = &
n\cdot f(1)
& = &
0.
\end{array}$$

\noindent Hence $f = \zero$.

\item If $f + g = f + h = \ground \colon \Z \rightarrow G$, then
  $f,g,h$ can be identified with elements of $G$ satisfying $f + g = f
  + h = 1$. By subtracting $f$ on both sides we obtain $g = h$.
\end{enumerate}

\noindent It is now easy to see that $\op{\OUG}$ is the subcategory
$\Causal(\op{\underline{\OUG}})$ and is thus an effectus. Indeed, a
map $f\colon G \rightarrow H$ in $\underline{\OUG}$ with $f \after
\ground = \ground$ is unital: $f(1) = f(\ground(1)) = \ground(1) =
1$. Hence $f$ is a map in the category $\OUG$ of order unit groups.

In the same way one can prove that the opposite $\op{\OUS}$ of the
category $\OUS$ of order unit spaces is an effectus. Since there is an
equivalence of category $\OUS \simeq \EMod$, see~\cite{JacobsM12b} for
details, the opposite $\op{\EMod}$ of the category of effect modules
(over $[0,1]$) is also an effectus.

In an analogous way one can define a category $\underline{\vNA}$ of
von Neumann algebras with completely positive maps and show that its
opposite is a grounded biproduct category. The category of von Neumann
algebras with completely positive unital maps is the associated
effectus of causal maps. In particular, it forces this quantum model
to be `non-signalling', see~\cite[\S5.3]{CoeckeK15}.
\end{example}

\begin{remark}
\label{rem:CPGrounded}
As mentioned in the beginning of this section, the
$\CP^*$-categories from~\cite{CoeckeHK14} form an inspiration
for the construction in this section. A category
$\CP^{*}(\cat{C})$ is obtained from a dagger-compact category
$\cat{C}$, and forms a grounded biproduct if:
\begin{enumerate}
\item $\cat{C}$ has finite biproducts;

\item the dagger in $\cat{C}$ yields a ``positive definite inner
  product'', that is, for all $\psi \colon I \rightarrow X$, if
  $\psi^{\dagger} \after \psi = 0$, then $\psi = 0$;

\item ``uniqueness of positive resolutions'' holds: for all positive
  maps $p, q, q'$, if $p + q = \idmap = p + q'$, then $q = q'$. We
  recall that a map $p$ is called positive in a dagger-category if
  there exists $g$ such that $p = g^{\dagger} \after g$.
\end{enumerate}
\end{remark}

The main point of this section is to show that inside a grounded
biproduct category there is an effectus of causal maps. The grounded
biproduct category can then be seen as a larger, ambient category of
the effectus. One can also go in the other direction, that is, produce
an ambient grounded biproduct category from an effectus via a
(universal) totalisation construction.  This will be described
elsewhere.


\section{Kernels and images}\label{sec:kerimg}

Kernels and images (and cokernels) are well-known constructions in
(categorical) algebra. Standardly, for a map $f\colon A \rightarrow B$
its kernel $\ker(f)$ and image $\img(f)$ are understood as
\emph{maps}, of the form $\ker(f) \colon K \rightarrowtail A$ and
$\img(f) \colon B \twoheadrightarrow C$, satisfying certain universal
properties. Here we define kernels and images as
\emph{predicates},\index{S}{cokernel!-- predicate}\index{S}{kernel!--
  predicate} namely $\ker(f) \colon A \rightarrow 1+1$ and $\img(f)
\colon B \rightarrow 1+1$. Later on, in the presence of comprehension
and quotients, our kernels and images (as predicates) will give rise
to kernels and images/cokernels in the traditional sense, see in
particular Lemma~\ref{lem:partcmpr}~\eqref{lem:partcmpr:kermap} and
Lemma~\ref{lem:quot}~\eqref{lem:quot:coker}. When confusion might
occur we speak of kernel/image \emph{predicates} versus kernel/image
\emph{maps}.\index{S}{cokernel!-- map}\index{S}{kernel!-- map} But our
default meaning is: predicate.

In the context of effectuses there is an important difference between
kernels and images: kernels always exist, but the presence of images
is an additional property of an effectus.

We recall from~\eqref{diag:subst} on page~\pageref{diag:subst} the two
forms of substitution that we use, for total and partial maps, namely:
$\tbox{f}(p) = p \tafter f$ and $\pbox{g}(p) = [p, \kappa_{1}] \tafter
g = (p^{\bot} \pafter g)^{\bot}$, with their basic properties
described in Exercise~\ref{exc:subst}. The predicate $\pbox{g}(p)$ may
be read as the `weakest liberal precondition' of $p$ for the partial
computation $g$. Informally it says: if $g$ terminates, then $p$
holds.  The `strong' version $\pdiam{g}(p) = p \pafter g$ says: $g$
terminates and then $p$ holds. As usual, there is the De Morgan
relationship $\pdiam{g}(p^{\bot}) = \pbox{g}(p)^\bot$.

At this stage we know that the predicates $\Pred(X)$ on an object $X$
in an effectus form an effect module, and are in particular partially
ordered (see Exercise~\ref{exc:EA}). We know that the total
substitution map $\tbox{f}$ preserves the effect module
structure. This is different for the partial substitution map
$\pbox{g}$. So far we only know that it preserves truth, see
Exercise~\ref{exc:subst}~\eqref{exc:subst:partialtruth}. But there is
a bit more that we can say now.

\begin{lemma}
\label{lem:partialmonotone}
For a partial map $g\colon X \pto Y$ in an effectus the partial
substitution map $\pbox{g} \colon \Pred(Y) \rightarrow \Pred(X)$
preserves $\one, \owedge$ and is monotone --- where $\owedge$ is the
De Morgan dual of $\ovee$ given by $p \owedge q = (p^{\bot} \ovee
q^{\bot})^{\bot}$.
\end{lemma}

\begin{proof}
We use that partial pre-composition $(-) \pafter g$ preserves $\ovee$,
see
Proposition~\ref{prop:effectusPCM}~\eqref{prop:effectusPCM:pres}. If
$p^{\bot} \orthogonal q^{\bot}$, then:
$$\begin{array}{rcl}
\pbox{g}(p \owedge q)
\hspace*{\arraycolsep}=\hspace*{\arraycolsep}
\big((p\owedge q)^{\bot} \pafter g\big)^{\bot}
& = &
\big((p^{\bot} \ovee q^{\bot}) \pafter g\big)^{\bot} \\
& = &
\big((p^{\bot} \pafter g) \ovee (q^{\bot} \pafter g)\big)^{\bot} \\
& = &
(p^{\bot} \pafter g)^{\bot} \owedge (q^{\bot} \pafter g)^{\bot} \\
& = &
\pbox{g}(p) \owedge \pbox{g}(q).
\end{array}$$

\noindent As a consequence $\pbox{g}$ is monotone, since $p \leq q$
iff $q^{\bot} \leq p^{\bot}$ iff $q^{\bot} \ovee r^{\bot} = p^{\bot}$ for some
$r$, that is, $q \owedge r = p$. \QED

\auxproof{
Explicitly, if $p \leq q$, then $q^{\bot} \leq p^{\bot}$ so that $q^{\bot} \ovee r
= p^{\bot}$ for some $r$. Hence:
$$\begin{array}{rcccl}
(q^{\bot} \pafter g) \ovee (r \pafter g)
& = &
(q^{\bot} \ovee r) \pafter g
& = &
p^{\bot} \pafter g.
\end{array}$$

\noindent This gives $q^{\bot} \pafter g \leq p^{\bot} \pafter g$, and
thus:
$$\begin{array}{rcccccl}
\pbox{g}(p)
& = &
\big(p^{\bot} \pafter g\big)^{\bot}
& \leq &
\big(q^{\bot} \pafter g\big)^{\bot}
& = &
\pbox{g}(q).
\end{array}$$
}
\end{proof}

\subsection{Kernels}\label{subsec:ker}

In linear algebra the kernel of a (linear) map $f\colon A \rightarrow
B$ is the subspace $\setin{a}{A}{f(a) = 0}$ of those elements that get
mapped to $0$. This also works for partial maps $X \tto Y+1$ where we
intuitively understand the kernel as capturing those elements of $X$
that get sent to $1$ in the outcome $Y+1$.

\begin{definition}
\label{def:ker}
Let $\cat{B}$ be an effectus. The \emph{kernel}\index{S}{kernel} of a
partial map $f\colon X \pto Y$ is the predicate $\ker(f) \colon X \tto
1+1$ given by:
\begin{equation}
\label{eqn:ker}
\begin{array}{rcl}
\ker(f)
& \defeq &
\pbox{f}(\zero) \\
& = &
[\zero, \kappa_{1}] \tafter f
\hspace*{\arraycolsep}=\hspace*{\arraycolsep}
[\kappa_{2}, \kappa_{1}] \tafter (\bang\tplus\idmap) \tafter f
\hspace*{\arraycolsep}=\hspace*{\arraycolsep}
\big((\bang\tplus\idmap) \tafter f\big)^{\bot}.
\end{array}\index{N}{$\ker(f)$, kernel of a map $f$}
\end{equation}

\noindent We call $f$ an \emph{internal mono}\index{S}{internal!--
  mono} if $\ker(f) = \zero$.

Sometimes it is easier to work with the orthosupplement of a kernel,
and so we introduce special notation $\kerbot$\index{N}{$\kerbot(f)$,
  kernel-supplement of a map $f$} for it:
$$\begin{array}{rcl}
\kerbot(f)
& \defeq &
\ker(f)^{\bot} \\
& = &
(\bang\tplus\idmap) \tafter f
\hspace*{\arraycolsep}=\hspace*{\arraycolsep}
\klin{\bang} \pafter f
\hspace*{\arraycolsep}=\hspace*{\arraycolsep}
[\kappa_{1} \tafter \bang, \kappa_{2}] \tafter f
\hspace*{\arraycolsep}=\hspace*{\arraycolsep}
\one \pafter f.
\end{array}$$
\end{definition}

This $\kerbot(f)$ is sometimes called the `domain predicate' since it
captures the `domain of $f$', where $f$ is defined, that is,
`non-undefined'. Here we call it the kernel-supplement. We shortly see
that internal monos are not related to `external' monos, in the
underling category, but they are part of a factorisation system, see
Proposition~\ref{prop:quotfactorisation}. Similarly, the internal epis
that will be defined later on in this section form part of a
factorisation system, see
Proposition~\ref{prop:cmprquotimg}~\eqref{prop:cmprquotimg:factsyst}.

We review the situation for our running examples.

\begin{example}
\label{ex:ker}
In the effectus $\Sets$ the kernel of a partial function $f\colon X
\to Y+1$ the partial substitution map $\pbox{f} \colon \Pow(Y)
\rightarrow \Pow(X)$ is given by:
\begin{equation}
\label{eqn:partsubstSets}
\begin{array}{rcl}
\pbox{f}(V)
& = &
\setin{x}{X}{\allin{y}{Y}{f(x) = \kappa_{1}y \Rightarrow y\in V}}.
\end{array}
\end{equation}

\noindent Hence we obtain as kernel predicate:
$$\begin{array}{rcl}
\ker(f)
\hspace*{\arraycolsep}=\hspace*{\arraycolsep}
\pbox{f}(\zero)
& = &
\setin{x}{X}{\allin{y}{Y}{f(x) = \kappa_{1}y \Rightarrow y\in \emptyset}} \\
& = &
\setin{x}{X}{f(x) = *},
\end{array}$$

\noindent where we write $*$ for the sole element of the final
set/object $1$.

In the effectus $\Kl(\Dst)$ for discrete probability a partial map
$f\colon X \to Y+1$ may be described either as a function $f\colon X
\rightarrow \Dst(Y+1)$ or as $f\colon X \rightarrow \sDst(Y)$, see
Example~\ref{ex:effectusKlD}.  We consider both cases.
\begin{enumerate}
\item For $f\colon X \rightarrow \Dst(Y+1)$ the partial substitution
  map $\pbox{f} \colon [0,1]^{Y} \rightarrow [0,1]^{X}$ is given
  by:
\begin{equation}
\label{eqn:partsubstKlD}
\begin{array}{rcl}
\pbox{f}(q)(x)
& = &
\sum_{y\in Y} f(x)(y)\cdot q(y) + f(x)(*).
\end{array}
\end{equation}

\noindent Hence the $\ker(f) \in [0,1]^{X}$ is the fuzzy predicate:
$$\begin{array}{rcccccl}
\ker(f)(x)
& = &
\pbox{f}(\zero)(x)
& = &
f(x)(*)
& = &
1 - \sum_{y\in Y}f(x)(y).
\end{array}$$

\noindent The kernel thus assigns to $x\in X$ the probability that
$f(x)$ is undefined.

\item For a function $f\colon X \rightarrow \sDst(Y)$ we have:
\begin{equation}
\label{eqn:partsubstKlsD}
\begin{array}{rcl}
\pbox{f}(q)(x)
& = &
\sum_{y\in Y} f(x)(y)\cdot q(y) + \big(1 - \sum_{y}f(x)(y)\big).
\end{array}
\end{equation}

\noindent It gives essentially the same description of the kernel:
$$\begin{array}{rclcrcl}
\ker(f)(x)
& = &
1 - \sum_{y}f(x)(y)
& \quad\mbox{so that}\quad &
\kerbot(f)(x)
& = &
\sum_{y}f(x)(y).
\end{array}$$
\end{enumerate}

In the effectus $\op{\OUG}$ of order unit groups the kernel of a
partial map $f\colon G \pto H$ can also be described in two ways,
via the correspondences from Example~\ref{ex:effectusOUG}. We choose
to understand $f$ as a positive subunital map $H \rightarrow G$. Then,
for an effect $e\in [0,1]_{H}$ we have:
\begin{equation}
\label{eqn:partsubstOUG}
\begin{array}{rcl}
\pbox{f}(e)
& = &
f(e) + f(1)^{\bot}.
\end{array}
\end{equation}

\noindent As a result, if identify $e\in [0,1]_{H}$ with the
corresponding map $e\colon H \rightarrow \Z\oplus \Z$, then:
\begin{equation}
\label{eqn:compappOUG}
\begin{array}{rcll} 
e \after f
& = &
\pbox{f}(e^{\bot})^{\bot} & 
   \mbox{by Exercise~\ref{exc:subst}~\eqref{exc:subst:partial}} \\
& = &
1 - f(e^{\bot}) - f(1)^{\bot} & \mbox{by~\eqref{eqn:partsubstOUG}} \\
& = &
1 - f(1) + f(e) - f(1)^{\bot} \qquad \\
& = &
f(1)^{\bot} + f(e) - f(1)^{\bot} \\
& = &
f(e).
\end{array}\hspace*{3em}
\end{equation}

\noindent Moreover, we simply have:
$$\begin{array}{rclcrcl}
\ker(f)
& = &
f(1)^{\bot}
& \quad\mbox{and so}\quad &
\kerbot(f)
& = &
f(1).
\end{array}$$

\noindent The same descriptions apply in the effectus $\op{\vNA}$ of
von Neumann algebras.
\end{example}

We continue with some basic properties of kernels.

\begin{lemma}
\label{lem:ker}
The kernel operation $\ker(-)$ in an effectus $\cat{B}$ satisfies:
\begin{enumerate}
\item \label{lem:ker:pred} $\ker(p) = p^{\bot}$ for a predicate
  $p\colon X \pto 1$, and so $\kerbot(p) = p$;

\item \label{lem:ker:zero} $\ker(f) = \one$ iff $\kerbot(f) = \zero$
  iff $f = \zero$;

\item \label{lem:ker:proj} $\ker(\rhd_{1}) = [\zero,\one]$ and
  $\ker(\rhd_{2}) = [\one,\zero]$, for the partial projections
  $\rhd_i$ defined in~\eqref{diag:partprojtot};

\item \label{lem:ker:mono} $f$ is internally monic, that is $\ker(f) =
  \zero$, iff $f$ is total;




\item \label{lem:ker:comp} $\ker(g\pafter f) = \pbox{f}\big(\ker(g)\big)$

\item \label{lem:ker:ord} $\ker(g \pafter f) \geq \ker(f)$;

\item \label{lem:ker:tot} $\ker(\klin{h} \pafter f) = \ker(f)$;

\item \label{lem:ker:monocat} the internally monic maps form a
  subcategory of $\Par(\cat{B})$;

\item \label{lem:ker:compintmonic} if $g\pafter f$ and $g$ are
  internally monic, then so is $f$;

\item \label{lem:ker:squareneg} $\pbox{f}(p^{\bot}) =
  \pbox{f}(p)^{\bot} \ovee \ker(f)$.
\end{enumerate}
\end{lemma}

\begin{proof}
This involves some elementary reasoning in an effectus.
\begin{enumerate}
\item Clearly, $\ker(p) = [\zero, \kappa_{1}] \tafter p = [\kappa_{2}
  \tafter \bang, \kappa_{2}] \tafter p = [\kappa_{2}, \kappa_{1}] \tafter
  p = p^{\bot}$.

\item First, $\ker(\zero) = \ker(\kappa_{2} \tafter \bang) = [\zero,
  \kappa_{1}] \tafter \kappa_{2} \tafter \bang = \kappa_{1} \tafter
  \bang = \one$. Next, if $\one = \ker(f)$, then $(\bang\tplus\idmap)
  \tafter f = \kerbot(f) = \zero$, so that $f = \zero$ by
  Lemma~\ref{lem:zero}.

\item We simply calculate:
$$\begin{array}{rcl}
\ker(\rhd_{1})
& = &
[\zero, \kappa_{1}] \tafter (\idmap\tplus\bang) 
\hspace*{\arraycolsep}=\hspace*{\arraycolsep}
[\zero, \kappa_{1} \tafter \bang] 
\hspace*{\arraycolsep}=\hspace*{\arraycolsep}
[\zero,\one]
\\
\ker(\rhd_{2})
& = &
[\zero, \kappa_{1}] \tafter [\kappa_{2} \tafter \bang, \kappa_{1}] 
\hspace*{\arraycolsep}=\hspace*{\arraycolsep}
[\kappa_{1} \tafter \bang, \zero] 
\hspace*{\arraycolsep}=\hspace*{\arraycolsep}
[\one, \zero].
\end{array}$$

\item On the one hand the kernel of a total map is $\zero$ since:
$$\begin{array}{rcccccccl}
\ker(\klin{g})
& = &
[\kappa_{2} \tafter \bang, \kappa_{1}] \tafter \kappa_{1} \tafter g
& = &
\kappa_{2} \tafter \bang \tafter g
& = &
\kappa_{2} \tafter \bang
& = &
\zero.
\end{array}$$

\noindent On the other hand, if $\ker(f) = \zero$, then $\one \pafter
f = \kerbot(f) = \one$ so that we $f$ is total by
Lemma~\ref{lem:zero}.

\item We have by Exercise~\ref{exc:subst}~\eqref{exc:subst:partialfun}:
$$\begin{array}{rcccccl}
\ker(g \pafter f)
& = &
\pbox{(g \pafter f)}(\zero)
& = &
\pbox{f}\big(\pbox{g}(\zero)\big)
& = &
\pbox{f}\big(\ker(g)\big).
\end{array}$$

\item From $\zero \leq \ker(g)$ we get by Lemma~\ref{lem:partialmonotone}
and point~\eqref{lem:ker:comp}:
$$\begin{array}{rcccccl}
\ker(f)
& = &
\pbox{f}(\zero) 
& \leq &
\pbox{f}(\ker(g))
& = &
\ker(g \pafter f).
\end{array}$$

\item By a straightforward computation:
$$\begin{array}{rcccccccl}
\ker(\klin{h} \pafter f)
& = &
[\zero, \kappa_{1}] \tafter (h\tplus\idmap) \tafter f
& = &
[\zero, \kappa_{1}] \tafter f
& = &
\pbox{f}(\zero)
& = &
\ker(f).
\end{array}$$

\item Internally monic maps are total by point~\eqref{lem:ker:mono},
  and these total map form a subcategory by Lemma~\ref{lem:zero}.

\item Let $g\pafter f$ and $g$ be internally monic. Write $g =
  \klin{h}$. Then $f$ is internally monic since by
  point~\eqref{lem:ker:tot},
$$\begin{array}{rcccccl}
\ker(f)
& = &
\ker(\klin{h} \pafter f)
& = &
\ker(g \pafter f)
& = &
\zero.
\end{array}$$

\item Let $f\colon X \pto Y$ be a partial map and $p$ a predicate
on $Y$. We first observe that:
$$\begin{array}{rcl}
\one
\hspace*{\arraycolsep}=\hspace*{\arraycolsep}
\big(\one \pafter f\big) \ovee \big(\one \pafter f\big)^{\bot} 
& = &
\big((p \ovee p^{\bot}) \pafter f\big) \ovee \ker(f) \\
& = &
\big(p \pafter f\big) \ovee \big(p^{\bot} \pafter f\big) \ovee \ker(f) \\
& = &
\pbox{f}(p^{\bot})^{\bot} \ovee \pbox{f}(p)^{\bot} \ovee \ker(f).
\end{array}$$

\noindent This last equation is based on
Exercise~\ref{exc:subst}~\eqref{exc:subst:partial}. By uniqueness of
orthosupplements we obtain $\pbox{f}(p^{\bot}) = \pbox{f}(p)^{\bot}
\ovee \ker(f)$. \QED

\auxproof{
Old proof: Let $f\colon X \pto Y$ be a partial map and $q$ a predicate
on $Y$. We prove $\pbox{f}(p^{\bot}) = \pbox{f}(p)^{\bot} \ovee
\ker(f)$ by reasoning in $\Pred(Y+1)$ and using that the total
substitution map $\tbox{f}$ preserves the effect algebra structure.
$$\begin{array}[b]{rcl}
\pbox{f}(p^{\bot})
\hspace*{\arraycolsep}=\hspace*{\arraycolsep}
[p^{\bot}, \kappa_{1}] \tafter f
\hspace*{\arraycolsep}=\hspace*{\arraycolsep}
\tbox{f}\big([p^{\bot}, \one]\big)
& = &
\tbox{f}\big([p^{\bot}, \zero] \ovee [\zero, \one]\big) \\
& = &
\tbox{f}\big([p, \one]^{\bot} \ovee [\zero, \one]\big) \\
& = &
\tbox{f}\big([p, \one]\big)^{\bot} \ovee \tbox{f}\big([\zero, \one]\big) \\
& = &
\pbox{f}(p)^{\bot} \ovee \ker(f).
\end{array}\eqno{\QEDbox}$$
}
\end{enumerate}
\end{proof}

We can now see that `internally monic' does not imply `monic'.  Indeed
`internally monic' means `total', and not every total function in
$\Par(\Sets)$ is also monic in $\Par(\Sets)$, like the truth function
$\{0,1\} \tto 1+1$.

More generally, it is not hard to see that a partial map $f\colon X
\tto Y+1$ is monic in the category $\Par(\cat{B})$ of partial maps of
an effectus $\cat{B}$ if and only if its Kleisli extension
$[f,\kappa_{2}] \colon X+1 \rightarrow Y+1$ is monic in $\cat{B}$.

\auxproof{
\begin{itemize}
\item Let $f\colon X \tto Y+1$ be monic in $\Par(\cat{B})$, and let
  $a, b \colon Y \tto X+1$ satisfy $[f,\kappa_{2}] \after a = f\pafter
  a = f \pafter b = [f,\kappa_{2}] \after b$. Then $a=b$.

\item Conversely, let $[f,\kappa_{2}] \colon X+1 \rightarrow Y+1$ be
  monic in $\cat{B}$, and let $a, b\colon Y\pto X$ satisfy $f \pafter
  a = f \pafter b$, that is, $[f,\kappa_{2}] \after a = [f,\kappa_{2}]
  \after b$. Then $a=b$.
\end{itemize}
}

But also: not all `external monos' are `internal monos'. For instance
the partial function $f\colon X \rightarrow \Dst(X+1)$ given by $f(x)
= \frac{3}{4}\ket{x} + \frac{1}{4}\ket{\!*\!}$ has $\ker(f)(x) = f(x)(*) =
\frac{1}{4}$, so $\ker(f) \neq 0$ in $[0,1]^{X}$, and thus $f$ is not
internally monic.  Still $f$ is an external mono in $\Par(\Kl(\Dst))$:
if $\omega,\omega' \in \Dst(X+1)$ satisfy $\tstat{f}(\omega) =
\tstat{f}(\omega')$, then for each $x\in X$ we have $\omega(x) =
\frac{4}{3}\cdot \tstat{f}(\omega)(x) = \frac{4}{3}\cdot
\tstat{f}(\omega')(x) = \omega'(x)$, and $\omega(*) = 4\cdot
\tstat{f}(\omega)(*) = 4 \cdot \tstat{f}(\omega')(*) =
\omega'(*)$.


The following technical but important result about kernel-supplements,
and the subsequent proposition, are due to~\cite{Cho15}.

\begin{lemma}
\label{lem:kerbot}
Let $\cat{B}$ be an effectus, with objects $X,Y\in\cat{B}$.  The
kernel-supplement map
$$\xymatrix@C+1pc{
\Par(\cat{B})\big(X,Y)\ar[r]^-{\kerbot} & \Pred(X)
}$$

\noindent is a map of PCMs, that preserves and reflects both $0$ and
orthogonality $\orthogonal$.

Moreover, this kernel-supplement map is `dinatural' in the situation:
$$\xymatrix@C+1pc{
\op{\cat{B}}\times\cat{B}
   \rrtwocell^{\Par(\cat{B})(-,-)}_{\Pred_{1} = \Pred \after \pi_{1}}{\quad\;\kerbot}
   & & \Sets
}$$

\noindent Indeed, for each map $f\colon X \tto Y$ in $\cat{B}$ we
have a commuting diagram:
\begin{equation}
\label{diag:kerbotdinatural}
\hspace*{-1em}\vcenter{\xymatrix@C-1.5pc@R-.5pc{
& {\Par(\cat{B})(X,X)\ar[rr]^-{\kerbot_{X}}} & &
   \Pred_{1}(X,X)\rlap{$\; = \Pred(X)$}
      \ar@{=}[dr]^(0.6){\qquad\Pred_{1}(X,f) = \idmap}
\\
\Par(\cat{B})(Y,X)\ar[ur]^{\Par(\cat{B})(f,X) = (-) \tafter f\qquad}
   \ar[dr]_{\Par(\cat{B})(Y,f) = (f\tplus\idmap) \tafter (-)\qquad\quad} & & & &
\llap{$\Pred(X)=\;$}\Pred_{1}(X,Y)
\\
& \Par(\cat{B})(Y,Y)\ar[rr]_-{\kerbot_{Y}} & &
   \Pred_{1}(Y,Y)\rlap{$\;=\Pred(Y)$}
      \ar[ur]_(0.6){\qquad\Pred_{1}(f,Y) = \tbox{f}}
}}
\end{equation}
\end{lemma}

\begin{proof}
First, the map $\kerbot$ preserves and reflects $\zero$, since
$\kerbot(f) = \zero$ iff $f=\zero$ by
Lemma~\ref{lem:ker}~\eqref{lem:ker:zero} --- or equivalently,
Lemma~\ref{lem:zero}. It also preserves $\ovee$ by
Proposition~\ref{prop:effectusPCM}~\eqref{prop:effectusPCM:pres}: if
parallel partial maps $f,g$ are orthogonal, then:
$$\begin{array}{rcccccl}
\kerbot(f \ovee g)
& = &
\one \pafter (f\ovee g)
& = &
(\one \pafter f) \ovee (\one \pafter g)
& = &
\kerbot(f) \ovee \kerbot(g).
\end{array}$$

\noindent The main challenge is to prove that $\kerbot(f) \orthogonal
\kerbot(g)$ implies $f \orthogonal g$. So let $b\colon X \pto 1+1$ be
a bound for $\kerbot(f)$ and $\kerbot(g)$. Then $\rhd_{1} \pafter b =
\kerbot(f) = \one \pafter f = \klin{\bang} \pafter f$ and similarly
$\rhd_{2} \pafter b = \klin{\bang} \pafter g$. We use the projection
pullbacks in $\Par(\cat{B})$ from~\eqref{diag:effectuspartprojpb} in
Lemma~\ref{lem:effectuspb}~\eqref{lem:effectuspb:par} in 
two steps below, first to get a map $c$, and then $d$.
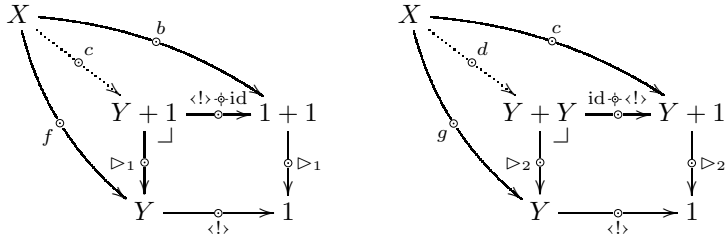
$$\xymatrix{
X\ar@/_2ex/[ddr]|-{\pafter}_{f}\ar@/^2ex/[drr]|-{\pafter}^-{b}
   \ar@{..>}[dr]|-{\pafter}^-{c} & &
&
X\ar@/_2ex/[ddr]|-{\pafter}_{g}\ar@/^2ex/[drr]|-{\pafter}^-{c}
   \ar@{..>}[dr]|-{\pafter}^-{d}
\\
& Y+1\ar[r]|-{\pafter}^-{\klin{\bang}\pplus\idmap}
   \ar[d]|-{\pafter}_{\rhd_{1}}\pullback & 
   1+1\ar[d]|-{\pafter}^{\rhd_{1}}
&
& Y+Y\ar[r]|-{\pafter}^-{\idmap\pplus\klin{\bang}}
   \ar[d]|-{\pafter}_{\rhd_{2}}\pullback & 
   Y+1\ar[d]|-{\pafter}^{\rhd_{2}}
\\
& Y\ar[r]|-{\pafter}_-{\klin{\bang}} & 1
&
& Y\ar[r]|-{\pafter}_-{\klin{\bang}} & 1
}$$

\noindent The outer diagram on the right commutes since:
$$\begin{array}{rcccccl}
\rhd_{2} \pafter c
& = &
\rhd_{2} \pafter (\klin{\bang}\pplus\idmap) \pafter c
& = &
\rhd_{2} \pafter b
& = &
\klin{\bang} \pafter g.
\end{array}$$

\noindent This $d$ is a bound for $f,g$, showing $f \orthogonal g$. By
construction $\rhd_{2} \pafter d = g$, and:
$$\begin{array}{rcccccl}
\rhd_{1} \pafter d
& = &
\rhd_{1} \pafter (\idmap\pplus\klin{\bang}) \pafter d
& = &
\rhd_{1} \pafter c
& = &
f.
\end{array}$$

Finally we check commutation of the dinaturality\index{S}{dinaturality}
diagram~\eqref{diag:kerbotdinatural}: for a total map $f\colon X \tto
Y$ and a partial map $g \in \Par(\cat{B})(Y, X)$ we have:
$$\begin{array}[b]{rcl}
\big(\tbox{f} \after \kerbot_{Y} \after ((f\tplus\idmap) \tafter -)\big)(g)
& = &
\tbox{f}\big(\kerbot_{Y}((f\tplus\idmap) \tafter g)\big) \\
& = &
(\bang\tplus\idmap) \tafter (f\tplus\idmap) \tafter g \tafter f \\
& = &
(\bang\tplus\idmap) \tafter g \tafter f \\
& = &
\kerbot_{X}(g \tafter f) \\
& = &
\big(\kerbot_{X} \after (- \tafter f)\big)(g).
\end{array}\eqno{\QEDbox}$$
\end{proof}

This reflection of orthogonality is a quite powerful property. First
we use it to say more about the PCM-structure $(\ovee, \zero)$ on
homsets of partial maps from Proposition~\ref{prop:effectusPCM}.

\begin{proposition}
\label{prop:effectusPCMorder}
Let $\cat{B}$ be an effectus. Then:
\begin{enumerate}
\item \label{prop:effectusPCMorder:pos} the sum $\ovee$ on partial
  homsets $\Par(\cat{B})(X,Y)$ is positive:\index{S}{positive!-- sum}
  $f\ovee g = \zero$ implies $f = g = \zero$;

\item \label{prop:effectusPCMorder:canc} it is also
  cancellative\index{S}{cancellation} in the sense: $f \ovee g = g$
  implies $f = \zero$;

\item \label{prop:effectusPCMorder:order} $\Par(\cat{B})(X,Y)$ is a
  partial order via $f \leq g$ iff $f \ovee h = g$ for some $h$; pre-
  and post-composition of partial maps is thus monotone.
\end{enumerate}
\end{proposition}

As a result of the last point the category $\Par(\cat{B})$ of partial
maps is not only enriched over PCMs, but also over pointed posets
(with a bottom element, preserved under composition).

\begin{proof}
If $f\ovee g = \zero$, then $\kerbot(f) \ovee \kerbot(g) =
\kerbot(f\ovee g) = \zero$, so that $\kerbot(f) = \kerbot(g) = \zero$
since the effect module $\Pred(X)$ is positive. But then $f = g =
\zero$ by Lemma~\ref{lem:ker}~\eqref{lem:ker:zero}.

If $f \ovee g = g$, then by similarly applying $\kerbot(-)$ we get
$\kerbot(f) = \zero$ in $\Pred(X)$, and thus $f = \zero$.

Obviously, the order $\leq$ on $\Par(\cat{B})(X,Y)$ is reflexive and
transitive. But it is also anti-symmetric: if $f \leq g$ and $g \leq
f$, say via $f \ovee h = g$ and $g \ovee k = f$, then $f \ovee (h\ovee
k) = f$, so that $h\ovee k = \zero$ by
point~\eqref{prop:effectusPCMorder:canc}, and thus $h = k = \zero$ by
point~\eqref{prop:effectusPCMorder:pos}. Hence $f = g$. \QED
\end{proof}

We include a second result that uses reflection of orthogonality by
kernel-supplements. The fact that normalisation of substates holds in
the two examples $\Kl(\Dst)$ and $\op{\vNA}$ in
Remark~\ref{rem:normalisation} is an instance of the following result
due to Sean Tull.

\begin{lemma}
\label{lem:unitnormalisation}
Let $\cat{B}$ be an effectus whose scalars are given by the unit
interval $[0,1]$ of $\R$. Then normalisation\index{S}{normalisation}
holds in $\cat{B}$: non-zero substates can be normalised to proper
states, that is, for each non-zero $\omega \colon 1 \pto X$ there is a
unique $\rho \colon 1 \tto X$ with $\omega = \klin{\rho} \pafter r$
for the (non-zero) scalar $r = \one \pafter \omega \colon 1 \pto 1$.
\end{lemma}

In particular, this means that in the context of~\cite{JacobsWW15},
where all effectuses have $[0,1]$ as scalars, normalisation comes for
free.

\begin{proof}
Let $\omega \colon 1 \pto X$ be a non-zero substate, with
corresponding scalar $r = \one \pafter \omega \in [0,1]$. Since $r\neq
0$ --- by Lemma~\ref{lem:zero} --- we can find an $n\in\NNO$ and
$r'\in [0,1]$ with $r' \leq r$ and $n\cdot r + r' = 1$.  More
abstractly, we can find scalars $s_{1}, \ldots, s_{m}\in [0,1]$ with
$\bigovee_{i} s_{i} \cdot r = 1$. We now form the scalar
multiplication $\omega \pafter s_{i} \colon 1 \pto X$ as in
Lemma~\ref{lem:substatePCMod}. The scalars $\one \pafter \omega
\pafter s_{i} = r \pafter s_{i}$ are orthogonal, so the maps $\omega
\pafter s_{i}$ are orthogonal too, since $\one \pafter (-) =
\ker^{\bot}$ reflects orthogonality, by Lemma~\ref{lem:kerbot}.  But
then we have the following equalities of maps $1 \pto 1$.
$$\begin{array}{rcccccl}
\one \pafter \bigovee_{i} (\omega \pafter s_{i})
& = &
\bigovee_{i} \one \pafter \omega \pafter s_{i}
& = &
\bigovee_{i} r \pafter s_{i}
& = &
\one.
\end{array}$$

\noindent Lemma~\ref{lem:zero} tells that the partial map
$\bigovee_{i} (\omega \pafter s_{i}) \colon 1 \pto X$ is total, so we
can write it as $\klin{\rho}$, for a unique state $\rho\colon 1 \tto
X$. By construction we have $\klin{\rho} \pafter \one \pafter \omega =
\omega$ as in~\eqref{eqn:normalisation}. 

If also $\rho'\colon 1\tto X$ satisfies $\klin{\rho'} \pafter r =
\omega = \klin{\rho} \pafter r$, then we obtain $\rho' = \rho$ from
faithfulness of $\klin{-} \colon \cat{B} \rightarrow \Par(\cat{B})$:
$$\begin{array}[b]{rcl}
\klin{\rho}
\hspace*{\arraycolsep}=\hspace*{\arraycolsep}
\klin{\rho} \pafter \bigovee_{i} r \pafter s_{i}
& = &
\bigovee_{i} \klin{\rho} \pafter r \pafter s_{i} \\
& = &
\bigovee_{i} \klin{\rho'} \pafter r \pafter s_{i}
\hspace*{\arraycolsep}=\hspace*{\arraycolsep}
\klin{\rho'} \pafter \bigovee_{i} r \pafter s_{i}
\hspace*{\arraycolsep}=\hspace*{\arraycolsep}
\klin{\rho'}.
\end{array}\eqno{\QEDbox}$$
\end{proof}

Tull's result is a bit more general and generalises the crucial
property used in the lemma to the requirement that the scalars
satisfy: for each non-zero scalar $r$ there are scalars $s_{1},
\ldots, s_{m}$ with $\bigovee_{i} s_{i}\cdot r = r = \bigovee_{i} r
\cdot s_{i}$. This condition does not apply to effectuses that have a
cube $[0,1]^{n}$ as scalars.  For instance $(r,0)\in [0,1]^{2}$ is
non-zero, for $r\neq 0$, but there is no way to raise $(r,0)$ to the
top element $(1,1)\in [0,1]^{2}$ via multiply-and-add. This means that
the product of two effectuses with normalisation need not have
normalisation.

By generalising the requirement to all sets of predicates one can
normalise all non-zero partial maps, via the scalar multiplication
from Lemma~\ref{lem:monPCMod}.

Once again using reflection of orthogonality we can give an
alternative description of the partial pairing $\dtuple{f,g}$ from
Lemma~\ref{lem:pairing}. It exists for maps $f,g$ satisfying $\one
\pafter f = (\one \pafter g)^{\bot}$. In the terminology of the
present section we can rephrase this assumption as $\kerbot(f) =
\ker(g)$, or as $\kerbot(f) \ovee \kerbot(g) = \one$.

\begin{lemma}
\label{lem:sumpairing}
Let $f\colon Z \pto X$ and $g\colon Z \pto Y$ be to partial maps in an
effectus with $\kerbot(f) = \ker(g)$. Then we can describe the
resulting total pairing map $\dtuple{f,g} \colon Z \tto X+Y$ as the
unique one satisfying in the homset of partial maps $Z \pto X+Y$:
\begin{equation}
\label{eqn:dtupleovee}
\begin{array}{rcl}
\klin{\dtuple{f,g}}
& = &
(\kappa_{1}\pafter f) \ovee (\kappa_{2}\pafter g).
\end{array}
\end{equation}

\noindent As special case of this equation we obtain, using the
equation $\dtuple{\rhd_{1},\rhd_{2}} = \idmap$
from~\eqref{eqn:dtupleuniqueness}, a new relationship between
projections and coprojections:
$$\begin{array}{rcl}
\idmap
& = &
(\kappa_{1}\pafter \rhd_{1}) \ovee (\kappa_{2}\pafter \rhd_{2}).
\end{array}$$

More generally, there is the following bijective correspondence
(from~\cite{Cho15}):
\begin{equation}
\label{bijcor:sumpairing}
\begin{prooftree}
{\xymatrix{Z\ar[r]|-{\pafter}^-{f} & X_{1} + \cdots + X_{n}}}
\Justifies
{\xymatrix{Z\ar[r]|-{\pafter}_-{f_i} & X_{i}} \;
   \mbox{ with $\kerbot(f_{i})$ orthogonal}}
\end{prooftree}
\end{equation}

\noindent Moreover, the map $f$ above the lines is total if and only
if the maps $f_i$ below the lines satisfy $\bigovee_{i} \kerbot(f_{i})
= \one$.
\end{lemma}

\begin{proof}
For the sum $\ovee$ on the right in the
equation~\eqref{eqn:dtupleovee} we take as bound $b \colon Z \pto
(X+Y) + (X+Y)$ the map $b = \klin{\kappa_{1}\tplus\kappa_{2}}
\tafter \dtuple{f,g} = (\klin{\kappa_{1}} \pplus \klin{\kappa_{2}})
\tafter \dtuple{f,g}$. Then:
$$\begin{array}{rcccccccl}
\rhd_{1} \pafter b
& = &
\rhd_{1} \pafter (\klin{\kappa_{1}} \pplus \klin{\kappa_{2}})
  \tafter \dtuple{f,g} 
& = &
\klin{\kappa_{1}} \pafter \rhd_{1} \tafter \dtuple{f,g} 
& = &
\klin{\kappa_{1}} \pafter f 
& = &
\kappa_{1} \pafter f.
\end{array}$$

\noindent Similarly we get $\rhd_{2} \pafter b = \kappa_{2} \pafter g$. Then:
$$\begin{array}{rcl}
(\kappa_{1}\pafter f) \ovee (\kappa_{2}\pafter g)
\hspace*{\arraycolsep}=\hspace*{\arraycolsep}
\nabla \pafter b 
& = &
\klin{\nabla} \pafter \klin{\kappa_{1}\tplus\kappa_{2}}
   \tafter \dtuple{f,g} \\
& = &
\klin{\idmap} \tafter \dtuple{f,g} \\
& = &
\klin{\dtuple{f,g}}.
\end{array}$$

\noindent The equation $(\kappa_{1}\pafter \rhd_{1}) \ovee
(\kappa_{2}\pafter \rhd_{2}) = \idmap$ holds
by~\eqref{eqn:dtupleuniqueness}, for $k=\idmap$.

In the bijective correspondence~\eqref{bijcor:sumpairing} we send a
map $f\colon Z \pto X_{1} + \cdots + X_{n}$ to the $n$-tuple of maps
$f_{i} = \rhd_{i} \pafter f$. These maps are all orthogonal, via $f$
as bound. Hence the maps $\kerbot(f_{i}) = \one \after f_{i}$ are also
orthogonal. If $f$ is total, then by an $n$-ary version
of~\eqref{eqn:dtupleovee},
$$\begin{array}{rcccccccccl}
\one
& = &
\one \pafter f
& = &
\one \pafter \bigovee_{i}(\kappa_{i} \pafter f_{i})
& = &
\bigovee_{i} \one \pafter \kappa_{i} \pafter f_{i}
& = &
\bigovee_{i} \one \pafter f_{i}
& = &
\bigovee_{i} \ker(f_{i}).
\end{array}$$

\noindent In the other direction, let $f_{i}\colon Z \rightarrow
X_{i}$ be maps for which the kernel-supplements $\kerbot(f_{i}) = \one
\after f_{i}\colon Z \pto 1$ are orthogonal. The maps $\kappa_{i}
\pafter f_{i} \colon Z \rightarrow X_{1}+\cdots+X_{n}$ are then
orthogonal too since the maps $\one \pafter \kappa_{i} \pafter f_{i} =
\one \pafter f_{i}$ are orthogonal and $\kerbot = \one \pafter (-)$
reflects orthogonality, by Lemma~\ref{lem:kerbot}. Hence we take $f =
\bigovee_{i}(\kappa_{i} \pafter f_{i})$. This map $f$ is total if
$\bigovee\ker(f_{i}) = \one$ by following the previous chain of
equations backwards. 

We get a bijective correspondence since $\bigovee_{i}(\kappa_{i} \pafter
\rhd_{i} \pafter f) = f$ like in~\eqref{eqn:dtupleovee}, and:
$$\begin{array}{rcccccl}
\rhd_{j} \pafter \bigovee_{i}(\kappa_{i} \pafter f_{i})
& = &
\bigovee_{i}(\rhd_{j} \pafter \kappa_{i} \pafter f_{i})
& = &
\bigovee_{i}\left\{\begin{array}{ll}
\idmap \pafter f_{i} \; & \mbox{if } j = i \\
\zero & \mbox{if } j \neq i
\end{array}\right\}
& = &
f_{j}.
\end{array}\eqno{\QEDbox}$$
\end{proof}

The bijective correspondence~\eqref{bijcor:sumpairing} extends the
pairing from Lemma~\ref{lem:pairing} to $n$-ary form and to partial
maps, see the end for Discussion~\ref{dis:partialtotal} for more
details.

\auxproof{
\begin{remark}
\label{rem:pairing}
Now that we have the terminology of kernels (and their complements) we
can get a better understanding of the pullback on the left
in~\eqref{diag:effectuspb} in the definition of effectus. This
pullback can be seen as a form of \emph{pairing} $\dtuple{f,g}$
for partial maps with `compatible' kernels. Let $f\colon Z \pto X$ and
$g\colon Z \pto Y$ be two partial maps with $\ker(f) =
\kerbot(g)$. The latter says that the outer diagram commutes in:
$$\xymatrix{
Z\ar@/_2ex/[ddr]|-{\tafter}_{[\kappa_{2}, \kappa_{1}] \tafter g}
   \ar@/^2ex/[drr]|-{\tafter}^{f}
   \ar@{..>}[dr]|-{\tafter}^(0.7){\dtuple{f,g}}
\\
& X+Y\ar[r]|-{\tafter}^-{\idmap\tplus \bang}
   \ar[d]|-{\tafter}^{\bang\tplus \idmap}\pullback & 
   X + 1\ar[d]|-{\tafter}^{\bang\tplus\idmap}
\\
& 1+ Y\ar[r]|-{\tafter}_-{\idmap\tplus \bang} & 1 + 1
}$$

\auxproof{
$$\begin{array}{rcl}
(\idmap\tplus \bang) \tafter [\kappa_{2}, \kappa_{1}] \tafter g
& = &
[\kappa_{2}, \kappa_{1}] \tafter (\bang\tplus\idmap) \tafter g \\
& = &
[\kappa_{2}, \kappa_{1}] \tafter \kerbot(g) \\
& = &
[\kappa_{2}, \kappa_{1}] \tafter \ker(f) \\
& = &
\kerbot(f) \\
& = &
(\bang\tplus\idmap) \tafter f.
\end{array}$$
}

\noindent This pairing $\dtuple{f,g}$ is the unique total map
with:
$$\begin{array}{rclcrcl}
\rhd_{1} \tafter \dtuple{f,g}
& = &
f
& \qquad\mbox{and}\qquad &
\rhd_{2} \tafter \dtuple{f,g}
& = &
g.
\end{array}$$

\noindent Uniqueness gives:
\begin{equation}
\label{eqn:dtupleuniqueness}
\begin{array}{rclcrcl}
\dtuple{f,g} \tafter h
& = &
\dtuple{f \tafter h, g\tafter h}
& \qquad\mbox{and}\qquad &
\dtuple{\rhd_{1} \tafter k, \rhd_{2}\tafter k} 
& = &
k,
\end{array}
\end{equation}

\noindent for each $h\colon W \rightarrow Z$ and $k\colon Z \tto X+Y$.

\auxproof{
$$\begin{array}{rcl}
\rhd_{1} \tafter \dtuple{f,g}
& = &
(\idmap\tplus \bang) \tafter \dtuple{f,g} \\
& = &
f
\\
\rhd_{1} \tafter \dtuple{f,g}
& = &
[\kappa_{2} \tafter \bang, \kappa_{1}] \tafter \dtuple{f,g} \\
& = &
[\kappa_{2}, \kappa_{1}] \after (\bang\tplus\idmap) \tafter \dtuple{f,g} \\
& = &
[\kappa_{2}, \kappa_{1}] \after [\kappa_{2}, \kappa_{1}] \after g \\
& = &
g.
\end{array}$$

\noindent Uniqueness of $\dtuple{f,g}$ follows from
Lemma~\ref{lem:effectusjm}.
}

Moreover, the pairing is also unique in satisfying in the homset of
partial maps $Z \pto X+Y$:
\begin{equation}
\label{eqn:dtupleovee}
\begin{array}{rcl}
\klin{\dtuple{f,g}}
& = &
(\kappa_{1}\pafter f) \ovee (\kappa_{2}\pafter g).
\end{array}
\end{equation}

\noindent As special case of this equation we obtain a new relationship
between projections and coprojections:
$$\begin{array}{rcl}
\idmap
& = &
(\kappa_{1}\pafter \rhd_{1}) \ovee (\kappa_{2}\pafter \rhd_{2}).
\end{array}$$

\auxproof{
Indeed, we take as bound $b \colon Z \pto (X+Y) + (X+Y)$ the map $b =
\kappa_{1} \tafter (\kappa_{1}\tplus\kappa_{2}) \tafter
\dtuple{f,g}$. Then:
$$\begin{array}{rcl}
\rhd_{1} \pafter b
& = &
[\idmap\tplus \bang, \kappa_{2}] \tafter \kappa_{1} \tafter 
   (\kappa_{1}\tplus\kappa_{2}) \tafter \dtuple{f,g} \\
& = &
(\idmap\tplus \bang) \tafter 
   (\kappa_{1}\tplus\kappa_{2}) \tafter \dtuple{f,g} \\
& = &
(\kappa_{1}\tplus\idmap) \tafter (\idmap\tplus \bang) \tafter \dtuple{f,g} \\
& = &
(\kappa_{1}\tplus\idmap) \tafter f \\
& = &
\kappa_{1} \pafter f
\\
\rhd_{2} \pafter b
& = &
[[\kappa_{2}\tafter \bang, \kappa_{1}], \kappa_{2}] \tafter \kappa_{1} \tafter 
   (\kappa_{1}\tplus\kappa_{2}) \tafter \dtuple{f,g} \\
& = &
[\kappa_{2}\tafter \bang, \kappa_{1}] \tafter 
   (\kappa_{1}\tplus\kappa_{2}) \tafter \dtuple{f,g} \\
& = &
[\kappa_{2}\tafter \bang, \kappa_{1} \tafter \kappa_{2}] \tafter 
   \dtuple{f,g} \\
& = &
[\kappa_{2}, \kappa_{1} \tafter \kappa_{2}] \tafter 
   (\bang\tplus\idmap) \tafter \dtuple{f,g} \\
& = &
[\kappa_{2}, \kappa_{1} \tafter \kappa_{2}] \tafter 
   [\kappa_{2}, \kappa_{1}] \tafter g \\
& = &
[\kappa_{1} \tafter \kappa_{2}, \kappa_{2}] \tafter g \\
& = &
(\kappa_{2}\tplus\idmap) \tafter g \\
& = &
\kappa_{2} \pafter g
\\
(\kappa_{1}\pafter f) \ovee (\kappa_{2}\pafter g)
& = &
\nabla \pafter b \\
& = &
(\nabla\tplus\idmap) \tafter \kappa_{1} \tafter 
   (\kappa_{1}\tplus\kappa_{2}) \tafter \dtuple{f,g} \\
& = &
\kappa_{1} \tafter \nabla \tafter 
   (\kappa_{1}\tplus\kappa_{2}) \tafter \dtuple{f,g} \\
& = &
\kappa_{1} \tafter [\kappa_{1}, \kappa_{2}] \tafter \dtuple{f,g} \\
& = &
\klin{\dtuple{f,g}}.
\end{array}$$

The equation $(\kappa_{1}\pafter \rhd_{1}) \ovee (\kappa_{2}\pafter
\rhd_{2}) = \idmap$ holds since:
\begin{itemize}
\item $\ker(\rhd_{1}) = [0,1] = [1,0]^{\bot} = \ker(\rhd_{2})^{\bot} =
  \kerbot(\rhd_{2})$ by Lemma~\ref{lem:ker}~\eqref{lem:ker:proj}.

\item And: $\rhd_{1} = \idmap\tplus \bang \colon X+Y \tto X+1$ and
  $[\kappa_{2}, \kappa_{1}] \tafter \rhd_{2} = [\kappa_{2},
  \kappa_{1}] \tafter [\kappa_{2} \tafter \bang, \kappa_{1}] = [\kappa_{1}
  \tafter \bang, \kappa_{2}] = \bang\tplus\idmap$.
\end{itemize}
}

\noindent Finally we notice that for total maps $h,k$ we have:
\begin{equation}
\label{eqn:dtuplenat}
\begin{array}{rcl}
(h\tplus k) \tafter \dtuple{f,g}
& = &
\dtuple{\klin{h} \pafter f, \klin{k} \pafter g}.
\end{array}
\end{equation}

\auxproof{
$$\begin{array}{rcl}
(\idmap\tplus\bang) \tafter (h\tplus k) \tafter \dtuple{f,g}
& = &
(h\tplus\idmap) \tafter (\idmap\tplus\bang) \tafter \dtuple{f,g} \\
& = &
(h\tplus\idmap) \tafter f \\
& = &
\klin{h} \pafter f
\\
(\bang\tplus\idmap) \tafter (h\tplus k) \tafter \dtuple{f,g}
& = &
(\idmap\tplus k) \tafter (\bang\tplus\idmap) \tafter \dtuple{f,g} \\
& = &
(\idmap\tplus k) \tafter [\kappa_{2}, \kappa_{1}] \tafter g \\
& = &
[\kappa_{2}, \kappa_{1}] \tafter (k\tplus \idmap) \tafter g \\
& = &
[\kappa_{2}, \kappa_{1}] \tafter (\klin{k} \pafter g).
\end{array}$$
}

Thus, the coproduct $+$ in the category of partial maps of an effectus
looks almost like a biproduct: it does have projections, and pairing,
but pairing only works conditionally.

The `pairing' described here can be generalised, as in~\cite{Cho15},
to a bijective correspondence between:
$$\begin{prooftree}
{\xymatrix{Z\ar[r]|-{\pafter}^-{f} & X_{1} + \cdots + X_{n}}}
\Justifies
{\xymatrix{Z\ar[r]|-{\pafter}_-{f_i} & X_{i}} \;
   \mbox{ with $\kerbot(f_{i})$ orthogonal}}
\end{prooftree}$$

\noindent Moreover, the map $f$ above the lines is total if and only
if the maps $f_i$ below the lines satisfy $\bigovee_{i} \kerbot(f_{i})
= \one$.
\end{remark}
}

\subsection{Images}\label{subsec:img}

A kernel is a predicate on the domain of a partial map. An image is a
predicate on its codomain. An effectus always has kernels. But the
existence of images must be required explicitly.

There is one more notion that we need in the description of images.
In an effect algebra an element $s$ is called
\emph{sharp}\index{S}{sharp!-- element of an effect algebra} if $s
\wedge s^{\bot} = 0$. One may argue that the meet $\wedge$ may not
exist. But the definition of sharpness only requires that the
particular meet $s \wedge s^{\bot}$ exists and is equal to
$0$. Equivalently, without meets: $s$ is sharp if for each element $x$
one has: $x\leq s$ and $x\leq s^{\bot}$ implies $x = 0$. Notice that
$\zero$ and $\one$ are sharp elements, and that if $s$ is sharp then
its orthosupplement $s^\bot$ is sharp too.

\begin{definition}
\label{def:img}
We say that an effectus has \emph{images}\index{S}{image} if for each
partial map $f\colon X \pto Y$ there is a least predicate $q$ on $Y$
with $\pbox{f}(q) = \one$. In that case we write
$\img(f)$\index{N}{$\img(f)$, image of a map $f$} for this predicate
$q$. We say that the effectus has \emph{sharp
  images}\index{S}{image!sharp --} if these image predicates $\img(f)$
are sharp.

Like kernel-supplements $\kerbot$ we also uses image-complements
$\imgbot$ defined as $\imgbot(f) =
\img(f)^{\bot}$.\index{N}{$\imgbot(f)$, image-supplement of a map $f$}

We call $f$ an \emph{internal epi}\index{S}{internal!-- epi} if $\img(f)
= \one$.
\end{definition}

In all our examples images are sharp. Since many basic results about
images can be proven without assuming sharpness, we shall not use
sharp images until we really need them.

Let's see if we have images in our running examples.

\begin{example}
\label{ex:img}
In the effectus $\Sets$, each partial map $f\colon X \pto Y$ has an
image, namely the subset of $Y$ given by:
$$\begin{array}{rcl}
\img(f)
& = &
\setin{y}{Y}{\exin{x}{X}{f(x) = \kappa_{1}y}}.
\end{array}$$

\noindent It is easy to see that $\img(f)$ is the least subset
$V\subseteq Y$ with $\pbox{f}(V) = 1$, using the definition of
$\pbox{f}$ from~\eqref{eqn:partsubstSets}.

\auxproof{
$$\begin{array}{rcccccl}
\pbox{f}(\img(f))
& = &
\set{x}{\all{y}{f(x) = y \Rightarrow y\in \img(f)}} 
& = &
X
& = &
1 \in \Pow(X).
\end{array}$$

\noindent And if $\pbox{f}(V) = 1$, then for all $x$ and $y$ we
have $f(x) = y \Rightarrow y\in V$. We need to prove $\img(f)\subseteq
V$.  So let $y\in\img(f)$, say via $y = f(x)$. But then $y\in V$.
}

The effectus $\Kl(\Dst)$ also has images: for a partial map $f\colon X
\rightarrow \Dst(Y+1)$ take the (sharp) predicate $\img(f) \colon Y
\rightarrow [0,1]$ given by:
$$\begin{array}{rcl}
\img(f)(y)
& = &
\left\{\begin{array}{ll}
1 \quad & \mbox{if there is an $x\in X$ with $f(x)(y) > 0$} \\
0 & \mbox{otherwise.}
\end{array}\right.
\end{array}$$

\noindent We have, using the description of $\pbox{f}$
from~\eqref{eqn:partsubstKlD},
$$\begin{array}{rcl}
\pbox{f}(\img(f))(x)
& = &
\sum_{y} f(x)(y) \cdot \img(f)(y) + f(x)(*) \\
& = &
\sum_{y} f(x)(y) + f(x)(*) \\
& = &
1.
\end{array}$$

\noindent Further, if $q\in [0,1]^{Y}$ satisfies $\pbox{f}(q) = 1$,
then $\sum_{y} f(x)(y)\cdot q(y) + f(x)(*) = 1$ for each $x\in X$. But
this can only happen if $q(y) = 1$ if $f(x)(y) > 0$, that is, if
$\img(f) \leq q$.

In the effectus $\op{\OUG}$ of order unit groups, images need not
exist. To see this, let~$A=[0,1]^{2}$ be the unit square, which is
clearly a convex set.  The set $\mathrm{Aff}(A)$ of bounded affine
functions $A\to \mathbb{R}$ forms an order unit group with
coordinatewise operations and order.  Let $\varphi\colon
\mathrm{Aff}(A)\to \mathbb{R}$ be the map given by $\varphi(f)=f(0,0)$
for all $f\in\mathrm{Aff}(A)$.  We claim that~$\varphi$, seen as arrow
$\mathbb{R} \to \mathrm{Aff}(A)$ in $\op{\OUG}$, has no image.
Towards a contradiction, let $\varphi$ have image $f = \img(\varphi)$
in the unit interval of predicates $[0,1]_{\mathrm{Aff}(A)}$, see
Example~\ref{ex:effectusOUG}. Then:
$$\begin{array}{rcccccccl}
1
& = &
\pbox{\varphi}(f)
& \smash{\stackrel{\eqref{eqn:partsubstOUG}}{=}} &
\varphi(f) + \varphi(\one)^{\bot}
& = &
f(0,0) + \one(0,0)^{\bot}
& = &
f(0,0).
\end{array}$$

\noindent Consider the functions $f_{1}, f_{2} \in \mathrm{Aff}(A)$
given by $f_{1}(x,y) = 1-x$ and $f_{2}(x,y) = 1-y$. Clearly,
$\pbox{\varphi}(f_{1}) = f_{1}(0,0) = 1$ and $\pbox{\varphi}(f_{2}) =
f_{2}(0,0) = 1$. Hence $f \leq f_{1}, f_{2}$ by minimality of images.
Then $f(1,y) \leq f_{1}(1,y) = 0$ and $f(x,1) \leq f_{2}(x, 1) = 0$.
In particular $f(1,1) = 0$. We now consider the middle point
$(\frac{1}{2}, \frac{1}{2}) \in A$. It can be written in two ways as
convex combination of extreme points, namely as:
$$\begin{array}{rcccl}
\frac{1}{2}(1,0) + \frac{1}{2}(0,1)
& \smash{\stackrel{(a)}{=}} &
(\frac{1}{2}, \frac{1}{2})
& \smash{\stackrel{(b)}{=}} &
\frac{1}{2}(0,0) + \frac{1}{2}(1,1).
\end{array}$$

\noindent By applying the affine function $f$ on both sides we obtain
a contradiction. Starting from the above equation~$(a)$ we get:
$$\begin{array}{rcccccl}
f\big(\frac{1}{2}, \frac{1}{2}\big)
& = &
f\big(\frac{1}{2}(1,0) + \frac{1}{2}(0,1)\big)
& = &
\frac{1}{2}f(1,0) + \frac{1}{2}f(0,1)
& = &
0.
\end{array}$$

\noindent But starting from the equation~$(b)$ we obtain a different
outcome:
$$\begin{array}{rcccccl}
f\big(\frac{1}{2}, \frac{1}{2}\big)
& = &
f\big(\frac{1}{2}(0,0) + \frac{1}{2}(1,1)\big)
& = &
\frac{1}{2}f(0,0) + \frac{1}{2}f(1,1)
& = &
\frac{1}{2}.
\end{array}$$

\noindent The conclusion is that the map $\varphi$ in the effectus
$\op{\OUG}$ has no image.

In the effectus $\op{\vNA}$ of von Neumann algebra the image
of a subunital map $f\colon \mathscr{A} \rightarrow \mathscr{B}$ does
exist, and is given by the following (sharp) effect of $\mathscr{A}$.
\begin{equation}
\label{eqn:imgvNA}
\begin{array}{rcl}
\img(f)
& = &
\displaystyle\bigwedge\setin{p}{\mathscr{A}}{p 
   \mbox{ is a projection with } f(p) = f(1)}.
\end{array}
\end{equation}

\noindent Here we use that projections are the sharp elements and form
a complete lattice, see \textit{e.g.}~\cite[Prop.~1.10.2]{Sakai71}
One can prove: $f(a) = 0$ implies $\img(f) \leq a^\bot$,
for an arbitrary element $a\in
[0,1]_{\mathscr{A}}$, see also Lemma~\ref{lem:img}~\eqref{lem:img:post} below.

\auxproof{
If $f(a) = 0$, then $f(\floor{a}) = 0$, and so $f(\floor{a}^{\bot}) =
f(1 - \floor{a}) = f(1) - f(\floor{a}) = f(1)$.  Hence $\img(f) \leq
\floor{a}^{\bot} \leq a^{\bot}$.
}

It can be shown that in a von Neumann algebra the sharp elements are
precisely the projections, that is, the effects $e$ with $e\cdot e =
e$. Later on, in
Proposition~\ref{prop:telos}~\eqref{prop:telos:sharpidemp}, we prove
this in abstract form.
\end{example}

In general, images need not exist in an effectus, but they do exist
for a few special maps: identity maps, coprojections, and partial
projections.

\begin{lemma}
\label{lem:imgexists}
In the category $\Par(\cat{B})$ of partial maps of an effectus $\cat{B}$,
\begin{enumerate}
\item \label{lem:imgexists:id} $\img(\idmap) = \one$, for the identity
  map $\idmap \colon X \pto X$;

\item \label{lem:imgexists:zero} $\img(\zero) = \zero$, for the zero
  map $\zero \colon X \pto Y$;

\item \label{lem:imgexists:coproj} $\img(\kappa_{1}) = [\one,\zero]$
  and $\img(\kappa_{2}) = [\zero,\one]$;

\item \label{lem:imgexists:proj} $\img(\rhd_{1}) = \one$ and also
  $\img(\rhd_{2}) = \one$, so that the partial projections $\rhd_{i}$
  are internally epic.
\end{enumerate}
\end{lemma}

\begin{proof}
We show each time that the claimed predicate has the universal
property of an image (the least sharp one such that \ldots).
\begin{enumerate}
\item By Exercise~\ref{exc:subst}~\eqref{exc:subst:partialfun} we have
  $\pbox{\idmap}(p) = p$. Hence $\pbox{\idmap}(p) = \one$ iff
  $p=\one$, so that $\img(\idmap) = \one$.

\item Let $\pbox{0}(p) = \one$. Then $\one = [p, \kappa_{1}] \tafter
  \kappa_{2} \tafter \bang = \kappa_{1} \tafter \bang$. This equation
  imposes no restrictions on $p$, so the least predicate for which
  this holds is the falsity predicate $0$. Hence $\img(\zero) =
  \zero$.

\item Formally a coprojection in $\Par(\cat{B})$ is of the form
  $\klin{\kappa_{i}} = \kappa_{1} \tafter \kappa_{i}$. By
  Exercise~\ref{exc:subst}~\eqref{exc:subst:totalpartial} we have
  $\pbox{\klin{\kappa_{1}}}([\one,\zero]) =
  \tbox{\kappa_{1}}([\one,\zero]) = [\one,\zero] \tafter \kappa_{1} =
  \one$. Further, if also $\pbox{\klin{\kappa_{1}}}(p) = \one$, then
  $p \tafter \kappa_{1} = \one$. Since $p \after \kappa_{2} \geq
  \zero$, we get $p \geq [\one, \zero]$. Similarly one shows that
  $\img(\kappa_{2}) = [\zero,\one]$.

\auxproof{ 
We have $\pbox{\klin{\kappa_{2}}}([0,1]) =
\tbox{\kappa_{2}}([0,1]) = [0,1] \after \kappa_{2} = 1$. And if
also $\pbox{\klin{\kappa_{2}}}(p) = 1$, then $p \tafter \kappa_{2} = 1$,
so $p \geq [0,1]$.
}

\item Let $\pbox{\rhd_{1}}(p) = \one$. Then $\one = [p, \kappa_{1}]
  \tafter \rhd_{1} = [p, \kappa_{1}] \tafter (\idmap\tplus \bang) =
          [p, \kappa_{1} \tafter \bang] = [p, \one]$. But then $p =
          [p,\one] \tafter \kappa_{1} = \one \tafter \kappa_{1} =
          \one$. The $\rhd_{2}$-case is left to the reader. \QED

\auxproof{
Let $\pbox{\rhd_{1}}(p) = 1$. Then $1 = [p, \kappa_{1}] \tafter
[\kappa_{2} \tafter \bang, \kappa_{1}] = [\kappa_{1} \tafter \bang, p] =
[1,p]$. Hence $p = 1$.
}
\end{enumerate}
\end{proof}

Having seen this result we can say that the following is an `internal'
short exact sequence in the category of partial maps an effectus.
\begin{equation}
\label{diag:coprodshortexact}
\xymatrix{
0 \ar[r] & X\ar[r]^-{\kappa_{1}} & X+Y\ar[r]^-{\rhd_2} & Y\ar[r] & 0
}
\end{equation}

\noindent Exactness means that the image of one map is the kernel of
the next one. This works because the coprojections are total in the
category of partial maps, and thus internally monic, the partial
projections are internally epic, and $\img(\kappa_{1}) = [\one,\zero]
= \ker(\rhd_{2})$. But there is more: the
sequence~\ref{diag:coprodshortexact} involves two splittings, making
$\kappa_{1}$ a split mono, and $\rhd_{2}$ a split epi:
\begin{equation}
\label{diag:coprodshortexactsplit}
\xymatrix{
0 \ar[r] & X\ar[r]^-{\kappa_{1}} & 
   X+Y\ar[r]^-{\rhd_2}\ar@/^3ex/[l]^{\rhd_1} & 
   Y\ar[r]\ar@/^3ex/[l]^{\kappa_{2}} & 0
}
\end{equation}

\noindent This is a reformulation of the `butterfly'
diagram~\eqref{diag:butterfly}.

\begin{lemma}
\label{lem:img}
In an effectus $\cat{B}$ with images one has:
\begin{enumerate}
\item \label{lem:img:ker} $\ker(f) = \pbox{f}\big(\imgbot(f)\big)$;

\item \label{lem:img:kerone} $\img(f) = \zero$ iff $\ker(f) = \one$;

\item \label{lem:img:post} $p \pafter f = \zero$ iff $p \leq
  \imgbot(f)$; in particular, $\imgbot(f) \pafter f = \zero$;

\item \label{lem:img:comp} $\img(g \pafter f) \leq \img(g)$;

\item \label{lem:img:compsquare$} $\img(f) \leq \pbox{g}\big(\img(g
  \pafter f)\big)$;

\item \label{lem:img:compeq} $\img(g \pafter f) = \img(g)$ if $f$ is
  a internal epi;

\item \label{lem:img:compzero} $g \pafter f = \zero$ iff $\img(f) \leq
  \ker(g)$;

\item \label{lem:img:compintepic} if $g\pafter f$ and $f$ are
  internally epic, then so is $g$;

\item \label{lem:img:cat} internal epis form a subcategory;

\item \label{lem:img:epichar} a partial map $f$ is internally epic iff
  for each predicate $q$ on its codomain: $\pbox{f}(q) = \one
  \Longleftrightarrow q = \one$;

\item \label{lem:img:ext} `external' epis are internal epis: each epi
  $f$ in $\Par(\cat{B})$ satisfies $\img(f) = \one$;

\item \label{lem:img:join} $\img([f,g]) = \img(f) \vee \img(g)$, for
  predicates $f,g$ with the same codomain.
\end{enumerate}
\end{lemma}

\begin{proof}
All of the points are obtained by elementary arguments.
\begin{enumerate}
\item The equation $\ker(f) = \pbox{f}\big(\imgbot(f)\big)$
  follows directly from Lemma~\ref{lem:ker}~\eqref{lem:ker:squareneg}:
$$\begin{array}{rcccccl}
\pbox{f}\big(\imgbot(f)\big)
& = &
\pbox{f}(\img(f))^{\bot} \ovee \ker(f)
& = &
\one^{\bot} \ovee \ker(f)
& = &
\ker(f).
\end{array}$$

\item If $\img(f) = 0$, then by the previous point:
$$\begin{array}{rcccccccl}
\ker(f)
& = &
\pbox{f}\big(\imgbot(f)\big)
& = &
\pbox{f}(\zero^{\bot})
& = &
\pbox{f}(\one)
& = &
\one.
\end{array}$$

\noindent Conversely, if $1 = \ker(f) = \pbox{f}(0)$, then 
$\img(f) \leq 0$ since it is minimal.

\item First, we have:
$$\begin{array}{rcl}
\imgbot(f) \pafter f
& = &
[\imgbot(f), \kappa_{2}] \tafter f \\
& = &
[\kappa_{2}, \kappa_{1}] \tafter [\img(f), \kappa_{1}] \tafter f 
\hspace*{\arraycolsep}=\hspace*{\arraycolsep}
\pbox{f}(\img(f))^{\bot}
\hspace*{\arraycolsep}=\hspace*{\arraycolsep}
\one^{\bot}
\hspace*{\arraycolsep}=\hspace*{\arraycolsep}
\zero.
\end{array}$$

\noindent Hence if $p \leq \imgbot(f)$, then $p \pafter f \leq
\imgbot(f) \pafter f = \zero$, using Lemma~\ref{lem:partialmonotone}.

In the other direction, if $p \pafter f = \zero$, then:
$$\begin{array}{rcccccl} \pbox{f}(p^{\bot}).  & = & (p \pafter
  f)^{\bot} & = & \zero^{\bot} & = & \one.
\end{array}$$

\noindent Hence $\img(f) \leq p^{\bot}$ by minimality of images,
and so $p \leq \img(f)^{\bot} = \imgbot(f)$.

\item In order to prove $\img(g \pafter f) \leq \img(g)$ it suffices
  to show $\pbox{(g \pafter f)}(\img(g)) = \one$. But the latter is
  easy, since $\pbox{(g \pafter f)}(\img(g)) =
  \pbox{f}\big(\pbox{g}(\img(g))\big) = \pbox{f}(\one) = \one$.

\item The inequality $\img(f) \leq \pbox{g}\big(\img(g \pafter
  f)\big)$ follows by minimality of images from:
$$\begin{array}{rcccl}
\pbox{f}\Big(\pbox{g}\big(\img(g \pafter f)\big)\Big)
& = &
\pbox{(g \pafter f)}\big(\img(g \pafter f)\big)
& = &
\one.
\end{array}$$

\item By point~\eqref{lem:img:comp} we only have to prove $\img(g
  \pafter f) \geq \img(g)$. This follows if we can show
  $\pbox{g}\big(\img(g \pafter f)\big) = \one$. But this is a
  consequence of the previous point, since $\img(f)=\one$ because $f$
  is by assumption an internal epi.

\item Let $g \pafter f = \zero$. Then, by
  Lemma~\ref{lem:ker}~\eqref{lem:ker:comp}, we have $\pbox{f}(\ker(g))
  = \ker(g \pafter f) = \ker(\zero) = \one$, so that $\img(f) \leq
  \ker(g)$.  Conversely, if $\img(f) \leq \ker(g)$, then $\one =
  \pbox{f}(\img(f)) \leq \pbox{f}(\ker(g)) = \ker(g \pafter f)$. But
  then $g \pafter f = \zero$ by Lemma~\ref{lem:ker}~\eqref{lem:ker:zero}.

\item Let $g\pafter f$ and $f$ be internally epic. Then $g$ is
  internally epic too, by point~\eqref{lem:img:compeq} since: $\img(g)
  = \img(g \pafter f) = \one$.

\item Lemma~\ref{lem:imgexists}~\eqref{lem:imgexists:id} says that the
  partial identity is internally epic. And if $g, f$ are internal
  epis, then by point~\eqref{lem:img:compeq}, $g \pafter f$ too, since
  $\img(g \pafter f) = \img(g) = \one$.

\item Let $f$ be internally epic. Clearly, if $q = \one$, then
  $\pbox{f}(q) = \pbox{f}(\one) = \one$. In the other direction, if
  $\pbox{f}(q) = \one$, then, by minimality of images, $\one = \img(f)
  \leq q$.

Conversely, assume $\pbox{f}(q) = \one \Longleftrightarrow q=\one$ for
each predicate $q$.  Take $q = \img(f)$; since by definition
$\pbox{f}(\img(f)) = \one$, we get $q=\img(f)=\one$, so that $f$ is
internally epic.

\item Let $f\colon X \pto Y$ be epic in $\Par(\cat{B})$.  By
  point~\eqref{lem:img:post} we have $\imgbot(f) \pafter f = \zero$.
  But also $0 \pafter f = \zero$. Hence $\imgbot(f) = \zero$ since $f$
  is epic, and thus $\img(f) = \one$, so that $f$ is internally
  epic.

\item We use point~\eqref{lem:img:post} to show that for an arbitrary
  predicate $p$,
$$\begin{array}[b]{rcl}
\img([f, g]) \leq p
& \Longleftrightarrow &
p^{\bot} \leq \imgbot([f,g]) \\
& \Longleftrightarrow &
p^{\bot} \pafter [f,g] = [p^{\bot} \pafter f, p^{\bot} \pafter g] = \zero \\
& \Longleftrightarrow &
p^{\bot} \pafter f = \zero \mbox{ and } p^{\bot} \pafter g = \zero \\
& \Longleftrightarrow &
p^{\bot} \leq \imgbot(f) = \zero \mbox{ and } p^{\bot} \leq \imgbot(g) = \zero \\
& \Longleftrightarrow &
\img(f) \leq p \mbox{ and } \img(g) \leq p.
\end{array}\eqno{\QEDbox}$$
\end{enumerate}
\end{proof}

\section{Relating total and partial maps}\label{sec:partialtotal}

This section contains the main result of~\cite{Cho15}, giving a
precise relation between total and partial maps in an effectus.
Briefly certain requirements $\mathcal{R}$ are identified so that:
\begin{itemize}
\item for a category $\cat{C}$ satisfying $\mathcal{R}$, the
  subcategory $\Tot(\cat{C})$ of `total' maps in $\cat{C}$ is an
  effectus, and $\Par(\Tot(\cat{C})) \cong \cat{C}$;

\item for an effectus $\cat{B}$, the category of partial maps
  $\Par(\cat{B})$ satisfies $\mathcal{R}$, and $\Tot(\Par(\cat{B}))
  \cong \cat{B}$.
\end{itemize}

\noindent A category satisfying $\mathcal{R}$ is called a \emph{FinPAC
  with effects} in~\cite{Cho15}. This terminology is explained below.

This correspondence is made precise in~\cite{Cho15} in the form of a
2-equivalence of 2-categories. Here however, we discuss the essentials
in more concrete form, and refer the reader to \textit{loc.\ cit.} for
further details. This close correspondence between total and partial
maps can be seen as a confirmation of the appropriateness of the
notion of effectus. The correspondence and its consquences --- for
notation and terminology --- are discussed at the end of this section.

We first have to explore the notion of finitely partially additive
category, or FinPAC.  This is based on the notion of partially
additive category, introduced in~\cite{ArbibM1980}.

\subsection{FinPACs}\label{subsec:FinPAC}

Recall that an arbitrary category $\cat{C}$ is $\PCM$-enriched if all
its homsets are PCMs with $\ovee, 0$, and pre- and post-composition
preserve the PCM-structure.  If such a category $\cat{C}$ has
coproducts $+$ too, then we can define `projections' $\rhd_{i}$ via
the zero map $\zero$, like in~\eqref{diag:partprojpart}, as:
\begin{equation}
\label{diag:partprojFinPAC}
\vcenter{\xymatrix@C+2pc{
X & X+Y\ar[l]_-{\rhd_{1} = [\idmap, \zero]}\ar[r]^-{\rhd_{2} = [\zero, \idmap]} 
   & Y
}}
\end{equation}

\noindent These projections are automatically natural: $\rhd_{i}
\after (f_{1}+f_{2}) = f_{i} \after \rhd_{i}$.

\auxproof{
$$\begin{array}{rcl}
\rhd_{1} \after (f_{1}+f_{2})
& = &
[\idmap, 0] \after (f_{1}+f_{2}) \\
& = &
[f_{1}, 0 \after f_{2}] \\
& = &
[f_{1}, 0] \\
& = &
f_{1} \after [\idmap, 0] \\
& = &
f_{1} \after \rhd_{1}.
\end{array}$$
}

Since $\cat{C}$ is any category --- not necessarily an effectus --- we
shall use ordinary notation, like $\after$ and $\rightarrow$, for
composition and maps, and not the special notation $\tafter, \pafter$
and $\tto, \pto$ for effectuses.

\begin{definition}
\label{def:FinPAC}
A finitely partially additive category\index{S}{finitely additive
  category}\index{S}{category!finitely additive --} (a
FinPAC,\index{S}{FinPAC|see{finitely additive category}} for short) is a
$\PCM$-enriched category $\cat{C}$ with finite coproducts $(+, 0)$
satisfying both the:
\begin{itemize}
\item \emph{Compatible Sum Axiom}: if there is a bound $b \colon X
  \rightarrow Y+Y$ for maps $f,g \colon X \rightarrow Y$ with
  $\rhd_{1} \after b = f$ and $\rhd_{2} \after b = g$, then $f
  \orthogonal g$ in the PCM $\cat{C}(X, Y)$;

\item \emph{Untying Axiom}: if $f \orthogonal g$ then $(\kappa_{1}
\after f) \orthogonal (\kappa_{2} \after g)$.
\end{itemize}
\end{definition}

The names for these two axioms come from the theory of partially
additive categories, see~\cite{ArbibM1980,ManesA86}
or~\cite{Haghverdi00}. The maps $\rhd_{i}$ are often called
\emph{quasi} projections in that context.

We need some basic results about FinPACs,
following~\cite{ManesA86,Cho15}.

\begin{lemma}
\label{lem:FinPAC}
In a FinPAC $\cat{C}$,
\begin{enumerate}
\item \label{lem:FinPAC:zero} The initial object $0\in\cat{C}$ is also
  final, and thus a zero object; the resulting zero map $X \rightarrow
  0 \rightarrow Y$ is the zero element $\zero$ of the PCM-structure on
  the homset $\cat{C}(X, Y)$;

\item \label{lem:FinPAC:id} the maps $\kappa_{1} \after \rhd_{1},
  \kappa_{2} \after \rhd_{2} \colon X+Y \rightarrow X+Y$ are
  orthogonal with sum $(\kappa_{1} \after \rhd_{1}) \ovee (\kappa_{2}
  \after \rhd_{2}) = \idmap$;

\item \label{lem:FinPAC:f} any map $f \colon Z \rightarrow X+Y$ can be
  written as sum $f = (\kappa_{1} \after \rhd_{1} \after f) \ovee
  (\kappa_{2} \after \rhd_{2} \after f)$;

\item \label{lem:FinPAC:jm} the two projection maps $\rhd_{i} \colon
  X_{1} + X_{2} \rightarrow X_{i}$ are jointly monic;

\item \label{lem:FinPAC:ortho} $f_{1} \orthogonal f_{2}$ iff there is
  a necessarily unique bound $b$ with $\rhd_{i} \after b = f_{i}$ and
  $f_{1} \ovee f_{2} = \nabla \after b$.
\end{enumerate}
\end{lemma}

\begin{proof}
For the sake of completeness, we include the details.
\begin{enumerate}
\item For each object $X\in\cat{C}$ there is a map $\zero \colon X
  \rightarrow 0$, namely the PCM-zero. It is the only such map, since
  each $f\colon X \rightarrow 0$ satisfies: $f = \idmap \after f = 0
  \after f = 0$. The resulting map $X \rightarrow 0 \rightarrow Y$ is
  thus $\bang \after \zero = \zero$.

\item Take as bound $b = \kappa_{1}+\kappa_{2} \colon X+Y \rightarrow
  (X+Y)+(X+Y)$. Then: 
$$\begin{array}{rcccccccl}
\rhd_{1} \after b
& = &
[\idmap, \zero] \after (\kappa_{1}+\kappa_{2})
& = &
[\kappa_{1}, \zero]
& = &
\kappa_{1} \after [\idmap, \zero]
& = &
\kappa_{1} \after \rhd_{1}.
\end{array}$$

\noindent Similarly one obtains $\rhd_{2} \after b = \kappa_{2} \after
\rhd_{2}$. This gives $(\kappa_{1} \after \rhd_{1}) \orthogonal
(\kappa_{2} \after \rhd_{2})$ by the Compatible Sum Axiom. We take
their sum $s = (\kappa_{1} \after \rhd_{1}) \ovee (\kappa_{2} \after
\rhd_{2})$ in the homset of maps $X+Y \rightarrow X+Y$ and obtain $s =
\idmap$ from $s \after \kappa_{i} = \kappa_{i}$, as in:
$$\begin{array}{rcl}
s \after \kappa_{2}
& = &
\big((\kappa_{1} \after \rhd_{1}) \ovee (\kappa_{2} \after \rhd_{2})\big) 
   \after \kappa_{2} \\
& = &
(\kappa_{1} \after [\idmap, \zero] \after \kappa_{2}) \ovee 
   (\kappa_{2} \after [\zero, \idmap] \after \kappa_{2}) \\
& = &
(\kappa_{1} \after \zero) \ovee \kappa_{2} \\
& = &
0 \ovee \kappa_{2} \\
& = &
\kappa_{2}.
\end{array}$$

\item Directly by the previous point, for $f\colon Z \rightarrow X+Y$,
$$\!\!\begin{array}{rcccccl}
f
& = &
\idmap \after f
& = &
\big((\kappa_{1} \after \rhd_{1}) \ovee (\kappa_{2} \after \rhd_{2})\big) \after f
& = &
(\kappa_{1} \after \rhd_{1} \after f) \ovee 
   (\kappa_{2} \after \rhd_{2} \after f).
\end{array}$$

\item Suppose $f,g\colon Z \rightarrow X_{1}+X_{2}$ satisfy $\rhd_{i}
  \after f = \rhd_{i} \after g$, for $i=1,2$. Then, by the previous
  point:
$$\begin{array}{rcccccl}
f
& = &
(\kappa_{1} \after \rhd_{1} \after f) \ovee 
   (\kappa_{2} \after \rhd_{2} \after f)
& = &
(\kappa_{1} \after \rhd_{1} \after g) \ovee 
   (\kappa_{2} \after \rhd_{2} \after g)
& = &
g.
\end{array}$$

\item If a bound exists, then it is unique because the $\rhd_{i}$ are
  jointly monic. The Compatible Sum Axiom says that existence of a
  bound $b$ for $f_{1}, f_{2}$ gives $f_{1} \orthogonal f_{2}$. We
  have to prove the converse. So let $f_{1} \orthogonal f_{2}$, and
  thus $(\kappa_{1} \after f_{1}) \orthogonal (\kappa_{2} \after
  f_{2})$ by the Untying Axiom. We take $b = (\kappa_{1} \after f_{1})
  \ovee (\kappa_{2} \after f_{2})$. Then $\rhd_{i} \after b = f_{i}$,
  making $b$ a bound. Moreover:
$$\begin{array}{rcccl}
\nabla \after b
& = &
(\nabla \after \kappa_{1} \after f_{1}) \ovee 
   (\nabla \after \kappa_{2} \after f_{2})
& = &
f_{1} \ovee f_{2}.
\end{array}\eqno{\QEDbox}$$

\auxproof{
$$\begin{array}{rcccccl}
\rhd_{1} \after b
& = &
(\rhd_{1} \after \kappa_{1} \after f_{1}) \ovee 
   (\rhd_{1} \after \kappa_{2} \after f_{2})
& = &
f_{1} \ovee (0 \after f_{2})
& = &
f_{1}.
\end{array}$$
}
\end{enumerate}
\end{proof}

The next result shows the relevance of FinPACs in the current setting.

\begin{lemma}
\label{lem:effectusFinPAC}
The category $\Par(\cat{B})$ of partial maps in an effectus $\cat{B}$
is a FinPAC.
\end{lemma}

\begin{proof}
We know that $\Par(\cat{B})$ is enriched over $\PCM$ by
Proposition~\ref{prop:effectusPCM}, and inherits coproducts from
$\cat{B}$, like any Kleisli category. The Compatible Sum Axiom holds
by definition of orthogonality, see
Proposition~\ref{prop:effectusPCM}~\eqref{prop:effectusPCM:hom}.  For
the Untying Axiom, let $f_{1}, f_{2} \colon X \pto Y$ satisfy $f_{1}
\orthogonal f_{2}$ via bound $b \colon X \pto Y+Y$. Then $c =
(\kappa_{1}\pplus\kappa_{2}) \pafter b \colon X \pto (Y+Y)+(Y+Y)$ is a
bound for $\kappa_{i} \after f_{i}$ since:
$$\begin{array}{rcccccl}
\rhd_{i} \pafter c
& = &
\rhd_{i} \pafter (\kappa_{1}\pplus\kappa_{2}) \pafter b
& = &
\kappa_{i} \pafter \rhd_{i} \pafter b
& = &
\kappa_{i} \pafter f_{i}.
\end{array}\eqno{\QEDbox}$$
\end{proof}

\subsection{FinPACs with effects}\label{subsec:FinPACwE}

We now come to the axiomatisation of the category of partial maps in
an effectus. We use the name `FinPAC with effects' from~\cite{Cho15}.
This is a temporary name, see Discussion~\ref{dis:partialtotal} below.

\begin{definition}
\label{def:FinPACwE}
A category $\cat{C}$ is called a \emph{FinPAC with
  effects}\index{S}{FinPAC!-- with effects} if it is a FinPAC with a
special object $I\in\cat{C}$ such that:
\begin{enumerate}
\item \label{def:FinPACwE:EA} the homset $\cat{C}(X,I)$ is not only a
  PCM, but an effect algebra, for each object $X\in\cat{C}$;

\item \label{def:FinPACwE:ortho} the top/truth element $\one \in \cat{C}(X,
  I)$ satisfies: for all $f,g \colon Y \rightarrow X$,
$$\begin{array}{rcl}
(\one \after f) \orthogonal (\one \after g)
& \Longrightarrow &
f \orthogonal g
\end{array}$$

\item \label{def:FinPACwE:zero} the bottom/falsity element $\zero\in
  \cat{C}(X, I)$ satisfies: for all $f\colon Y \rightarrow X$,
$$\begin{array}{rcl}
\one \after f = \zero
& \Longrightarrow &
f = \zero.
\end{array}$$
\end{enumerate}

\noindent These last two points say that the function $\one \after
(-)$ reflects orthogonality and zero.

A map $f\colon X \rightarrow Y$ in such a FinPAC with effects is
called \emph{total}\index{S}{total map!-- in a FinPAC with effects} if
$\one_{Y} \after f = \one_{X}$. We write $\Tot(\cat{C})
\hookrightarrow \cat{C}$\index{N}{cat@$\Tot(\cat{C})$, category of
  total maps in a FinPAC with effects $\cat{C}$} for the `wide'
subcategory (with the same objects) of total maps in $\cat{C}$.
\end{definition}

These top maps $\one \colon X \rightarrow I$ resemble the ground maps
$\ground \colon X \rightarrow I$ in Definition~\ref{def:biprodground}.
Recall that causal maps $f$ satisfy $\ground \after f = \ground$.  The
corresponding property $\one \after f = \one$ describes the total
maps, as defined above.

Before arriving at the main result of this section, we collect a few
facts about FinPACs with effects.

\begin{lemma}
\label{lem:FinPACwE}
In an FinPAC with effects $(\cat{C}, I)$,
\begin{enumerate}
\item \label{lem:FinPACwE:splitmonic} split monics are total, so
  in particular all coprojections and isomorphisms are total;

\item \label{lem:FinPACwE:idI} $\idmap[I] = \one \colon I \rightarrow
  I$; as a result, $I\in\cat{C}$ is final in $\Tot(\cat{C})$;

\item \label{lem:FinPACwE:onecoptuple} $[\one, \one] = \one \colon X+Y
  \rightarrow I$;

\item \label{lem:FinPACwE:pres} for each total map $f\colon X
  \rightarrow Y$ pre-composition $(-) \after f \colon \cat{C}(Y, I)
  \rightarrow \cat{C}(X,I)$ is a map of effect algebras;

\item \label{lem:FinPACwE:coprod} $\Tot(\cat{C})$ inherits finite
  coproducts $(+,0)$ from $\cat{C}$.
\end{enumerate}
\end{lemma}

\begin{proof}
We reason in the category $\cat{C}$.
\begin{enumerate}
\item Let $f \after m = \idmap$, making $m$ a split monic. We have
  $\one \after f \leq \one$, since $1$ is by definition the top
  element. Hence by post-composing with $m$ we get: $\one = \one
  \after f \after m \leq \one \after m$, so that $\one = \one \after
  m$.  Coprojections $\kappa_{i}$ are split monics in $\cat{C}$, since
  $\rhd_{i} \after \kappa_{i} = \idmap$.

\item In the effect algebra $\cat{C}(I, I)$ we have $\idmap
  \orthogonal \idmap^{\bot}$. Hence $(\one \after \idmap) \orthogonal
  (\one \after \idmap^{\bot})$ since post-composition is a
  PCM-map. But then $\one \after \idmap^{\bot} = \zero$ by
  Definition~\ref{def:EA}~\eqref{def:EA:ortho}. This gives
  $\idmap^{\bot} = \zero$ by
  Definition~\ref{def:FinPACwE}~\eqref{def:FinPACwE:zero}, and thus
  $\idmap = \one$.

For each object $X$ there is a total map $\one_{X} \colon X
\rightarrow I$, since $\one_{X} = \idmap[I] \after \one_{X} = \one_{I}
\after \one_{X}$. If $f\colon X \rightarrow I$ is total, then $f =
\idmap[I] \after f = \one_{I} \after f = \one_{X}$.

\item For the equation $[\one,\one] = \one \colon X+Y \rightarrow I$
  we use that coprojections are total, by
  point~\eqref{lem:FinPACwE:splitmonic}:
$$\begin{array}{rcccccl}
\one
& = &
\one \after [\kappa_{1}, \kappa_{2}]
& = &
[\one \after \kappa_{1}, \one \after \kappa_{2}]
& = &
[\one,\one].
\end{array}$$

\item The map $(-) \after f$ preserves the PCM-structure by
  definition.  And it preserves truth $\one$ since $f$ is total. Hence
  it is a map of effect algebras, see Definition~\ref{def:EA}.

\item The object $0$ is initial in $\Tot(\cat{C})$, since the unique
  map $\bang \colon 0 \rightarrow X$ is total: $\one_{X} \after
  \bang_{X} = \bang_{I} = \one_{0}$ by initiality in
  $\cat{C}$. Coprojections are total by
  point~\eqref{lem:FinPACwE:splitmonic}. If $f,g$ are total, then
  so is $[f,g]$ since $\one \after [f,g] = [\one \after f, \one \after
    g] = [\one,\one] = \one$ by
  point~\eqref{lem:FinPACwE:onecoptuple}. \QED
\end{enumerate}
\end{proof}

We now come to the main result of this section.

\begin{theorem}
\label{thm:partialtotal}
\emph{(From~\cite{Cho15})}
\begin{enumerate}
\item For an effectus $\cat{B}$, the category of partial maps
  $\Par(\cat{B})$ with special object $1$ is a FinPAC with effects,
  and $\Tot(\Par(\cat{B})) \cong \cat{B}$.

\item For a FinPAC with effects $(\cat{C}, I)$, the subcategory
  $\Tot(\cat{C})$ of total maps is an effectus, and
  $\Par(\Tot(\cat{C})) \cong \cat{C}$.
\end{enumerate}
\end{theorem}

\begin{proof}
Let $\cat{B}$ be an effectus. Lemma~\ref{lem:effectusFinPAC} tells
that $\Par(\cat{B})$ is a FinPAC. We take $I = 1$, so that
$\Par(\cat{B})(X, 1) = \cat{B}(X, 1+1) = \Pred(X)$ is an effect
algebra. Next if $\one \pafter f = \kerbot(f) \orthogonal \kerbot(g) =
\one \pafter g$, then $f \orthogonal g$ since $\kerbot$ reflects
orthogonality --- and zero too --- by
Lemma~\ref{lem:kerbot}. Reflection of zero proves
requirement~\eqref{def:FinPACwE:zero} in
Definition~\ref{def:FinPACwE}.  In order to prove $\Tot(\Par(\cat{B}))
\cong \cat{B}$ we have to prove that a map $f\colon X \pto Y$ is total
iff $\one \pafter f = \one$.  But we already know this from
Lemma~\ref{lem:zero}.

For the second point, let $(\cat{C}, I)$ be a FinPAC with effects.
From Lemma~\ref{lem:FinPACwE} we know that the category
$\Tot(\cat{C})$ of total maps has $I$ as finial object, and has
coproducts $(+,0)$ as in $\cat{C}$. We first show that the rectangles
in Definition~\ref{def:effectus}~\eqref{def:effectus:pb} are pullbacks
in $\Tot(\cat{C})$. We may thus assume that we have total maps $f,g,h$
in commuting (outer) diagrams:
$$\xymatrix{
Z\ar@/_2ex/[ddr]_{g}\ar@/^2ex/[drr]^-{f}\ar@{..>}[dr]^(0.6){k}
& & &
W\ar@/_2ex/[ddr]_{h}\ar@/^2ex/[drr]^-{\one}\ar@{..>}[dr]^(0.6){\rhd_{1} \after h}
\\
& X+Y\ar[r]^-{\idmap + \one}\ar[d]^{\one+\idmap} & X + I\ar[d]^{\one+\idmap}
&
& X\ar[r]^-{\one}\ar[d]_{\kappa_1} & I\ar[d]^{\kappa_1}
\\
& I + Y\ar[r]_-{\idmap+\one} & I + I
& 
& X+Y\ar[r]_-{\one+\one} & I + I
}$$

\noindent We first concentrate on the situation on the left. By
assumption, $(\one+\idmap) \after f = (\idmap+\one) \after g = b$,
say. Then $(\rhd_{1} \after b) \orthogonal (\rhd_{2} \after b)$, by
definition of orthogonality. But:
$$\begin{array}{rcccccl}
\rhd_{1} \after b
& = &
\rhd_{1} \after (\one+\idmap) \after f
& = &
\one \after \rhd_{1} \after f
& = &
\one \after \kappa_{1} \after \rhd_{1} \after f 
\\
\rhd_{2} \after b
& = &
\rhd_{2} \after (\idmap+\one) \after g
& = &
\one \after \rhd_{2} \after g
& = &
\one \after \kappa_{2} \after \rhd_{2} \after g.
\end{array}$$

\noindent Hence we have $(\one \after \kappa_{1} \after \rhd_{1}
\after f) \orthogonal (\one \after \kappa_{2} \after \rhd_{2} \after
g)$, from which we can conclude $(\kappa_{1} \after \rhd_{1} \after f)
\orthogonal (\kappa_{2} \after \rhd_{2} \after g)$ using
Definition~\ref{def:FinPACwE}~\eqref{def:FinPACwE:ortho}. Thus we can
define:
$$\begin{array}{rcl}
k
& = &
(\kappa_{1} \after \rhd_{1} \after f) \ovee
   (\kappa_{2} \after \rhd_{2} \after g) \;\colon\; W \longrightarrow X+Y.
\end{array}$$

\noindent Then:
$$\begin{array}{rcl}
(\idmap+\one) \after k
& = &
((\idmap+\one) \after \kappa_{1} \after \rhd_{1} \after f) \ovee
   ((\idmap+\one) \after \kappa_{2} \after \rhd_{2} \after g) \\
& = &
(\kappa_{1} \after \rhd_{1} \after f) \ovee
   (\kappa_{2} \after \one \after \rhd_{2} \after g) \\
& = &
(\kappa_{1} \after \rhd_{1} \after f) \ovee
   (\kappa_{2} \after \rhd_{2} \after (\idmap+\one) \after g) \\
& = &
(\kappa_{1} \after \rhd_{1} \after f) \ovee
   (\kappa_{2} \after \rhd_{2} \after (\one+\idmap) \after f) \\
& = &
(\kappa_{1} \after \rhd_{1} \after f) \ovee
   (\kappa_{2} \after \rhd_{2} \after f) \\
& = &
f \qquad \mbox{by Lemma~\ref{lem:FinPAC}~\eqref{lem:FinPAC:f}.}
\end{array}$$

\noindent Similarly one proves $(\one+\idmap) \after k = g$. For
uniqueness, let $\ell \colon W \rightarrow X+Y$ also satisfy
$(\idmap+\one) \after \ell = f$ and $(\one+\idmap) \after \ell = g$, then:
$$\begin{array}{rcl}
k
& = &
(\kappa_{1} \after \rhd_{1} \after f) \ovee
   (\kappa_{2} \after \rhd_{2} \after g) \\
& = &
(\kappa_{1} \after \rhd_{1} \after (\idmap+\one) \after \ell) \ovee
   (\kappa_{2} \after \rhd_{2} \after (\one+\idmap) \after \ell) \\
& = &
(\kappa_{1} \after \rhd_{1} \after \ell) \ovee
   (\kappa_{2} \after \rhd_{2} \after \ell) \\
& = &
\ell.
\end{array}$$

In the above diagram on the right we have $(\one+\one) \after h =
\kappa_{1} \after \one$. We claim $\rhd_{2} \after h = \zero$. This
follows by Definition~\ref{def:FinPACwE}~\eqref{def:FinPACwE:zero}
from
$$\begin{array}{rcccccccl}
\one \after \rhd_{2} \after h
& = &
\rhd_{2} \after (\one+\one) \after h
& = &
\rhd_{2} \after \kappa_{1} \after \one
& = &
\zero \after \one
& = &
\zero.
\end{array}$$

\noindent The map $h\colon W \rightarrow X+Y$ then satisfies,
by Lemma~\ref{lem:FinPAC}~\eqref{lem:FinPAC:f},
$$\begin{array}{rcccccl}
h 
& = &
(\kappa_{1} \after \rhd_{1} \after h) \ovee
   (\kappa_{2} \after \rhd_{2} \after h) 
& = &
(\kappa_{1} \after \rhd_{1} \after h) \ovee (\kappa_{2} \after \zero)
& = &
\kappa_{1} \after \rhd_{1} \after h.
\end{array}$$

\noindent Hence $\rhd_{1} \after h \colon W \rightarrow X$ is a
mediating map. It is the unique one, since if also $k\colon W
\rightarrow X$ satisfies $\kappa_{1} \after k = h$, then $\rhd_{1}
\after h = \rhd_{1} \after \kappa_{1} \after k = k$.

We still have to prove that the two maps $\IV, \XI \colon (I+I)+I
\rightarrow I+I$ are jointly monic in $\Tot(\cat{C})$, where $\IV =
            [\idmap,\kappa_{2}]$ and $\XI = [[\kappa_{2}, \kappa_{1}],
              \kappa_{2}]$. So let $f,g \colon X \rightarrow (I+I) +
            I$ satisfy $\IV \after f = \IV \after g$ and $\XI \after f
            = \XI \after g$. Write:
$$\begin{array}{rclcrclcrcl}
f_{1}
& = &
\rhd_{1} \after \rhd_{1} \after f
& \qquad &
f_{2}
& = &
\rhd_{2} \after \rhd_{1} \after f
& \qquad &
f_{3}
& = &
\rhd_{2} \after f.
\end{array}$$

\noindent Then $f = (\kappa_{1} \after \kappa_{1} \after f_{1}) \ovee
(\kappa_{1} \after \kappa_{2} \after f_{1}) \ovee (\kappa_{2} \after
f_{3})$. We can write the map $g$ in a similar manner. The
equation $\IV \after f = \IV \after g$ yields,
$$\begin{array}{rcl}
(\kappa_{1} \after f_{1}) \ovee (\kappa_{2} \after f_{2}) \ovee 
   (\kappa_{2} \after f_{3})
& = &
(\kappa_{1} \after g_{1}) \ovee (\kappa_{2} \after g_{2}) \ovee 
   (\kappa_{2} \after g_{3}).
\end{array}$$

\noindent Hence by post-composing with $\rhd_{1}$ and with $\rhd_{2}$ we
get:
$$\begin{array}{rclcrcl}
f_{1}
& = &
g_{1}
& \qquad\mbox{and}\qquad &
f_{2} \ovee f_{3}
& = &
g_{2} \ovee g_{3}.
\end{array}$$

\noindent Similarly, the equation $\XI \after f = \XI \after g$
yields:
$$\begin{array}{rcl}
(\kappa_{2} \after f_{1}) \ovee (\kappa_{1} \after f_{2}) \ovee 
   (\kappa_{2} \after f_{3})
& = &
(\kappa_{2} \after g_{1}) \ovee (\kappa_{1} \after g_{2}) \ovee 
   (\kappa_{2} \after g_{3}).
\end{array}$$

\noindent Post-composing with $\rhd_{1}$ yields $f_{2} = g_{2}$.  By
substitution in our previous finding we get $f_{2} \ovee f_{3} = f_{2}
\ovee g_{3}$. Cancellation in the effect algebra $\cat{C}(X, I)$ gives
$f_{3} = g_{3}$. Hence $f = g$.

Finally we show that we have an identity-on-objects, full and faithful
functor $F\colon \Par(\Tot(\cat{C})) \rightarrow \cat{C}$, defined on
maps by $F(g) = \rhd_{1} \after g$. It is easy to see that $F$ is
indeed a functor. We construct an inverse functor $G \colon \cat{C}
\rightarrow \Par(\Tot(\cat{C}))$.

\auxproof{
$$\begin{array}{rcl}
F(\idmap)
& = &
F(\kappa_{1}) \\
& = &
\rhd_{1} \after \kappa_{1} \\
& = &
\idmap
\\
F(h \pafter g)
& = &
\rhd_{1} \after [h, \kappa_{2}] \after g \\
& = &
[\rhd_{1} \after h, \rhd_{1} \after \kappa_{2}] \after g \\
& = &
[F(h), 0] \after g \\
& = &
F(h) \after [\idmap, 0] \after g \\
& = &
F(h) \after \rhd_{1} \after g \\
& = &
F(h) \after F(g).
\end{array}$$
}

Let $f\colon X \rightarrow Y$ be a map in $\cat{C}$. We form $\one
\after f \colon X \rightarrow I$, so that $(\one \after f) \orthogonal
(\one \after f)^{\bot}$. Now note that $\one \after f = \one \after
\kappa_{1} \after f$, for $\kappa_{1} \colon Y \rightarrow Y+1$. Next,
$$\begin{array}{rcccccl}
(\one \after f)^{\bot}
& = &
\idmap[I] \after (\one \after f)^{\bot}
& = &
\one \after (\one \after f)^{\bot}
& = &
\one \after \kappa_{2} \after (\one \after f)^{\bot},
\end{array}$$

\noindent where $\kappa_{2} \colon 1 \rightarrow Y+1$. Hence $(\one
\after \kappa_{1} \after f) \orthogonal (\one \after \kappa_{2} \after
(\one \after f)^{\bot})$, which gives by
Definition~\ref{def:FinPACwE}~\eqref{def:FinPACwE:ortho} an
orthogonality $(\kappa_{1} \after f) \orthogonal (\kappa_{2} \after
(\one \after f)^{\bot})$ in the homset of maps $X \rightarrow Y+1$. We
now define:
$$\begin{array}{rcl}
G(f)
& = &
(\kappa_{1} \after f) \ovee (\kappa_{2} \after (\one \after f)^{\bot})
   \;\colon\; X \longrightarrow Y+I.
\end{array}$$

\noindent Clearly, 
$$\begin{array}{rcl}
FG(f)
\hspace*{\arraycolsep} = \hspace*{\arraycolsep}
\rhd_{1} \after G(f)
& = &
(\rhd_{1} \after \kappa_{1} \after f) \ovee 
   (\rhd_{1} \after \kappa_{2} \after (\one \after f)^{\bot}) \\
& = &
f \ovee (\zero \after (\one \after f)^{\bot})
\hspace*{\arraycolsep} = \hspace*{\arraycolsep}
f \ovee \zero
\hspace*{\arraycolsep} = \hspace*{\arraycolsep}
f.
\end{array}$$

\noindent In order to see $GF(g) = g$, for a total map $g\colon X
\rightarrow Y+I$, write $g = (\kappa_{1} \after \rhd_{1} \after g)
\ovee (\kappa_{2} \after \rhd_{2} \after g)$ by
Lemma~\ref{lem:FinPAC}~\eqref{lem:FinPAC:f}, and compare it with:
$$\begin{array}{rcl}
GF(g)
& = &
(\kappa_{1} \after \rhd_{1} \after g) \ovee 
   (\kappa_{2} \after (\one \after \rhd_{1} \after g)^{\bot}).
\end{array}$$

\noindent Hence it suffices to show $\rhd_{2} \after g = (\one \after
\rhd_{1} \after g)^{\bot}$. This is done as follows. The map
$(\one+\idmap) \after g \colon X \rightarrow I+I$ satisfies $(\rhd_{1}
\after (\one+\idmap) \after g) \orthogonal (\rhd_{2} \after (\one+\idmap)
\after g)$ and thus:
$$\begin{array}{rcccccccl}
(\one \after \rhd_{1} \after g) \ovee (\rhd_{2} \after g)
& = &
\nabla \after (\one+\idmap) \after g
& = &
[\one,\one] \after g
& = &
\one \after g
& = &
\one.
\end{array}$$

\noindent But then we are done by uniqueness of orthosupplements in
the effect algebra $\cat{C}(X,I)$. \QED
\end{proof}

\begin{discussion}
\label{dis:partialtotal}
Now that we have seen the equivalence of `effectus' and `FinPAC with
effects' we have a choice --- or a dilemma, if you like: which notion
to use? Let's start by listing some pros and contras.
\begin{enumerate}
\item The notion of effectus has the definite advantage that its
  definition is simple and elegant --- see
  Definition~\ref{def:effectus}. Surprisingly many results can be
  obtained from this relatively weak structure, which are best
  summarised in the resulting state-and-effect
  triangle~\eqref{diag:effectustriangle}.

A disadvantage of using the notion of effectus is that we have to
explicitly distinguish total and partial maps, for which we have even
introduced separate notation. Another disadvantage is that from a
computational perspective the category $\Par(\cat{B})$ of partial maps
in an effectus $\cat{B}$ is the more interesting structure, and not
$\cat{B}$ itself. In support of the notion of effectus one could claim
that the main examples are most naturally described as effectus, and
not as FinPAC with effects: thus, for instance the categories $\Sets$
and $\Kl(\Dst)$ with total functions and distributions are in a sense
more natural descriptions, than the categories $\Par(\Sets)$ and
$\Par(\Kl(\Dst)) \cong \Kl(\sDst)$ with partial maps and
subdistributions.

\item The definition of `FinPAC with effects' is much less elegant,
  see Definitions~\ref{def:FinPACwE} and~\ref{def:FinPAC}: it is not
  only much more verbose, but also involves `structure', namely the
  special object $I$, of which it is even not clear that it is
  determined up-to-isomorphism. On the other hand, a definite
  advantage is that in a FinPAC with effects the total maps are a
  natural subclass of all the maps (understood as the partial ones),
  and there is no need for separate notation for total and partial
  maps. Moreover, the notion of FinPAC with effects gives you in many,
  computational situations directly the structure that is of most
  interest, namely partial maps. This is especially the case when we
  discuss comprehension and quotients later on.
\end{enumerate}

\noindent How to weigh these arguments? How to proceed from here? We
can choose to work from now on (1)~only with effectuses, (2)~only with
FinPACs with effects, or (3)~switch freely between them, depending on
whatever works best in which situation.

The first two options are easiest, but provide limited
flexibility. Therefore we will choose the third approach. We do
realise that it does not make the theory of effectuses easier, since
one has to been keenly aware of which description applies. But we hope
that the reader will reach such a level of enlightment that the
differences become immaterial --- and Wittgenstein's proverbial ladder
can be thrown away, after one has climbed it.

More concretely, in the sequel we will start definitions and results
with either ``let $\cat{B}$ an effectus in total
form'',\index{S}{effectus!-- in total form} or with ``let $\cat{C}$ be
an effectus in partial form''.\index{S}{effectus!-- in partial form}
The latter expression will replace the term `FinPAC with effects'; it
will not be used anymore in the sequel of this document.

We take another important decision: up to now we have used separate
notation for total ($\tafter, \tto, \tplus$) and partial ($\pafter,
\pto, \pplus$) maps in effectuses (in total form). From now on:
\begin{itemize}
\item we use ordinary categorical notation in an \emph{effectus in
  total form} (replacing $\tafter, \tto, \tplus$ by $\after, \to, +$)
  but continue to use special notation ($\pafter, \pto, \pplus$) in
  the category of partial maps, \textit{i.e.}~in the Kleisli category
  of the lift monad.

\item we also use ordinary categorical notation in an \emph{effectus
  in partial form}.
\end{itemize}

In line with such easy switching of contexts we will freely use
notation that we have introduced for \emph{partial} maps in an
effectus in total form for \emph{ordinary} maps in effectuses in
partial form --- where $\pafter$ simply becomes $\after$. Thus for
instance, for such a map $f\colon X \rightarrow Y$ in an effectus in
partial form we write:
$$\begin{array}{rcl}
\pbox{f}(q)
& = &
(q^{\bot} \after f)^{\bot} 
\\
\ker(f)
& = &
\pbox{f}(\zero)
\hspace*{\arraycolsep} = \hspace*{\arraycolsep}
(\one \after f)^{\bot}
\\
\kerbot(f)
& = &
\one \after f.
\end{array}$$

\noindent Theorem~\ref{thm:partialtotal} allows us to translate back
and forth between the total and partial world. Thus, the properties
of, for instance, Lemma~\ref{lem:ker}, which are formulated for an
effectus in total form, also make sense for an effectus in partial
form. Further, if $f$ is total, then the two forms of substitution
$\pbox{f}$ and $\tbox{f}$ concide:
$$\begin{array}{rcccl}
\pbox{f}(q)
& = &
q \after f
& = &
\tbox{f}(q).
\end{array}$$

\noindent This follows from uniqueness of orthosupplements:
$$\begin{array}{rcccccl}
(q^{\bot} \after f) \ovee (q \after f)
& = &
(q^{\bot} \ovee q) \after f
& = &
\one \after f
& = &
\one.
\end{array}$$

There is another topic that we can now understand in greater
generality, namely the partial pairing $\dtuple{f,g}$. It was
introduced in Lemma~\ref{lem:pairing} for maps $f,g$ satisfying
$\ker^{\bot}(f) \ovee \ker^{\bot}(g) = \one$, and produced a
\emph{total} map $\dtuple{f,g}$. The bijective
correspondence~\eqref{bijcor:sumpairing}, in upwards direction,
extends this pairing in two ways, namely to $n$-ary pairing and to
partial maps $f_{i} \colon Z \pto X_{i}$. for which the
kernel-supplements $\kerbot(f_{i})$ are only orthogonal (instead of
adding up to $\one$). The resulting pairing $\dtuple{f_{1}, \ldots,
  f_{n}} \colon Z \pto X_{1}+\cdots+X_{n}$ is then only a partial map,
defined as $\bigovee_{i}(\kappa_{i} \after f_{i})$. It is unique in
satisfying:
$$\begin{array}{rcl}
\rhd_{i} \after \dtuple{f_{1}, \ldots, f_{n}} 
& = &
f_{i}.
\end{array}$$

\noindent Thus, in an effectus in partial form we have partial pairing
too.

Recall that this map $\dtuple{f_{1}, \ldots, f_{n}}$ is total if
$\bigovee_{i}\kerbot(f_{i}) = \one$. In that case we are back in the
situation of Lemma~\ref{lem:pairing}. In the sequel we use this
pairing in this more general form, as essentially given by the
correspondence~\eqref{bijcor:sumpairing}.
\end{discussion}

\section{Commutative and Boolean effectuses}\label{sec:commbool}

Partial endomaps $X \pto X$ play an important role in effectus theory.
They give rise to predicates, by taking their kernel, but they may
also be obtained from predicates, as their associated `side
effect'. This is a topic that will return many times in the
sequel. For this reason we introduce special notation, and write
$\End(X)$\index{N}{$\End(X)$, homset of partial maps $X \rightarrow
  X$} for the set of partial maps $X \pto X$, in an effectus in total
form. This set $\End(X)$ is a partial commutative monoid (PCM) by
Proposition~\ref{prop:effectusPCM} via $\ovee, \zero$, it carries a
partial order by
Proposition~\ref{prop:effectusPCMorder}~\eqref{prop:effectusPCMorder:order},
and it is a (total) monoid via partial composition $\pafter, \idmap$.
We recall from Lemma~\ref{lem:kerbot} that the kernel-orthosupplement
$\kerbot$ forms a PCM-homomorphism
$$\xymatrix@C+1pc{
\End(X) \ar[r]^-{\kerbot} & \Pred(X)
}$$

\noindent It reflects $\zero$ and $\orthogonal$. Explicitly, $\kerbot(f) =
\one \pafter f = (\bang+\idmap) \after f$.

We shall write $\sEnd(X)\hookrightarrow \End(X)$\index{N}{$\sEnd(X)$,
  set of partial maps $X \rightarrow X$ below the identity} for the
subset:
$$\begin{array}{rcl}
\sEnd(X)
& = &
\set{f\colon X \pto X}{f \leq \idmap[X]},
\end{array}$$

\noindent where $\idmap[X] = \kappa_{1}$ is the partial identity $X
\pto X$. We shall understand endomaps in $\sEnd(X)$ as
\emph{side-effect free}\index{S}{side-effect free map} morphisms.

\begin{definition}
\label{def:commbool}
An effectus in total form is called
\emph{commutative}\index{S}{commutative!--
  effectus}\index{S}{effectus!commutative --} if for each object $X$
both:
\begin{enumerate}
\item \label{def:commbool:asrt} the map $\kerbot \colon \sEnd(X)
  \rightarrow \Pred(X)$ is an isomorphism; we shall write the inverse
  as $p \mapsto \asrt_{p}$,\index{N}{$\asrt_p$, assert map for
    predicate $p$!in a commutative effectus} and call it `assert';

\item \label{def:commbool:comm} $\asrt_{p} \pafter \asrt_{q} =
  \asrt_{q} \pafter \asrt_{p}$, for each pair of predicates
  $p,q\in\Pred(X)$.
\end{enumerate}

\noindent For $p,q\in\Pred(X)$ we define a new `product' predicate
$p\andthen q \in \Pred(X)$\index{N}{$p \andthen q$, sequential
  composition of predicates $p,q$} via:
$$\begin{array}{rcccccl}
p\andthen q
& = &
\kerbot\big(\asrt_{p} \pafter \asrt_{q}\big)
& = &
\kerbot\big(\asrt_{q} \pafter \asrt_{p}\big)
& = &
q\andthen p.
\end{array}$$

An effectus is called \emph{Boolean}\index{S}{Boolean!--
  effectus}\index{S}{effectus!Boolean --} if it is commutative
and satisfies: $\asrt_{p} \pafter \asrt_{p^{\bot}} = 0$, for each
predicate $p\in \Pred(X)$.
\end{definition}

We shall read the predicate $p \andthen q$ as `$p$ andthen $q$'. This
$\andthen$ is a commutative operation in the present commutative
context, but it is non-commutative in a more general setting, see
Section~\ref{sec:noncomm}.

The conditions in this definition are given in such a way that they
can easily be re-formulated for an effectus in partial form. Hence, in
the sequel, we freely speak about a commutative/Boolean effectus in
partial form.

Later on in Subsection~\ref{subsec:copier} we will show that the
presence of copiers makes an effectus commutative.

\begin{example}
\label{ex:commbool}
We describe three examples of these subclasses of effectuses.
\begin{enumerate}
\item \label{ex:commbool:Sets} The effectus $\Sets$ is Boolean. For a
  predicate $P\subseteq X$, the partial assert function $\asrt_{P}
  \colon X \to X+1$ is given by:
$$\begin{array}{rcl}
\asrt_{P}(x)
& = &
\left\{\begin{array}{ll}
\kappa_{1}x \quad & \mbox{if } x\in P \\
\kappa_{2}* & \mbox{if } x\not\in P.
\end{array}\right.
\end{array}$$

\noindent For convenience we often omit these coprojections
$\kappa_{1}, \kappa_{2}$ in such descriptions. We have $\asrt_{P} \leq
\idmap$ since $\asrt_{P} \ovee \asrt_{P^\bot} = \idmap$, see the
description of $\ovee$ on partial maps in Example~\ref{ex:PCM}. We
have:
$$\begin{array}{rcccccl}
\kerbot(\asrt_{P})
& = &
\set{x}{\asrt_{P}(x) \neq *}
& = &
\set{x}{x\in P}
& = &
P.
\end{array}$$

\noindent Next, let $f\colon X \rightarrow X+1$ satisfy $f \leq
\idmap$.  This means $f(x) \neq * \Rightarrow f(x) = x$. Now we take
as predicate $P = \set{x}{f(x) \neq *}$, so that $f =
\asrt_{P}$. Hence we have an isomorphism $\sEnd(X) \cong \Pred(X)$.

We further have:
$$\begin{array}{rcccl}
\big(\asrt_{P} \pafter \asrt_{Q}\big)(x)
& = &
\left\{\begin{array}{ll}
x \quad & \mbox{if } x\in P\cap Q \\
* & \mbox{otherwise}
\end{array}\right\}
& = &
\big(\asrt_{Q} \pafter \asrt_{P}\big)(x).
\end{array}$$

\noindent The product predicate $P\andthen Q$ is thus
intersection/conjunction $P\cap Q$.

The effectus $\Sets$ is Boolean since:
$$\begin{array}{rcccccl}
\big(\asrt_{P} \pafter \asrt_{P^\bot}\big)(x)
& = &
\left\{\begin{array}{ll}
x \quad & \mbox{if } x\in P\cap \neg P \\
* & \mbox{otherwise}
\end{array}\right\}
& = &
*
& = &
\zero(x).
\end{array}$$

\item \label{ex:commbool:KlD} As may be expected, the effectus
  $\Kl(\Dst)$ is commutative. For a predicate $p\in [0,1]^{X}$ we have
  an assert map $\asrt_{p} \colon X \rightarrow \Dst(X+1)$ given by
  the convex sum:
$$\begin{array}{rcl}
\asrt_{p}(x)
& = &
p(x)\ket{x} + (1-p(x))\ket{\!*\!}.
\end{array}$$

\noindent We have $\asrt_{p} \leq \idmap$ since $\asrt_{p} \ovee
\asrt_{p^\bot} = \idmap$. Then, following the description of kernels
in $\Kl(\Dst)$ from Example~\ref{ex:ker}, we get:
$$\begin{array}{rcl}
\kerbot(\asrt_{p})(x)
\hspace*{\arraycolsep}=\hspace*{\arraycolsep}
1 - \ker(\asrt_{p})(x)
& = &
1 - \asrt_{p}(*) \\
& = &
1 - (1 - p(x))
\hspace*{\arraycolsep}=\hspace*{\arraycolsep}
p(x).
\end{array}$$

\noindent We check that we get an isomorphism $\sEnd(X) \cong
\Pred(X)$. Let $f\colon X \rightarrow \Dst(X+1)$ satisfy $f\leq
\idmap$, where $\idmap(x) = 1\ket{x}$. Then $f(x)(x') \neq 0
\Rightarrow x' = x$. Taking $p(x) = f(x)(x)$ then yields $f =
\asrt_{p}$.

Next, these assert maps satisfy:
$$\begin{array}{rcl}
\big(\asrt_{p} \pafter \asrt_{q}\big)(x)
& = &
p(x)\cdot q(x)\ket{x} + (1 - p(x)\cdot q(x))\ket{\!*\!} \\
& = &
\big(\asrt_{q} \pafter \asrt_{p}\big)(x).
\end{array}$$

\noindent The product predicate $p\andthen q$ is thus the pointwise
multiplication $(p\andthen q) = p(x) \cdot q(x)$.

It is instructive to see why $\Kl(\Dst)$ is not a Boolean
effectus:
\noindent It satisfies:
$$\begin{array}{rcl}
\big(\asrt_{p} \pafter \asrt_{p^\bot}\big)(x)
& = &
p(x)\cdot (1 -p(x))\ket{x} + (1 - p(x)\cdot (1-p(x)))\ket{\!*\!}.
\end{array}$$

\noindent If $\asrt_{p} \pafter \asrt_{p^\bot} = \zero$, then $p(x)
\cdot (1-p(x)) = 0$ for each $x$. But this requires that $p(x)$ is a
Boolean predicate, with $p(x) \in \{0,1\}$ for each $x$, so that $p$
restricts to $X \rightarrow \{0,1\}$. But of course, not every
predicate in $\Kl(\Dst)$ is Boolean.

\item \label{ex:commbool:vNA} The effectus $\op{\CvNA}$ of
  \emph{commutative} von Neumann algebras is commutative. For an
  effect $e\in [0,1]_{\mathscr{A}}$ in such a commutative algebra
  $\mathscr{A}$, we define $\asrt_{e} \colon \mathscr{A} \pto
  \mathscr{A}$ by $\asrt_{e}(a) = e\cdot a$.  Clearly, this is a
  linear subunital map. We use commutativity to show that it is
  positive. For $a \geq 0$, say $a = b\cdot b^{*}$ we obtain a
  positive element:
$$\begin{array}{rcccccccl}
\asrt_{e}(a)
& = &
e\cdot a
& = &
\sqrt{e} \cdot \sqrt{e} \cdot b \cdot b^{*}
& = &
\sqrt{e} \cdot b \cdot b^{*} \cdot \sqrt{e}
& = &
(\sqrt{e} \cdot b) \cdot (\sqrt{e} \cdot b)^{*}.
\end{array}$$

\noindent Here we use that each positive element $x$ has a positive
square root $\sqrt{x}$ and is self-adjoint, that is, satisfies $x =
x^{*}$.

Like before, we have $\asrt_{e} \ovee \asrt_{e^{\bot}} = \idmap$, so
that $\asrt_{e} \leq \idmap$. Further, following the description in
Example~\ref{ex:ker},
$$\begin{array}{rcccccl}
\kerbot(\asrt_{e})
& = &
\asrt_{e}(1)
& = &
e\cdot 1
& = &
e.
\end{array}$$

\noindent Showing that $\kerbot$ is an isomorphism requires more work.
Let $f\colon \mathscr{A} \pto \mathscr{A}$ be a subunital positive map
below the identity. This means $f(a) \leq a$ for all $a \geq
0$. Define $g\colon \mathscr{A}\oplus \mathscr{A} \rightarrow
\mathscr{A}$ by $g(x,y) = f(x) + y - f(y)$. This map is positive,
because if $x, y \geq 0$, then $f(x) \geq 0$ and $y - f(y) \geq 0$,
since $f(y) \leq y$, and thus $g(x,y) \geq 0$. Clearly, $g$ is
unital. More generally, $g(x,x) = x$, which can be written more
abstractly as $g \after \Delta = \idmap$, where the diagonal $\Delta
\colon \mathscr{A} \rightarrow \mathscr{A}\oplus\mathscr{A}$ preserves
multiplication. Hence `Tomiyama'~\cite{Tomiyama57} applies,
see~\cite[Lemma~38]{Jacobs15a}, so that $z\cdot g(x,y) =
g(\Delta(z)\cdot (x,y)) = g(z\cdot x, z\cdot y)$. The element $g(1,0)
= f(1)$ is central in $\mathscr{A}$ since:
$$\begin{array}{rcccccccl}
x\cdot g(1, 0)
& = &
g\big(\Delta(x)\cdot (1,0)\big)
& = &
g(x, 0)
& = &
g\big((1,0)\cdot \Delta(x)\big)
& = &
g(1, 0) \cdot x.
\end{array}$$

\noindent We thus obtain:
$$\begin{array}{rcccccccccl}
f(x)
& = &
g(x, 0)
& = &
g(1, 0)\cdot x
& = &
f(1)\cdot x
& = &
\asrt_{f(1)}(x)
& = &
\asrt_{\kerbot(f)}(x).
\end{array}$$

Finally, we have, for arbitrary effects $e,d\in [0,1]_{\mathscr{A}}$,
$$\begin{array}{rcccccl}
\big(\asrt_{e} \pafter \asrt_{d}\big)(a)
& = &
e\cdot (d\cdot a)
& = &
d \cdot (e\cdot a)
& = &
\big(\asrt_{d} \pafter \asrt_{e}\big)(a).
\end{array}$$

\noindent As a result, the product predicate $e\andthen d$ is given
by multiplication $\cdot$ in the von Neumann algebra.
\end{enumerate}
\end{example}

The assert maps, if they exist, give rise to many interesting
properties.

\begin{lemma}
\label{lem:commeffectus}
Let $\cat{B}$ be a commutative effectus in total form.
\begin{enumerate}
\item \label{lem:commeffectus:EA} Each PCM $\sEnd(X)$ is an effect
  algebra; in fact, with composition $\pafter, \idmap$ it is a
  commutative effect monoid.

\item \label{lem:commeffectus:EAmap} The assert map $\Pred(X)
  \rightarrow \sEnd(X)$ is an (iso)morphism of effect algebras, and
  makes $\Pred(X)$ with $\andthen$ also a commutative effect monoid.

\item \label{lem:commeffectus:zero} $p \pafter f = \zero$ iff
  $\asrt_{p} \pafter f = \zero$.

\item \label{lem:commeffectus:scal} $\asrt_{s} = s$ for each scalar
  $s\colon 1 \pto 1$.

\item \label{lem:commeffectus:cotuple} $\asrt_{[p,q]} = \asrt_{p}
  \pplus \asrt_{q} \colon X+Y \pto X+Y$.

\item \label{lem:commeffectus:instr} The \emph{instrument} map
  $\instr_{p} = \dtuple{\asrt_{p}, \asrt_{p^\bot}} \colon X
  \rightarrow X+X$\index{N}{$\instr_p$, instrument map for predicate $p$} satisfies $\nabla \after \instr_{p} = \idmap$.

\item \label{lem:commeffectus:order} $p\andthen q = q \pafter \asrt_{p}$,
  and thus $p\andthen q \leq q$ and also $p\andthen q \leq p$.

\item \label{lem:commeffectus:subst} $\pbox{\big(\asrt_{p}\big)}(q) =
  (p\andthen q) \ovee p^{\bot}$.

\item \label{lem:commeffectus:BA} The following points are
equivalent, for an arbitrary predicate $p$.
\begin{enumerate}
\item \label{lem:commeffectus:BA:idempasrt} $\asrt_{p} \pafter
  \asrt_{p} = \asrt_{p}$
\item \label{lem:commeffectus:BA:idempandthen} $p \andthen p = p$
\item \label{lem:commeffectus:BA:zeroasrt} $\asrt_{p} \pafter
  \asrt_{p^\bot} = \zero$
\item \label{lem:commeffectus:BA:zeroandthen} $p \andthen p^{\bot} = \zero$
\item \label{lem:commeffectus:BA:sharp} $p$ is sharp, that is,
  $p\wedge p^{\bot} = \zero$.
\end{enumerate}
\end{enumerate}
\end{lemma}

The instrument map $\instr_{p} \colon X \rightarrow X+X$ from
point~\eqref{lem:commeffectus:instr} can be understood as a `case'
expression, sending the input to the left component in $X+X$ if $p$
holds, and to the right component otherwise. The property $\nabla
\after \instr_{p} = \idmap$ says that this instrument map has no
side-effects.  We will encounter instrument maps in a non-commutative
setting in Section~\ref{sec:noncomm} where they do have
side-effects. Such side-effects are an essential aspect of the quantum
world. They don't exist in the current commutative setting.

\begin{proof}
We handle these points one-by-one.
\begin{enumerate}
\item First we have to produce an orthosupplement for an arbitrary
  $f\in\sEnd(X)$. Since $f \leq \idmap$, we have $f \ovee g = \idmap$
  for some $g \in \End(X)$. Clearly $g\leq \idmap$. It is unique with
  this property, since if $f \ovee h = \idmap = f \ovee g$, then
  $\kerbot(f) \ovee \kerbot(g) = \one = \kerbot{f} \ovee
  \kerbot(h)$. Hence $\kerbot(g) = \kerbot(h)$ by cancellation in the
  effect algebra $\Pred(X)$. But then $g = h$ since $\kerbot$ is an
  isomorphism.

We see that $\idmap = \zero^{\bot} \in \sEnd(X)$. If $\idmap
\orthogonal f$ in $\sEnd(X)$, then $\kerbot(\idmap) = \one \orthogonal
\kerbot(f)$.  Hence $\kerbot(f) = \zero$ in $\Pred(X)$, and thus $f =
\zero$ since $\kerbot$ reflects $\zero$.

The partial composition operation $\pafter$ on $\sEnd(X)$ preseves
$\ovee, \zero$ by
Proposition~\ref{prop:effectusPCM}~\eqref{prop:effectusPCM:pres}. Moreover,
the top element $\idmap \in \sEnd(X)$ obviously satisfies $p \pafter
\idmap = p = \idmap \pafter p$. Hence $\sEnd(X)$ is an effect monoid,
see Definition~\ref{def:EMon}. We show that it is commutative. For
arbitrary maps $f,g\colon X \pto X$ write $p = \kerbot(f)$ and $q =
\kerbot(g)$, so that $f = \asrt_{p}$ and $g = \asrt_{q}$. Then we are
done by Definition~\ref{def:commbool}~\eqref{def:commbool:comm}:
$$\begin{array}{rcccccl}
f \pafter g
& = &
\asrt_{p} \pafter \asrt_{q}
& = &
\asrt_{q} \pafter \asrt_{p}
& = &
g \pafter f.
\end{array}$$

\item We have $\asrt_{\one} = \idmap$, since $\one =
  \kerbot(\idmap)$. Next, if $p \orthogonal q$ in $\Pred(X)$, then
  $\kerbot(\asrt_{p}) = p \orthogonal q = \kerbot(\asrt_{q})$, and
  thus $\asrt_{p} \orthogonal \asrt_{q}$ since $\kerbot$ reflects
  $\orthogonal$. In that case $\asrt_{p\ovee q} = \asrt_{p} \ovee
  \asrt_{q}$ since:
$$\begin{array}{rcl}
\kerbot(\asrt_{p \ovee q})
\hspace*{\arraycolsep}=\hspace*{\arraycolsep}
p \ovee q
& = &
\kerbot(\asrt_{p}) \ovee \kerbot(\asrt_{q}) \\
& = &
\kerbot\big(\asrt_{p} \ovee \asrt_{q}\big).
\end{array}$$

\noindent By construction, $\kerbot \colon \sEnd(X)
\conglongrightarrow \Pred(X)$ sends $\pafter$ to $\andthen$, so that
$\Pred(X)$ becomes a commutative effect monoid, and $\kerbot$ an
isomorphism of effect monoids.

\auxproof{
We have $0\andthen p = 0$ and $1 \andthen p  = p$ since:
$$\begin{array}{rcl}
0\andthen p
& = &
\kerbot(\asrt_{p} \pafter \asrt_{0})
\hspace*{\arraycolsep}=\hspace*{\arraycolsep}
\kerbot(\asrt_{p} \pafter 0)
\hspace*{\arraycolsep}=\hspace*{\arraycolsep}
\kerbot(0)
\hspace*{\arraycolsep}=\hspace*{\arraycolsep}
0
\\
1\andthen p
& = &
\kerbot(\asrt_{p} \pafter \asrt_{1})
\hspace*{\arraycolsep}=\hspace*{\arraycolsep}
\kerbot(\asrt_{p} \pafter \idmap)
\hspace*{\arraycolsep}=\hspace*{\arraycolsep}
\kerbot(\asrt_{p})
\hspace*{\arraycolsep}=\hspace*{\arraycolsep}
p.
\end{array}$$

\noindent Associativity of the product is easy:
$$\begin{array}{rcl}
p \andthen (q\andthen r)
& = &
\kerbot\big(\asrt_{q \andthen r} \pafter \asrt_{p}\big) \\
& = &
\kerbot\big(\asrt_{r} \pafter \asrt_{q} \pafter \asrt_{p}\big) \\
& = &
\kerbot\big(\asrt_{r} \pafter \asrt_{p\andthen q}\big) \\
& = &
(p\andthen q) \andthen r.
\end{array}$$

\noindent Further,
$$\begin{array}[b]{rcl}
(p \ovee p') \andthen q
& = &
\kerbot\big(\asrt_{q} \pafter \asrt_{p \ovee p'}\big) \\
& = &
\kerbot\big(\asrt_{q} \pafter (\asrt_{p} \ovee \asrt_{p'})\big) \\
& = &
\kerbot\big((\asrt_{q} \pafter \asrt_{p}) \ovee 
   (\asrt_{q} \pafter \asrt_{p'})\big) \\
& = &
\kerbot\big(\asrt_{q} \pafter \asrt_{p}\big) \ovee 
   \kerbot\big(\asrt_{q} \pafter \asrt_{p'}\big) \\
& = &
(p\andthen q) \ovee (p' \ovee q).
\end{array}$$
}

\item By Lemma~\ref{lem:zero}, using that $\one \pafter \asrt_{p} = 
\kerbot(\asrt_{p}) = p$,
$$\begin{array}{rcccl}
\asrt_{p} \pafter f = \zero
& \Longleftrightarrow &
\one \pafter \asrt_{p} \pafter f = \zero
& \Longleftrightarrow &
p \pafter f = \zero.
\end{array}$$

\item We obtain $\asrt_{s} = s$ for a scalar $s$ from
  Lemma~\ref{lem:ker}~\eqref{lem:ker:pred}:
$$\begin{array}{rcccl}
\kerbot(\asrt_{s})
& = &
s
& = &
\kerbot(s).
\end{array}$$

\item The equation $\asrt_{[p,q]} = \asrt_{p} \pplus \asrt_{q}$ follows from:
$$\begin{array}{rcl}
\kerbot\big(\asrt_{p} \pplus \asrt_{q}\big)
& = &
\one \pafter [\kappa_{1} \pafter \asrt_{p}, \kappa_{2} \pafter \asrt_{q}] \\
& = &
[\one \pafter \kappa_{1} \pafter \asrt_{p}, 
   \one \pafter \kappa_{2} \pafter \asrt_{q}] \\
& = &
[\one \pafter \asrt_{p}, \one \pafter \asrt_{q}] 
   \quad \mbox{since coprojections are total} \\
& = &
[\kerbot(\asrt_{p}), \kerbot(\asrt_{q})] \\
& = &
[p, q] \\
& = &
\kerbot(\asrt_{[p,q]}).
\end{array}$$

\item For a predicate $p\in\Pred(X)$ we have by definition:
  $\ker(\asrt_{p}) = p^{\bot} = \kerbot(\asrt_{p})$. Hence we can use
  the pairing from Lemma~\ref{lem:pairing} and can form the (total)
  instrument map $\instr_{p} = \dtuple{\asrt_{p}, \asrt_{p^\bot}}
  \colon X \rightarrow X+X$. We obtain $\nabla \after \instr_{p} =
  \idmap$ from:
$$\begin{array}{rcl}
\klin{\nabla \after \instr_{p}}
\hspace*{\arraycolsep}=\hspace*{\arraycolsep}
\nabla \after \klin{\instr_{p}}
& \smash{\stackrel{\eqref{eqn:dtupleovee}}{=}} &
\nabla \after 
   \big((\kappa_{1} \after \asrt_{p}) \ovee (\kappa_{2} \after \asrt_{p^\bot})\big)
   \\
& = &
(\nabla \after \kappa_{1} \after \asrt_{p}) \ovee 
   (\nabla \after \kappa_{2} \after \asrt_{p^\bot}) \\
& = &
\asrt_{p} \ovee \asrt_{p^\bot} \\
& = &
\klin{\idmap}.
\end{array}$$

\item Simply: $p\andthen q = \kerbot(\asrt_{q} \pafter \asrt_{p}) = \one
  \pafter \asrt_{q} \pafter \asrt_{p} = q \pafter \asrt_{p}$.

Since $\asrt_{p} \leq \idmap$, we get $p\andthen q = q \pafter \asrt_{p}
\leq q \pafter \idmap = q$. By commutativity we obtain: $p\andthen q = q
\andthen p \leq p$.

\item We have:
$$\begin{array}[b]{rcl}
\pbox{\big(\asrt_{p}\big)}(q)
& = &
\pbox{\big(\asrt_{p}\big)}(q^{\bot\bot}) \\
& = &
\pbox{\big(\asrt_{p}\big)}(q^{\bot})^{\bot} \ovee \ker(\asrt_{p}) 
   \qquad \mbox{by Lemma~\ref{lem:ker}~\eqref{lem:ker:squareneg}} \\
& = &
(q \pafter \asrt_{p}) \ovee p^{\bot} \\
& = &
(p\andthen q) \ovee p^{\bot}.
\end{array}$$

\item The equivalences $\eqref{lem:commeffectus:BA:idempasrt}
  \Leftrightarrow \eqref{lem:commeffectus:BA:idempandthen}$ and
  $\eqref{lem:commeffectus:BA:zeroasrt} \Leftrightarrow
  \eqref{lem:commeffectus:BA:zeroandthen}$ are obvious, via the
  isomorphism $\kerbot$. We prove
  $\eqref{lem:commeffectus:BA:idempasrt} \Leftrightarrow
  \eqref{lem:commeffectus:BA:zeroasrt}$ and
  $\eqref{lem:commeffectus:BA:sharp} \Leftrightarrow
  \eqref{lem:commeffectus:BA:idempasrt}$.

For $\eqref{lem:commeffectus:BA:idempasrt} \Rightarrow
\eqref{lem:commeffectus:BA:zeroasrt}$, let predicate $p$ satisfy
$\asrt_{p} \pafter \asrt_{p} = \asrt_{p}$.  Then:
$$\begin{array}{rcl}
\asrt_{p}
\hspace*{\arraycolsep}=\hspace*{\arraycolsep}
\asrt_{p} \pafter \idmap
& = &
\asrt_{p} \pafter (\asrt_{p} \ovee \asrt_{p^\bot}) \\
& = &
(\asrt_{p} \pafter \asrt_{p}) \ovee (\asrt_{p} \pafter \asrt_{p^\bot}) \\
& = &
\asrt_{p} \ovee (\asrt_{p} \pafter \asrt_{p^\bot}).
\end{array}$$

\noindent Hence $\asrt_{p} \pafter \asrt_{p^\bot} = \zero$ by
Lemma~\ref{prop:effectusPCMorder}~\eqref{prop:effectusPCMorder:canc}.

In the reverse direction, assuming $\asrt_{p} \pafter \asrt_{p^\bot} =
\zero$ we obtain:
$$\begin{array}{rcl}
\asrt_{p}
\hspace*{\arraycolsep}=\hspace*{\arraycolsep}
\asrt_{p} \pafter \idmap
& = &
\asrt_{p} \pafter (\asrt_{p} \ovee \asrt_{p^\bot}) \\
& = &
(\asrt_{p} \pafter \asrt_{p}) \ovee (\asrt_{p} \pafter \asrt_{p^\bot}) \\
& = &
(\asrt_{p} \pafter \asrt_{p}) \ovee \zero \\
& = &
\asrt_{p} \pafter \asrt_{p}.
\end{array}$$

For $\eqref{lem:commeffectus:BA:zeroasrt} \Rightarrow
\eqref{lem:commeffectus:BA:sharp}$ we assume $\asrt_{p} \pafter
\asrt_{p^\bot} = \zero$. Let $q \leq p$ and $q \leq p^{\bot}$.  If we
show $q = 0$, then $p \wedge p^{\bot} = \zero$. First, the inequality
$q \leq p$ gives:
$$\begin{array}{rcccccccl}
q \pafter \asrt_{p^\bot}
& \leq &
p \pafter \asrt_{p^{\bot}}
& = &
\one \pafter \asrt_{p} \pafter \asrt_{p^\bot}
& = &
\one \pafter \zero
& = &
\zero.
\end{array}$$

\noindent Hence:
$$\begin{array}{rcccccccl}
q \pafter \asrt_{p}
& = &
(q \pafter \asrt_{p}) \ovee (q \pafter \asrt_{p^\bot})
& = &
q \pafter (\asrt_{p} \ovee \asrt_{p^\bot})
& = &
q \pafter \idmap
& = &
q.
\end{array}$$

\noindent But now the inequality $q \leq p^{\bot}$ gives the required
result:
$$\begin{array}{rcccccccccl}
q 
& = &
q \pafter \asrt_{p}
& \leq &
p^{\bot} \pafter \asrt_{p}
& = &
\one \pafter \asrt_{p^\bot} \pafter \asrt_{p}
& = &
\one \pafter \zero
& = &
\zero.
\end{array}$$

Finally, for $\eqref{lem:commeffectus:BA:sharp} \Rightarrow
\eqref{lem:commeffectus:BA:zeroasrt}$ let $p \wedge p^{\bot} =
\zero$. We have $p \andthen p^{\bot} \leq p$ and also $p \andthen
p^{\bot} \leq p^{\bot}$ by point~\eqref{lem:commeffectus:order}.
Hence $\zero = p \andthen p^{\bot} = \kerbot(\asrt_{p} \pafter
\asrt_{p^\bot})$. Since $\kerbot$ reflects $\zero$, we obtain
$\asrt_{p} \pafter \asrt_{p^\bot} = \zero$. \QED
\end{enumerate}
\end{proof}

We show how Bayes' rule can be described abstractly in the current
setting. A type-theoretic formulation of these ideas is elaborated
in~\cite{AdamsJ15}.

\begin{example}
\label{ex:Bayes}
Let $(\cat{C},I)$ be a commutative effectus, in partial form, with
normalisation,\index{S}{normalisation} as described in
Remark~\ref{rem:normalisation}.  Consider a total state $\omega\colon
1 \rightarrow X$ and a predicate $p\colon X \rightarrow I$ on the same
object $X$. We obtain a substate $\asrt_{p} \after \omega \colon 1
\rightarrow X$ with:
$$\begin{array}{rcccl}
\one \after \asrt_{p} \after \omega
& = &
p \after \omega
& \smash{\stackrel{\eqref{eqn:Born}}{=}} &
\omega \models p.
\end{array}$$

\noindent We thus have by Lemma~\ref{lem:zero}:
$$\begin{array}{rcl}
\asrt_{p} \after \omega = \zero
& \Longleftrightarrow &
\big(\omega\models p\big) = \zero.
\end{array}$$

\noindent Now let the validity $\omega\models p$ be non-zero. Then we
can normalise the substate $\asrt_{p} \after \omega$. We write the
resulting total `conditional' state as $\omega|_{p} \colon 1
\rightarrow X$.\index{N}{$\omega"|"_{p}$, conditional
  state}\index{S}{conditional state}\index{S}{state!conditional --} It
satisfies by construction, see~\eqref{eqn:normalisation}:
\begin{equation}
\label{eqn:condstate}
\begin{array}{rcl}
\omega|_{p} \after (\omega\models p)
& = &
\asrt_{p} \after \omega,
\end{array}
\end{equation}

\noindent where composition $\after$ on the left is scalar
multiplication, see Lemma~\ref{lem:substatePCMod}. This new
`conditional' state $\omega|_{p}$ should be read as the update of
state $\omega$ after learning $p$.

We claim that we now have the following abstract version of Bayes'
rule:\index{S}{Bayes' rule} for an arbitrary predicate $q$ on $X$,
\begin{equation}
\label{eqn:Bayes}
\begin{array}{rcl}
\big(\omega|_{p} \models q\big) \cdot \big(\omega \models p\big)
& = &
\big(\omega \models p\andthen q\big).
\end{array}
\end{equation}

\noindent In presence of division, this equation can be recast in
more familiar form:
$$\begin{array}{rcl}
\omega|_{p} \models q
& = &
\displaystyle\frac{\omega\models p\andthen q}{\omega \models p}.
\end{array}$$

\noindent The proof of Bayes' equation~\eqref{eqn:Bayes} is easy:
$$\begin{array}{rcl}
\big(\omega|_{p} \models q\big) \cdot \big(\omega \models p\big)
& = &
q \after \omega|_{p} \after \big(\omega \models p\big) \\
& \smash{\stackrel{\eqref{eqn:condstate}}{=}} &
q \after \asrt_{p} \after \omega \\
& = &
(p \andthen q) \after \omega \\
& = &
\omega \models p\andthen q.
\end{array}$$

\noindent This abstract description suggests how to do conditional
probability in a non-commutative setting, in presence of assert maps,
see Section~\ref{sec:noncomm}, and see also~\cite{LeiferS13}.

At this abstract level the \emph{total probability law}, also known as
\emph{the rule of belief propagation}~\cite{LeiferS13} also holds: for
a state $\omega$ on $X$, an $n$-test $p_{1}, \ldots, p_{n}$ on $X$,
and an arbitrary predicate $q$ on $X$,
$$\begin{array}{rcl}
\big(\omega\models q\big)
& = &
\bigovee_{i}\big(\omega|_{p_i}\models q\big) \cdot \big(\omega\models p_{i}\big).
\end{array}$$

\noindent The proof is left to the interested reader.

\auxproof{
$$\begin{array}{rcl}
\lefteqn{\textstyle\bigovee_{i}\big(\omega|_{p_i}\models q\big) \cdot 
   \big(\omega\models p_{i}\big)} \\
& \smash{\stackrel{\eqref{eqn:Bayes}}{=}} &
\bigovee_{i} \omega \models p_{i} \andthen q \\
& = &
\omega \models \bigovee_{i} (p_{i} \andthen q) \\
& & \qquad
   \mbox{since $\omega\models (-) = (-) \after \omega = \tbox{\omega}$
      is a map of effect modules} \\
& = &
\omega \models (\bigovee_{i} p_{i}) \andthen q \\
& = &
\omega \models \one \andthen q \\
& = &
\omega\models q.
\end{array}$$
}

In the current setting we can take the commutative effectus
$\Kl(\Dst)$ as example. For a total state $\omega\in\Dst(X)$ on a set
$X$ and a (fuzzy) predicate $p\in [0,1]^{X}$ with $\omega\models p =
\sum_{x} \omega(x)\cdot p(x) \neq 0$ we obtain as normalised state
$\omega|_{p} \in \Dst(X)$,
$$\begin{array}{rcl}
\omega|_{p}
& = &
{\displaystyle\sum}_{x}
   \displaystyle\frac{\omega(x)\cdot p(x)}{\omega\models p}\bigket{x}.
\end{array}$$

\noindent Then indeed, for $q\in [0,1]^{X}$,
$$\begin{array}{rcccccl}
\omega|_{p} \models q
& = &
\sum_{x} \omega_{p}(x)\cdot q(x)
& = &
\sum_{x} \displaystyle\frac{\omega(x)\cdot p(x) \cdot q(x)}{\omega\models p}
& = &
\displaystyle\frac{\omega\models p\andthen q}{\omega\models p}.
\end{array}$$

We briefly mention the continuous probabilistic case, given by the
Kleisli category $\Kl(\Giry)$ of the Giry monad $\Giry$ on measurable
spaces. For a measurable space $X$, with set $\Sigma_X$ of measurable
subsets, let $\omega \in \Giry(X)$ be a probablity distribution. Each
measurable subset $M\in\Sigma_{X}$ gives rise to predicate $\indic{M}
\colon X \rightarrow [0,1]$ with $\indic{M}(x) = 1$ if $x\in M$ and
$\indic{M}(x) = 0$ otherwise. We claim that the conditional state
$\omega|_{\indic{M}} \in \Giry(X)$, as described above, is the
conditional probability measure $\omega(-\mid M)$, if $\omega(M) \neq
0$.

Indeed, the subprobability measure $\asrt_{\indic{M}} \after \omega
\colon \Sigma_{X} \rightarrow [0,1]$ is given by $A \mapsto \int
\indic{M\cap A} \intd \omega = \omega(M \cap A)$. Normalisation gives
the conditional probability:
$$\begin{array}{rcccl}
\omega|_{\indic{M}}(A)
& = &
\displaystyle\frac{\omega(M\cap A)}{\omega(M)}
& = &
\omega(A\mid M).
\end{array}$$
\end{example}

We turn to Boolean effectuses and collect some basic results.

\begin{lemma}
\label{lem:booleffectus}
Let $\cat{B}$ now be a Boolean effectus, that is, a commutative
effectus in which $\asrt_{p} \pafter \asrt_{p^\bot} = \zero$ holds for
each predicate $p$.
\begin{enumerate}
\item \label{lem:booleffectus:idemp} All assert maps are idempotent,
  that is, $\asrt_{p} \pafter \asrt_{p} = \asrt_{p}$, and all
  predicates $p$ are sharp, that is, $p \wedge p^{\bot} = \zero$.

\item \label{lem:booleffectus:conj} The predicate $p \andthen q$ is
  the meet/conjunction $p\wedge q$ in $\Pred(X)$. Disjunctions then
  also exist via De Morgan: $p \vee q = (p^{\bot} \wedge
  q^{\bot})^{\bot}$.

\item \label{lem:booleffectus:ortho} $p \orthogonal q$ iff $\asrt_{p}
  \pafter \asrt_{q} = \zero$ iff $p\wedge q = \zero$.

\item \label{lem:booleffectus:ovee} If $p \orthogonal q$, then $p
  \ovee q = p \vee q$.

\item \label{lem:booleffectus:BA} Conjunction $\wedge$ distributes
  over disjunction $\vee$, making each effect algebra $\Pred(X)$ a
  Boolean algebra.\index{S}{Boolean!-- algebra}
\end{enumerate}
\end{lemma}

\begin{proof}
Most of these points are relatively easy, except the last one.
\begin{enumerate}
\item Directly by
  Lemma~\ref{lem:commeffectus}~\eqref{lem:commeffectus:BA}.

\item We have $p\andthen p = \kerbot(\asrt_{p} \pafter \asrt_{p}) =
  \kerbot(\asrt_{p}) = p$. This allows us to show that $p\andthen q$
  is the meet of $p,q$. We already have $p\andthen q \leq p$ and
  $p\andthen q \leq q$ by
  Lemma~\ref{lem:commeffectus}~\eqref{lem:commeffectus:order}. Next,
  let $r$ be a predicate with $r \leq p$ and $r \leq q$. Since product
  $r\andthen (-)$ preserves $\ovee$, by
  Lemma~\ref{lem:commeffectus}~\eqref{lem:commeffectus:EA}, it is
  monotone. Hence: $r = r\andthen r \leq p\andthen q$.

\item The equivalence $\asrt_{p} \pafter \asrt_{q} = \zero$ iff
  $p\wedge q = \zero$ follows from the previous point and 
  Lemma~\ref{lem:commeffectus}~\eqref{lem:commeffectus:BA}. So let $p
  \orthogonal q$, so that $q \leq p^{\bot}$. Then $p \wedge q \leq p
  \wedge p^{\bot} = \zero$. Conversely, if $\asrt_{p} \pafter
  \asrt_{q} = \zero$ then:
$$\begin{array}{rcl}
\asrt_{p}
\hspace*{\arraycolsep}=\hspace*{\arraycolsep}
\asrt_{p} \pafter \idmap
& = &
\asrt_{p} \pafter (\asrt_{q} \ovee \asrt_{q^\bot}) \\
& = &
(\asrt_{p} \pafter \asrt_{q}) \ovee (\asrt_{p} \pafter \asrt_{q^\bot}) \\
& = &
\zero \ovee (\asrt_{p} \pafter \asrt_{q^\bot}) \\
& = &
\asrt_{p} \pafter \asrt_{q^\bot}.
\end{array}$$
\noindent Hence $p = p \wedge q^{\bot}$, so that $p \leq q^{\bot}$,
and thus $p \orthogonal q$.

\item Let $p \orthogonal q$. We intend to prove that the sum $p \ovee
  q$, if it exists, is the join $p \vee q$.  In any effect algebra, $p
  \ovee q$ is an upperbound of both $p$ and $q$, so we only need to
  prove that it is the least upperbound. Let $p \leq r$ and $q\leq
  r$. The inequality $p \leq r$ says $p \orthogonal r^{\bot}$ and thus
  $\asrt_{r^\bot} \pafter \asrt_{p} = \zero$ by the previous
  point. Similarly $q \leq r$ gives $\asrt_{r^\bot} \pafter \asrt_{q}
  = 0$. But then:
$$\begin{array}{rcl}
\asrt_{r^\bot} \pafter \asrt_{p\ovee q}
& = &
\asrt_{r^\bot} \pafter \big(\asrt_{p} \ovee \asrt_{q}\big) \\
& = &
\big(\asrt_{r^\bot} \pafter \asrt_{p}\big) \ovee 
   \big(\asrt_{r^\bot} \pafter \asrt_{q}\big) 
\hspace*{\arraycolsep}=\hspace*{\arraycolsep}
0 \ovee 0 
\hspace*{\arraycolsep}=\hspace*{\arraycolsep}
0.
\end{array}$$
\noindent Hence $p\ovee q \orthogonal r^{\bot}$, again by the previous
point, and thus $p\ovee q \leq r$.


\item This result can be traced back to~\cite[Thm.~3.11]{BennettF95}.
  However, we give our own proof of distributivity $p \wedge (q \vee
  r) = (p \wedge q) \vee (p \wedge r)$, for all predicates
  $p,q,r\in\Pred(X)$. We repeatedly use the equivalence
$$\begin{array}{rcl}
x \leq y^{\bot} 
& \Longleftrightarrow &
x \wedge y = \zero
\end{array}\eqno{(*)}$$
\noindent from point~\eqref{lem:booleffectus:ortho}.

The inequality $p \wedge (q \vee r) \geq (p \wedge q) \vee (p \wedge r)$
always holds, so we need to prove the inequality:
$$\begin{array}{rcccl}
p \wedge (q \vee r)
& \leq &
(p \wedge q) \vee (p \wedge r)
& = &
\big[(p \wedge q)^{\bot} \wedge (p \wedge r)^{\bot}\big]^{\bot}.
\end{array}$$
\noindent That is, by~$(*)$ we have to prove:
$$\begin{array}{rcl}
p \wedge (q \vee r) \wedge (p \wedge q)^{\bot} \wedge (p \wedge r)^{\bot}
& = &
\zero.
\end{array}$$
\noindent We pick an arbitrary predicate $x$ for which:
$$(a)\; x\leq p
\qquad
(b)\; x \leq p\vee r
\qquad
(c)\; x \leq (p \wedge q)^{\bot}
\qquad
(d)\; x \leq (p \wedge r)^{\bot}.$$
\noindent We need to prove $x = \zero$. The inequality~(a) is
equivalent to $x \wedge p = x$. By~$(*)$, (c) and~(d) are equivalent
to $x \wedge p \wedge q = \zero$ and $x \wedge p \wedge r =
\zero$. But then $x \wedge q = \zero$ and $x \wedge r = \zero$, so
that $x \leq q^{\bot}$ and $x \leq r^{\bot}$ again by~$(*)$. We now
have $x \leq q^{\bot} \wedge r^{\bot} = (p \vee r)^{\bot}$, which,
together with~(b) and point~\eqref{lem:booleffectus:sharp} gives the
required conclusion:
$$\begin{array}{rcccl}
x
& \leq &
(p\vee r) \wedge (p\vee r)^{\bot}
& = &
\zero.
\end{array}\eqno{\QEDbox}$$
\end{enumerate}
\end{proof}

\begin{remark}
\label{rem:BooleanEffectusCommutation}
The commutation requirement $\asrt_{p} \pafter \asrt_{q} = \asrt_{q}
\pafter \asrt_{p}$ from
Definition~\ref{def:commbool}~\eqref{def:commbool:comm} holds
automatically in a Boolean effectus. More precisely, it follows from
the `Boolean' requirement $\asrt_{p} \pafter \asrt_{p^\bot} = \zero$
in Definition~\ref{def:commbool}, assuming the assert isomorphism in
point~\eqref{def:commbool:asrt}. This works as follows.

Let $p,q$ be arbitrary predicates on the same object. The associated
assert maps satisfy:
$$\begin{array}{rcl}
\zero
\hspace*{\arraycolsep}=\hspace*{\arraycolsep}
\asrt_{p} \pafter \asrt_{p^\bot}
& = &
\asrt_{p} \pafter \idmap \pafter \asrt_{p^\bot} \\
& = &
\asrt_{p} \pafter \big(\asrt_{q} \ovee \asrt_{q^{\bot}}\big) \pafter \asrt_{p^\bot} \\
& = &
\big(\asrt_{p} \pafter \asrt_{q} \pafter \asrt_{p^\bot}\big)
   \ovee \big(\asrt_{p} \pafter \asrt_{q^{\bot}} \pafter \asrt_{p^\bot}\big).
\end{array}$$

\noindent By Proposition~\ref{prop:effectusPCMorder}~\eqref{prop:effectusPCMorder:pos} we obtain:
$$\begin{array}{rclcrcl}
\asrt_{p} \pafter \asrt_{q} \pafter \asrt_{p^\bot}
& \smash{\stackrel{(a)}{=}} &
\zero.
& \qquad\mbox{Similarly,}\qquad &
\asrt_{p^\bot} \pafter \asrt_{q} \pafter \asrt_{p}
& \smash{\stackrel{(b)}{=}} &
\zero.
\end{array}$$

\noindent We can now derive commutation of assert maps:
$$\begin{array}{rcl}
\asrt_{p} \pafter \asrt_{q}
& = &
\asrt_{p} \pafter \asrt_{q} \pafter \big(\asrt_{p} \ovee \asrt_{p^\bot}\big) \\
& = &
\big(\asrt_{p} \pafter \asrt_{q} \pafter \asrt_{p}\big) 
   \ovee \big(\asrt_{p} \pafter \asrt_{q} \pafter \asrt_{p^\bot}\big) \\
& \smash{\stackrel{(a)}{=}} &
\big(\asrt_{p} \pafter \asrt_{q} \pafter \asrt_{p}\big) \ovee \zero \\
& \smash{\stackrel{(b)}{=}} &
\big(\asrt_{p} \pafter \asrt_{q} \pafter \asrt_{p}\big)
   \ovee \big(\asrt_{p^\bot} \pafter \asrt_{q} \pafter \asrt_{p}\big) \\
& = &
\big(\asrt_{p} \ovee \asrt_{p^\bot}\big) \pafter \asrt_{q} \pafter \asrt_{p} \\
& = &
\asrt_{q} \pafter \asrt_{p}.
\end{array}$$

\noindent It is not known if this commutation property can be derived
from the assert isomorphism in
Definition~\ref{def:commbool}~\eqref{def:commbool:asrt}.
\end{remark}

The next result justifies the name \emph{Boolean} effectuses. It uses
the category $\BA$\index{N}{cat@$\BA$, category of Boolean algebras}
of Boolean algebras, which is a subcategory $\BA \hookrightarrow \EA$
of the category of effect algebras.

\begin{proposition}
\label{prop:booleffectus}
For a Boolean effectus $\cat{B}$, all predicates are sharp, and the
predicate functor restricts to $\Pred \colon \cat{B} \rightarrow
\op{\BA}$.
\end{proposition}

\begin{proof}
For each $X\in\cat{B}$ the collection $\Pred(X)$ contains only sharp
predicates and is a Boolean algebra, by
Lemma~\ref{lem:booleffectus}~\eqref{lem:booleffectus:BA}. We have to
prove that for each map $f\colon Y \rightarrow X$ in $\cat{B}$ the
total substitution functor $\tbox{f} \colon \Pred(X) \rightarrow
\Pred(Y)$ is a map of Boolean algebras. For this it suffices that
$\tbox{f}$ preserves disjunctions $\vee$.

First, let $p,q\in\Pred(X)$ be disjoint, that is, $p \wedge q = 0$. By
Lemma~\ref{lem:booleffectus}~\eqref{lem:booleffectus:ortho}
and~\eqref{lem:booleffectus:ovee}, using that $\tbox{f}$ is a map
of effect algebras,
$$\begin{array}{rcccccl}
\tbox{f}(p \vee q)
& = &
\tbox{f}(p \ovee q)
& = &
\tbox{f}(p) \ovee \tbox{f}(q)
& = &
\tbox{f}(p) \vee \tbox{f}(q).
\end{array}$$

\noindent For arbitrary $p,q$ we can always rewrite the join $p \vee q$
in a Boolean algebra as a disjoint join:
$$\begin{array}{rcl}
p \vee q
& = &
(p \wedge q^{\bot}) \vee (p \wedge q) \vee (q \wedge p^{\bot}).
\end{array}$$

\noindent Since substitution $\tbox{f}$ preserves disjoint joins we
get:
$$\begin{array}[b]{rcl}
\tbox{f}(p \vee q)
& = &
\tbox{f}(p \wedge q^{\bot}) \vee \tbox{f}(p \wedge q) \vee 
   \tbox{f}(q \wedge p^{\bot}) \\
& = &
\tbox{f}(p \wedge q^{\bot}) \vee \tbox{f}(p \wedge q) \vee 
   \tbox{f}(p \wedge q) \vee \tbox{f}(q \wedge p^{\bot}) \\
& = &
\tbox{f}\big((p \wedge q^{\bot}) \vee (p \wedge q)\big) \vee 
   \tbox{f}\big((p \wedge q) \vee (q \wedge p^{\bot})\big) \\
& = &
\tbox{f}(p) \vee \tbox{f}(q).
\end{array}\eqno{\QEDbox}$$
\end{proof}

Later, in Section~\ref{sec:extcat}, we return to Boolean effectuses
and describe their close relationship to extensive categories. But we
first need the notions of comprehension and quotient, 
see Section~\ref{sec:cmpr} and~\ref{sec:quot}.

\section{Monoidal effectuses}\label{sec:monoidal}

In this section we shall consider effectuses with tensors $\otimes$.
Such tensors are used for parallel composition, which is an important
part of quantum theory. For background information on symmetric
monoidal categories we refer to~\cite{Maclane71}.

\begin{definition}
\label{def:monoidal}
An effectus in total form $\cat{B}$ is called
\emph{monoidal}\index{S}{monoidal effectus}\index{S}{effectus!monoidal
  --} if it is a symmetric monoidal category such that:
\begin{enumerate}
\item the tensor unit is the final object $1$;

\item the tensor distributes over finite coproducts $(+,0)$; this
  means that the following canonical maps are isomorphisms:
\begin{equation}
\label{diag:tensorsumdistr}
\vcenter{\xymatrix{
(X\otimes A) + (Y\otimes A)\ar[r] & (X+Y)\otimes A
&
0\ar[r] & 0\otimes A
}}
\end{equation}

\noindent This says that the functor $(-)\otimes A$ preserves finite
coproducts $(+,0)$. By symmetry $A\otimes(-)$ then also preserves
$(+,0)$.
\end{enumerate}
\end{definition}

We recall the notation that is standardly used for monoidal
natural isomorphisms, and add notation for the distributivity
isomorphism~\eqref{diag:tensorsumdistr}.
$$\begin{array}{c}
\vcenter{\xymatrix@C-.5pc{
1\otimes X\ar[r]^-{\lambda}_-{\cong} & X
& 
X\otimes Y\ar[r]^-{\gamma}_-{\cong} & Y\otimes Y
&
X\otimes (Y\otimes Z)\ar[r]^-{\alpha}_-{\cong} &
  (X\otimes Y)\otimes Z
}}
\\
\vcenter{\quad\xymatrix{
X\otimes I\ar[r]^-{\rho = \lambda\after \gamma} & X
& 
(X\otimes A)+(Y\otimes A)
   \ar@/^1.5ex/[rr]^-{[\kappa_{1}\otimes\idmap,\kappa_{2}\otimes\idmap]} & 
   \cong & (X+Y)\otimes A\qquad \ar@/^1.5ex/[ll]^-{\dis}
}}
\end{array}$$

Because the final object $1$ is the tensor unit there are projections
$X \leftarrow X\otimes Y \rightarrow Y$ defined in:
\begin{equation}
\label{diag:tensorproj}
\vcenter{\xymatrix@C+1pc@R-.5pc{
X & X\otimes Y\ar[l]_-{\pi_{1}}\ar[dl]^{\idmap\otimes \bang}
   \ar[r]^-{\pi_{2}}\ar[dr]_{\bang\otimes\idmap} & Y
\\
X\otimes 1\ar[u]^{\rho}_{\cong} & & 1\otimes Y\ar[u]_{\lambda}^{\cong}
}}\index{N}{$\pi_i$, $i$-th projection in a monoidal effectus}
\end{equation}

\noindent These projections take the marginal, see
Example~\ref{ex:monoidal}. They are natural in the sense that $\pi_{i}
\after (f_{1}\otimes f_{2}) = f_{i} \after \pi_{i}$. A monoidal
category in which the tensor unit is final is sometimes called
semicartesian.\index{N}{semicartesian}

We have described tensors in the `total' case, for effectuses. But
they can equivalently be described in the `partial' case. This is the
topic of the next result.

\begin{proposition}
\label{prop:monparttot}
Let $\cat{B}$ be an effectus in total form. Then: $\cat{B}$ is
monoidal effectus if and only if
\begin{itemize}
\item $\Par(\cat{B})$ is symmetric monoidal, and the monoidal and
  distributivity isomorphisms~\eqref{diag:tensorsumdistr} are total
  --- so that $\otimes$ distributes over $(+,0)$ in $\Par(\cat{B})$;

\item the tensor $f\otimes g$ of two total morphisms $f,g$ in
  $\Par(\cat{B})$ is again total;

\item the final object $1\in\cat{B}$ forms a tensor unit in
  $\Par(\cat{B})$.
\end{itemize}
\end{proposition}


When we talk about a `monoidal effectus in partial form' we mean an
effectus in partial form that satisfies the properties described above
for the category of partial maps $\Par(\cat{B})$, with $1$ as tensor
unit. The above second point is equivalent to: the following diagram
of partial maps commutes.
$$\xymatrix@R-.5pc@C-1pc{
& X\otimes Y\ar[dl]|-{\pafter}_{\one\otimes\one}\ar[dr]|-{\pafter}^-{\one} &
\\
1\otimes 1\ar[rr]|-{\pafter}^-{\cong} & & 1
}$$

\begin{proof}
Given the tensor $\otimes$ on $\cat{B}$ we obtain $\otimes$ on
$\Par(\cat{B})$ in the same way on objects. Given partial maps
$h\colon X \pto A$ and $k\colon Y\pto B$ in $\Par(\cat{B})$ we obtain
$h\otimes k \colon X\otimes Y \pto A\otimes B$ in $\Par(\cat{B})$ as:
$$\xymatrix@R-.5pc@C-1pc{
X\otimes Y\ar[d]^-{h\otimes k}
\\
(A+1)\otimes (B+1)\ar@{=}[r]^-{\sim} &
(A\otimes B) + (A\otimes 1) + (1\otimes B) + (1\otimes 1)
   \ar[d]^{[\kappa_{1}, \kappa_{2}\after \bang, \kappa_{2}\after \bang, \kappa_{2}\after \bang]}
\\
& (A\otimes B)+1
}$$

\noindent Then $\klin{f}\otimes\klin{g} = \klin{f\otimes
  g}$. Remaining details are left to the reader. \QED
\end{proof}

\begin{example}
\label{ex:monoidal}
We shall look at monoidal structure in three of our four running
examples. It is not clear if order unit groups have tensors.
\begin{enumerate}
\item The effectus $\Sets$ has finite products $(\times, 1)$, where
  functors $A\times(-)$ preserve finite coproducts. Obviously, the
  final object is the unit for $\times$.

\item The effectus $\Kl(\Dst)$ for discrete probability has tensors
  because the distribution monad $\Dst$ is monoidal: there are natural
  maps $\Dst(X)\times\Dst(Y) \rightarrow \Dst(X\times Y)$ which
  commute appropriately with the unit and multiplication of the monad.
  These maps send a pair $(\omega,\rho)\in\Dst(X)\times\Dst(Y)$ to the
  distribution in $\Dst(X\times Y)$ given by $(x,y) \mapsto
  \omega(x)\cdot \rho(y)$.

The resulting tensor on $\Kl(\Dst)$ is given on objects by 
cartesian product $\times$. For morphisms $f\colon A \rightarrow
\Dst(X)$ and $g\colon B \rightarrow \Dst(Y)$ there is map
$f\otimes g$ in $\Kl(\Dst)$ given by the composite:
$$\xymatrix{
A\times B\ar[r]^-{f\times g} & \Dst(X) \times \Dst(Y)\ar[r] &
   \Dst(X\times Y)
}$$

\noindent the projection map $X\otimes Y \rightarrow X$ in
$\Kl(\Dst)$, is, following~\eqref{diag:tensorproj}, the function
$$\xymatrix{
X\times Y \ar[r]^-{\pi_{1}} & \Dst(X)\qquad\mbox{with}\qquad
   \pi_{1}(x,y) = 1\ket{x}
}$$

\noindent When applied to a state $\omega\in\Dst(X\times Y) =
\Stat(X\times Y)$ we obtain the marginal distribution
$\Stat(\pi_{1})(\omega) = \tstat{(\pi_{1})}(\omega) \in\Stat(X)$ given
by:
$$\begin{array}{rcl}
\tstat{(\pi_{1})}(\omega)
& = &
{\displaystyle\sum}_{x}\big(\sum_{y}\omega(x,y)\big)\bigket{x}.
\end{array}$$

\item The effectus $\op{\vNA}$ of von Neumann algebras is also
  monoidal. Here it is essential that the morphisms of~$\vNA$
  are completely positive. The details are quite complicated, see
  \textit{e.g.}~\cite{Takesaki01} and~\cite{Cho14} for distributivity
  results using the `minimal' tensor. The tensor unit is the algebra
  $\C$ of complex numbers, which is final in $\op{\vNA}$. Distribution
  of $\otimes$ over coproducts $(\oplus, 0)$ is investigated
  in~\cite{Cho14}.
\end{enumerate}
\end{example}

The next result is proven explicitly in~\cite{Jacobs15a}. Here we give
a more abstract argument, using partial maps.

\begin{corollary}
\label{cor:scalmultcomm}
In a monoidal effectus the effect monoid of scalars $\Pred(1)$ is
commutative.
\end{corollary}

\begin{proof}
It is a well-known fact that the endomaps $I\rightarrow I$ on the
tensor unit in a monoidal category form a commutative monoid under
composition. This relies on the `Eckmann-Hilton' argument see
\textit{e.g.}~\cite{KellyL80}. We apply this to the monoidal category
$\Par(\cat{B})$ of a monoidal effectus $\cat{B}$. The scalars are
precisely the partial maps $1 \pto 1$ on the tensor unit $1$, and
their multiplication is partial/Kleisli composition, see
Definition~\ref{def:effectusScalars}. \QED
\end{proof}

We need to know a bit more about the monoidal structure in categories
of partial maps.

\begin{lemma}
\label{lem:monpartprops}
Let $\cat{B}$ be a monoidal effectus in total form. Then, in
$\Par(\cat{B})$,
\begin{enumerate}
\item \label{lem:monpartprops:zero} $f\otimes \zero = \zero \colon
  X\otimes Y \pto A\otimes B$;

\item \label{lem:monpartprops:partproj} the following diagram commutes:
$$\xymatrix@C-1pc@R-.5pc{
X_{1}\otimes A + X_{2}\otimes A\ar[rr]|-{\pafter}^-{\cong}
   \ar[dr]|-{\pafter}_{\rhd_{i}} & &
   (X_{1}+X_{2})\otimes A\ar[dl]|-{\pafter}^{\rhd_{i}\otimes\idmap}
\\
& X_{i}\otimes A &
}$$

\item \label{lem:monpartprops:ovee} $f\otimes (g_{1} \ovee g_{2}) =
  (f\otimes g_{1}) \ovee (f\otimes g_{2})$.

\item \label{lem:monpartprops:proj} The functor $\klin{-} \colon
  \cat{B} \rightarrow \Par(\cat{B})$ preserves projections: these
  projections can be described in $\Par(\cat{B})$ as:
$$\begin{array}{rcl}
\klin{\pi_{1}}
& = &
\vcenter{\xymatrix@C-1pc{
\Big(X\otimes Y \ar[rr]^-{\idmap\otimes\one} & & 
   X\otimes 1 \ar[r]^-{\rho}_-{\cong} & X\Big)
}}
\\
\klin{\pi_{2}}
& = &
\vcenter{\xymatrix@C-1pc{
\Big(X\otimes Y \ar[rr]^-{\one\otimes\idmap} & & 
   1\otimes Y \ar[r]^-{\lambda}_-{\cong} & Y\Big)
}}
\end{array}$$

\noindent These projection maps are natural only wrt.\ total maps.
\end{enumerate}
\end{lemma}

\begin{proof}
We follow the relevant definitions.
\begin{enumerate}
\item We have $f\otimes \zero = \zero$ via the following diagram, using that
  $0$ is the zero object in $\Par(\cat{B})$.
$$\xymatrix@R-1pc@C+1pc{
X\otimes A\ar[r]|-{\pafter}^-{f\otimes \bang}\ar[dr]|-{\pafter}_{\bang} & 
   Y\otimes 0\ar[r]|-{\pafter}^-{\idmap\otimes \bang}\ar@{=}[d]^{\wr} & 
   Y\otimes B
\\
& 0\ar[ur]|-{\pafter}_{\bang} &
}$$

\item We calculate in $\Par(\cat{B})$, for $i = 1$,
$$\begin{array}{rcl}
\lefteqn{(\rhd_{1}\otimes\idmap) \pafter 
   [\kappa_{1}\otimes\idmap, \kappa_{2}\otimes\idmap]} \\
& = &
[([\idmap,\zero]\otimes\idmap) \pafter (\kappa_{1}\otimes\idmap), 
   ([\idmap,\zero]\otimes\idmap) \pafter (\kappa_{2}\otimes\idmap)] \\
& = &
[\idmap\otimes\idmap, \zero\otimes\idmap] \\
& = &
[\idmap, \zero] \qquad \mbox{by point~\eqref{lem:monpartprops:zero}} \\
& = &
\rhd_{1}.
\end{array}$$

\item Let $b\colon Y \pto B+B$ be a bound for $g_{1}, g_{2} \colon Y
  \pto B$. Then we take as new bound:
$$\xymatrix{
c = \Big(X\otimes Y\ar[r]|-{\pafter}^-{f\otimes b} &
   A\otimes (B+B)\ar[r]|-{\pafter}^-{\cong} & A\otimes B + A\otimes B\Big)
}$$

\noindent Then, $c$ is a bound for $f\otimes g_{i}$, since by
point~\eqref{lem:monpartprops:partproj},
$$\begin{array}{rcccccl}
\rhd_{i} \pafter c
& = &
(\idmap\otimes\rhd_{i}) \pafter (f\otimes b)
& = &
f\otimes (\rhd_{i}\otimes b)
& = &
f\otimes g_{i}.
\end{array}$$

\noindent Hence:
$$\begin{array}[b]{rcl}
(f\otimes g_{1}) \ovee (f\otimes g_{2})
\hspace*{\arraycolsep}=\hspace*{\arraycolsep}
\nabla \pafter c
& = &
(\idmap\otimes\nabla) \pafter (f\otimes b) \\
& = &
f \otimes (\nabla\pafter b) \\
& = &
f\otimes (g_{1} \ovee g_{2}).
\end{array}$$

\item We use that the functor $\klin{-} \colon \cat{B} \rightarrow
\Par(\cat{B})$ preserves the tensor, and that $\one \colon X \rightarrow 1$
is the unique map to the final object in $\Par(\cat{B})$. Hence the
first projection in $\Par(\cat{B})$ is:
$$\begin{array}[b]{rcl}
\rho \pafter (\idmap\otimes\one)
\hspace*{\arraycolsep}=\hspace*{\arraycolsep}
\rho \pafter (\klin{\idmap}\otimes\klin{\bang}) 
& = &
(\rho\otimes\idmap) \after \klin{\idmap\otimes\bang} \\
& = &
\kappa_{1} \after \rho \after (\idmap\otimes\bang) \\
& = &
\klin{\pi_{1}}.
\end{array}\eqno{\QEDbox}$$
\end{enumerate}
\end{proof}

In Proposition~\ref{prop:effectusPCM} we have seen that homsets of
partial maps $X \pto Y$ in an effectus form a partial commutative
monoid (PCM). In the presence of tensors one also obtains scalar
multiplication on such homsets.

\begin{lemma}
\label{lem:monPCMod}
If $\cat{B}$ is a monoidal effectus in total form, then each partial
homset $\Par(\cat{B})(X, Y)$ is a partial commutative module; partial
pre- and post-composition preserves scalar
multiplication,\index{S}{scalar multiplication!-- on homsets, in
  presence of tensor}
and thus the module structure.

Moreover, the homsets $\cat{B}(X,Y)$ are convex sets, and the partial
homsets $\Par(\cat{B})(X, Y)$ are subconvex sets; again this structure
is preserved by pre- and post-composition.
\end{lemma}

\begin{proof}
For a partial map $f\colon X\pto Y$ and a scalar $s\colon 1 \pto 1$ we
define $s\cdot f \colon X \pto Y$ in $\Par(\cat{B})$ as:
$$\xymatrix@C-.8pc{
s\cdot f = \Big(X\ar[r]|-{\pafter}^-{\cong} & 
   1\otimes X\ar[rr]|-{\pafter}^-{s\otimes f} & &
   1\otimes Y\ar[r]|-{\pafter}^-{\cong} & Y\Big)
}$$

\noindent It is obviously an action, and preserves $0,\ovee$ in each
argument by Lemma~\ref{lem:monpartprops}. 

We only show that the homset $\cat{B}(X,Y)$ is convex. We essentially
proceed as in (the proof of) Lemma~\ref{lem:effectusConv}: for scalars
$r_{1}, \ldots, r_{n} \in \Pred(1)$ with $\bigovee_{i} r_{i} = \one$,
the bound $b\colon 1 \pto n\cdot 1$ with $\rhd_{i} \pafter b = r_{i}$
and $\nabla \pafter b = \one \colon 1 \pto 1$ is a total map of the
form $b = \kappa_{1} \tafter s$. For an $n$-tuple of total maps $f_{i}
\colon X \to X$ we now define the convex sum $\sum_{i}r_{i}f_{i}$ to be
the composite in $\cat{B}$:
$$\xymatrix@C-0.5pc{
X\ar[r]^-{\cong} & 
   1\otimes X\ar[r]^-{s\otimes \idmap} &
   (n\cdot 1) \otimes X\ar[r]^-{\cong} & 
   n\cdot (1\otimes X)\ar[r]^-{\cong} & 
   n\cdot X\ar[rr]^-{[f_{1}, \ldots, f_{n}]} & & Y
}\eqno{\QEDbox}$$
\end{proof}

In Definition~\ref{def:puresstate} we have described when substates $1
\pto X$ are pure, via scalar multiplication for such states. With the
scalar multiplication for arbitrary partial maps we can extend the
definition of purity.

\begin{definition}
\label{def:purepartmap}
A partial map $f\colon X \pto Y$ in a monoidal effectus is called
\emph{pure} if for each pair of orthogonal partial maps $f_{1}, f_{2}
\colon X \pto Y$ with $f = f_{1} \ovee f_{2}$ there is a scalar $s$
with:
$$\begin{array}{rclcrcl}
f_{1}
& = &
s \cdot f
& \qquad\mbox{and}\qquad &
f_{2}
& = &
s^{\bot} \cdot f.
\end{array}$$

\noindent An object $X$ is called pure if the partial identity map
$\idmap \colon X \pto X$ is pure, in the above sense.
\end{definition}

It can be shown that a von Neumann algebra $\mathscr{A}$ is pure if
and only if it is a factor, that is, its center is $\C$, the algebra
of compex numbers. The proof will be given elsewhere.

Here is another notion that requires tensors. It is seen as the key
property of quantum mechanics according to~\cite{ChiribellaAP11}. It
will be further elaborated elsewhere.

\begin{definition}
\label{def:purification}
The \emph{purification} of a substate $\omega \colon 1 \pto X$ is a
pure substate $\rho \colon 1 \pto X\otimes Y$ with $\pi_{1} \pafter
\rho = \omega$.
\end{definition}

\subsection{Commutative effectuses via copiers}\label{subsec:copier}

This section concentrates on special tensors $\otimes$ which come
equipped with copier operations $X \rightarrow X\otimes X$. It is
shown that in presence of such copiers --- satisfying suitable
requirements --- an effectus is commutative. This result is the
logical contrapositive of the familiar fact that no-cloning holds in
the quantum world: if we do have cloning, we are not in the quantum,
but in the probabilistic, world.

\begin{definition}
\label{def:copier}
We say that a monoidal effectus $\cat{B}$ in total form has
copiers\index{N}{copier} if for each object $X\in\cat{B}$ there is
a copier, or diagonal, map $\delta \colon X \rightarrow X\otimes X$ 
such that the following diagrams commute:
$$\xymatrix@C-.5pc@R-.5pc{
& X 
& 
& X\otimes X\ar[dd]_{\cong}^{\gamma}
&
X\ar@{=}[r]\ar[d]_{\delta} & X\ar[d]^{\delta}
\\
X\ar[r]^-{\delta}\ar@{=}[ur]\ar@{=}[dr] & X\otimes X\ar[u]_{\pi_1}\ar[d]^{\pi_2}
&
X\ar[ur]^-{\delta}\ar[dr]_-{\delta} & 
&
X\otimes X\ar[d]_{\idmap\otimes\delta} & X\otimes X\ar[d]^{\delta\otimes\idmap}
\\
& X
&
& X\otimes X
&
X\otimes (X\otimes X)\ar[r]_-{\alpha}^-{\cong} & (X\otimes X)\otimes X
}$$

\noindent and additionally, for each partial map $f\colon X \pto X$
which is side-effect-free, that is satisfies $f\leq \idmap$, one has
in $\Par(\cat{B})$:
\begin{equation}
\label{eqn:copiersef}
\begin{array}{rcccl}
(f\otimes\idmap) \pafter \klin{\delta}
& = &
\klin{\delta} \pafter f
& = &
(\idmap\otimes f)\pafter \klin{\delta}.
\end{array}
\end{equation}

\noindent (The two equations are equivalent.)

\auxproof{
Assuming the left-equation we get:
$$\begin{array}{rcl}
(\idmap\otimes f)\pafter \klin{\delta}
& = &
(\idmap\otimes f)\pafter \gamma \pafter \klin{\delta} \\
& = &
\gamma \pafter (f\otimes\idmap) \pafter \klin{\delta} \\
& = &
\gamma \pafter \klin{\delta} \pafter f \\
& = &
\klin{\delta} \pafter f.
\end{array}$$
}

\end{definition}

We now define an assert map in terms of copying. It's easiest to do
this in partial form. For a predicate $p\colon X \pto 1$ write
$\asrt_{p}$ for the composite:
\begin{equation}
\label{diag:asrtpartcopy}
\xymatrix@C-.5pc{
\asrt_{p} = \Big(X\ar[r]|-{\pafter}^-{\klin{\delta}} & 
   X\otimes X\ar[r]|-{\pafter}^-{p\otimes\idmap} & 
   1\otimes X\ar[r]|-{\pafter}^-{\lambda}_-{\cong} & X\Big)
}\index{N}{$\asrt_p$, assert map for
    predicate $p$!in a monoidal effectus with copiers}
\end{equation}


\begin{theorem}
\label{thm:copier}
A monoidal effectus with copiers is commutative.
\end{theorem}

\begin{proof}
Let $\cat{B}$ be a monoidal effectus in total form, with copiers. We
reason in $\Par(\cat{B})$. We first show that the map $\asrt_{p}\colon
X \pto X$ from~\eqref{diag:asrtpartcopy} is side-effect-free,
\textit{i.e.}~that it is below the (partial) identity $X \pto X$. This
follows from:
$$\begin{array}{rcll}
\lefteqn{\asrt_{p} \ovee \asrt_{p^{\bot}}} \\
& = &
\big(\lambda \pafter (p\otimes\idmap) \pafter \klin{\delta}\big) \ovee
   \big(\lambda \pafter (p^{\bot}\otimes\idmap) \pafter \klin{\delta}\big) \\
& = &
\lambda \pafter \big((p\otimes\idmap)\ovee(p^{\bot}\otimes\idmap)\big)
   \pafter \klin{\delta} &
   \mbox{by Proposition~\ref{prop:effectusPCM}~\eqref{prop:effectusPCM:pres}} \\
& = &
\lambda \pafter \big((p\ovee p^{\bot})\otimes\idmap\big)
   \pafter \klin{\delta} & 
   \mbox{by Lemma~\ref{lem:monpartprops}~\eqref{lem:monpartprops:ovee}} \\
& = &
\lambda \pafter \big(\one\otimes\idmap\big) \pafter \klin{\delta} \\
& = &
\klin{\pi_{2}} \pafter \klin{\delta} &
   \mbox{by Lemma~\ref{lem:monpartprops}~\eqref{lem:monpartprops:proj}} \\
& = &
\idmap.
\end{array}$$

\noindent Next we show that the mapping $p \mapsto \asrt_{p}$ makes the map
$\kerbot \colon \sEnd(X) \rightarrow \Pred(X)$ bijective, see
Definition~\ref{def:commbool}~\eqref{def:commbool:asrt}.
$$\begin{array}{rcll}
\kerbot(\asrt_{p})
& = &
\one \pafter \lambda \pafter (p\otimes\idmap) \pafter \klin{\delta} \\
& = &
\lambda \pafter (\idmap\otimes\one) \pafter (p\otimes\idmap) \pafter 
   \klin{\delta} \\
& = &
\rho \pafter (p\otimes\idmap) \pafter (\idmap\otimes\one) \pafter 
   \klin{\delta} \quad & 
   \mbox{since $\lambda=\rho\colon 1\otimes 1 \pto 1$} \\
& = &
p \pafter \rho \pafter (\idmap\otimes\one) \pafter \klin{\delta} \\
& = &
p \pafter \klin{\pi_{1}} \pafter \klin{\delta} &
   \mbox{by Lemma~\ref{lem:monpartprops}~\eqref{lem:monpartprops:proj}} \\
& = &
p.
\end{array}$$

\noindent In the other direction, let $f\colon X \pto X$ satisfy
$f\leq \idmap$. The associated predicate $\kerbot(f) = \one \pafter f
\colon X \pto 1$ satisfies:
$$\begin{array}{rcll}
\asrt_{\kerbot(f)}
& = &
\lambda \pafter (\one\otimes\idmap) \pafter (f\otimes\idmap) 
   \pafter \klin{\delta} \quad \\
& = &
\klin{\pi_{1}} \pafter \klin{\delta} \pafter f &
   \mbox{by~\eqref{eqn:copiersef} and 
      Lemma~\ref{lem:monpartprops}~\eqref{lem:monpartprops:proj}} \\
& = &
f.
\end{array}$$

Finally we have to prove that the assert maps commute: $\asrt_{p}
\pafter \asrt_{q} = \asrt_{q} \pafter \asrt_{p}$. Since the
kernel-supplement map $\kerbot$ is bijective it suffices to prove the
middle equation in:
$$\begin{array}{rcccccl}
\kerbot\big(\asrt_{q} \pafter \asrt_{p}\big)
& = &
p\andthen q
& = &
q\andthen p
& = &
\kerbot\big(\asrt_{p} \pafter \asrt_{q}\big).
\end{array}$$

\noindent Hence we are done by:
$$\begin{array}[b]{rcll}
p \andthen q
& = &
q \pafter \asrt_{p} &
   \mbox{by Lemma~\ref{lem:commeffectus}~\eqref{lem:commeffectus:order}} \\
& = &
q \pafter \lambda \pafter (p\otimes\idmap) \pafter \klin{\delta} \\
& = &
\lambda \pafter (\idmap\otimes q) \pafter (p\otimes\idmap) \pafter 
   \klin{\delta} \quad \\
& = &
\lambda \pafter \gamma \pafter (p\otimes q) \pafter \klin{\delta} &
   \mbox{since }
   \lambda = \rho = \lambda \pafter \gamma \colon 1\otimes 1 \pto 1 \\
& = &
\lambda \pafter (q\otimes p) \pafter \gamma \pafter \klin{\delta} \\
& = &
p \pafter \lambda \pafter (q\otimes\idmap) \pafter \klin{\delta} \\
& = &
q \andthen p.
\end{array}\eqno{\QEDbox}$$

\auxproof{
We first show that the map $\asrt_{p}\colon X \rightarrow X+1$ is
side-effect-free, that is, is below the (partial) identity $X
\rightarrow X+1$. In fact, we show that $\asrt_{p} \ovee
\asrt_{p^\bot} = \idmap \colon X \rightarrow X+1$. We take as bound $b
= \kappa_{1} \after (\lambda+\lambda) \after \dis \after
(p\otimes\idmap) \after \delta \colon X \rightarrow (X+X)+1$. Then:
$$\begin{array}{rcll}
\rhd_{1} \pafter b
& = &
[\rhd_{1}, \kappa_{2}] \after \kappa_{1} \after (\lambda+\lambda) \after 
   \dis \after (p\otimes\idmap) \after \delta \\
& = &
(\idmap+\bang) \after (\lambda+\lambda) \after 
   \dis \after (p\otimes\idmap) \after \delta \\
& = &
(\lambda+\bang) \after \dis \after (p\otimes\idmap) \after \delta \\
& = &
\asrt_{p} 
\\
\rhd_{2} \pafter b
& = &
[\rhd_{2}, \kappa_{2}] \after \kappa_{1} \after 
   (\lambda+\lambda) \after \dis \after (p\otimes\idmap) \after \delta \\
& = &
[\kappa_{2}\after\bang, \kappa_{1}] \after 
   (\lambda+\lambda) \after \dis \after (p\otimes\idmap) \after \delta \\
& = &
(\lambda+\bang) \after [\kappa_{2}, \kappa_{1}] \after 
   \dis \after (p\otimes\idmap) \after \delta \\
& = &
(\lambda+\bang) \after \dis \after ([\kappa_{2}, \kappa_{1}]\otimes\idmap) 
   \after (p\otimes\idmap) \after \delta \qquad
   & \mbox{see~\eqref{diag:copier:swap}} \\
& = &
\asrt_{p^\bot} 
\\
\asrt_{p} \ovee \asrt_{p^\bot}
& = &
\nabla \pafter b \\
& = &
\kappa_{1} \after \nabla \after (\lambda+\lambda) \after 
   \dis \after (p\otimes\idmap) \after \delta \\
& = &
\kappa_{1} \after \lambda \after \nabla \after 
   \dis \after (p\otimes\idmap) \after \delta \\
& = &
\kappa_{1} \after \lambda \after (\nabla\otimes\idmap) 
   \after (p\otimes\idmap) \after \delta 
   & \mbox{see~\eqref{diag:copier:swap}} \\
& = &
\kappa_{1} \after \lambda \after (!\otimes\idmap) \after \delta \\
& = &
\kappa_{1} \after \pi_{2} \after \delta \\
& = &
\kappa_{1}.
\end{array}$$

Next we show that the mapping $p \mapsto \asrt_{p}$ makes the map
$\kerbot \colon \sEnd(X) \rightarrow \Pred(X)$ bijective, see
Definition~\ref{def:commbool}~\eqref{def:commbool:asrt}.
$$\begin{array}{rcll}
\kerbot(\asrt_{p})
& = &
(!+\idmap) \after (\lambda+\bang) \after \dis \after (p\otimes\idmap) 
   \after \delta \\
& = &
(\bang+\bang) \after \dis \after (p\otimes\idmap) \after \delta \\
& = &
(\pi_{1}+\pi_{1}) \after \dis \after (p\otimes\idmap) \after \delta
   & \mbox{since }! = \pi_{1} \colon 1\otimes X \rightarrow 1 \\
& = &
\pi_{1} \after (p\otimes\idmap) \after \delta &
   \mbox{see the left diagram in~\eqref{diag:copier:distr}} \\
& = &
p \after \pi_{1} \after \delta \\
& = &
p.
\end{array}$$

\noindent In the other direction, let $f\colon X \rightarrow X+1$
satisfy $f\leq \idmap$. The associated predicate $\kerbot(f) = 
(\bang+\idmap) \after f \colon X \rightarrow 1+1$ satisfies:
$$\begin{array}{rcl}
\asrt_{\kerbot(f)}
& = &
(\lambda+\bang) \after \dis \after ((\bang+\idmap)\otimes\idmap) \after 
   (f\otimes\idmap) \after \delta \\
& = &
(\lambda+\bang) \after ((\bang\otimes\idmap)+\idmap) \after \dis \after 
   (f\otimes\idmap) \after \delta \\
& = &
(\pi_{1}+\pi_{1}) \after \dis \after (f\otimes\idmap) \after \delta \\
& = &
\pi_{1} \after (f\otimes\idmap) \after \delta 
   \qquad \mbox{by } \\
& = &
f \after \pi_{1} \after \delta \\
& = &
f.
\end{array}$$

Finally we have to prove that the induced andthen operation $p\andthen
q = q \pafter \asrt_{p}$ is commutative. We first note:
$$\begin{array}{rcl}
p \andthen q
& = &
[q,\kappa_{2}] \after (\lambda+\bang) \after \dis \after 
   (p\otimes\idmap) \after \delta \\
& = &
[q \after \lambda ,\kappa_{2}\after\bang] \after \dis \after 
   (p\otimes\idmap) \after \delta \\
& = &
[\lambda \after (\idmap\otimes q) , \zero \after (\idmap\otimes q)] \after 
   \dis \after (p\otimes\idmap) \after \delta \\
& = &
[\lambda,\zero] \after ((\idmap\otimes q)+(\idmap\otimes q)) \after \dis \after 
   (p\otimes\idmap) \after \delta \\
& = &
[\lambda,\zero] \after \dis \after (\idmap\otimes q) \after 
   (p\otimes\idmap) \after \delta \\
& = &
[\lambda,\zero] \after \dis \after (p\otimes q) \after \delta.
\end{array}$$

\noindent But then we are done by using the swap equations $\gamma
\after \delta = \delta$ and $\dis \after \gamma = \dis$ in:
$$\begin{array}[b]{rcl}
p \andthen q
& = &
[\lambda,\zero] \after \dis \after (p\otimes q) \after \delta \\
& = &
[\lambda,\zero] \after \dis \after (p\otimes q) \after \gamma \after \delta \\
& = &
[\lambda,\zero] \after \dis \after \gamma \after (q\otimes p) \after \delta \\
& = &
[\lambda,\zero] \after \dis \after (q\otimes p) \after \delta \\
& = &
q \andthen p.
\end{array}\eqno{\QEDbox}$$
}
\end{proof}

We conclude by describing copiers in two of our examples.

\begin{examples}
\label{ex:copier}
In the commutative effectus $\Kl(\Dst)$ one has copiers
$\delta \colon X \rightarrow X\otimes X$ given by:
$$\begin{array}{rcl}
\delta(x)
& = &
1\ket{(x,x)}.
\end{array}$$

\noindent These copiers are \emph{not} natural in $X$. But we do not
need naturality for the above theorem. We do need to check
equation~\eqref{eqn:copiersef}. So let $f\colon X \rightarrow
\sDst(X)$ be side-effect-free map. Then $f$ is of the form $f(x) =
p(x)\ket{x}$, for a predicate $p\in [0,1]^{X}$, see
Example~\eqref{ex:commbool}~\eqref{ex:commbool:KlD}. But then:
$$\begin{array}{rcccl}
\big((f\otimes \idmap) \after \delta\big)(x)
& = &
p(x)\ket{(x,x)}
& = &
\big(\delta \after f\big)(x).
\end{array}$$

In the commutative effectus $\op{\CvNA}$ of commutative von Neumann
algebras copiers $\delta \colon \mathscr{A}\otimes\mathscr{A}
\rightarrow \mathscr{A}$ are given by multiplication:
$$\begin{array}{rcl}
\delta(a\sotimes b)
& = &
a\cdot b.
\end{array}$$

\noindent This map is positive since the multiplication of two
positive elements is positive again in a commutative von Neumann
algebra.

If $f\colon \mathscr{A} \pto \mathscr{A}$ is a subunital map below the
identity, then $f(x) = f(1)\cdot x$, see
Example~\ref{ex:commbool}~\eqref{ex:commbool:vNA}.  Hence the
equation~\ref{eqn:copiersef} holds:
$$\begin{array}{rcl}
\big(\delta \after (f\otimes\idmap)\big)(a\sotimes b)
\hspace*{\arraycolsep}=\hspace*{\arraycolsep}
\delta\big(f(a)\sotimes b\big) 
& = &
(f(1)\cdot a) \cdot b \\
& = &
f(1) \cdot (a\cdot b) 
\hspace*{\arraycolsep}=\hspace*{\arraycolsep}
\big(f \after \delta\big)(a\sotimes b).
\end{array}$$
\end{examples}

\section{Comprehension}\label{sec:cmpr}

Comprehension is the operation that assigns to a predicate $p$ on a
type/object $X$ a new type/object $\cmpr{X}{p}$ which intuitively
contains those elements of $X$ that satisfy $p$. This comprehension
type comes equipped with an inclusion map $\pi_{p} \colon \cmpr{X}{p}
\rightarrow X$ which we call a comprehension map.

In categorical logic comprehension is described nicely via an
adjunction, see \textit{e.g.}~\cite{Jacobs99}, like any logical
operation that comes equipped with an introduction and elimination
rule. The slogan is: comprehension is right adjoint to truth.  Later
in Section~\ref{sec:quot} we shall see a similar situation for
quotients, namely: quotients are left adjoint to falsity.

This truth functor exists as functor that maps an object $X$ to the
truth/top element $1(X)$ in the set $\Pred(X)$ of predicates on $X$.
Our first task is to combine these set $\Pred(X)$ into a single
category $\Pred(\cat{B})$ of predicates in an effectus $\cat{B}$.  In
the partial case there is a slightly different category of predicates.

\begin{definition}
\label{def:totpredcat}
Let $\cat{B}$ be an effectus in total form. We write
$\Pred(\cat{B})$\index{N}{cat@$\Pred(\cat{B})$ category of predicates
  in an effectus in total form $\cat{B}$} for the category with:
\begin{itemize}
\item objects are pairs $(X,p)$ where $p\in\Pred(X)$ is a predicate
  $p\colon X \rightarrow 1+1$ on $X$; often we simply talk about
  objects $p$ in $\Pred(\cat{B})$ when the carrier $X$ is clear from
  the context;

\item morphisms $f\colon (X,p) \rightarrow (Y,q)$ in $\Pred(\cat{B})$
  are maps $f\colon X \rightarrow Y$ in $\cat{B}$ with an inequality:
$$\begin{array}{rcl}
p
& \leq &
\tbox{f}(q).
\end{array}$$

\noindent We recall that $\tbox{f}(q) = q \after f$ is total
substitution, which is a map of effect modules by
Theorem~\ref{thm:effectusEMod}.
\end{itemize}

\noindent We obtain a category $\Pred(\cat{B})$ by
Exercise~\ref{exc:subst}~\eqref{exc:subst:totalfun}. It comes equipped
with a forgetful functor $\Pred(\cat{B}) \rightarrow \cat{B}$ which is
so trivial that we don't bother to give it a name.

There are truth and falsity functors $\one, \zero \colon \cat{B}
\rightarrow \Pred(\cat{B})$\index{N}{$\zero$,!falsity
  functor}\index{N}{$\one$,!truth functor} given by the truth/top
element $\one(X)\in\Pred(X)$ and the falsity/bottom $\zero(X)
\in\Pred(X)$ respectively. These functors form left and right adjoints
to the forgetful functor in:
\begin{equation}
\label{diag:totpredzeroone}
\vcenter{\xymatrix@R-.5pc{
\Pred(\cat{B})\ar[d]_{\dashv\;}^{\;\dashv}
\\
\cat{B}\ar@/^4ex/[u]^(0.4){\zero\!}\ar@/_4ex/[u]_(0.4){\!\one}
}}
\end{equation}
\end{definition}

The forgetful functor $\Pred(\cat{B}) \rightarrow \cat{B}$ is an
example of a `fibration', or `indexed category', which is a basic
structure in categorical logic and type theory~\cite{Jacobs99}. Here
however we use this functor concretely, without going into the general
theory.

The adjunction $\zero \dashv \mbox{forget} \dashv \one$
in~\eqref{diag:totpredzeroone} exists because there are obvious
bijective correspondences:
$$\begin{prooftree}
{\xymatrix{X\ar[r]^-{f} & Y \rlap{\hspace*{7em}in $\cat{B}$}}}
\Justifies
{\xymatrix{(X,p)\ar[r]_-{f} & 
  (Y,\one)\rlap{$\;=\one(Y)$\hspace*{2em}in $\Pred(\cat{B})$}}}
\end{prooftree}$$

\noindent And:
$$\begin{prooftree}
{\xymatrix{X\ar[r]^-{f} & Y\rlap{\hspace*{7.8em}in $\cat{B}$}}}
\Justifies
{\xymatrix{\zero(Y)=\;(X,\zero)\ar[r]_-{f} & 
  (Y,q)\rlap{\hspace*{4.3em}in $\Pred(\cat{B})$}}}
\end{prooftree}$$

\noindent These trivial correspondences exist because one always
has $p \leq \tbox{f}(\one) = \one$ and $0 \leq \tbox{f}(q)$.

There is a similar, but slightly different, way to combine predicates
in a category in the partial case.

\begin{definition}
\label{def:partpredcat}
Let $(\cat{C},I)$ be an effectus in partial form. We write
$\PPred(\cat{C})$\index{N}{cat@$\PPred(\cat{C})$, category of
  predicates in an effectus in partial form $\cat{C}$} for the
category with:
\begin{itemize}
\item objects $(X,p)$, where $p\colon X \rightarrow I$ is a predicate
on $X$;

\item morphisms $f\colon (X,p) \rightarrow (Y,q)$ are maps $f\colon X
  \to Y$ in $\cat{C}$ satisfying:
$$\begin{array}{rcl}
p
& \leq &
\pbox{f}(q)
\end{array}$$

\noindent where $\pbox{f}(q) = (q^{\bot} \after f)^{\bot}$ is
partial substitution, which is monotone and preserves truth by
Lemma~\ref{lem:partialmonotone}. This yields a category
by Exercise~\ref{exc:subst}~\eqref{exc:subst:partialfun}.
\end{itemize}

\noindent Like in the total case we have falsity and truth functors as
left and right adjoints to the forgetful functor.
\begin{equation}
\label{diag:partpredzeroone}
\vcenter{\xymatrix@R-.5pc{
\PPred(\cat{C})\ar[d]_{\dashv\;}^{\;\dashv}
\\
\cat{C}\ar@/^4ex/[u]^(0.4){\zero\!}\ar@/_4ex/[u]_(0.4){\!\one}
}}
\end{equation}
\end{definition}

The adjoint correspondences in~\eqref{diag:partpredzeroone} work just
like in the total case. They use that partial substution $\pbox{f}$
preserves truth.

The above definition for an effectus in partial form applies in
particular to the category $\Par(\cat{B})$ of partial maps of an
effectus $\cat{B}$. This is the topic of the next result.

\begin{lemma}
\label{lem:totpartpred}
Let $\cat{B}$ be an effectus in total form, with its inclusion functor
$\klin{-} \colon \cat{B} \rightarrow \Par(\cat{B})$. Then we have a
diagram in which everything from left to right commutes:
$$\xymatrix@R-.5pc{
\Pred(\cat{B})\ar[d]_{\dashv\;}^{\;\dashv}\ar[rr]
& &
\PPred(\Par(\cat{B}))\ar[d]_{\dashv\;}^{\;\dashv}
\\
\cat{B}\ar@/^4ex/[u]^(0.4){\zero\!}\ar@/_4ex/[u]_(0.4){\!\one}\ar[rr]^-{\klin{-}}
& &
\Par(\cat{B})\ar@/^4ex/[u]^(0.4){\zero\!}\ar@/_4ex/[u]_(0.4){\!\one}
}$$
\end{lemma}

\begin{proof}
The functor $\klin{-} \colon \cat{B} \rightarrow \Par(\cat{B})$ lifts
to $\Pred(\cat{B}) \rightarrow \PPred(\Par(\cat{B}))$ since a map
$f\colon (X,p) \rightarrow (Y,q)$ in $\Pred(\cat{B})$ yields a map
$\klin{f} \colon (X,p) \rightarrow (Y,q)$. The reason is
Exercise~\ref{exc:subst}~\eqref{exc:subst:totalpartial}, giving:
$p \leq \tbox{f}(q) = \pbox{\klin{f}}(q)$. \QED
\end{proof}

We now come to the definition of comprehension, as right adjoint to
truth. We first give seperate formulations for the total and partial
case, but we prove that they are equivalent later on --- see
Theorem~\ref{thm:totpartcmpr}. Comprehension will be a functor from
predicates to their underlying objects. We shall write it as $(X,p)
\mapsto \cmpr{X}{p}$. When we consider it as a functor we write
$\cmpr{-}{-}$.

\begin{definition}
\label{def:cmpr}
We give separate formulations of essentially the same notion.
\begin{enumerate}
\item \label{def:cmpr:tot} An effectus in total form $\cat{B}$ has
  \emph{comprehension}\index{S}{comprehension}\index{S}{effectus!--
    with comprehension} when its truth functor $\one \colon \cat{B}
  \rightarrow \Pred(\cat{B})$ has a right adjoint $\cmpr{-}{-}$ such
  that for each predicate $p\colon X \rightarrow 1+1$ and object
  $Y\in\cat{B}$ the canonical map
\begin{equation}
\label{diag:cmprdefdistr}
\xymatrix{
\cmpr{X}{p}+Y\ar[r] & \cmpr{X+Y\,}{\,[p,\one]}
}\index{N}{$\cmpr{X}{p}$, comprehension of predicate $p$ on $X$}
\end{equation}

\noindent is an isomorphism.

\item \label{def:cmpr:part} An effectus in partial form $(\cat{C},I)$
  has \emph{comprehension} if its truth functor $\one \colon \cat{C}
  \rightarrow \PPred(\cat{C})$ has a right adjoint $\cmpr{-}{-}$ and
  each counit component $\cmpr{X}{p} \rightarrow X$ in $\cat{C}$ is a
  total map.
\end{enumerate}
\end{definition}

Formally this counit is a map in $\PPred(\cat{C})$ of the form
$\pi_{p} \colon (\one, \cmpr{X}{p}) \rightarrow
(X,p)$,\index{D}{$\pi_p$, comprehension projection for predicate $p$}
but in the definition we have written it as a map in $\cat{B}$, for
simplicity.  These comprehension maps $\pi_p$ should not be confused
with the tensor projections $\pi_i$ in~\eqref{diag:tensorproj}.

\begin{example}
\label{ex:cmpr}
We describe comprehension for our four running examples. We shall
alternate between the total and partial descriptions. This
comprehension structure occurs already in~\cite{ChoJWW15}, except for
order unit groups.
\begin{enumerate}
\item \label{ex:cmpr:Sets} Recall that predicates on an object/set $X$
  in the effectus $\Sets$ correspond to subsets of $X$. The category
  $\Pred(\Sets)$ has such subsets $(P\subseteq X)$ as objects;
  morphisms $f\colon (P\subseteq X) \rightarrow (Q\subseteq Y)$ are
  functions $f\colon X \rightarrow Y$ satisfying $P\subseteq
  f^{-1}(Q)$, that is: $x\in P \Rightarrow f(x)\in Q$.

The comprehension functor $\Pred(\Sets) \rightarrow \Sets$ is given by
$(Q\subseteq Y) \mapsto Q$. The adjoint correspondences for $\one \dashv
\cmpr{-}{-}$ amount to:
$$\begin{prooftree}
\xymatrix{\llap{$\one(X) =\;$}(X\subseteq X) \ar[r]^-{f} & (Q\subseteq Y)}
\Justifies
\xymatrix@C+1pc{X\ar[r]_-{g} & Q\rlap{$\;=\cmpr{Y}{Q}$}}
\end{prooftree}$$


\noindent Obviously, if $X\subseteq f^{-1}(Q)$, then $f(x) \in Q$ for
each $x\in X$, so that $f$ restricts to a unique map $\overline{f}
\colon X \rightarrow Q = \cmpr{Y}{Q}$ with $\pi_{Q} \after
\overline{f} = f$, where $\pi_{Q} \colon \cmpr{Y}{Q} \hookrightarrow
Y$ is the inclusion.

For a predicate $P\subseteq X$ the predicate $[P,\one] \subseteq X+Y$
used in~\eqref{diag:cmprdefdistr} is given by $[P,\one] =
\set{\kappa_{1}x}{x\in P} \cup \set{\kappa_{2}y}{y\in Y}$. Hence it is
immediate that we have: $\cmpr{X}{P}+Y = P+Y = [P,\one] =
\cmpr{X+Y\,}{\,[P,\one]}$.

\item \label{ex:cmpr:KlD} For discrete probability we use the effectus
  in partial form $\Kl(\sDst) \cong \Par(\Kl(\Dst))$. The category
  $\PPred(\Kl(\sDst))$ has fuzzy predicates $p\in [0,1]^{X}$ as
  objects. A morphism $f\colon (p\in [0,1]^{X}) \rightarrow (q\in
  [0,1]^{Y})$ is a function $f \colon X \rightarrow \sDst(Y)$
  satisfying, for each $x\in X$,
$$\begin{array}{rcccl}
p(x)
& \leq &
\pbox{f}(q)(x)
& \smash{\stackrel{\eqref{eqn:partsubstKlD}}{=}} & 
\sum_{y} f(x)(y)\cdot q(y) + 1 - \sum_{y}f(x)(y).
\end{array}$$

\noindent Comprehension sends a predicate $q\in [0,1]^{Y}$ to the set
$\cmpr{Y}{q} = \setin{y}{Y}{q(y) = 1}$ of elements where the predicate
$q$ holds with certainty. Let's check the correspondences:
$$\begin{prooftree}
\xymatrix{(\one\in [0,1]^{X}) \ar[r]^-{f} & 
   (q\in [0,1]^{Y})\rlap{\hspace*{2.8em}in $\PPred(\Kl(\sDst))$}}
\Justifies
\xymatrix@C+1pc{X\ar[r]_-{g} & 
  \setin{y}{Y}{q(y)=1}\rlap{\hspace*{3em}in $\Kl(\sDst)$}}
\end{prooftree}\hspace*{8em}$$

\noindent Let $f\colon X \rightarrow \sDst(Y)$ satisfy $1 \leq
\sum_{y} f(x)(y)\cdot q(y) + 1 - \sum_{y}f(x)(y)$ for each $x\in
X$. Hence $\sum_{y}f(x)(y) \leq \sum_{y}f(x)(y) \cdot q(y)$. This can
only happen if $f(x)(y) > 0$ implies $q(y) = 1$, for each $x$. But
this means that we can write each $f(x) \in \sDst(Y)$ as subconvex sum
$f(x) = \sum_{y} f(x)(y)\ket{y}$ with $q(y)=1$, and thus
$y\in\cmpr{Y}{q}$. Hence we have $f(x) \in \sDst(\cmpr{Y}{q})$, so
that $f$ restricts to a unique function $\overline{f} \colon X
\rightarrow \sDst(\cmpr{Y}{q})$ with $\pi_{Q} \after \overline{f} =
f$, where the comprehension map $\pi_{Q} \colon \cmpr{Y}{q}
\rightarrow \sDst(Y)$ is given by $\pi_{Q}(y) = 1\ket{y}$. Clearly,
it is total.

\item \label{ex:cmpr:OUG} The effectus $\op{\OUS}$ of order unit
  groups also has comprehension. The description is rather confusing
  because we work in the opposite category. It turns out that in this
  model comprehension is given by a quotient! We follow the relevant
  description meticulously, in the total case.

The category $\Pred(\op{\OUG})$ of predicates has effects $e\in
[0,1]_{G}$ in order unit groups $G$ as objects. A morphism $f\colon
(e\in [0,1]_{G}) \rightarrow (d\in [0,1]_{H})$ is a map $f\colon G
\rightarrow H$ in $\op{\OUG}$ with $e \leq \tbox{f}(d)$. This means
that $f$ is a homomorphism of order unit groups $f\colon H \rightarrow
G$ with $e \leq f(d)$ in $G$.

For an effect $d\in H$ we write $\pideal{d}{H} \subseteq
H$\index{D}{$\pideal{d}{H}$, ideal generated by $d$ in order unit
  group $H$} for the subset given by:
\begin{equation}
\label{eqn:pideal}
\begin{array}{rcl}
\pideal{d}{H}
& = &
\setin{x}{H}{\exin{n}{\NNO}{-n\cdot d \leq x \leq n\cdot d}}.
\end{array}
\end{equation}

\noindent It is not hard to see that $\pideal{d}{H}$ is a subgroup of
$H$. It is an ideal, since it satisfies: $-x \leq y \leq x$ and $x\in
\tuple{d}_{G}$ implies $y\in \pideal{d}{H}$. In fact, $\pideal{d}{H}$
is the smallest ideal containing $d$. It inherits the order from $H$
and is an order unit group itself, with $d\in\pideal{d}{H}$ as unit.

We now define comprehension as quotient group:
$$\begin{array}{rcl}
\cmpr{H}{d}
& = &
H/\pideal{d^\bot}{H}
\end{array}
\qquad\mbox{with}\qquad
\left\{\vcenter{\xymatrix@R-2.2pc{
H\ar[r]^-{\pi_{d}} & \cmpr{H}{d} \\
x\ar@{|->}[r] & x + \pideal{d^\bot}{H}
}}\right.$$

First we have to check that $\cmpr{H}{d}$ is an order unit group.  It
is easy to see that the quotient of an ordered Abelian group with an
ideal is again an ordered Abelian group, so we concentrate on checking
that $\cmpr{H}{d}$ has an order unit. We claim that $\pi_{d}(1) = 1 +
\pideal{d^\bot}{H}$ is the order unit in $H/\pideal{d^\bot}{H} =
\cmpr{H}{d}$. Indeed, for an arbitrary $x\in H$ there is an $n\in\NNO$
with $-n\cdot 1 \leq x \leq n\cdot 1$. But then:
$$\begin{array}{rcccccccl}
-n\cdot \pi_{d}(1)
& = &
-n\cdot 1 + \pideal{d^\bot}{H}
& \leq &
x + \pideal{d^\bot}{H}
& \leq &
n\cdot 1 + \pideal{d^\bot}{H}
& = &
n\cdot \pi_{d}(1).
\end{array}$$

The claimed comprehension adjunction involves a bijective
correspondence between:
$$\begin{prooftree}
\xymatrix{\llap{$\one(G) = \;$}(1\in [0,1]_{G})\ar[r]^-{f} & (H,d)
   \rlap{\hspace*{2.4em}in $\Pred\big(\op{\OUG}\big)$}}
\Justifies
\xymatrix@C+1pc{G\ar[r]_-{g} & \cmpr{H}{d}
   \rlap{\hspace*{4em}in $\op{\OUG}$}}
\end{prooftree}\hspace*{6em}$$


\noindent That is, between maps in $\OUG$:
$$\begin{prooftree}
\xymatrix{H\ar[r]^-{f} & G \mbox{ with $f(d) = 1$}}
\Justifies
\xymatrix{\cmpr{H}{d}\ar[r]_-{g} & G}
\end{prooftree}$$

\noindent This works as follows.
\begin{itemize}
\item Let $f\colon H \rightarrow G$ in $\OUG$ satisfy $f(d) = 1$. Then
  for each $x\in\pideal{d^{\bot}}{H}$ we have $f(x) = 0$. Indeed, if
  $-n\cdot d^{\bot} \leq x \leq n\cdot d^{\bot}$, then, because
  $f(d^{\bot}) = f(1 - d) = f(1) - f(d) = 1 - 1 = 0$, we get:
$$\begin{array}{rcccccccccccl}
0
& = &
-n\cdot f(d^{\bot})
& = &
f(-n\cdot d^{\bot})
& \leq &
f(x)
& \leq &
f(n\cdot d^{\bot})
& = &
n\cdot f(d^{\bot})
& = &
0.
\end{array}$$

\noindent Thus there is a unique group homomorphism $\overline{f}
\colon \cmpr{H}{d} = H/\pideal{d^\bot}{H} \rightarrow G$ by
$\overline{f}(x + \pideal{d^\bot}{H}) = f(x)$. This map $\overline{f}$
is clearly monotone and unital. 

\auxproof{
let $x + \pideal{d^\bot}{H} \leq y + \pideal{d^\bot}{H}$, say via $x +
u \leq y$, for $u\in\pideal{d^\bot}{H}$. Then $f(u) = 0$, as just shown,
and thus:
$$\begin{array}{rcccccl}
f(x)
& = &
f(x) + f(u)
& = &
f(x+u)
& \leq &
f(y).
\end{array}$$
}

\item The other direction is easy: given $g\colon \cmpr{H}{d}
  \rightarrow G$, take $\overline{g} = g \after \pi_{d} \colon H
  \rightarrow G$. Then:
$$\begin{array}{rcl}
\overline{g}(d)
\hspace*{\arraycolsep}=\hspace*{\arraycolsep}
g\big(\pi_{d}(d)\big) 
& = &
g\big(d + \pideal{d^\bot}{H}\big) \\
& = &
g\big(d + d^{\bot} + \pideal{d^\bot}{H}\big) \\
& = &
g\big(1 + \pideal{d^{\bot}}{H}\big)
\hspace*{\arraycolsep}=\hspace*{\arraycolsep}
g\big(\one_{\cmpr{H}{d}}\big)
\hspace*{\arraycolsep}=\hspace*{\arraycolsep}
1.
\end{array}$$
\end{itemize}

\noindent Clearly, \smash{$\overline{\overline{f}} = \overline{f}
  \after \pi_{d} = f$}. And \smash{$\overline{\overline{g}} = g$}
holds by uniqueness, since $\overline{g}(x) = g(\pi_{d}(x)) = g(x +
\pideal{d^\bot}{H})$.

Finally, for $e\in [0,1]_{G}$ the canonical map
$\cmpr{G\oplus H}{(e,1)} \rightarrow \cmpr{G}{e} \oplus H$ is the
function:
$$\xymatrix@R-2pc{
(G\oplus H)/\pideal{(e,1)^{\bot}}{} \ar[rr] & &
   G/\pideal{e^\bot}{}\oplus H \\
(x,y) + \pideal{(e,1)^{\bot}}{} \ar@{|->}[rr] & &
   (x+\pideal{e^\bot}{}, y).
}$$

\noindent The equivalence relation on $G\oplus H$ induced by the ideal
$\pideal{(e,1)^{\bot}}{} = \pideal{(e^{\bot}, 0)}{}$ is given by
$(x,y) \sim (x',y')$ iff $(x-x', y-y') \in \pideal{(e^{\bot}, 0)}{}$.
The latter means that there is an $n\in\NNO$ with $-n\cdot (e^{\bot},
0) \leq (x-x', y-y') \leq n\cdot (e^{\bot}, 0)$.  This is equivalent
to: $-n\cdot e^{\bot} \leq x-x' \leq n\cdot e^{\bot}$ and $y=y'$.  The
above function is thus well-defined, and clearly an isomorphism, with
inverse $(x+\pideal{e^\bot}{}, y) \mapsto (x,y) +
\pideal{(e,1)^\bot}{}$.

\item \label{ex:cmpr:vNA} The effectus in partial form
  $\Par(\op{\vNA})$ of von Neumann algebras and subunital maps also
  has comprehension. The relevant category of predicates
  $\PPred(\Par(\op{\vNA}))$ has effects $e\in[0,1]_{\mathscr{A}}$ in a
  von Neumann algebra $\mathcal{A}$ as objects. Morphisms $f \colon
  (e\in [0,1]_{\mathscr{A}}) \rightarrow (d\in[0,1]_{\mathscr{B}})$
  are subunital completely positive normal maps $f\colon \mathscr{B}
  \rightarrow \mathscr{A}$ with $e \leq f(d)$ in $\mathscr{A}$.

Recall the sharp predicates\index{S}{sharp!-- predicate!-- on a von
  Neumann algebra} on a von Neumann algebra are precisely the
projections.  For an effect $d\in[0,1]_{\mathscr{B}}$ we write
$\floor{d}\in[0,1]_{\mathscr{B}}$\index{D}{$\floor{d}$, greatest sharp
  element below $d$ in a von Neumann algebra} for the greatest sharp
element $s\in\mathscr{B}$ with $s\leq d$. The definition $\ceil{d} =
\floor{d^\bot}^{\bot}$\index{D}{$\ceil{d}$, least sharp element above
  $d$ in a von Neumann algebra} then yields the least sharp element
above $d$.

The projection~$\ceil{d}$ is known as the \emph{support projection}
of~$d$.  In other texts it is often denoted as~$r(d)$, as it is as an
the projection onto the closed range of~$d$.  Before we construct the
comprehension in~$\op\vNA$, we list a few properties of~$\ceil{d}$,
which will be useful later on.
\begin{enumerate}
\item The ascending sequence~$d \leq d^{\nicefrac{1}{2}} \leq
  d^{\nicefrac{1}{4}} \leq \cdots$ has supremum~$\ceil{d}$.  This can
  be shown using the spectral theorem.  It follows~$\ceil{d} =
  \ceil{\sqrt{d}} = \ceil{d^2}$.

\item We have~$d = \ceil{d}\cdot d = d \cdot \ceil{d}$.  In fact,
  $\ceil{d}$ is the least projection with this property.  This is a
  consequence of Lemma~\ref{lem:vNdownsetproj}.
  In combination with the previous
  point we obtain, for instance: $\sqrt{d}\cdot \ceil{d} = \sqrt{d}
  \cdot \ceil{\sqrt{d}} = \sqrt{d} = \ceil{d}\cdot \sqrt{d}$.
\end{enumerate}

Now we can define comprehension as:
$$\begin{array}{rcccl}
\cmpr{\mathscr{B}}{d}
& = &
\floor{d}\mathscr{B}\floor{d}
& = &
\set{\floor{d}\cdot x \cdot \floor{d}}{x\in\mathscr{B}}.
\end{array}$$

\noindent This subset $\cmpr{\mathscr{B}}{d}$ is itself a von Neumann
algebra, with $\floor{d}$ as unit. The associated comprehension map is
the map $\pi_{d} \colon \mathscr{B} \rightarrow \cmpr{\mathscr{B}}{d}$
in $\vNA$ given by $\pi_{d}(x) = \floor{d}\cdot x\cdot \floor{d}$.

We must show that given a von Neumann algebra~$\mathscr{A}$ and a
subunital map~$f\colon \mathscr{B}\to \mathscr{A}$ with~$f(d)=f(1)$
there is a unique subunital map $g\colon
\floor{d}\mathscr{B}\floor{d}\to\mathscr{B}$ with~$g(\floor{d}\cdot
x\cdot \floor{d})=f(x)$.  Put $g(x)=f(x)$; the difficulty it to show
that~$f(\floor{d}\cdot x\cdot \floor{d})=f(d)$.  For the details,
see~\cite{WesterbaanW15}.  We sketch the proof here.  By a variant of
Cauchy-Schwarz inequality for the completely positive map~$f$
(see~\cite[Exc.~3.4]{paulsen2002}) we can reduce this problem to
proving that~$f(\floor{d})=f(1)$, that is, $f(\ceil{d^\bot})=0$.
Since~$\ceil{d^\bot}$ is the supremum of~$d^\bot\leq
(d^\bot)^{\nicefrac{1}{2}} \leq
(d^\bot)^{\nicefrac{1}{4}}\leq\dotsb$ and~$f$ is normal,
$f(\ceil{d^\bot})$ is the supremum of~$f(d^\bot)\leq
f((d^\bot)^{\nicefrac{1}{2}}) \leq
f((d^\bot)^{\nicefrac{1}{4}})\leq\dotsb$, all of which turn out to be
zero by Cauchy-Schwarz since~$f(d)=f(1)$.  Thus~$f(\ceil{d^\bot})=0$,
and we are done.

We note that with the same argument, one shows that for any
$2$-positive normal subunital map~$h\colon \mathscr{A} \rightarrow
\mathscr{B}$, we have, for all effects $a\in [0,1]_{\mathscr{A}}$,
\begin{equation} 
\label{eqn:vNAceilMap}
\begin{array}{rcl}
h(a) = 0 
& \Longleftrightarrow &
h(\ceil{a}) = 0.
\end{array}
\end{equation}
\end{enumerate}
\end{example}

We come to the promised result that relates comprehension in the
total and in the partial case.

\begin{theorem}
\label{thm:totpartcmpr}
An effectus in total form $\cat{B}$ has comprehension iff its effectus
in partial form $\Par(\cat{B})$ has comprehension.

In both cases the comprehension maps $\pi_{p} \colon \cmpr{X}{p}
\rightarrow X$ are monic --- in $\cat{B}$ and in $\Par(\cat{B})$
respectively.
\end{theorem}

\begin{proof}
Since the left adjoint truth functors $\one \colon \cat{B} \rightarrow
\Pred(\cat{B})$ and $\one \colon \Par(\cat{B}) \rightarrow
\PPred(\Par(\cat{B}))$ are faithful in both cases, the counits of the
adjunctions are monic by a general categorical result --- see
\textit{e.g.}~(the dual of)~\cite[IV.3~Thm.1]{Maclane71}.  These
counits are monic maps $(\cmpr{X}{p},\one) \rightarrow (X,p)$ in the
categories $\Pred(\cat{B})$ and $\PPred(\Par(\cat{B}))$
respectively. It is easy to see, using the truth functor $\one$, that
the underlying maps are then monic in $\cat{B}$ and $\Par(\cat{B})$.

\auxproof{
If $f,g\colon Y \rightarrow \cmpr{X}{p}$ satisfy $\pi_{p} \after f
= \pi_{p} \after g$, then $1f = f, 1g = g\colon (Y,1) \rightarrow
(\cmpr{X}{p}, 1)$ satisfy $\pi_{p} \after f = \pi_{p} \after g$
in the category of predicates. Hence $f = g$. This argument applies
both in the total and in the partial case.
}

First, let $\one\colon \cat{B} \rightarrow \Pred(\cat{B})$ have a
right adjoint $\cmpr{-}{-}$. The counit is a map $\pi_{p} \colon
(\cmpr{X}{p},\one) \rightarrow (X,p)$ in $\Pred(\cat{B})$ satisfies
$\one \leq \tbox{\pi_{p}}(p) = p \after \pi_{p}$ by definition. The
underlying map $\pi_{p} \colon \cmpr{X}{p} \rightarrow X$ is monic in
$\cat{B}$, as just noted.

We first investigate the canonical map~\eqref{diag:cmprdefdistr}, call
it $\sigma$, is the unique one in:
\begin{equation}
\label{diag:cmprdistriso}
\vcenter{\xymatrix@R-.5pc@C-1pc{
\cmpr{X}{p} + Y\ar@{..>}[rr]^-{\sigma_p}\ar[dr]_{\pi_{p}+\idmap} & & 
   \cmpr{X+Y\,}{\,[p, \one]}\ar@{>->}[dl]^{\pi_{[p,\one]}} \\
& X+A &
}}
\end{equation}

\noindent This map $\sigma$ exists by comprehension because:
$$\begin{array}{rcl}
\tbox{(\pi_{p}+\idmap)}([p,\one])
& = &
[p, \kappa_{1}\after\bang] \after (\pi_{p}+\idmap) \\
& = &
[p \after \pi_{p}, \kappa_{1}\after\bang] 
\hspace*{\arraycolsep}=\hspace*{\arraycolsep}
[\kappa_{1} \after \bang, \kappa_{1}\after\bang]
\hspace*{\arraycolsep}=\hspace*{\arraycolsep}
\kappa_{1} \after [\bang,\bang]
\hspace*{\arraycolsep}=\hspace*{\arraycolsep}
\kappa_{1} \after \bang
\hspace*{\arraycolsep}=\hspace*{\arraycolsep}
\one.
\end{array}$$

\noindent By assumption, this $\sigma$ is an isomorphism, for each
predicate $p$ and object $Y$.

We now show that the truth functor $\one\colon \cat{B} \rightarrow
\PPred(\Par(\cat{B}))$ has a right adjoint, on an object $(X,p)$ given
by $\cmpr{X}{p}$. There is total map $\klin{\pi_{p}} = \kappa_{1}
\after \pi_{1} \colon \cmpr{X}{p} \pto X$ in $\Par(\cat{B})$, which
satisfies by Exercise~\ref{exc:subst}~\eqref{exc:subst:totalpartial}:
$$\begin{array}{rcccccl}
\pbox{\klin{\pi_{p}}}(p)
& = &
\tbox{\big(\pi_{p}\big)}(p)
& = &
p \after \pi_{p}
& = &
\one.
\end{array}$$

Let $f\colon Y \pto X$ be an arbitrary partial map satisfying $1 \leq
\pbox{f}(p) = [p,\kappa_{1}] \after f = [p,\one] \after f =
\tbox{f}([p,\one])$.  Then there is a unique total map $\overline{f} \colon Y
\tto \cmpr{X+1\,}{\,[p,\one]}$ with $\pi_{[p,\one]} \after
\overline{f} = f$. We now take $\smash{\widehat{f}} = \sigma^{-1} \after
\overline{f} \colon Y \tto \cmpr{X}{p} + 1$, using the
isomorphism from~\eqref{diag:cmprdistriso} with $Y=1$. Then:
$$\begin{array}{rcccccl}
\klin{\pi_{p}} \pafter \smash{\widehat{f}}
& = &
(\pi_{p}+\idmap) \after \sigma^{-1} \after \overline{f}
& \smash{\stackrel{\eqref{diag:cmprdistriso}}{=}} &
\pi_{[p,\one]} \after \overline{f}
& = &
f.
\end{array}$$

\noindent If $g\colon Y \pto \cmpr{X}{p}$ also satisfies
$\klin{\pi_{1}} \pafter g = f$, then $\sigma \after g\colon Y \to
\cmpr{X+1\,}{\,[p,\one]}$ satisfies $\pi_{[p,\one]} \after \sigma_{p}
\after g = (\pi_{p}+\idmap) \after g = \klin{\pi_{p}} \pafter g =
f$. Hence $\sigma \after g = \overline{f}$, and thus $g = \sigma^{-1}
\after \overline{f} = \smash{\widehat{f}}$.

In the other direction, we show how comprehension in the partial case
yields comprehension in the total case with the distributivity
isomorphism~\eqref{diag:cmprdefdistr}. Therefore, assume that a right
adjoint $\cmpr{-}{-}$ to $\one \colon \cat{B} \rightarrow
\PPred(\Par(\cat{B}))$ exists so that each counit map $\cmpr{X}{p}\pto
X$ is total. We shall write this counit as $\klin{\pi_{p}}$, for
$\pi_{p} \colon \cmpr{X}{p} \rightarrow X$ in $\cat{B}$. As noted in
the very beginning of this proof, the map $\klin{\pi_{p}}$ is monic in
$\Par(\cat{B})$. Further, we have:
$$\begin{array}{rcccccl}
p \after \pi_{p}
& = &
\tbox{\big(\pi_{p}\big)}(p)
& = &
\pbox{\klin{\pi_{p}}}(p)
& = &
\one.
\end{array}$$

\noindent The last equation holds because the counit is a map
$\klin{\pi_{p}} \colon (\cmpr{X}{p},\one) \rightarrow (X,p)$ in
$\PPred(\Par(\cat{B}))$.

We now show that there is also a right adjoint to the truth
functor $\one\colon\cat{B} \rightarrow \Pred(\cat{B})$, on objects given
by $(X,p) \mapsto \cmpr{X}{p}$.

Let $f\colon Y \to X$ be a total map satisfying $\tbox{f}(p) =
\one$. Then $\klin{f} = \kappa_{1} \after f \colon Y \pto X$ in
$\Par(\cat{B})$ satisfies $\pbox{\klin{f}}(p) = \tbox{f}(p) =
\one$. Hence there is a unique partial map $\overline{f} \colon Y \pto
\cmpr{X}{p}$ with $\klin{\pi_{p}} \pafter \overline{f} = \klin{f}$.
The latter means that the outer diagram below commutes in the effectus
$\cat{B}$.
$$\xymatrix{
Y\ar@/^2ex/[drr]^-{f}\ar@/_2ex/[ddr]_{\overline{f}}
   \ar@{..>}[dr]^(0.6){\smash{\widehat{f}}}
\\
& \cmpr{X}{p}\ar[r]_-{\pi_p}
   \ar@{ >->}[d]_{\kappa_1}\pullback & 
   X\ar@{ >->}[d]^{\kappa_1} 
\\
& \cmpr{X}{p} + 1\ar[r]_-{\pi_{p}+\idmap} & X+1
}$$

\noindent The rectangle in the middle is a pullback
by~\eqref{diag:effecutsderivedpb}. This mediating map
$\smash{\widehat{f}}$ is what we need, since it satisfies $\pi_{p}
\after \smash{\widehat{f}} = f$ by construction. Moreover, if $g\colon
Y \to \cmpr{X}{p}$ in $\cat{B}$ also satisfies $\pi_{p} \after g = f$,
then $\klin{g} = \kappa_{1} \after g\colon Y \pto \cmpr{X}{p}$ in
$\Par(\cat{B})$ satisfies $\klin{\pi_{p}} \pafter \klin{g} =
\klin{\pi_{p} \after g} = \klin{f}$. Hence $\klin{g} = \overline{f}$
by uniqueness. But then $\kappa_{1} \after g = \klin{g} = \overline{f}
= \kappa_{1} \after \smash{\widehat{f}}$, and so $g =
\smash{\widehat{f}}$ since $\kappa_{1}$ is monic. Thus we have shown
that the truth functor $\one\colon \cat{B} \rightarrow \Pred(\cat{B})$
has a right adjoint.

We still have to prove that the canonical map $\sigma \colon
\cmpr{X}{p}+Y \rightarrow \cmpr{X+Y\,}{\,[p,\one]}$
from~\eqref{diag:cmprdistriso} is an isomorphism. The inverse
$\sigma^{-1}$ for $\sigma$ is obtained via the following
pullback~\eqref{def:effectus:pb} from the definition of effectus:
$$\xymatrix{
\cmpr{X+Y\,}{\,[p,\one]}
   \ar@/^3ex/[drr]^-{(\bang+\idmap) \after \pi_{[p,\one]}}
   \ar@/_3ex/[ddr]_-{f}\ar@{..>}[dr]^-{\sigma^{-1}}
\\
& \cmpr{X}{p}+Y\ar[r]^(0.5){\bang+\idmap}
   \ar[d]_{\idmap+\bang}\pullback & 
   1+Y\ar[d]^{\idmap+\bang}
\\
& \cmpr{X}{p} + 1\ar[r]_-{\bang+\idmap} & 1+1
}$$

\noindent The auxiliary map $f$ in this diagram is the unique one
with $\klin{\pi_{p}} \pafter f = (\idmap+\bang) \after \pi_{[p,\one]} \colon
\cmpr{X+Y\,}{\,[p,\one]} \rightarrow X+1$. It exists since:
$$\begin{array}{rcl}
\pbox{\big((\idmap+\bang) \after \pi_{[p,\one]}\big)}(p)
& = &
[p,\one] \after (\idmap+\bang) \after \pi_{[p,\one]} \\
& = &
[p,\one] \after \pi_{[p,\one]} \\
& = &
\tbox{\pi_{[p,\one]}}([p,\one]) \\
& = &
\one.
\end{array}$$

\noindent We get $(\pi_{p}+\idmap) \after \sigma^{-1} =
\pi_{[p,\one]} \colon \cmpr{X+Y\,}{\,[p,\one]} \rightarrow X+Y$ via a
similar pullback, with $X+Y$ in the upper left corner, since:
$$\begin{array}{rcl}
(\bang+\idmap) \after (\pi_{p}+\idmap) \after \sigma^{-1}
& = &
(\bang+\idmap) \after \sigma^{-1} \\
& = &
(\bang+\idmap) \after \pi_{[p,\one]} 
\\
(\idmap+\bang) \after (\pi_{p}+\idmap) \after \sigma^{-1}
& = &
(\pi_{p}+\idmap) \after (\idmap+\bang) \after \sigma^{-1} \\
& = &
(\pi_{p}+\idmap) \after f \\
& = &
(\idmap+ \bang) \after \pi_{[p,\one]}.
\end{array}$$

\noindent From this we can conclude $\sigma \after \sigma^{-1} =
\idmap$, using that $\pi_{[p,\one]}$ is monic in $\cat{B}$, and:
$$\begin{array}{rcccl}
\pi_{[p,\one]} \after \sigma \after \sigma^{-1}
& = &
(\pi_{p}+\idmap) \after \sigma^{-1}
& = &
\pi_{[p,\one]}.
\end{array}$$

\noindent We obtain $\sigma^{-1} \after \sigma = \idmap$ via uniqueness of
mediating maps in the above pullback defining $\sigma^{-1}$:
$$\begin{array}{rcl}
(\bang+\idmap) \after \sigma^{-1} \after \sigma
& = &
(\bang+\idmap) \after \pi_{[p,\one]} \after \sigma \\
& = &
(\bang+\idmap) \after (\pi_{p}+\idmap) \\
& = &
(\bang+\idmap) 
\\
(\idmap+\bang) \after \sigma^{-1} \after \sigma
& = &
f \after \sigma \\
& \smash{\stackrel{(*)}{=}} &
(\idmap+\bang).
\end{array}$$

\noindent The latter, marked equation uses that the map
$\klin{\pi_{p}}$ is monic in $\Par(\cat{B})$:
$$\begin{array}[b]{rcl}
\klin{\pi_{p}} \pafter (f \after \sigma)
& = &
(\pi_{p}+\idmap) \after f\after \sigma \\
& = &
(\idmap+\bang) \after \pi_{[p,\one]} \after \sigma \\
& = &
(\idmap+\bang) \after (\pi_{p}+\idmap) \\
& = &
\klin{\pi_{p}} \pafter (\idmap+\bang).
\end{array}\eqno{\QEDbox}$$
\end{proof}

In the remainder of this section we list several properties of
comprehension. Some of these properties are more naturally formulated
in the total case, and some in the partial case.  Therefore we split
these results up in two parts.

\begin{lemma}
\label{lem:totcmpr}
Let $\cat{B}$ be an effectus with comprehension, in total form.
\begin{enumerate}
\item \label{lem:totcmpr:equalpb} A (comprehension) projection map
  $\pi_{p} \colon \cmpr{X}{p} \rightarrowtail X$ can be characterised
  either as the equaliser in $\cat{B}$:
$$\vcenter{\xymatrix{
\cmpr{X}{p}\ar@{ >->}[r]^-{\pi_p} & X\ar@/^2ex/[r]^-{p}
   \ar@/_2ex/[r]_-{\one = \kappa_{1} \after \bang} & 1\rlap{$+1$}
}}
\qquad\mbox{or as pullback:}\qquad
\vcenter{\xymatrix{
\cmpr{X}{p}\ar@{ >->}[d]_{\pi_p}\ar[r]^-{\bang}\pullback & 
   1\ar@{ >->}[d]^{\kappa_1}
\\
X\ar[r]_-{p} & 1+1
}}$$

\item \label{lem:totcmpr:bot} One has $p^{\bot} \after \pi_{p} = 0$, and
  moreover the diagrams below are, respectively, an equaliser and a
  pullback in $\cat{B}$.
$$\vcenter{\xymatrix{
\cmpr{X}{p^\bot}\ar@{ >->}[r]^-{\pi_{p^\bot}} & 
   X\ar@/^2ex/[r]^-{p}
   \ar@/_2ex/[r]_-{\zero = \kappa_{2} \after\bang} & 1\rlap{$+1$}
}}
\hspace*{6em}
\vcenter{\xymatrix{
\cmpr{X}{p^\bot}\ar@{ >->}[d]_{\pi_{p^\bot}}
   \ar[r]^-{\bang}\pullback & 
   1\ar@{ >->}[d]^{\kappa_2}
\\
X\ar[r]_-{p} & 1+1
}}$$

\item \label{lem:totcmpr:one} A projection $\pi_{p} \colon \cmpr{X}{p}
  \rightarrowtail X$ is an isomorphism if and only if $p=\one$.

\item \label{lem:totcmpr:zero} The comprehension $\cmpr{X}{\zero}$ is
  initial in $\cat{B}$.

\item \label{lem:totcmpr:pb} Projection maps are closed under pullback
  in $\cat{B}$: for each predicate $p$ on $X$ and map $f\colon Y
  \to X$ we have a pullback in $\cat{B}$:
\begin{equation}
\label{diag:totcmpr:pb}
\vcenter{\xymatrix{
\cmpr{Y}{\tbox{f}(p)}\ar@{..>}[r]
   \ar@{ >->}[d]_{\pi_{\tbox{f}(p)}}\pullback & 
   \cmpr{X}{p}\ar@{ >->}[d]^{\pi_p}
\\
Y\ar[r]_-{f} & X
}}
\end{equation}

\item \label{lem:totcmpr:coproj} The pullback of a coprojection along
  an arbitrary map exists in $\cat{B}$, and is given as in:
$$\vcenter{\xymatrix@R-.5pc{
\cmpr{X}{p}\ar[r]\ar@{ >->}[d]_{\pi_p}\pullback & 
   Y\ar@{ >->}[d]^{\kappa_1}
\\
X\ar[r]_-{f} & Y+Z
}}\qquad\mbox{where}\qquad p = (\bang+\bang) \after f \colon X \rightarrow 1+1$$

\noindent In particular, the coprojection $Y \rightarrowtail Y+Z$ is
itself (isomorphic to) a comprehension map, namely to $Y \cong
\cmpr{Y+Z\,}{\,[\one,\zero]} \rightarrowtail Y+Z$.

\item \label{lem:totcmpr:pbzero} Projections of orthogonal predicates are
  disjoint: if $p \orthogonal q$, then the diagram below is a
  pullback.
$$\vcenter{\xymatrix{
0\pullback\ar[d]\ar[r] & \cmpr{X}{p}\ar@{ >->}[d]^{\pi_{p}} 
\\
\cmpr{X}{q}\ar@{ >->}[r]_-{\pi_{q}} & X
}}$$

\item \label{lem:totcmpr:sumproj} For predicates $p$ on $X$ and $q$ on
  $Y$ the sum map $\pi_{p}+\pi_{q} \colon \cmpr{X}{p} +
  \cmpr{Y}{q} \rightarrow X+Y$ is monic in $\cat{B}$.

\item \label{lem:totcmpr:sumiso} For predicates $p$ on $X$ and $q$ on
  $Y$ there is a (canonical) isomorphism as on the left below.  Using
  point~\eqref{lem:totcmpr:equalpb}, this implies that the square on
  the right is a pullback.
$$\vcenter{\xymatrix@C-3pc{
\cmpr{X}{p} + \cmpr{Y}{q}\ar@{=}[rr]^-{\textstyle\sim}
   \ar@{ >->}[dr]_{\pi_{p}+\pi_{q}} & & 
   \cmpr{X+Y\,}{\,[p,q]}\ar@{ >->}[dl]^{\pi_{[p,q]}}
\\
& X+Y &
}}
\hspace*{3em}
\vcenter{\xymatrix@C-.5pc{
\cmpr{X}{p} + \cmpr{Y}{q}\ar@{ >->}[d]_{\pi_{p}+\pi_{q}}
   \ar[r]\pullback & 
   1\ar@{ >->}[d]^{\kappa_{1}}
\\
X+Y\ar[r]_-{[p,q]} & 1+1
}}$$
\end{enumerate}
\end{lemma}

\begin{proof}
We use the formulation of comprehension for effectuses in
Definition~\ref{def:cmpr}~\eqref{def:cmpr:tot}.
\begin{enumerate}
\item This is just a reformulation of the universal property of
  comprehension.

\item We have $p^{\bot} \after \pi_{p} = [\kappa_{2}, \kappa_{1}]
  \after p \after \pi_{p} = [\kappa_{2}, \kappa_{1}] \after \kappa_{1}
  \after \bang = \kappa_{2} \after \bang = \zero$. If $f\colon Y \to
  X$ satisfies $p \after f = \zero$, then $p^{\bot} \after f =
  [\kappa_{2}, \kappa_{1}] \after p \after f = [\kappa_{2},
    \kappa_{1}] \after \kappa_{2} \after \bang = \kappa_{1} \after
  \bang = \one$. Hence $f$ factors through $\cmpr{X}{p^\bot}$ making
  the diagrams in point~\eqref{lem:totcmpr:bot} an equaliser and a
  pullback.

\item We first show that $\pi_{\one} \colon \cmpr{X}{\one}
  \rightarrowtail X$ is an isomorphism. The identity map $\idmap
  \colon X \to X$ trivially satisfies $\one = \tbox{\idmap}(\one)$.
  Hence there is a unique map $f \colon X \to \cmpr{X}{\one}$ with
  $\pi_{\one} \after f = \idmap$. The equation $f \after \pi_{\one} =
  \idmap$ follows because the projections are monic: $\pi_{\one}
  \after (f\after \pi_{\one}) = \pi_{\one} = \pi_{\one} \after
  \idmap$.

Conversely, if $\pi_{p} \colon \cmpr{X}{p} \rightarrowtail X$ is an
isomorphism, then, using the pullback in
point~\eqref{lem:totcmpr:equalpb} we obtain $p=1$ by writing:
$$\xymatrix{
p = \Big(X\ar[r]^-{\pi_{p}^{-1}}_-{\cong} &
   \cmpr{X}{p}\ar[r]^-{\bang} & 
   1\ar@{ >->}[r]^-{\kappa_1} & 
   1+1\Big) = \Big(X\ar[r]^-{\one} & 1+1\Big)
}$$

\item The projection $\pi_{\zero} \colon \cmpr{X}{\zero}
  \rightarrowtail X$ gives rise to an equality of predicates
  $\one=\zero\colon \cmpr{X}{\zero} \to 1+1$ via:
$$\begin{array}{rcccccl}
\one
& = &
\tbox{\pi_{\zero}}(\zero)
& = &
\zero \after \pi_{\zero}
& = &
\zero.
\end{array}$$

\noindent Hence we have a situation:
$$\xymatrix{
\cmpr{X}{\zero}\ar@/^2ex/[drr]^-{\bang}
   \ar@/_2ex/[ddr]_-{\bang}\ar@{..>}[dr]
\\
& 0\ar@{ >->}[d]\ar@{ >->}[r]\pullback &
  1\ar@{ >->}[d]^{\kappa_{1}} 
\\
& 1\ar@{ >->}[r]_-{\kappa_2} & 1+1
}$$

\noindent Proposition~\ref{prop:effectuscoproj} says that the
rectangle is a pullback, and that $0$ is strict. This means that the
map $\cmpr{X}{\zero} \to 0$ is an isomorphism.

\item The dashed arrow in diagram~\eqref{diag:totcmpr:pb} exists since:
$$\begin{array}{rcccccl}
\tbox{(f \after \pi_{\tbox{f}(p)})}(p)
& = &
p \after f \after \pi_{\tbox{f}(p)}
& = &
\tbox{f}(p) \after \pi_{\tbox{f}(p)}
& = &
\one.
\end{array}$$

\noindent The square~\eqref{diag:totcmpr:pb} is a pullback, since if
$g\colon Z \to Y$ and $h\colon Z \to \cmpr{X}{p}$ satisfy $f \after
g = \pi_{p} \after h$, then $g$ factors through
$\pi_{\tbox{f}(p)}$ since:
$$\begin{array}{rcccccccl}
\tbox{g}\big(\tbox{f}(p)\big)
& = &
p \after f \after g
& = &
p \after \pi_{p} \after h
& = &
\one \after h
& = &
\one.
\end{array}$$

\item First we have to check that the rectangle in
  point~\eqref{lem:totcmpr:coproj} commutes. It arrises in a situation:
$$
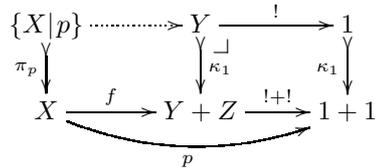
\xymatrix@R-.5pc{
\cmpr{X}{p}\ar@{..>}[r]\ar@{ >->}[d]_{\pi_p} & 
   Y\ar@{ >->}[d]^{\kappa_1}\ar[r]^-{\bang}\pullback & 
   1\ar@{ >->}[d]_{\kappa_{1}}
\\
X\ar[r]^-{f}\ar@/_3ex/[rr]_-{p} & 
   Y+Z\ar[r]^-{\bang+\bang} & 1+1
}$$

\noindent The outer rectangle is a pullback by
point~\eqref{lem:totcmpr:equalpb}. Hence the rectangle on the left is
a pullback, by the Pullback Lemma.

\item We use that diagram~\eqref{diag:totcmpr:pb} is a pullback, with
  $f = \pi_{q}$. It yields that
  $\cmpr{\cmpr{X}{q}}{\tbox{\pi_{q}}(p)}$ forms a pullback. But since
  $p \orthogonal q$, and so $p \leq q^{\bot}$, we have
  $\tbox{\pi_{q}}(p) \leq \tbox{\pi_{q}}(q^{\bot}) = \zero$. This
  gives $\cmpr{\cmpr{X}{q}}{\pi_{q}^{*}(p)} \cong 0$ by
  point~\eqref{lem:totcmpr:zero}. Hence the rectangle in
  point~\eqref{lem:totcmpr:pbzero} is a pullback.

\item Let $f,g\colon Y \to \cmpr{X}{p} + \cmpr{X}{q}$ satisfy
  $(\pi_{p}+\pi_{q}) \after f = (\pi_{p}+\pi_{q}) \after
  g$. We must show that $f = g$. The next diagram is a pullback in
  $\cat{B}$, see Definition~\ref{def:effectus}.
$$\xymatrix{
\cmpr{X}{p} + \cmpr{X}{q}\pullback\ar[r]^-{\bang+\idmap}
   \ar[d]_{\idmap+\bang} &
   1 + \cmpr{X}{q}\ar[d]^{\idmap+\bang}
\\
\cmpr{X}{p} + 1\ar[r]_-{\bang+\idmap} & 1+1
}$$

\noindent Hence we are done if we can prove $f_{1} = g_{1}$ and $f_{2}
= g_{2}$ where:
$$\left\{\begin{array}{rcl}
f_{1} & = & (\idmap+\bang) \after f \\
g_{1} & = & (\idmap+\bang) \after g \\
\end{array}\right.
\qquad\qquad
\left\{\begin{array}{rcl}
f_{2} & = & (\bang+\idmap) \after f \\
g_{2} & = & (\bang+\idmap) \after g \\
\end{array}\right.$$

\noindent We recall from Theorem~\ref{thm:totpartcmpr} that the
projection maps $\klin{\pi_{q}} = \kappa_{1} \after \pi_{q}$ are monic
in $\Par(\cat{B})$. We then get:
$$\begin{array}{rcl}
\klin{\pi_{p}} \pafter f_{1}
& = &
(\pi_{p}+\idmap) \after (\idmap+\bang) \after f \\
& = &
(\idmap+\bang) \after (\pi_{p}+\pi_{q}) \after f \\
& = &
(\idmap+\bang) \after (\pi_{p}+\pi_{q}) \after g \\
& = &
(\pi_{p}+\idmap) \after (\idmap+\bang) \after g \\
& = &
\klin{\pi_{p}} \pafter g_{1}.
\end{array}$$

\noindent Similarly one gets $f_{2} = g_{2}$.

\item For the isomorphism $\cmpr{X}{p} + \cmpr{Y}{q} \cong
  \cmpr{X+Y\,}{\,[p,q]}$ we define total maps in both directions. First,
  consider the map $\kappa_{1} \after \pi_{p} \colon
  \cmpr{X}{p} \rightarrow X+Y$, satisfying:
$$\begin{array}{rcccccl}
\tbox{\big(\kappa_{1} \after \pi_{p}\big)}([p,q])
& = &
[p,q] \after \kappa_{1} \after \pi_{p}
& = &
p \after \pi_{p}
& = &
\one.
\end{array}$$

\noindent This yields a unique map $\varphi_{1} \colon \cmpr{X}{p}
\to \cmpr{X+Y\,}{\,[p,q]}$ with $\pi_{[p,q]} \after \varphi_{1} =
\kappa_{1} \after \pi_{p}$. In a similar way we get $\varphi_{2}
\colon \cmpr{Y}{q} \to \cmpr{X+Y\,}{\,[p,q]}$ with $\pi_{[p,q]}
\after \varphi_{2} = \kappa_{2} \after \pi_{q}$. The cotuple $\varphi
= [\varphi_{1}, \varphi_{2}] \colon \cmpr{X}{p} +\cmpr{Y}{q}
\to \cmpr{X+Y\,}{\,[p,q]}$ gives a map in one direction.
It satisfies, by construction:
$$\begin{array}{rcccccl}
\pi_{[p,q]} \after \varphi
& = &
[\pi_{[p,q]} \after \varphi_{1}, \pi_{[p,q]} \after \varphi_{2}]
& = &
[\kappa_{1} \after \pi_{p}, \kappa_{2} \after \pi_{q}]
& = &
\pi_{p} + \pi_{q}.
\end{array}$$

For the other direction we consider the map $\rhd_{1} \after \pi_{[p,q]}
\colon \cmpr{X+Y\,}{\,[p,q]} \to X+1$. It satisfies:
$$\begin{array}{rcl}
\pbox{\big(\rhd_{1} \after \pi_{[p,q]}\big)}(p)
\hspace*{\arraycolsep}=\hspace*{\arraycolsep}
[p, \kappa_{1}] \after (\idmap+\bang) \after \pi_{[p,q]} 
& = &
[p,\one] \after \pi_{[p,q]} \\
& = &
\tbox{\pi_{[p,q]}}([p,\one]) \\
& \geq &
\tbox{\pi_{[p,q]}}([p,q]) \\
& = &
\one.
\end{array}$$

\noindent Hence by (the partial version of) comprehension there is a
uniqe map $\psi_{1} \colon \cmpr{X+Y\,}{\,[p,q]} \to \cmpr{X}{p}+1$ with
$(\pi_{p}+\idmap) \after \psi_{1} = \rhd_{1} \after
\pi_{[p,q]}$. In a similar way there is a unique $\psi_{2} \colon
\cmpr{X+Y\,}{\,[p,q]} \to \cmpr{Y}{q}+1$ with $(\pi_{q}+\idmap)
\after \psi_{2} = \rhd_{2} \after \pi_{[p,q]}$. We then
obtain the map $\psi = \dtuple{\psi_{1}, \psi_{2}} \colon
\cmpr{X+Y\,}{\,[p,q]} \to \cmpr{X}{p} + \cmpr{Y}{q}$, like in
Lemma~\ref{lem:pairing}. This pairing exists since:
$$\begin{array}{rcl}
\ker(\psi_{1})
& = &
[\kappa_{2}, \kappa_{1}] \after (\bang+\idmap) \after \psi_{1} \\
& = &
[\kappa_{2}, \kappa_{1}] \after (\bang+\idmap) \after 
   (\pi_{p}+\idmap) \after \psi_{1} \\
& = &
\rhd_{2} \after \rhd_{1} \after \pi_{[p,q]} \\
& = &
(\bang+\idmap) \after \rhd_{2} \after \pi_{[p,q]} \\
& = &
(\bang+\idmap) \after (\pi_{q}+\idmap) \after \psi_{2} \\
& = &
(\bang+\idmap) \after \psi_{2} \\
& = &
\kerbot(\psi_{2}).
\end{array}$$

\auxproof{
We rely on:
$$\begin{array}{rcl}
\pbox{\big(\rhd_{2} \after \pi_{[p,q]}\big)}(q)
& = &
[q, \kappa_{1}] \after [\kappa_{2}\after \bang, \kappa_{1}] \after \pi_{[p,q]} \\
& = &
[\kappa_{1}\after \bang, q] \after \pi_{[p,q]} \\
& = &
[1, q] \after \pi_{[p,q]} \\
& = &
\tbox{\pi_{[p,q]}}([1,q]) \\
& \geq &
\tbox{\pi_{[p,q]}}([p,q]) \\
& = &
1
\\
\rhd_{2} \after \rhd_{1} 
& = &
[\kappa_{2}\after \bang, \kappa_{1}] \after [\kappa_{1}, \kappa_{2} \after \bang] \\
& = &
[\kappa_{2}\after \bang, \kappa_{1} \after \bang] \\
& = &
[\kappa_{1} \after \bang, \kappa_{2}] \after [\kappa_{2}\after \bang, \kappa_{1}] \\
& = &
(\bang+\idmap) \after \rhd_{2}
\end{array}$$
}

\noindent By~\eqref{eqn:dtuplenat} and the uniqueness of
$\dtuple{-,-}$ we obtain:
$$\begin{array}{rcl}
(\pi_{p}+\pi_{q}) \after \psi
& = &
\dtuple{\klin{\pi_{p}}\pafter \psi_{1}, \klin{\pi_{q}} \pafter \psi_{2}} \\
& = &
\dtuple{\rhd_{1} \after \pi_{[p,q]}, \rhd_{2} \after \pi_{[p,q]}} 
\hspace*{\arraycolsep}=\hspace*{\arraycolsep}
\pi_{[p,q]}.
\end{array}$$

\noindent Our final aim is to show that $\varphi$ and $\psi$ are each
others inverses. This is easy since both $\pi_{[p,q]}$ and
$\pi_{p}+\pi_{q}$ are monic --- the latter by
point~\eqref{lem:totcmpr:sumproj}. Hence we are done with:
$$\begin{array}[b]{rcl}
\pi_{[p,q]} \after \varphi \after \psi
& = &
(\pi_{p}+\pi_{q}) \after \psi
\hspace*{\arraycolsep}=\hspace*{\arraycolsep}
\pi_{[p,q]}
\\
(\pi_{p}+\pi_{q}) \after \psi \after \varphi
& = &
\pi_{[p,q]}  \after \varphi
\hspace*{\arraycolsep}=\hspace*{\arraycolsep}
(\pi_{p}+\pi_{q}).
\end{array}\eqno{\QEDbox}$$

\auxproof{
The equation $(\pi_{p}+\pi_{q}) \after \psi = \pi_{[p,q]}$
follows by uniqueness of mediating maps in the pullback diagram:
$$\xymatrix@R-.5pc{
\cmpr{X+Y\,}{\,[p,q]}
   \ar@/^3ex/[drr]^-{(\bang+\idmap) \after \pi_{[p,q]}}
   \ar@/_3ex/[ddr]_-{(\idmap+\bang) \after \pi_{[p,q]}}
   \ar@{..>}[dr]
\\
& X+Y\ar[r]^(0.5){\bang+\idmap}
   \ar[d]_{\idmap+\bang}\pullback & 
   1+Y\ar[d]^{\idmap+\bang}
\\
& X + 1\ar[r]_-{\bang+\idmap} & 1+1
}$$

\noindent The map $(\pi_{p}+\pi_{q}) \after \psi$ is mediating by the
way things are defined:
$$\begin{array}{rcl}
(\idmap+\bang) \after (\pi_{p}+\pi_{q}) \after \psi
& = &
(\pi_{p}+\idmap) \after (\idmap+\bang) \after \psi \\
& = &
(\pi_{p}+\idmap) \after \psi_{1} \\
& = &
(\idmap+\bang) \after \pi_{[p,q]} 
\\
(\bang+\idmap) \after (\pi_{p}+\pi_{q}) \after \psi
& = &
(\idmap+\pi_{q}) \after (\bang+\idmap) \after \psi \\
& = &
(\idmap+\pi_{q}) \after \psi_{2} \\
& = &
(\bang+\idmap) \after \pi_{[p,q]}.
\end{array}$$
}
\end{enumerate}
\end{proof}

We turn to comprehension properties in the partial case. Of particular
interest is Diagram~\eqref{diag:partcmpr:kermap} below, which shows
that in presence of comprehension, kernel predicates give rise to
kernel maps.

\begin{lemma}
\label{lem:partcmpr}
Let $(\cat{C},I)$ be an effectus with comprehension, in partial form.
\begin{enumerate}
\item \label{lem:partcmpr:pred} For each predicate $p$ one has
  $p \after \pi_{p} = \one$ and $p^{\bot} \after \pi_{p} = \zero$.

\item \label{lem:partcmpr:kermap} The category $\cat{C}$ has `total'
  kernel maps:\index{S}{kernel!-- map!-- in an effectus} for each map
  $f\colon X \to Y$ the following diagram is an equaliser:
\begin{equation}
\label{diag:partcmpr:kermap}
\vcenter{\xymatrix@C+1pc{
\cmpr{X}{\ker(f)}\ar@{ >->}[r]^-{\pi_{\ker(f)}} &
   X\ar@/^1.5ex/[r]^-{f}\ar@/_1.5ex/[r]_-{\zero} & Y
}}
\end{equation}

\noindent where the comprehension map $\pi_{\ker(f)}$ is total, by
definition.

\item \label{lem:partcmpr:kermapbot} For each map $f$, the composite
  $f \after \pi_{\kerbot(f)}$ is total.

\item \label{lem:partcmpr:pb} Comprehension maps are closed under
  pullback in $\cat{C}$: for each map $f\colon X \rightarrow Y$ and
  predicate $p$ on $Y$ there is a pullback:
\begin{equation}
\label{diag:partcmpr:pb}
\vcenter{\xymatrix{
\cmpr{Y}{\pbox{f}(p)}\ar@{..>}[r]\ar@{ >->}[d]_{\pi_{\pbox{f}(p)}}\pullback & 
   \cmpr{X}{p}\ar@{ >->}[d]^{\pi_{p}}
\\
X\ar[r]_-{f} & Y
}}
\end{equation}

\item \label{lem:partcmpr:imgfact} In presence of images, every map
  $f\colon X \to Y$ factors through the comprehension map of its
  image:
$$\xymatrix@R-.5pc{
X\ar[rr]^-{f}\ar@{..>}@/_1ex/[dr] & & Y
\\
& \cmpr{Y}{\img(f)}\ar@{ >->}@/_1ex/[ur]_-{\pi_{\img(f)}}
}$$


\end{enumerate}
\end{lemma}

\begin{proof}
We use the formulation of comprehension in the partial case from
Definition~\ref{def:cmpr}~\eqref{def:cmpr:part}.
\begin{enumerate}
\item By definition the counit is a map $\pi_{p} \colon (\cmpr{X}{p},
  \one) \rightarrow (X,p)$ in $\PPred(\cat{C})$. This means $\one \leq
  \pbox{\pi_{p}}(p) = (p^{\bot} \after \pi_{p})^{\bot}$. Hence
  $p^{\bot} \after \pi_{p} = \zero$.

Since $\pi_{p}$ is total, pre-composition $(-) \after \pi_{p}$ is a
map of effect algebras, see
Lemma~\ref{lem:FinPACwE}~\eqref{lem:FinPACwE:pres}. Therefore:
$$\begin{array}{rcccccccl}
p \after \pi_{p}
& = &
p^{\bot\bot} \after \pi_{p}
& = &
\big(p^{\bot} \after \pi_{p}\big)^{\bot}
& = &
\zero^{\bot}
& = &
\one.
\end{array}$$

\item For the kernel predicate $\ker(f) = (\one \after f)^{\bot}
  \colon X \rightarrow I$ of a map $f\colon X \rightarrow Y$ consider
  the comprehension map $\pi_{\ker(f)} \colon \cmpr{X}{\ker(f)}
  \rightarrowtail X$. It satisfies, by the previous point:
$$\begin{array}{rcccl}
\one \after f \after \pi_{\ker(f)}
& = &
\kerbot(f) \after \pi_{\ker(f)}
& = &
\zero.
\end{array}$$

\noindent Hence $f \after \pi_{\ker(f)} = \zero$ by
Definition~\ref{def:FinPACwE}~\eqref{def:FinPACwE:zero}.

Next, let $g\colon Z \rightarrow X$ satisfy $f \after g = \zero \after g =
\zero$. Then:
$$\begin{array}{rcll}
\pbox{g}\big(\ker(f)\big)
& = &
\ker(f \after g) \qquad & \mbox{by Lemma~\ref{lem:ker}~\eqref{lem:ker:comp}} \\
& = &
\ker(\zero) \\
& = &
\one & \mbox{by Lemma~\ref{lem:ker}~\eqref{lem:ker:zero}.}
\end{array}$$

\noindent Hence there is a necessarily unique map $\overline{g} \colon
Z \rightarrow \cmpr{X}{\ker(f)}$ with $\pi_{\ker(f)} \after
\overline{g} = g$.

\item For each $f\colon X \rightarrow Y$ we have a total map
$f \after \pi_{\kerbot(f)}$ since by point~\eqref{lem:partcmpr:pred}:
$$\begin{array}{rcccl}
\one \after f \after \pi_{\kerbot(f)}
& = &
\kerbot(f) \after \pi_{\kerbot(f)}
& = &
\one.
\end{array}$$

\item The dashed arrow in diagram~\eqref{diag:partcmpr:pb} exists
  because by Exercise~\ref{exc:subst}~\eqref{exc:subst:partialfun}:
$$\begin{array}{rcccl}
\pbox{(f \after \pi_{\pbox{f}(p)})}(p)
& = &
\pbox{\pi_{\pbox{f}(p)}}\big(\pbox{f}(p)\big)
& = &
\one.
\end{array}$$

\noindent Next, assume we have maps $g\colon Z \rightarrow X$ and
$h\colon Z \rightarrow \cmpr{X}{p}$ with $f \after g = \pi_{p} \after
h$. The map $g$ then factors through $\pi_{\pbox{f}(p)}$ since:
$$\begin{array}{rcccccccccl}
\pbox{g}\big(\pbox{f}(p)\big)
& = &
\pbox{(g \after f)}(p)
& = &
\pbox{(\pi_{p} \after h)}(p)
& = &
\pbox{h}\big(\pbox{\pi_{p}}(p)\big)
& = &
\pbox{h}(1)
& = &
\one.
\end{array}$$

\item By definition $\pbox{f}(\img(f)) = \one$, so that $f$ factors
  through $\pi_{\img(f)}$. \QED
\end{enumerate}
\end{proof}

\auxproof{
Lemma~\ref{lem:totcmpr}~\eqref{lem:totcmpr:pb} and
Lemma~\ref{lem:partcmpr}~\eqref{lem:partcmpr:pb} imply that for an
effectuses $\cat{B}$ in total form and $\cat{C}$ in partial form there
are `fibred' functors in commuting triangles:
$$\xymatrix@C-1pc@R-.5pc{
\Pred(\cat{B})\ar[dr]\ar[rr]^-{\pi_{(-)}} & & \Sub(\cat{B})\ar[dl]
& &
\PPred(\cat{C})\ar[dr]\ar[rr]^-{\pi_{(-)}} & & \Sub(\cat{C})\ar[dl]
\\
& \cat{B} &
& &
& \cat{C} &
}$$

\noindent where, in general, $\Sub(\cat{A})$ is the category with
monics in $\cat{A}$ as objects, and commuting squares as morphisms.
When restricted to \emph{sharp} predicates, these functors $\pi_{(-)}$
are often full and faithful. This will be used for instance in
Proposition~\ref{prop:cmprquotimg}.
}

\section{Quotients}\label{sec:quot}

Comprehension sends a predicate $p$ on $X$ to the object $\cmpr{X}{p}$
of elements of $X$ for which $p$ holds. Quotients form a dual
operation: it sends a predicate $p$ on $X$ to the object $X/p$ in
which elements for which $p$ holds are identified with $0$.

We briefly describe comprehension and quotients for vector spaces and
Hilbert spaces (like in~\cite{ChoJWW15}). Let
$\LSub(\Vect)$\index{D}{$\Vect$, category of vector
  spaces}\index{D}{$\LSub(\Vect)$, category of linear subspaces of
  vector spaces} and $\CLSub(\Hilb)$\index{D}{$\Hilb$, category of
  Hilbert spaces}\index{D}{$\CLSub(\Hilb)$, category of closed linear
  subspaces of Hilbert spaces} be the categories of linear
(resp.\ closed linear) subspaces $S \subseteq V$, where $V$ is a
vector (resp.\ Hilbert) space. Morphisms are maps between the
underlying spaces that restrict appropriately. The obvious forgetful
functors $\LSub(\Vect) \rightarrow \Vect$ and $\CLSub(\Hilb)
\rightarrow \Hilb$ have two adjoints on each side:
$$\xymatrix{
\LSub(\Vect)\ar[d]_{\dashv\;}^{\;\dashv}
   \ar@/_8ex/[d]^{\;\dashv}_{\begin{array}{c}\scriptstyle (S\subseteq V) \\[-.7em]
                        \scriptstyle\mapsto V/S\end{array}}
   \ar@/^8ex/[d]_{\dashv\;}^{\begin{array}{c}\scriptstyle (S\subseteq V) \\[-.7em]
                        \scriptstyle\mapsto S\end{array}}
& \hspace*{4em} &
\CLSub(\Hilb)\ar[d]_{\dashv\;}^{\;\dashv}
   \ar@/_8ex/[d]^{\;\dashv}_{\begin{array}{c}\scriptstyle (S\subseteq V) \\[-.7em]
                        \scriptstyle\mapsto S^{\bot}\end{array}}
   \ar@/^8ex/[d]_{\dashv\;}^{\begin{array}{c}\scriptstyle (S\subseteq V) \\[-.7em]
                        \scriptstyle\mapsto S\end{array}}
\\
\Vect\ar@/^4ex/[u]^(0.4){\zero\!}\ar@/_4ex/[u]_(0.4){\!\one}
& &
\Hilb\ar@/^4ex/[u]^(0.4){\zero\!}\ar@/_4ex/[u]_(0.4){\!\one}
}$$

\noindent We see that in both cases comprehension is right adjoint to
the truth functor $\zero$, and is given by mapping a subspace
$S\subseteq V$ to $S$ itself. This is like for $\Sets$, see
Example~\ref{ex:cmpr}~\eqref{ex:cmpr:Sets}.

Categorically, quotients have dual description to comprehension: not
as right adjoint to truth, but as left adjoint to falsity $\zero$. We
briefly describe the associated dual correspondences, on the left
below for vector spaces, and on the right for Hilbert spaces.
$$\begin{prooftree}
\xymatrix{(S\subseteq V)\ar[r]^-{f} & (\{0\} \subseteq W)}
\Justifies
\xymatrix{V/S\ar[r]_-{g} & W}
\end{prooftree}
\hspace*{5em}
\begin{prooftree}
\xymatrix{(S\subseteq V)\ar[r]^-{f} & (\{0\} \subseteq W)}
\Justifies
\xymatrix{S^{\bot}\ar[r]_-{g} & W}
\end{prooftree}$$

\noindent The correspondence on the left says that a linear map
$f\colon V \rightarrow W$ with $S \subseteq f^{-1}(\{0\}) = \ker(f)$
corresponds to a map $g\colon V/S \rightarrow W$. Indeed, this $g$ is
given by $g(v+S) = f(v)$. This is well-defined, because $f(v) = f(v')$
if $v \sim v'$, since the latter means $v - v' \in S$, and thus
$f(v)-f(v') = f(v-v') = 0$. This is the standard universal property of
quotients in algebra.

The situation is more interesting for Hilbert spaces. A map $f \colon
V \rightarrow W$ with $S \subseteq \ker(f)$ is completely determined
by what it does on the orthosupplement $V^{\bot}$, since each vector
$v\in V$ can be written as sum $v = x+y$ with $x\in S$ and $y\in
S^{\bot}$. Hence $f(v) = f(y)$. This explains why $f$ corresponds
uniquely to a map $g\colon S^{\bot} \rightarrow W$, namely its
restriction.

Below we shall see more examples where quotients are given by
complements of comprehension. But first we have to say what it means
for an effectus to have quotients. Recall that comprehension can be
defined both for effectuses in total and partial form in an equivalent
manner, see Theorem~\ref{thm:totpartcmpr}. In contrast, quotients only
makes sense in the partial case.

\begin{definition}
\label{def:quot}
We say that an effectus in partial form $(\cat{C}, I)$ has
\emph{quotients}\index{S}{quotients}\index{S}{effectus!-- with
  quotients} if its zero functor $\zero \colon \cat{C} \rightarrow
\PPred(\cat{C})$ has a left adjoint.

We say that an effectus in total form $\cat{B}$ has quotients if its
category $\Par(\cat{B})$ of partial maps has quotients.
\end{definition}

For an effectus in partial form $\cat{C}$ with both comprehension and
quotients we have a `quotient-comprehension chain', like for vector
and Hilbert spaces:
$$\xymatrix{
\PPred(\cat{C})\ar[d]_{\dashv\;}^{\;\dashv}
   \ar@/_8ex/[d]^{\;\dashv}_{\begin{array}{c}\scriptstyle\mathrm{Quotient} \\[-.7em]
                        \scriptstyle(X,p) \mapsto X/p\end{array}}
   \ar@/^8ex/[d]_{\dashv\;}^{\begin{array}{c}\scriptstyle\mathrm{Comprehension} \\[-.7em]
                        \scriptstyle(X,p) \mapsto \cmpr{X}{p}\end{array}}
\\
\cat{C}\ar@/^4ex/[u]^(0.4){\zero\!}\ar@/_4ex/[u]_(0.4){\!\one}
}$$

\noindent We shall study such combinations in
Section~\ref{sec:cmprquot}.

The unit of the quotient adjunction is a map in
$\PPred(\cat{C})$ which we write as:
$$\xymatrix{
(X,p)\ar[r]^-{\xi_p} & (X/p, \zero)
}$$

\noindent Hence by definition it is a map $\xi_{p} \colon X
\rightarrow X/p$\index{D}{$X/p$, quotient of predicate $p$ on
  $X$}\index{D}{$\xi_p$, quotient map for predicate $p$} in $\cat{C}$
satisfying $p \leq \pbox{f}(\zero) = \ker(f)$.

\begin{example}
\label{ex:quot}
Each of our four running examples of effectuses has quotients.
\begin{enumerate}
\item \label{ex:quot:Sets} For the effectus $\Sets$, the associated
  category $\Par(\Sets)$ is the category of sets and partial
  functions. For a predicate $P\subseteq X$ we define the quotient set
  $X/P$ as comprehension of the complement:
\begin{equation}
\label{eqn:quot:Sets}
\begin{array}{rcccl}
X/P
& \defeq &
\cmpr{X}{\neg P}
& = &
\setin{x}{X}{x\not\in P},
\end{array}
\end{equation}

\noindent analogously to the Hilbert space example described
above. This gives a left adjoint to the zero functor $\zero \colon
\Par(\Sets) \rightarrow \PPred(\Par(\Sets))$, since there are
bijective correspondences:
$$\begin{prooftree}
\xymatrix{ (P\subseteq X) \ar[r]^-{f} & \zero(Y) 
   \rlap{$\;=(\emptyset \subseteq Y)$}}
\Justifies
\xymatrix{ \cmpr{X}{\neg P} \ar[r]_-{g} & Y }
\end{prooftree}$$

\noindent Indeed, given $f\colon (P\subseteq X) \rightarrow \zero(Y)$
in $\PPred(\Par(\Sets))$, then $f \colon X \pto Y$ is a partial
function satisfying $P \subseteq \pbox{f}(\zero) = \set{x}{f(x) = *}$,
see~\eqref{eqn:partsubstSets}.  Thus $f$ is determined by its outcome
on the complement $\neg P$, so that we can simply define a
corresponding partial function $\overline{f} \colon \cmpr{X}{\neg P}
\pto Y$ as $\overline{f}(x) = f(x)$. And, in the other direction, for
$g \colon \cmpr{X}{\neg P} \pto Y$ we define the extension
$\overline{g} \colon X \pto Y$ as:
$$\begin{array}{rcl}
\overline{g}(x)
& = &
\left\{\begin{array}{ll}
* \;\mbox{\textit{i.e.}~undefined} \quad & \mbox{if } x\in P \\
g(x) & \mbox{if } x \in \neg P
\end{array}\right.
\end{array}$$

\noindent By construction, this $\overline{g}$ is a map $(P\subseteq
X) \rightarrow (\emptyset \subseteq Y)$ in $\PPred(\Par(\Sets))$ since:
$$\begin{array}{rcccccl}
\ker(\overline{g})
& = &
\pbox{\overline{g}}(\zero)
& = &
\set{x}{g(x) = *}
& \supseteq &
P.
\end{array}$$

\noindent It is easy to see $\smash{\overline{\overline{f}} = f}$ and
$\smash{\overline{\overline{g}} = g}$. The unit $\xi_{P} \colon X
\rightarrow X/P$ is the partial function with $\xi_{P}(x) = x$ if
$x\in P$ and $\xi_{P}(x)$ undefined otherwise.

\auxproof{
$$\begin{array}[b]{rcl}
\overline{\overline{f}}(x)
& = &
\left\{\begin{array}{ll}
* \quad & \mbox{if } x\in P \\
\overline{f}(x) & \mbox{if } x \in \neg P
\end{array}\right. \\
& = &
\left\{\begin{array}{ll}
f(x) \quad & \mbox{if } x\in P \\
f(x) & \mbox{if } x \in \neg P
\end{array}\right. \\
& = &
f(x) \\
\overline{\overline{g}}(x)
& = &
\overline{g}(x) \\
& = &
\left\{\begin{array}{ll}
* \quad & \mbox{if } x\in P \\
g(x) & \mbox{if } x \in \neg P
\end{array}\right. \\
& = &
g(x) \qquad \mbox{since by assumption $x\in\neg P$.}
\end{array}$$
}

\item \label{ex:quot:KlD} We recall that the category of partial maps
  for the effectus $\Kl(\Dst)$ for discrete probability is the Kleisli
  category $\Kl(\sDst)$ of the subdistribution monad $\sDst$, see
  Example~\ref{ex:effectusKlD}. For a predicate $p\in [0,1]^{X}$ on a
  set $X$ we take as quotient:
\begin{equation}
\label{eqn:quot:KlD}
\begin{array}{rcl}
X/p
& \defeq &
\setin{x}{X}{p(x) < 1}.
\end{array}
\end{equation}

The quotient adjunction involves bijective correspondences:
$$\begin{prooftree}
\xymatrix{ (p\in [0,1]^{X}) \ar[r]^-{f} & (\zero\in [0,1]^{Y})
   \rlap{$\;=\zero(Y)$}}
\Justifies
\xymatrix{ X/p \ar[r]_-{g} & Y }
\end{prooftree}$$

\noindent This works as follows.
\begin{itemize}
\item Given $f\colon (p\in [0,1]^{X}) \rightarrow (\zero\in
  [0,1]^{Y})$ in $\PPred(\Kl(\sDst))$, then $f\colon X \rightarrow
  \sDst(Y)$ satisfies $p(x) \leq \pbox{f}(\zero)(x) = \ker(f)(x) = 1 -
  \sum_{y}f(x)(y)$ for each $x\in X$,
  see~\eqref{eqn:partsubstKlsD}. We then define $\overline{f} \colon
  X/p \rightarrow \sDst(Y)$ as:
$$\begin{array}{rcccl}
\overline{f}(x)
& = &
\sum_{y} \frac{f(x)(y)}{p^{\bot}(x)}\bigket{y}
& = &
\sum_{y} \frac{f(x)(y)}{1-p(x)}\bigket{y}
\end{array}$$

\noindent This is well-defined, since $p(x)\neq 1$ for $x\in X/p$.

\item In the other direction, given a function $g \colon X/p
  \rightarrow \sDst(Y)$ we define $\overline{g} \colon X \rightarrow
  \sDst(Y)$ as:
$$\begin{array}{rcl}
\overline{g}(x)
& = &
\sum_{y} p^{\bot}(x)\cdot g(x)(y)\bigket{y}
\end{array}$$

\noindent This $\overline{g}$ is a morphism $(p\in [0,1]^{X})
\rightarrow (\zero\in [0,1]^{Y})$ in $\PPred(\Kl(\sDst))$, since we
have an inequality $p \leq \ker(\overline{g})$ via:
$$\begin{array}{rcl}
\ker(\overline{g})(x)
\hspace*{\arraycolsep} = \hspace*{\arraycolsep}
1 - \sum_{y} \overline{g}(x)(y)
& = &
1 - \sum_{y} p^{\bot}(x)\cdot g(x)(y) \\
& = &
p(x) + p^{\bot}(x) - p^{\bot}(x)\cdot \sum_{y} g(x)(y) \\
& = &
p(x) - p^{\bot}(x)\cdot (1 - \sum_{y} g(x)(y)) \\
& = &
p(x) + p^{\bot}(x)\cdot \ker(g)(x) \\
& \geq &
p(x).
\end{array}$$
\end{itemize}

\noindent Clearly, the mappings $f\mapsto \overline{f}$ and $g \mapsto
\overline{g}$ are each other's inverses. The unit map $\xi_{p} \colon
X \rightarrow X/p$ is given by $\xi_{p}(x) = p^{\bot}(x)\ket{x} =
(1-p(x))\ket{x}$.

\auxproof{
$$\begin{array}[b]{rcl}
\overline{\overline{f}}(x)
& = &
\sum_{y} p^{\bot}(x)\cdot \overline{f}(x)(y)\ket{y} \\
& = &
\sum_{y} f(x)(y)\ket{y} \\
& = &
f(x) \\
\overline{\overline{g}}(x)
& = &
\sum_{y} \frac{\overline{g}(x)(y)}{p^{\bot}(x)}\ket{y} \\
& = &
\sum_{y} g(x)(y)\ket{y} \\
& = &
g(x).
\end{array}$$
}

We notice that also in this case we can express the quotient object as
comprehension, namely:
$$\begin{array}{rcccccl}
X/p
& = &
\setin{x}{X}{p(x) < 1}
& = &
\setin{x}{X}{\ceil{p^\bot}(x) = 1}
& = &
\cmpr{X}{\ceil{p^\bot}},
\end{array}$$

\noindent where, in general, $\ceil{q}$\index{D}{$\ceil{q}$, least
  sharp predicate above fuzzy predicate $q$} is the least sharp
predicate above $q$, given by:
\begin{equation}
\label{eqn:ceilKlD}
\begin{array}{rcl}
\ceil{q}(x)
& = &
\left\{\begin{array}{ll}
1 \quad & \mbox{if } q(x) > 0 \\
0 & \mbox{if } q(x) = 0
\end{array}\right.
\end{array}
\end{equation}

\auxproof{
Indeed:
$$\begin{array}{rcccl}
\ceil{p^\bot}(x) = 1
& \Longleftrightarrow &
p^{\bot}(x) = 1 - p(x) > 0
& \Longleftrightarrow &
p(x) < 1.
\end{array}$$
}

\item \label{ex:quot:OUG} For the effectus of order unit groups we
  recall from Example~\ref{ex:effectusOUG} that the category
  $\Par(\op{\OUG})$ of partial maps in $\op{\OUG}$ contains positive
  \emph{subunital} group homomorphisms as maps.  Because of the
  `opposite' involved, we define quotients via subgroups, in the
  following way. For an effect $e\in [0,1]_{G}$ in an order unit group
  $G$ we take:
\begin{equation}
\label{eqn:quot:OUG}
\begin{array}{rcl}
G/e
& \defeq &
\pideal{e^\bot}{G},
\end{array}
\end{equation}

\noindent where $\pideal{e^\bot}{G} \subseteq G$ is the subgroup given
by $\pideal{e}{G} = \setin{x}{G}{\exin{n}{\NNO}{-n\cdot e^{\bot} \leq
    x \leq n\cdot e^{\bot}}}$, see
Example~\ref{ex:cmpr}~\eqref{ex:cmpr:OUG}.  We write the inclusion
$\pideal{e^\bot}{G} \subseteq G$ as a map $\xi_{e} \colon
\pideal{e^\bot}{G} \rightarrow G$. This is a positive group
homomorphism, which is subunital since:
$$\begin{array}{rcccl}
\xi_{e}(\one_{G/e})
& = &
\xi_{e}(e^{\bot})
& = &
e^{\bot} \;\in\; [0,1]_{G}.
\end{array}$$

The quotient adjunction involves a bijective correspondence between:
$$\begin{prooftree}
\xymatrix{(G,e)\ar[r]^-{f} & \zero(H)
   \rlap{\hspace*{2em}in $\PPred\big(\op{\OUG}\big)$}}
\Justifies
\xymatrix{G/e\ar[r]_-{g} & H
   \rlap{\hspace*{3.1em}in \smash{$\Par(\op{\OUG})$}}}
\end{prooftree}\hspace*{5em}$$

\noindent That is, between positive subunital (PsU) group homomorphisms:
$$\begin{prooftree}
\xymatrix{H\ar[r]^-{f} & G \mbox{ with $e \leq \pbox{f}(\zero)$}}
\Justifies
\xymatrix{H\ar[r]_-{g} & G/e}
\end{prooftree}$$

\noindent where $\pbox{f}(\zero) = \ker(f) = f(1)^{\bot}$,
see~\eqref{eqn:partsubstOUG}. This works as follows.
\begin{itemize}
\item Given $f\colon H \rightarrow G$ with $e \leq \ker(f) =
  f(1)^{\bot}$, we have $f(1) \leq e^{\bot}$. For each $x\in H$ there
  is an $n$ with $-n\cdot 1 \leq x \leq n\cdot 1$. Using that $f$ is
positive/monotone we get:
$$\begin{array}{rcl}
-n\cdot e^{\bot} 
\hspace*{\arraycolsep}\leq\hspace*{\arraycolsep}
-n\cdot f(1) 
\hspace*{\arraycolsep}=\hspace*{\arraycolsep}
f(-n\cdot 1)
& \leq &
f(x) \\
& \leq &
f(n\cdot 1)
\hspace*{\arraycolsep}=\hspace*{\arraycolsep}
n\cdot f(1) 
\hspace*{\arraycolsep}\leq\hspace*{\arraycolsep}
n\cdot e^{\bot}.
\end{array}$$

\noindent Hence $f(x) \in \pideal{e^\bot}{G}$, so that we can simply
define $\overline{f} \colon H \rightarrow \pideal{e^\bot}{G} = G/e$ as
$\overline{f}(x) = f(x)$.

\item Given $g\colon H \rightarrow G/e = \pideal{e^{\bot}}{G}$, define
  $\overline{g} = \xi_{e} \after g \colon H \rightarrow G$. This is a
  subunital map with:
$$\begin{array}{rcl}
\pbox{\overline{g}}(\zero)
\hspace*{\arraycolsep}=\hspace*{\arraycolsep}
\overline{g}(\one)^{\bot}
\hspace*{\arraycolsep}=\hspace*{\arraycolsep}
\xi_{e^{\bot}}(g(1))^{\bot}
& = &
g(1)^{\bot} \\
& \geq &
\big(\one_{\pideal{e^\bot}{G}}\big)^{\bot}
\hspace*{\arraycolsep}=\hspace*{\arraycolsep}
\big(e^{\bot}\big)^{\bot}
\hspace*{\arraycolsep}=\hspace*{\arraycolsep}
e.
\end{array}$$
\end{itemize}

\noindent We have \smash{$\overline{\overline{f}} = \xi_{e} \after
  \overline{f} = f$}, and \smash{$\overline{\overline{g}} = g$} by
injectivity of $\xi_{e}$. \QED

\auxproof{
If $g\colon H \rightarrow \pideal{e^\bot}{G}$ satisfies $\xi_{e} \after
g = f$, then $g = \overline{f}$ since $\xi_{e}$ is injective.
}

\item \label{ex:quot:vNA} In the effectus $\op{\vNA}$ of von Neumann
  algebras the partial maps are also the subunital ones, see
  Example~\ref{ex:effectusNA}. The quotient of an effect $e\in
  [0,1]_{\mathscr{A}}$ in a von Neumann algebra $\mathscr{A}$ is
  defined as:
\begin{equation}
\label{eqn:quot:vNA}
\begin{array}{rcccl}
\mathscr{A}/e
& \defeq &
\ceil{e^\bot}\mathscr{A}\ceil{e^\bot}
& = &
\set{\ceil{e^\bot}\cdot x \cdot \ceil{e^\bot}}{x\in \mathscr{A}}.
\end{array}
\end{equation}

\noindent It uses the `ceiling' $\ceil{e^\bot} = \floor{e}^\bot$ from
Example~\ref{ex:cmpr}~\eqref{ex:cmpr:vNA}. The quotient map $\xi_{e}
\colon \mathscr{A} \rightarrow \mathscr{A}/e$ in $\op{\vNA}$ is given
by the subunital function $\mathscr{A}/e \rightarrow \mathscr{A}$ with:
\begin{equation}
\label{eqn:xi:vNA}
\begin{array}{rcl}
\xi_{e}(x)
& = &
\sqrt{e^\bot}\cdot x \cdot \sqrt{e^\bot}.
\end{array}
\end{equation}

\noindent This map incorporates L{\"u}ders rule, see
\textit{e.g.}~\cite[Eq.(1.3)]{BuschS98}.

The proof that these constructions yield a left adjoint to the zero
functor is highly non-trivial. For details we refer
to~\cite{WesterbaanW15}. We conclude that once again quotient objects
can be expressed via comprehension:
$$\begin{array}{rcccccl}
\cmpr{\mathscr{A}}{\ceil{e^\bot}}
& = &
\floor{\ceil{e^\bot}}\mathscr{A}\floor{\ceil{e^\bot}}
& = &
\ceil{e^\bot}\mathscr{A}\ceil{e^\bot}
& = &
\mathscr{A}/e.
\end{array}$$
\end{enumerate}
\end{example}

We continue with a series of small results that hold for quotients.

\begin{lemma}
\label{lem:quot}
Let $(\cat{C},I)$ be an effectus with quotients, in partial form.
\begin{enumerate}
\item \label{lem:quot:XiEqn} The unit $\xi_{p} \colon X \to X/p$
  satisfies $\ker(\xi_{p}) = p$ and thus $\kerbot(\xi_{p}) =
  p^{\bot}$.

\item \label{lem:quot:XiEpi} The unit map $\xi_{p} \colon X \to X/p$ is
  epic in $\cat{C}$.

\item \label{lem:quot:XiComp} For a predicate $q\colon X/p \rightarrow
  I$ one has $q \after \xi_{p} \leq p^{\bot}$.

\item \label{lem:quot:XiFandF} The functor 
$$\xymatrix@C+1pc{
\PPred(\cat{C})\ar[r]^-{\xi_{(-)}} & \Quot(\cat{C})
}$$

\noindent is full and faithful --- where $\Quot(\cat{C})$ is the
obvious category with epis $\twoheadrightarrow$ as objects and
commuting squares between them as arrows.

\item \label{lem:quot:univ} The universal property of quotient amounts
  to: for each $f\colon X \rightarrow Y$ with $p \leq \ker(f)$ there
  is a unique $\overline{f} \colon X/p \rightarrow Y$ with
  $\overline{f} \after \xi_{p} = f$. But we can say more: if there is
  an equality $p = \ker(f)$, then $\overline{f}$ is total.

\item \label{lem:quot:zeroone} There are total isomorphisms
  $\xi_{\zero} \colon X \conglongrightarrow X/\zero$ and $0
  \conglongrightarrow X/\one$.

\item \label{lem:quot:dc} For each predicate $p$ on $X$ there is a
  total `decomposition' map $\dc_{p} \colon X \rightarrow X/p^{\bot} +
  X/p$,\index{D}{$\dc_p$, decomposition map for predicate $p$} namely
  the unique map in:
\begin{equation}
\label{diag:quot:dc}
\vcenter{\xymatrix{
& X\ar[dl]_{\xi_{p^{\bot}}}\ar[dr]^{\xi_{p}}\ar@{..>}[d]^-{\dc_{p}} &
\\
X/p^{\bot} & X/p^{\bot}+X/p\ar[l]^-{\rhd_{1}}\ar[r]_-{\rhd_{2}} & X/p
}}
\end{equation}

\noindent Using the pairing notation of Lemma~\ref{lem:pairing}, we
can describe it as $\dc_{p} = \dtuple{\xi_{p^{\bot}}, \xi_{p}}$, and
thus also as total sum $\dc_{p} = (\kappa_{1} \after \xi_{p^{\bot}})
\ovee (\kappa_{2} \after \xi_{p})$ of two maps $X \rightarrow
X/p^{\bot} + X/p$, see Lemma~\ref{lem:sumpairing}.

\item \label{lem:quot:mapdc} Each total map $f\colon Y \rightarrow
  X_{1} + X_{2}$ can be decomposed in the subcategory $\Tot(\cat{C})
  \hookrightarrow \C$ of total maps as $f = (f_{1}+f_{2}) \after
  \dc_{p}$ where $p = (\bang+\bang) \after f$ and $f_{1}, f_{2}$ are
  total, and uniquely determined.

\item \label{lem:quot:sumxi} For predicates $p$ on $X$ and $q$ on $Y$
  the sum map $\xi_{p} + \xi_{q}\colon X+Y \rightarrow X/p + Y/q$ is
  epic in $\cat{C}$.

\item \label{lem:quot:sumiso} For predicates $p$ on $X$ and $q$ on $Y$
  there is a total (canonical) isomorphism in:
$$\xymatrix@R-.5pc@C-1.5pc{
(X/p)+(Y/q)\ar[rr]^-{\cong} & & (X+Y)/[p,q]
\\
& X+Y\ar@{->>}[ul]^{\xi_{p}+\xi_{q}}\ar@{->>}[ur]_{\xi_{[p,q]}} &
}$$

\item \label{lem:quot:partproj} Partial projections $\rhd_i$ are
  (isomorphic to) quotient maps, via total isomorphisms:
$$\qquad\quad\xymatrix@R-.5pc@C-1.5pc{
& X+Y\ar@{->>}[dl]_{\xi_{[\zero,\one]}}\ar@{->>}[dr]^{\rhd_1} &
& \qquad\qquad\quad &
& X+Y\ar@{->>}[dl]_{\xi_{[\one,\zero]}}\ar@{->>}[dr]^{\rhd_1} &
\\
\qquad\llap{$(X+Y)/[\zero,\one]$}\ar[rr]_-{\cong} & & X
& &
\qquad\llap{$(X+Y)/[\one,\zero]$}\ar[rr]_-{\cong} & & Y
}$$

\item \label{lem:quot:div} For a predicate $p$ on $X$ there is an
  isomorphism of effect modules:
\begin{equation}
\label{diag:quot:div}
\vcenter{\xymatrix{
\Pred(X/p^{\bot})\ar@/^2ex/[rr]^-{p\cdot (-)} & \cong &
  {\downset p} = \{q\in\rlap{$\Pred(X)\;|\;q \leq p\}$}
   \ar@/^2ex/[ll]^-{(-)/p}
}}\hspace*{7em}
\end{equation}

\noindent As a result, if $X/p \cong 0$, then $p = \one$.

\item \label{lem:quot:comp} Quotient maps are closed under
  composition, up to total isomorphisms. Given a predicate $p$ on $X$
  and $q$ on $X/p$, there is an isomorphism as on the left below. It
  uses the operation $p^{\bot} \cdot (-) = (-)\after \xi_{p}$
  from point~\eqref{lem:quot:div}.
$$\xymatrix@C-3pc@R-1pc{
(X/p)/q\ar@{=}[rrrr]^-{\sim}_-{\text{total}} & & & & 
   X/(p^{\bot}\cdot q^{\bot})^{\bot}
 & \qquad\quad &
 (X/p)/(r^{\bot}/p^{\bot})^{\bot}
   \ar@{=}[rrrr]^-{\sim}_-{\text{total}} & & & & X/r
\\
& X/p\ar@{->>}[lu]^{\xi_{q}} & \hspace*{3em} & &
& &
& X/p\ar@{->>}[lu]^{\xi_{(p^{\bot}/r^{\bot})^\bot}} & \hspace*{6em} & &
\\
& & X\ar@{->>}[lu]^{\xi_{p}}\ar@{->>}[rruu]_{\xi_{(p^{\bot}\cdot q^{\bot})^{\bot}}} & &
& &
& & X\ar@{->>}[lu]^{\xi_{p}}\ar@{->>}[rruu]_{\xi_{r}} & &
}$$

\noindent Equivalently, for predicates $p,r$ on $X$ with $p\leq r$,
there is a total isomorphism as on the right above.

\item \label{lem:quot:coker} If our effectus $\cat{C}$ has images,
  then it has cokernel maps:\index{S}{cokernel!-- map!-- in an
    effectus} for each map $f\colon X \rightarrow Y$ we have a
  coequaliser in $\cat{C}$ of the form:
$$\xymatrix@C+.5pc{
X\ar@/^1.5ex/[rr]^-{f}\ar@/_1.5ex/[rr]_{\zero} & & 
   Y\ar@{->>}[rr]^-{\coker(f) \;\defeq\; \xi_{\img(f)}} 
   & & Y/\img(f)
}$$

\noindent Thus: $\coker(f) = \zero$ iff $\img(f) = \one$, that is, iff
$f$ is internally epic.
\end{enumerate}
\end{lemma}

The decomposition property of maps from point~\eqref{lem:quot:mapdc}
plays an important role in Section~\ref{sec:extcat} when we relate
extensive categories and (Boolean) effectuses.

\begin{proof}
We reason in the effectus in partial form $\cat{C}$ and do not use
different notation for partial and total maps.
\begin{enumerate}
\item Since the unit $\xi_p$ is a map $(X,p) \rightarrow (X/p, \zero)$
  in $\PPred(\cat{C})$ we have an inequality $p \leq
  \pbox{\xi_{p}}(\zero) = \ker(\xi_{p})$ by definition. In order to
  get an equation $p = \ker(\xi_{p})$, notice that the predicate
  $p^{\bot} \colon X \rightarrow I$ is a map $(X,p) \rightarrow (I,
  \zero)$ in $\PPred(\cat{C})$, since:
$$\begin{array}{rcccccccl}
\pbox{(p^{\bot})}(\zero)
& = &
\big(\zero^{\bot} \after p^{\bot}\big)^{\bot}
& = &
\big(\one_{I} \after p^{\bot}\big)^{\bot}
& = &
\big(\idmap[I] \after p^{\bot}\big)^{\bot}
& = &
p.
\end{array}$$

\noindent Hence there is a unique map $\overline{p} \colon X/p
\rightarrow I$ with $\overline{p} \after \xi_{p} = p^{\bot}$. But then
we are done by Lemma~\ref{lem:ker}~\eqref{lem:ker:pred}
and~\eqref{lem:ker:ord}: $p = \ker(p^{\bot}) = \ker(\overline{p}
\after \xi_{p}) \geq \ker(\xi_{p})$.

\item Let maps $f,g \colon X/p \rightarrow Y$ satisfy $f \after
  \xi_{p} = g\after \xi_{p} = h$, say. This $h\colon X \rightarrow Y$
  then satisfies $p = \ker(\xi_{p}) \leq \ker(f \after \xi_{p}) =
  \ker(h) = \pbox{h}(\zero)$. Hence there is a unique $\overline{h}
  \colon X/p \rightarrow Y$ with $\overline{h} \after \xi_{p} = h$.
  But then $f = \overline{h} = g$ by uniqueness.

\item For a predicate $q$ on $X/p$ one has:
$$\begin{array}{rcll}
q \after \xi_{p}
\hspace*{\arraycolsep}=\hspace*{\arraycolsep}
\big(q^{\bot\bot} \after \xi_{p}\big)^{\bot\bot}
\hspace*{\arraycolsep}=\hspace*{\arraycolsep}
\pbox{\xi_{p}}(q^{\bot})^{\bot} 
& \leq &
\pbox{\xi_{p}}(\zero)^{\bot} \quad & \mbox{since $\zero \leq q^{\bot}$} \\
& = &
\ker(\xi_{p})^{\bot} \\
& = &
p^{\bot} & \mbox{by point~\eqref{lem:quot:XiEqn}.}
\end{array}$$

\item Let's assume we have a commuting diagram in $\cat{C}$ of the
  form:
$$\xymatrix@R-.5pc@C+1pc{
X\ar@{->>}[d]_{\xi_p}\ar[r]^-{f} & Y\ar@{->>}[d]^{\xi_q} 
\\
X/p\ar@{..>}[r]_-{g} & Y/q
}$$

\noindent We have to prove $p \leq \pbox{f}(q)$, making $f$ a map
$(X,p) \rightarrow (Y,q)$ in $\PPred(\cat{C})$. This follows from:
$$\begin{array}{rcl}
p
\hspace*{\arraycolsep}=\hspace*{\arraycolsep}
\ker(\xi_{p})
\hspace*{\arraycolsep}=\hspace*{\arraycolsep}
\pbox{\xi_{p}}(\zero)
& \leq &
\pbox{\xi_{p}}\big(\pbox{g}(\zero)\big) \\
& = &
\pbox{f}\big(\pbox{\xi_{q}}(\zero)\big)
\hspace*{\arraycolsep}=\hspace*{\arraycolsep}
\pbox{f}\big(\ker(\xi_{q})\big)
\hspace*{\arraycolsep}=\hspace*{\arraycolsep}
\pbox{f}(q).
\end{array}$$

\item Clearly, $p \leq \ker(f) = \pbox{f}(\zero)$ means that $f$ is a
  map $(X,p) \rightarrow (Y,\zero)$ in $\PPred(\cat{C})$. Hence the
  unique map $\overline{f} \colon X/p \rightarrow Y$ is the adjoint
  transpose.

If we have an equation $p = \ker(f)$, then we can show that
$\overline{f}$ is total, by using that the unit map $\xi_p$ is epic:
$$\begin{array}{rcccccccccl}
\one \after \overline{f} \after \xi_{p}
& = &
\one \after f 
& = &
\kerbot(f)
& = &
p^{\bot} 
& = &
\kerbot(\xi_{p}) 
& = &
\one \after \xi_{p}.
\end{array}$$

\item We first show $X/\zero \cong X$. The identity map in $X$
  satisfies: $\ker(\idmap[X]) = \zero$. Hence by
  point~\eqref{lem:quot:univ} there is a unique total map $f\colon
  X/\zero \rightarrow X$ satisfying $f \after \xi_{\zero} =
  \idmap$. The map $\xi_{\zero}$ is also total since $\one \after
  \xi_{\zero} = \kerbot(\xi_{\zero}) = \zero^{\bot} = \one$. We now get
  $\xi_{\zero} \after f = \idmap[X]$ by using that $\xi_{\zero}$ is
  epic: $(\xi_{\zero} \after f) \after \xi_{\zero} = \xi_{\zero} =
  \idmap \after \xi_{\zero}$.

Next, for the isomorphism $X/\one \cong 0$, we notice that the zero
map $\zero \colon X \rightarrow 0$ in $\cat{C}$ satisfies $\ker(\zero)
= (\zero^{\bot} \after \zero)^{\bot} = \zero^{\bot} = \one$, so that
there is a unique total map $f\colon X/\one \rightarrow 0$ with $f
\after \xi_{\one} = \zero$. Obviously, $f \after \bang_{X/\one} =
\idmap[\zero]$, where $\bang_{X/\one} \colon 0 \rightarrow X/\one$ is
total. We also have $\bang_{X/\one} \after f = \idmap[X/\one]$, since
$\xi_{\one}$ is epic and: $(\bang_{X/\one} \after f) \after \xi_{\one}
= \bang_{X/\one} \after \zero = \zero = \xi_{\one}$. This last
equation holds by
Definition~\ref{def:FinPACwE}~\eqref{def:FinPACwE:zero} since $\one
\after \xi_{\one} = \kerbot(\xi_{\one}) = \one^{\bot} = \zero$.

\item The diagram~\eqref{diag:quot:dc} uniquely defines $\dc_{p} =
  \dtuple{\xi_{p^{\bot}}, \xi_{p}}$ as described in
  Lemma~\ref{lem:pairing}. It satisfies $[\kappa_{2}, \kappa_{1}]
  \after \dc_{p} = \dc_{p^\bot}$ since:
$$\begin{array}{rcl}
\rhd_{1} \after [\kappa_{2}, \kappa_{1}] \after \dc_{p}
& = &
[[\idmap, \zero] \after \kappa_{2}, [\idmap, \zero] \after \kappa_{1}] 
   \after \dc_{p} \\
& = &
[\zero, \idmap] \after \dc_{p}
\hspace*{\arraycolsep}=\hspace*{\arraycolsep}
\rhd_{2} \after \dc_{p}
\hspace*{\arraycolsep}=\hspace*{\arraycolsep}
\xi_{p}.
\end{array}$$

\noindent Similarly one proves $\rhd_{2} \after [\kappa_{2},
  \kappa_{1}] \after \dc_{p} = \xi_{p^{\bot}}$, so that:
$$\begin{array}{rcccl}
[\kappa_{2}, \kappa_{1}] \after \dc_{p}
& = &
\dtuple{\xi_{p}, \xi_{p^{\bot}}}
& = &
\dc_{p^\bot}.
\end{array}$$

\item Let $f\colon Y \rightarrow X_{1}+X_{2}$ be total, and write $p =
  (\bang+\bang) \after f \colon Y \rightarrow 1+1$ in the subcategory
  $\Tot(\cat{C})$ of total maps. The partial map $\rhd_{1} \after f
  \colon Y \rightarrow X_{1}$ satisfies:
$$\begin{array}{rcccl}
\ker\big(\rhd_{1} \after f\big)
& = &
\big((\bang+\idmap) \after (\idmap+\bang) \after f\big)^{\bot}
& = &
p^{\bot}.
\end{array}$$

\noindent Hence there is a unique total map $f_{1}\colon Y/p^{\bot}
\rightarrow X_{1}$ with $f_{1} \after \xi_{p^\bot} = \rhd_{1} \after
f$.  Similarly, from $\ker\big(\rhd_{2} \after f\big) = p$ we obtain a
unique total $f_{2} \colon Y/p \rightarrow X_{2}$ with $f_{2} \after
\xi_{p} = \rhd_{2} \after f$. But then in $\Tot(\cat{C})$, using
Lemma~\ref{lem:pairing}:
$$\begin{array}{rcl}
(f_{1}+f_{2}) \after \dc_{p}
\hspace*{\arraycolsep}=\hspace*{\arraycolsep}
(f_{1}+f_{2}) \after \dtuple{\xi_{p^\bot}, \xi_{p}} 
& = &
\dtuple{f_{1} \after \xi_{p^\bot}, f_{2} \after \xi_{p}} \\
& = &
\dtuple{\rhd_{1} \after f, \rhd_{2} \after f} \\
& = &
f.
\end{array}$$

\item In an arbitrary category, if $f\colon X\twoheadrightarrow A$,
  $g\colon Y \twoheadrightarrow B$ are both epic, then so is the sum
  $f+g$. 

\auxproof{
Let $h,k\colon A+B \rightarrow C$ satisfy $h \after (f+g) = k \after
(f+g)$. Write $h = [h_{1}, h_{2}]$ and $k = [k_{1}, k_{2}]$. Then:
$$\begin{array}{rcccccl}
[h_{1} \after f, h_{2} \after g]
& = &
h \after (f+g)
& = &
k \after (f+g)
& = &
[k_{1} \after f, k_{2} \after g].
\end{array}$$

\noindent Hence we have $h_{1} \after f = k_{1} \after g$ and
$h_{2} \after f = k_{2} \after g$. Thus, since $f,g$ are epic, we
get $h_{1} = k_{1}$ and $h_{2} = k_{2}$, and so $h = k$ as required.
}

\item We first have to say what the canonical map $\tau \colon
  (X+Y)/[p,q] \rightarrow (X/p)+(Y/q)$ is. Consider the sum map
  $\xi_{p}+\xi_{q}$ in $\cat{C}$ with:
$$\begin{array}{rcll}
\kerbot(\xi_{p}+\xi_{q})
\hspace*{\arraycolsep}=\hspace*{\arraycolsep}
\one \after (\xi_{p}+\xi_{q})
& = &
[\one,\one] \after (\xi_{p}+\xi_{q}) &
 \mbox{by Lemma~\ref{lem:FinPACwE}~\eqref{lem:FinPACwE:onecoptuple}} \\
& = &
[\one \after \xi_{p}, \one \after \xi_{q}] \\
& = &
[\kerbot(\xi_{p}), \kerbot(\xi_{q})] \\
& = &
[p^{\bot}, q^{\bot}] \\
& = &
[p,q]^{\bot} & \mbox{by Theorem~\ref{thm:effectusEMod}.}
\end{array}$$

\noindent This last step uses the isomorphism of effect modules
$\Pred(X+Y) \cong \Pred(X) \times \Pred(Y)$. The resulting equation
$\ker(\xi_{p}+\xi_{q}) = [p,q]$ yields by point~\eqref{lem:quot:univ}
a unique total map $\tau \colon (X+Y)/[p,q] \rightarrow X/p+Y/q$ with
$\tau \after \xi_{[p,q]} = \xi_{p}+\xi_{q}$.

In the other direction, consider the two composites:
$$\xymatrix@R-2pc{
X\ar[dr]^-{\kappa_{1}} 
\\
& X+Y\ar[r]^-{\xi_{[p,q]}} & (X+Y)/[p,q]
\\
Y\ar[ur]_-{\kappa_2}
}$$

\noindent The first one satisfies:
$$\begin{array}{rcl}
\ker\big(\xi_{[p,q]} \after \kappa_{1}\big)
\hspace*{\arraycolsep}=\hspace*{\arraycolsep}
\big(\one \after \xi_{[p,q]} \after \kappa_{1}\big)^{\bot} 
& = &
\big(\kerbot(\xi_{[p,q]}) \after \kappa_{1}\big)^{\bot} \\
& = &
\big([p^{\bot},q^{\bot}] \after \kappa_{1}\big)^{\bot} \\
& = &
p.
\end{array}$$

\noindent Similarly, $\ker\big(\xi_{[p,q]} \after \kappa_{2}\big) =
q$.  Hence there are unique total maps $f\colon X/p \rightarrow
(X+Y)/[p,q]$ and $g\colon Y/q \rightarrow (X+Y)/[p,q]$ with $f \after
\xi_{p} = \xi_{[p,q]} \after \kappa_{1}$ and $g \after \xi_{q} =
\xi_{[p,q]} \after \kappa_{2}$. We claim that the cotuple $[f,g]
\colon X/p + Y/q \rightarrow (X+Y)/[p,q]$ is the inverse of the
canonical map $\tau$.

We obtain $[f,g] \after \tau = \idmap$ from the fact that $\xi$'s are
epic:
$$\begin{array}{rcl}
[f,g] \after \tau \after \xi_{[p,q]}
\hspace*{\arraycolsep}=\hspace*{\arraycolsep}
[f,g] \after (\xi_{p}+\xi_{q}) 
& = &
[f \after \xi_{p}, g \after \xi_{q}] \\
& = &
[\xi_{[p,q]} \after \kappa_{1}, \xi_{[p,q]} \after \kappa_{2}] \\
& = &
\xi_{[p,q]}.
\end{array}$$

\noindent In the other direction we use that the sum map
$\xi_{p}+\xi_{q}$ in $\cat{C}$ is epic:
$$\begin{array}[b]{rcl}
\tau \after [f,g] \after (\xi_{p}+\xi_{q})
& = &
\tau \after [f \after \xi_{p}, g \after \xi_{q}] \\
& = &
[\tau \after \xi_{[p,q]} \after \kappa_{1}, 
   \tau \after \xi_{[p,q]} \after \kappa_{2}] \\
& = &
[(\xi_{p}+\xi_{q}) \after \kappa_{1}, (\xi_{p}+\xi_{q}) \after \kappa_{2}] \\
& = &
(\xi_{p}+\xi_{q}).
\end{array}$$

\item By points~\eqref{lem:quot:sumiso} and~\eqref{lem:quot:zeroone}
  we have isomorphisms:
$$\xymatrix@R-.5pc@C-2pc{
& \qquad & & X+Y\ar@{->>}[dlll]_{\xi_{[\zero,\one]}}
   \ar@{->>}[dl]|{\xi_{\zero}+\xi_{\one}}
   \ar@{->>}[dr]|{\idmap+!}\ar@{->>}[drrr]^{\rhd_1}
   & & \qquad &
\\
(X+Y)/[\zero,\one]\ar@{=}[rr]_-{\sim} & & 
   (X/\zero) + (Y/\one)\ar@{=}[rr]_-{\sim} & & 
   X + 0\ar@{=}[rr]_-{\sim} & & X
}$$

\item For a predicate $q\leq p$, we have $p^{\bot} \leq q^{\bot} =
  \ker(q)$.  Hence there is a unique predicate, written as $q/p \colon
  X/p^{\bot} \rightarrow I$, with $(q/p) \after \xi_{p^\bot} = q$.  In
  the other direction, for $r\colon X/p^{\bot} \rightarrow I$ we write
  $p\cdot r = r\after \xi_{p^\bot} \colon X \rightarrow I$. Then
  $p\cdot r \leq p$ follows from:
$$\begin{array}{rcl}
(p\cdot r)^{\bot}
\hspace*{\arraycolsep}=\hspace*{\arraycolsep}
\ker(p\cdot r)
\hspace*{\arraycolsep}=\hspace*{\arraycolsep}
\ker(r \after \xi_{p^\bot})
& \geq &
\ker(\xi_{p^\bot}) \qquad \mbox{by Lemma~\ref{lem:ker}~\eqref{lem:ker:ord}} \\
& = &
p^{\bot}.
\end{array}$$

\noindent By construction, $p\cdot (q/p) = q$. But also $(p\cdot r)/p
= r$, because $\xi$'s are epic, and: $((p\cdot r)/p) \after
\xi_{p^{\bot}} = p\cdot r = r \after \xi_{p^\bot}$. It requires a bit
of work to verify that these mappings $(-)/p$ and $p\cdot (-)$
preserve the effect module structure.

\auxproof{
We first show that the function $(-)/p \colon \downset p \rightarrow
\Pred(X/p^{\bot})$ is a map of effect modules.
\begin{itemize}
\item The top element in $\downset p$ is $p$ itself. We get $p/p = 1$,
  since $1\after \xi_{p^{\bot}} = (!+\idmap) \after \xi_{p^{\bot}} =
    \Dp(\xi_{p^{\bot}}) = p$.

\item The sum of $q_{1}, q_{2} \in \downset p$ is defined as $q_{1}
  \ovee q_{2}$ provided that $q_{1} \orthogonal q_{2}$ and $q_{1}
  \ovee q_{2} \leq p$. Let's assume that $b\colon X \rightarrow
  (1+1)+1$ is a bound in $\cat{B}$ with $\IV \after b = q_{1}$ and
  $\XI \after b = q_{2}$, where $\IV = [\idmap,\kappa_{2}]$ and $\XI =
        [[\kappa_{2},\kappa_{1}], \kappa_{2}]$. Then $\Dp(b) =
        (!+\idmap) \after b = (\nabla+\idmap) \after b = q_{1} \ovee
        q_{2} \leq p = p^{\bot\bot}$. Hence there is a unique map $b/p
        \colon X/p^{\bot} \rightarrow (1+1)+1$ with $(b/p) \after
        \xi_{p^\bot} = b$. We claim that $b/p$ is a bound for
        $q_{1}/p$ and $q_{2}/p$. Indeed, we have $\IV \after b/p =
        q_{1}/q$ and $\XI \after b/p = q_{2}/p$ since $\xi_{p^\bot}$
        is epic and:
$$\begin{array}{rcl}
(\IV \after b/q) \after \xi_{p^\bot}
& = &
[\IV \after b/q, \kappa_{2}] \after \xi_{p^\bot} \\
& = &
\IV \after [b/q, \kappa_{2}] \after \xi_{p^\bot} \\
& = &
\IV \after b \\
& = &
q_{1} \\
& = &
(q_{1}/p) \after \xi_{p^\bot} \\
(\XI \after b/q) \after \xi_{p^\bot}
& = &
[\XI \after b/q, \kappa_{2}] \after \xi_{p^\bot} \\
& = &
\XI \after [b/q, \kappa_{2}] \after \xi_{p^\bot} \\
& = &
\XI \after b \\
& = &
q_{2} \\
& = &
(q_{2}/p) \after \xi_{p^\bot} \\
\end{array}$$

\noindent Now we get $(q_{1}/p) \ovee (q_{2}/p) = (\nabla+\idmap) \after
b/p$. But this $(q_{1}\ovee q_{2})/p$, since:
$$\begin{array}{rcl}
(q_{1}\ovee q_{2})/p \after \xi_{p^\bot}
& = &
q_{1} \ovee q_{2} \\
& = &
(\nabla+\idmap) \after b \\
& = &
(\nabla+\idmap) \after [b/p, \kappa_{2}] \after \xi_{p^\bot} \\
& = &
[(\nabla+\idmap) \after b/p, \kappa_{2}] \after \xi_{p^\bot} \\
& = &
[(q_{1}/p) \ovee (q_{2}/p), \kappa_{2}] \after \xi_{p^\bot} \\
& = &
\big((q_{1}/p) \ovee (q_{2}/p)\big) \after \xi_{p^\bot}.
\end{array}$$

\item Let $s\colon 1 \rightarrow 1+1$ be a scalar, and $q\leq p$.
  Then $s\scalar q \leq q \leq p$, where $s\scalar q = [s, \kappa_{2}]
  \after q = s\after q$. We get $s \scalar (q/p) = (s\scalar q)/p$
  by:
$$\begin{array}{rcl}
\big(s \scalar (q/p)\big) \after \xi_{p^\bot}
& = &
[s \scalar (q/b), \kappa_{2}] \after \xi_{p^\bot} \\
& = &
[[s,\kappa_{2}] \after (q/b), \kappa_{2}] \after \xi_{p^\bot} \\
& = &
[s,\kappa_{2}] \after [(q/b), \kappa_{2}] \after \xi_{p^\bot} \\
& = &
[s,\kappa_{2}] \after q \\
& = &
s\scalar q \\
& = &
\big((s\scalar q)/p\big) \after \xi_{p^\bot}.
\end{array}$$

\noindent Then:
$$\begin{array}{rcl}
(p\cdot q_{1}) \ovee (p\cdot q_{2})
& = &
(\nabla+\idmap) \after (p\cdot b) \\
& = &
(\nabla+\idmap) \after [b, \kappa_{2}] \after \xi_{p^\bot} \\
& = &
[(\nabla+\idmap) \after b, \kappa_{2}] \after \xi_{p^\bot} \\
& = &
[q_{1} \ovee q_{2}, \kappa_{2}] \after \xi_{p^\bot} \\
& = &
r\cdot (q_{1} \ovee q_{2}).
\end{array}$$

\item For a scalar $s\colon 1 \rightarrow 1+1$ and a predicate $q\colon
X/p^{\bot} \rightarrow 1+1$ we have:
$$\begin{array}{rcl}
s \scalar (p\cdot q)
& = &
[s, \kappa_{2}] \after [p, \kappa_{2}] \after \xi_{p^\bot} \\
& = &
[[s, \kappa_{2}] \after p, \kappa_{2}] \after \xi_{p^\bot} \\
& = &
[s\scalar p, \kappa_{2}] \after \xi_{p^\bot} \\
& = &
p\cdot (s\scalar p).
\end{array}$$
\end{itemize}

In the other direction, the map $p\cdot (-) \colon \Pred(X/p^{\bot})
\rightarrow \downset p$ also preserves the effect module structure.
\begin{itemize}
\item $p\cdot 1 = 1 \after \xi_{p^\bot} = (!+\idmap) \after
  \xi_{p^\bot} = p$

\item Let $q_{1}, q_{2} \colon X/p^{\bot} \rightarrow 1+1$ have bound
$b\colon X/p^{\bot} \rightarrow (1+1)+1$. We claim that $p\cdot b$
is then a bound for $p\cdot q_{1}$ and $p\cdot q_{2}$. Indeed:
$$\begin{array}{rcl}
\IV \after (p\cdot b)
& = &
\IV \after [b, \kappa_{2}] \after \xi_{p^\bot} \\
& = &
[\IV \after b, \kappa_{2}] \after \xi_{p^\bot} \\
& = &
[q_{1}, \kappa_{2}] \after \xi_{p^\bot} \\
& = &
p\cdot q_{1} \\
\XI \after (p\cdot b)
& = &
\XI \after [b, \kappa_{2}] \after \xi_{p^\bot} \\
& = &
[\XI \after b, \kappa_{2}] \after \xi_{p^\bot} \\
& = &
[q_{2}, \kappa_{2}] \after \xi_{p^\bot} \\
& = &
p\cdot q_{2}.
\end{array}$$
\end{itemize}
}

Finally, if $X/p \cong 0$, then $\downset p^{\bot} \cong \Pred(X/p)
\cong \Pred(0) \cong 1$. But then $\downset p^{\bot} = \{0\}$, so that
$p^{\bot} = \zero$ and thus $p = \one$.

\item Let $p$ be a predicate on $X$ and $q$ on $X/p$. Then
by Lemma~\ref{lem:ker}~\eqref{lem:ker:comp}:
$$\begin{array}{rcccccccl}
\ker(\xi_{q} \after \xi_{p})
& = &
\pbox{\xi_{p}}(\ker(\xi_{q}))
& = &
\pbox{\xi_{p}}(q) 
& = &
\big(q^{\bot} \after \xi_{p}\big)^{\bot} 
& = &
\big(p^{\bot} \cdot q^{\bot}\big)^{\bot}.
\end{array}$$

\noindent Hence there is a unique (total) map $f\colon
X/(p^{\bot}\cdot q^{\bot})^{\bot} \rightarrow (X/p)/q$ with $f \after
\xi_{X/(p^{\bot}\cdot q^{\bot})^{\bot}} = \xi_{q} \after \xi_{p}$. In
the other direction we proceed in two steps.  First, by construction,
$p^{\bot}\cdot q^{\bot} \leq p^{\bot}$, so $p \leq (p^{\bot}\cdot
q^{\bot})^{\bot}$. This yields a map $g\colon X/p \rightarrow
X/(p^{\bot}\cdot q^{\bot})^{\bot}$ with $g \after \xi_{p} =
\xi_{(p^{\bot}\cdot q^{\bot})^{\bot}}$. We wish to show that there is
a total map $h\colon (X/p)/q \rightarrow X/(p^{\bot}\cdot
q^{\bot})^{\bot}$ with $h\after \xi_{q} = g$. This works if $\ker(g) =
q \colon X/p \rightarrow I$. Because the $\xi$'s are epic in
$\cat{C}$ we are done by:
$$\begin{array}{rcl}
\kerbot(g) \after \xi_{p}
\hspace*{\arraycolsep}=\hspace*{\arraycolsep}
\one \after g \after \xi_{p} 
\hspace*{\arraycolsep}=\hspace*{\arraycolsep}
\one \after \xi_{(p^{\bot}\cdot q^{\bot})^{\bot}} 
& = &
\kerbot(\xi_{(p^{\bot}\cdot q^{\bot})^{\bot}}) \\
& = &
p^{\bot}\cdot q^{\bot} \\
& = &
q^{\bot} \after \xi_{p}.
\end{array}$$

\noindent We now have (total) maps $(X/p)/q \rightleftarrows
X/(p^{\bot}\cdot q^{\bot})^{\bot}$, which are each other's inverses
because they commute with the $\xi$'s.

We leave it to the reader to construct similar maps
$(X/p)/(r^{\bot}/p^{\bot})^{\bot} \rightleftarrows X/r$ for $p\leq r$.

\auxproof{
Let $p\leq r$. Then $r^{\bot} \leq p^{\bot}$, so we have
$r^{\bot}/p^{\bot} \colon X/p \rightarrow 1+1$ with $r^{\bot}/p^{\bot}
\after \xi_{p} = r^{\bot}$. We have:
$$\begin{array}{rcl}
\ker(\xi_{(r^{\bot}/p^{\bot})^{\bot}} \after \xi_{p})
& = &
\square(\xi_{p})\big(\ker(\xi_{(r^{\bot}/p^{\bot})^{\bot}})\big) \\
& = &
\square(\xi_{p})\big((r^{\bot}/p^{\bot})^{\bot}\big) \\
& = &
[(r^{\bot}/p^{\bot})^{\bot}, \kappa_{1}] \after \xi_{p} \\
& = &
[\kappa_{2}, \kappa_{1}] \after [r^{\bot}/p^{\bot}, \kappa_{2}] \after \xi_{p} \\
& = &
[\kappa_{2}, \kappa_{1}] \after r^{\bot} \\
& = &
r.
\end{array}$$

\noindent Hence there is a total map $f\colon X/r \rightarrow
(X/p)/(r^{\bot}/p^{\bot})^{\bot}$ with $J(f) \after \xi_{r} =
\xi_{(r^{\bot}/p^{\bot})^{\bot}} \after \xi_{p}$.

In the other direct we proceed in two steps. First, since $p \leq r$
we have by functoriality of quotients a unique map $g\colon X/p
\rightarrow X/r+1$ with $g \after \xi_{p} = \xi_{r}$. In a next step
we see a total map $h\colon (X/p)/(p^{\bot}/r^{\bot})^{\bot}
\rightarrow X/r$ with $J(h) \after \xi_{(p^{\bot}/r^{\bot})^{\bot}}
= g$. This requires an equality $\ker(g) =
(r^{\bot}/p^{\bot})^{\bot}$, and thus $\ker(g)^{\bot} \after \xi_{p}
= r^{\bot}$. This last equation is obtained as follows.
$$\begin{array}{rcl}
\ker(g)^{\bot} \after \xi_{p}
& = &
[[\kappa_{2}, \kappa_{1}] \after \ker(g), \kappa_{2}] \after \xi_{p} \\
& = &
[\kappa_{2}, \kappa_{1}] \after [\ker(g), \kappa_{2}] \after \xi_{p} \\
& = &
\square(\xi_{p})(\ker(g))^{\bot} \\
& = &
\ker(g \after \xi_{p})^{\bot} \\
& = &
\ker(\xi_{r})^{\bot} \\
& = &
r^{\bot}.
\end{array}$$
}

\item Assuming images in $\cat{C}$ we have for each map $f\colon X
  \rightarrow Y$ a predicate $\img(f) \colon Y \rightarrow I$, with
  quotient map $\xi_{\img(f)} \colon Y \rightarrow Y/\img(f)$. We
  claim that $\xi_{\img(f)} \after f = \zero \colon X \rightarrow
  X/\img(f)$. This follows by
  Definition~\ref{def:FinPACwE}~\eqref{def:FinPACwE:zero} from:
$$\begin{array}{rcccccccccl}
\one \after \xi_{\img(f)} \after f
& = &
\kerbot(\xi_{\img(f)}) \after f
& = &
\imgbot(f) \after f
& = &
\pbox{f}(\img(f))^{\bot}
& = &
\one^{\bot}
& = &
\zero.
\end{array}$$

\noindent Now suppose we have a map $g\colon Y \rightarrow Z$ with $g
\after f = g \after \zero = \zero$. Then $\img(f) \leq \ker(g)$ by
Lemma~\ref{lem:img}~\eqref{lem:img:compzero}. This yields the required
unique map $\overline{g} \colon Y/\img(f) \rightarrow Z$ with
$\overline{g} \after \xi_{\img(f)} = g$.

Finally, if $\img(f) = \one$, then $Y/\img(f) = Y/\one \cong 0$, by
point~\eqref{lem:quot:zeroone}, so that $\coker(f) = \xi_{\img(f)}
\colon Y \rightarrow 0$ is the zero map. Conversely, let $\coker(f) =
\zero$. By Lemma~\ref{lem:ker}~\eqref{lem:ker:pred} we have
$\ker(\imgbot(f)) = \img(f) \colon Y \rightarrow I$. Hence there is a
unique map $g\colon X/\img(f) \rightarrow I$ with $\imgbot(f) = g
\after \xi_{\img(f)} = g \after \coker(f) = g\after \zero =
\zero$. But then $\img(f) = \one$. \QED
\end{enumerate}
\end{proof}

Several of the points in Lemma~\ref{lem:quot} are used to prove that
quotients give rise to a factorisation system. Such a system is given
by two collections of map, called `abstract monos' and `abstract
epis', satisfying certain properties, see \textit{e.g.}~\cite{BarrW85}
for details.

\begin{proposition}
\label{prop:quotfactorisation}
Each effectus with quotients in partial form has a factorisation
system\index{S}{factorisation system!-- via quotient maps} given by
internal monos (\textit{i.e.}~total maps) and quotient maps $\xi$.
\end{proposition}

\begin{proof}
An arbitrary map $f\colon X \rightarrow Y$ can be factored as:
$$\xymatrix@R-.5pc{
X\ar[rr]^-{f}\ar@{->>}@/_1ex/[dr]_-{\xi_{\ker(f)}} & & Y
\\
& X/\ker(f)\ar@{..>}@/_1ex/[ur]_-{f'\text{, total}}
}$$

\noindent The fact that the map $f'$ is total follows from
Lemma~\ref{lem:quot}~\eqref{lem:quot:univ}.

By Lemma~\ref{lem:ker}~\eqref{lem:ker:monocat} the total (internally
monic) maps are closed under composition. The $\xi$'s are also closed
under composition, by Lemma~\ref{lem:quot}~\eqref{lem:quot:comp}. The
intersection of the these two classes consists of isomorphisms: if
$\xi_{p} \colon X \rightarrow X/p$ is total, then $\one = \one \after
\xi_{p} = \kerbot(\xi_{p}) = p^{\bot}$, so that $p=\zero$, making
$\xi_{p} = \xi_{\zero} \colon X \rightarrow X/\zero$ an isomorphism,
see Lemma~\ref{lem:quot}~\eqref{lem:quot:zeroone}.

Finally we check the diagonal-fill-in property. In a commuting diagram
as on the left below, there is a unique diagonal as on the right,
making both triangles commute.
$$\xymatrix@C+1pc{
X\ar@{->>}[r]^-{\xi_p}\ar[d]_{g} & X/p\ar[d]^{h}
& &
X\ar@{->>}[r]^-{\xi_p}\ar[d]_{g} & X/p\ar[d]^{h}\ar@{..>}[dl]
\\
Y\ar[r]_-{f}^-{\text{total}} & Z
& &
Y\ar[r]_-{f}^-{\text{total}} & Z
}$$

\noindent We have by Lemma~\ref{lem:ker}~\eqref{lem:ker:tot}
and~\eqref{lem:ker:ord}, using that $f$ is total,
$$\begin{array}{rcccccccl}
\ker(g)
& = &
\ker(f \after g)
& = &
\ker(h \after \xi_{p})
& \geq &
\ker(\xi_{p})
& = &
p.
\end{array}$$

\noindent Hence there is a unique map $k\colon X/p \rightarrow Y$ with
$k \after \xi_{p} = g$. We obtain $f \after k = h$ by using that
$\xi_{p}$ is epic. \QED
\end{proof}

\section{Boolean effectuses and extensive categories}\label{sec:extcat}

The main result of this section says that a Boolean effectus with
comprehension is the same thing as an extensive category. This is an
important coincidence of two categorical notions, which gives a deeper
understanding of what Boolean effectuses are. These extensive
categories are well-known from the literature~\cite{CarboniLW93} and
capture well-behaved coproducts.

In this section we shall work in the `total' language of effectuses,
and not the `partial' version. The reason is that we will work towards
a correspondence between extensive categories and certain effectuses
in total form--- and not between extensive categories and effectuses
in partial form.

We recall from Definition~\ref{def:commbool} that an effectus is
Boolean if it is commutative --- with assert maps $\asrt_{p} \colon X
\to X+1$ for predicates $p\colon X \to 1+1$ as inverse to $\kerbot$
--- which satisfy $\asrt_{p} \pafter \asrt_{p^\bot} = \zero$. We start
with an auxiliary result about comprehension in commutative
effectuses.

\begin{lemma}
\label{lem:commcmprasrt}
In a commutative effectus $\cat{B}$ with comprehension the following
two diagrams are equalisers in $\cat{B}$.
$$\xymatrix@C+0pc{
\cmpr{X}{p^\bot}\ar@{ >->}[r]^-{\pi_{p^\bot}} &
   X\ar@/^1.5ex/[r]^-{\asrt_{p}}\ar@/_1.5ex/[r]_-{\zero} & X\rlap{$+1$}
& &
\cmpr{X}{p}\ar@{ >->}[r]^-{\pi_{p}} &
   X\ar@/^1.5ex/[r]^-{\asrt_p}\ar@/_1.5ex/[r]_-{\kappa_{1}} & X\rlap{$+1$}
}$$
\end{lemma}

\begin{proof}
The equaliser on the left in $\cat{B}$ is in essence the kernel map
property in $\Par(\cat{B})$ from
Lemma~\ref{lem:partcmpr}~\eqref{lem:partcmpr:kermap}. For the one on
the right we recall that $\asrt_{p} \leq \idmap$ in the homset of
partial maps $X\pto X$, by definition. Hence:
$$\begin{array}{rcccccl}
\asrt_{p} \after \pi_{p}
& = &
\asrt_{p} \pafter \klin{\pi_{p}}
& \leq &
\idmap \pafter \klin{\pi_{p}}
& = &
\kappa_{1} \after \pi_{p}.
\end{array}$$

\noindent We wish to show that this inequality is actually an
equality.  So let $f\colon X\pto X$ satisfy $(\asrt_{p} \after
\pi_{p}) \ovee f = (\kappa_{1} \after \pi_{p})$. Applying the map of
effect algebras $\kerbot = (\bang+\idmap) \after (-)$ on both sides
yields in $\Pred(\cmpr{X}{p})$,
$$\begin{array}{rcll}
\one 
\hspace*{\arraycolsep}=\hspace*{\arraycolsep}
\kappa_{1} \after \bang
\hspace*{\arraycolsep}=\hspace*{\arraycolsep}
\kappa_{1} \after \bang \after \pi_{p}
& = &
\kerbot(\kappa_{1} \after \pi_{1}) \\
& = &
\kerbot(\asrt_{p} \after \pi_{p}) \ovee \kerbot(f) \\
& = &
((\bang+\idmap) \after \asrt_{p} \after \pi_{p}) \ovee \kerbot(f) \\
& = &
(p \after \pi_{p}) \ovee \kerbot(f) \\
& = &
\one \ovee \kerbot(f).
\end{array}$$

\noindent Hence $\kerbot(f) = \zero$ by cancellation in the effect
algebra $\Pred(\cmpr{X}{p})$, and thus $f = \zero$ by
Lemma~\ref{lem:ker}~\eqref{lem:ker:zero}.  But then $\kappa_{1} \after
\pi_{p} = (\asrt_{p} \after \pi_{p}) \ovee f = \asrt_{p} \after
\pi_{p}$.

Next, let $f\colon Y \to X$ in $\cat{B}$ satisfy $\asrt_{p} \after f
= \kappa_{1} \after f$. Then:
$$\begin{array}{rcl}
\tbox{f}(p)
\hspace*{\arraycolsep}=\hspace*{\arraycolsep}
p \after f
& = &
(\bang+\idmap) \after \asrt_{p} \after f \\
& = &
(\bang+\idmap) \after \kappa_{1} \after f
\hspace*{\arraycolsep}=\hspace*{\arraycolsep}
\kappa_{1} \after \bang \after f
\hspace*{\arraycolsep}=\hspace*{\arraycolsep}
\kappa_{1} \after \bang
\hspace*{\arraycolsep}=\hspace*{\arraycolsep}
\one.
\end{array}$$

\noindent Hence $f$ factors in a unique way through $\pi_{p}$. \QED

\auxproof{
We elaborate the one on the left in detail in $\cat{B}$, instead of in
$\Par(\cat{B})$, like in
Lemma~\ref{lem:partcmpr}~\eqref{lem:partcmpr:kermap}. We have:
$$\begin{array}{rcccccl}
1 \pafter (\asrt_{p} \after \pi_{p^\bot})
& = &
(1 \pafter \asrt_{p}) \after \pi_{p^\bot}
& = &
p \after \pi_{p^\bot}
& = &
0.
\end{array}$$

\noindent Hence $\asrt_{p} \after \pi_{p^\bot} = 0$. Further, if $f\colon
Y \rightarrow X$ satisfies $\asrt_{p} \after f = 0 \after f = 0$, then:
$$\begin{array}{rcccccccccl}
\tbox{f}(p^{\bot})
& = &
[\kappa_{2}, \kappa_{1}] \after p \after f
& = &
\big((\bang+\idmap) \after \asrt_{p} \after f\big)^{\bot}
& = &
\big((\bang+\idmap) \after 0\big)^{\bot}
& = &
0^{\bot}
& = &
1.
\end{array}$$

\noindent Hence $f$ factors uniquely through $\pi_{p^\bot}$.
}
\end{proof}

We turn to extensive categories. There are several equivalent
formulations, see~\cite{CarboniLW93}, but we use the most standard
one.

\begin{definition}
\label{def:extensive}
A category is called \emph{extensive}\index{S}{extensive
  category}\index{S}{category!extensive --} if it has finite
coproducts $(+,0)$ such that pullbacks along coprojections exist, and
in every diagram of the form,
$$\xymatrix{
Z_{1}\ar[r]\ar[d] & Z\ar[d] & Z_{2}\ar[l]\ar[d]
\\
X\ar[r]_-{\kappa_{1}} & X+Y & Y\ar[l]^-{\kappa_{2}}
}$$

\noindent the two rectangles are pullbacks if and only if the top row
is a coproduct, that is, the induced map $Z_{1}+Z_{2} \rightarrow Z$
is an isomorphism.
\end{definition}

There are many examples of extensive categories, with `well-behaved'
coproducts. For instance, every topos --- including $\Sets$ --- is an
extensive category, see \textit{e.g.}~\cite{CarboniLW93}. At the end
of this section we list some concrete examples.

For the record we recall the following basic observation
from~\cite{CarboniLW93}. The first point is the anologue of
Proposition~\ref{prop:effectuscoproj} for effectuses.

\begin{lemma}
\label{lem:extensive}
In an extensive category, 
\begin{enumerate}
\item \label{lem:extensive:coproj} coprojections are are monic and
  disjoint, and the initial object $0$ is strict;

\item \label{lem:extensive:pb} If the rectangles on the left below are
  pullbacks, for $i=1,2$, then the rectangle on the right is a pulback
  too.
$$\xymatrix@C+0.5pc{
A_{i}\ar[d]_{f_i}\ar[r]^-{g_i}\pullback & X\ar[d]^{f}
& &
A_{1}+A_{2}\ar[d]_{f_{1}+f_{2}}\ar[r]^-{[g_{1},g_{2}]}\pullback & X\ar[d]^{f}
\\
B_{i}\ar[r]_-{h_i} & Y
& &
B_{1}+B_{2}\ar[r]_-{[h_{1},h_{2}]} & Y
}$$

\end{enumerate}
\end{lemma}

\begin{proof}
For the first point, consider the rectangles:
$$\xymatrix{
0\ar[r]\ar[d] & Y\ar[d]^-{\kappa_2} & Y\ar@{=}[l]\ar@{=}[d]
\\
X\ar[r]_-{\kappa_{1}} & X+Y & Y\ar[l]^-{\kappa_{2}}
}$$

\noindent The top row is a coproduct diagram. Hence the two rectangles
are pullbacks. The one on the left says that coprojections are
disjoint, and the one on the right says that $\kappa_{2}$ is monic.

For strictness of $0$, we have to prove that a map $f\colon X
\rightarrow 0$ is an isomorphism. Consider the diagram:
$$\xymatrix{
X\ar@{=}[r]\ar[d]_{f}\pullback & 
   X\ar[d]|-{\kappa_{1} \after f = \kappa_{2} \after f} & 
   X\ar@{=}[l]\ar[d]^{f}\pullback[dl]
\\
0\ar[r]_-{\kappa_{1}} & 0+0 & Y\ar[l]^-{\kappa_{2}}
}$$

\noindent The two rectangles are pullbacks because the coprojections
are monic. Hence the top row is a coproduct diagram, so that the
codiagonal $\nabla = [\idmap,\idmap] \colon X+X \rightarrow X$ is an
isomorphism. Using $\nabla \after \kappa_{1} = \idmap = \nabla \after
\kappa_{2}$ we get $\kappa_{1} = \kappa_{2} \colon X \rightarrow X+X$.
Now we can now prove that $\bang\after f \colon X \rightarrow 0
\rightarrow X$ is the identity, via:
$$\begin{array}{rcccccl}
\bang \after f
& = &
[\bang\after f, \idmap[X]] \after \kappa_{1}
& = &
[\bang\after f, \idmap[X]] \after \kappa_{2}
& = &
\idmap[X].
\end{array}$$

\auxproof{
Let $g\colon Y \rightarrow X$ and $h\colon Y \rightarrow 0$
satisfy $\kappa_{1} \after f \after g = \kappa_{1} \after h$. Then
$f \after g = h$, so that $h$ is the unique mediating map.
}

We turn to the second point. Let $a\colon Z \rightarrow X$ and $b
\colon Z \rightarrow B_{1}+B_{2}$ satisfy $f \after a = [h_{1}, h_{2}]
\after b$. Form the pullbacks on the left, and the mediating maps
$c_{i}$ on the right:
$$\vcenter{\xymatrix{
Z_{1}\ar@{ >->}[r]^-{k_1}\ar[d]_{b_1}\pullback & 
   Z\ar[d]^{b} & 
   Z_{2}\ar@{ >->}[l]_-{k_2}\ar[d]^{b_2}\pullback[dl]
\\
B_{1}\ar@{ >->}[r]_-{\kappa_{1}} & B_{1}+B_{2} & B_{2}\ar@{ >->}[l]^-{\kappa_{2}}
}}
\qquad\mbox{and}\qquad
\vcenter{\xymatrix{
Z_{i}\ar@/_2ex/[ddr]_{b_i}\ar@/^2ex/[drr]^-{a \after k_{i}}
   \ar@{..>}[dr]^(0.6){c_i}
\\
& A_{i}\ar[d]_{f_i}\ar[r]^-{g_i}\pullback & X\ar[d]^{f}
\\
& B_{i}\ar[r]_-{h_i} & Y
}}$$

\noindent The outer diagram on the right commutes since:
$$\begin{array}{rcccccl}
f \after a \after k_{i}
& = &
[h_{1}, h_{2}] \after b \after k_{i}
& = &
[h_{1}, h_{2}] \after \kappa_{i} \after b_{i}
& = &
h_{i} \after b_{i}.
\end{array}$$

\noindent The unique mediating map is $c = (c_{1}+c_{2}) \after
[k_{1}, k_{2}]^{-1} \colon Z \rightarrow A_{1}+A_{2}$. \QED

\auxproof{
Indeed $[g_{1}, g_{2}] \after c = a$ and $(f_{1}+f_{2}) \after c = b$ since:
$$\begin{array}{rcl}
a \after [k_{1}, k_{2}]
& = &
[a \after k_{1}, a \after k_{2}] \\
& = &
[g_{1} \after c_{1}, g_{2} \after c_{2}] \\
& = &
[g_{1}, g_{2}] \after (c_{1}+c_{2})
\\
b \after [k_{1}, k_{2}]
& = &
[b \after k_{1}, b \after k_{2}] \\
& = &
[\kappa_{1} \after b_{1}, \kappa_{2} \after b_{2}] \\
& = &
[\kappa_{1} \after f_{1} \after c_{1}, \kappa_{2} \after f_{2} \after c_{2}] \\
& = &
(f_{1}+f_{2}) \after (c_{1}+c_{2}).
\end{array}$$

\noindent And if also $d \colon Z \rightarrow A_{1}+A_{2}$ satisfies
$[g_{1}, g_{2}] \after d = a$ and $(f_{1}+f_{2}) \after d = b$ then form
the pullbacks:
$$\xymatrix{
Z_{1}\ar@{ >->}[r]^-{k_1}\ar@{..>}[d]_{d_1}\pullback & 
   Z\ar[d]^{d} & 
   Z_{2}\ar@{ >->}[l]_-{k_2}\ar@{..>}[d]^{d_2}\pullback[dl]
\\
A_{1}\ar@{ >->}[r]_-{\kappa_{1}}\ar[d]_{f_1}\pullback & 
   A_{1}+A_{2}\ar[d]^{f_{1}+f_{2}} & 
   A_{2}\ar[d]^{f_2}\ar@{ >->}[l]^-{\kappa_{2}}\pullback[dl]
\\
B_{1}\ar@{ >->}[r]_-{\kappa_{1}} & B_{1}+B_{2} & B_{2}\ar@{ >->}[l]^-{\kappa_{2}}
}$$

\noindent Then $d_{i} = c_{i}$ since $f_{i} \after d_{i} = b_{i}$ by
construction, and:
$$\begin{array}{rcccccl}
g_{i} \after d_{i}
& = &
[g_{1}, g_{2}] \after \kappa_{i} \after d_{i}
& = &
[g_{1}, g_{2}] \after d \after k_{i}
& = &
a \after k_{i}.
\end{array}$$

\noindent Hence $c = d$ since:
$$\begin{array}{rcccccccl}
d \after [k_{1}, k_{2}]
& = &
[d \after k_{1}, d\after k_{2}]
& = &
[\kappa_{1} \after d_{1}, \kappa_{2} \after d_{2}]
& = &
[\kappa_{1} \after c_{1}, \kappa_{2} \after c_{2}]
& = &
c_{1}+c_{2}.
\end{array}$$
}
\end{proof}

\begin{proposition}
\label{prop:extensiveeffectus}
Each extensive category with a final object is a Boolean effectus with
comprehension.
\end{proposition}

\begin{proof}
This is quite a bit of work. So let $\cat{A}$ be an extensive category
with a final object $1$.  We start by showing that $\cat{A}$ is an
effectus. We first concentrate on the two pullbacks in
Definition~\ref{def:effectus}:
$$\vcenter{\xymatrix{
X+Y\ar[r]^-{\idmap+ \bang}\ar[d]_{\bang+ \idmap} & 
   X + 1\ar[d]^{\bang+\idmap}
& &
X\ar[r]^-{\bang}\ar@{ >->}[d]_{\kappa_1} & 1\ar@{ >->}[d]^{\kappa_1}
\\
1+ Y\ar[r]_-{\idmap+ \bang} & 1 + 1
& &
X+Y\ar[r]_-{\bang+\bang} & 1 + 1
}}\eqno{(*)}$$

\noindent For convenience, we start with the diagram on the right. We
consider it as turned left-part of the following diagram:
$$\xymatrix{
X\ar@{ >->}[r]^-{\kappa_1}\ar[d]_{\bang} & X+Y\ar[d]^{\bang+\bang} & 
   Y\ar@{ >->}[l]_-{\kappa_2}\ar[d]^{\bang}
\\
1\ar@{ >->}[r]_-{\kappa_{1}} & 1+1 & 1\ar@{ >->}[l]^-{\kappa_{2}}
}$$

\noindent Since the top row is a coproduct diagram, both rectangles
are pullbacks in an extensive category, by definition. 

For the diagram on the left in $(*)$, it suffices, by
Lemma~\ref{lem:extensive}~\eqref{lem:extensive:pb} to prove that the
two diagrams below are pullbacks.
$$\xymatrix{
X\ar[r]^-{\kappa_{1}}\ar[d]_{\bang} & X+1\ar[d]^{\bang+\idmap}
& &
Y\ar@{=}[d]\ar[r]^-{\bang} & 1\ar[r]^-{\kappa_{2}}\ar@{=}[d] & X+1\ar[d]^{\bang+\idmap}
\\
1\ar[r]_-{\kappa_{1}} & 1+1
& &
Y\ar[r]_-{\bang} & 1\ar[r]_-{\kappa_{2}} & 1+1
}$$

\noindent The rectangle on the left is a pullback just like the
diagram on the right in~$(*)$. Similarly, the right-rectangle on the
right is a pullback. The left-rectangle on the right is obviously a
pullback, so we are done by the Pullback Lemma.

\auxproof{
Explicitly: let $f\colon Z \rightarrow X+1$
and $g\colon Z \rightarrow 1+Y$ satisfy $(\bang+\idmap) \after f =
(\idmap+\bang)  \after g = h$, say. Form the diagram:
$$\xymatrix{
Z_{1}\ar@{ >->}[r]^-{k_1}\ar[d]\pullback & 
   Z\ar[d]^{h} & 
   Z_{2}\ar@{ >->}[l]_-{k_2}\ar[d]\pullback[dl]
\\
1\ar@{ >->}[r]_-{\kappa_{1}} & 1+1 & 1\ar@{ >->}[l]^-{\kappa_{2}}
}$$

\noindent Then we can form the following diagrams for $f$ and $g$.
$$\xymatrix{
Z_{1}\ar@{ >->}[r]^-{k_1}\ar@{..>}[d]_{f_1}\pullback & 
   Z\ar[d]^-{f} & 
   Z_{2}\ar@{ >->}[l]_-{k_2}\ar@{..>}[d]^{\bang}\pullback[dl]
& &
Z_{1}\ar@{ >->}[r]^-{k_1}\ar@{..>}[d]_{\bang}\pullback & 
   Z\ar[d]^-{g} & 
   Z_{2}\ar@{ >->}[l]_-{k_2}\ar@{..>}[d]^{g_2}\pullback[dl]
\\
X\ar@{ >->}[r]_-{\kappa_{1}}\ar[d]_{\bang}\pullback & 
   X+1\ar[d]^{\bang+\idmap} & 
   1\ar@{ >->}[l]^-{\kappa_{2}}\ar@{=}[d]\pullback[dl]
& &
1\ar@{ >->}[r]_-{\kappa_{1}}\ar@{=}[d]_{\bang}\pullback & 
   1+Y\ar[d]^{\idmap+\bang} & 
   Y\ar@{ >->}[l]^-{\kappa_{2}}\ar[d]^{\bang}\pullback[dl]
\\
1\ar@{ >->}[r]_-{\kappa_{1}} & 1+1 & 1\ar@{ >->}[l]^-{\kappa_{2}}
& &
1\ar@{ >->}[r]_-{\kappa_{1}} & 1+1 & 1\ar@{ >->}[l]^-{\kappa_{2}}
}$$

\noindent We can now construct as mediating map:
$$\xymatrix@C+1pc{
h = \Big(Z\ar[r]_-{\cong}^-{[k_{1},k_{2}]^{-1}} & 
   Z_{1}+Z_{2}\ar[r]^-{f_{1}+g_{2}} & X+Y\Big)
}$$

\noindent Then:
$$\begin{array}{rcl}
(\idmap+\bang) \after h
& = &
(f_{1}+\bang) \after [k_{1},k_{2}]^{-1} \\
& = &
[\kappa_{1} \after f_{1}, \kappa_{2} \after \bang] \after [k_{1},k_{2}]^{-1} \\
& = &
[f \after k_{1}, f \after k_{2}] \after [k_{1},k_{2}]^{-1} \\
& = &
f
\\
(\bang+\idmap) \after h
& = &
(\bang+g_{2}) \after [k_{1},k_{2}]^{-1} \\
& = &
[\kappa_{1} \after \bang, \kappa_{2} \after g_{2}]  \after [k_{1},k_{2}]^{-1} \\
& = &
[g \after k_{1}, g \after k_{2}]  \after [k_{1},k_{2}]^{-1} \\
& = &
g.
\end{array}$$

\noindent For uniqueness let $\ell\colon Z \rightarrow X+Y$ also
satisfy $(\idmap+\bang) \after \ell = f$ and $(\bang+\idmap) \after \ell =
g$. Then $(\bang+\bang) \after \ell = h$, so that we have a diagram:
$$\xymatrix{
Z_{1}\ar@{ >->}[r]^-{k_1}\ar@{..>}[d]_{\ell_1}\pullback & 
   Z\ar[d]^-{\ell} & 
   Z_{2}\ar@{ >->}[l]_-{k_2}\ar@{..>}[d]^{\ell_2}\pullback[dl]
\\
X\ar@{ >->}[r]_-{\kappa_{1}}\ar[d]_{\bang}\pullback & 
   X+Y\ar[d]^{\bang+\bang} & 
   Y\ar@{ >->}[l]^-{\kappa_{2}}\ar[d]^{\bang}\pullback[dl]
\\
1\ar@{ >->}[r]_-{\kappa_{1}} & 1+1 & Y\ar@{ >->}[l]^-{\kappa_{2}}
}$$

\noindent We claim $\ell_{1} = f_{1}$ and $\ell_{2} = g_{2}$. This follows
from:
$$\begin{array}{rcl}
\kappa_{1} \after f_{1}
& = &
f \after k_{1} \\
& = &
(\idmap+\bang) \after \ell \after k_{1} \\
& = &
(\idmap+\bang) \after \kappa_{1} \after \ell_{1} \\
& = &
\kappa_{1} \after \ell_{1}
\\
\kappa_{2} \after g_{2}
& = &
g \after k_{2} \\
& = &
(\bang+\idmap) \after \ell \after k_{2} \\
& = &
(\bang+\idmap) \after \kappa_{2} \after \ell_{2} \\
& = &
\kappa_{2} \after \ell_{2}.
\end{array}$$

\noindent Hence $\ell = h = (f_{1}+g_{2}) \after [k_{1},k_{2}]^{-1}$ since:
$$\begin{array}{rcl}
\ell \after [k_{1}, k_{2}]
& = &
[\ell \after k_{1}, \ell \after k_{2}] \\
& = &
[\kappa_{1} \after \ell_{1}, \kappa_{2} \after \ell_{2}] \\
& = &
[\kappa_{1} \after f_{1}, \kappa_{2} \after g_{2}] \\
& = &
f_{1}+g_{2}.
\end{array}$$

We continue with the diagram on the right in $(*)$. Let $f\colon Z
\rightarrow X+Y$ satisfy $(\bang+\bang) \after f = \kappa_{1} \after \bang$. We
form the diagram:
$$\xymatrix{
Z_{1}\ar@{ >->}[r]^-{k_1}\ar[d]_{f_1}\pullback & 
   Z\ar[d]^-{f} & 
   Z_{2}\ar@{ >->}[l]_-{k_2}\ar[d]^{f_2}\pullback[dl]
\\
X\ar@{ >->}[r]_-{\kappa_{1}}\ar[d]_{\bang}\pullback & 
   X+Y\ar[d]^{\bang+\bang} & 
   Y\ar@{ >->}[l]^-{\kappa_{2}}\ar[d]^{\bang}\pullback[dl]
\\
1\ar@{ >->}[r]_-{\kappa_{1}} & 1+1 & Y\ar@{ >->}[l]^-{\kappa_{2}}
}$$

\noindent The lower two rectangles are pullbacks by definition (in an
extensive category). The upper two rectangles are constructed as
pullbacks. As a result, the cotuple $[k_{1},k_{2}] \colon Z_{1}+Z_{2}
\rightarrow Z$ is an isomorphism. We claim $Z_{2} \cong 0$. This
follows by Lemma~\ref{lem:extensive}~\eqref{lem:extensive:coproj}
because coprojections are disjoint, and $0$ is strict, in:
$$\xymatrix{
Z_{2}\ar@/^2ex/[drr]^-{\bang}\ar@/_2ex/[ddr]_{\bang}\ar@{..>}[dr]^-{\cong}
\\
& 0\ar[d]\ar[r]\pullback & 1\ar@{ >->}[d]^{\kappa_1}
\\
& 1\ar@{ >->}[r]_-{\kappa_{2}} & 1+1
}$$

\noindent The outer diagram commutes since:
$$\begin{array}{rcccccccl}
\kappa_{1} \after \bang
& = &
\kappa_{1} \after \bang \after k_{2}
& = &
(\bang+\bang) \after f \after k_{2}
& = &
\kappa_{2} \after  \bang \after f_{2}
& = &
\kappa_{2} \after \bang.
\end{array}$$

\noindent If $Z_{2} \cong 0$, then $k_{1} \colon Z_{1} \rightarrow Z$
is an isomorphism. Hence we obtain as (unique) mediating map $f_{1}
\after k_{1}^{-1} \colon Z \rightarrow X$, since:
$$\begin{array}{rcccl}
\kappa_{1} \after f_{1} \after k_{1}^{-1}
& = &
f \after k_{1} \after k_{1}^{-1}
& = &
f.
\end{array}$$
}

The next step is to prove that the two maps $\IV =
[\idmap,\kappa_{2}], \XI = [[\kappa_{2},\kappa_{1}],\kappa_{2}]$ are
jointly monic. We sketch how to proceed. If $f,g\colon Y \rightarrow
(1+1)+1$ satisfy $\IV \after f = \IV \after g$ and $\XI \after f = \XI
\after g$ then we decompose $f,g$ each in three parts, via pulbacks
with approriate coprojections $1\rightarrow (1+1)+1$, and show that
these three parts are equal. This is left to the interested reader.

\auxproof{
Next we wish to show that the maps $\IV = [\idmap,\kappa_{2}], \XI =
[[\kappa_{2},\kappa_{1}],\kappa_{2}]$ are jointly monic. Let
$f,g\colon Y \rightarrow (1+1)+1$ be maps with $\IV \after f = \IV
\after g = h$, say, and $\XI \after f = \XI \after g = h'$, say. We
form the diagrams:
$$\xymatrix{
Y_{1}\ar@{ >->}[r]^-{k_1}\ar[d]\pullback & Y\ar[d]^{h} & 
   Y_{2}\ar@{ >->}[l]_-{k_2}\ar[d]\pullback[dl]
& &
Y'_{1}\ar@{ >->}[r]^-{k'_1}\ar[d]\pullback & Y\ar[d]^{h'} & 
   Y'_{2}\ar@{ >->}[l]_-{k'_2}\ar[d]\pullback[dl]
\\
1\ar@{ >->}[r]_-{\kappa_{1}} & 1+1 & 1\ar@{ >->}[l]^-{\kappa_{2}}
& &
1\ar@{ >->}[r]_-{\kappa_{1}} & 1+1 & 1\ar@{ >->}[l]^-{\kappa_{2}}
}$$

\noindent Then $f_{11} = g_{11}$ and $f_{12} = g_{12}$ in:
$$\xymatrix@C+.5pc{
Y_{1}\ar@{..>}@<-.5ex>[d]_{f_{11}}\ar@{..>}@<+.5ex>[d]^{g_{11}}
   \ar@{ >->}[r]^-{k_1} & 
   Y\ar@<-.5ex>[d]_{f}\ar@<+.5ex>[d]^{g} 
& &
Y'_{1}\ar@{..>}@<-.5ex>[d]_{f_{12}}\ar@{..>}@<+.5ex>[d]^{g_{12}}
   \ar@{ >->}[r]^-{k'_1} & 
   Y\ar@<-.5ex>[d]_{f}\ar@<+.5ex>[d]^{g} 
\\
1\ar@{ >->}[r]_-{\kappa_{1}\after\kappa_{1}}\ar[d]_{\bang}\pullback & 
   (1+1)+1\ar[d]^{\IV} 
& &
1\ar@{ >->}[r]_-{\kappa_{1}\after\kappa_{2}}\ar[d]_{\bang}\pullback & 
   (1+1)+1\ar[d]^{\XI} 
\\
1\ar@{ >->}[r]_-{\kappa_{1}} & 1+1
& &
1\ar@{ >->}[r]_-{\kappa_{1}} & 1+1
}$$

\noindent The two lower rectangles are pullbacks since we can
reorganise them as standard pullbacks:
$$\xymatrix@R-.5pc@C+1pc{
1\ar@{=}[d]\ar@{ >->}[r]_-{\kappa_{1}\after\kappa_{1}}
   \ar@/^4ex/[rr]^{\kappa_1}\pullback & 
   (1+1)+1\ar[d]^{[\idmap,\kappa_{2}]}\ar[r]^-{\alpha^{-1}}_-{\cong} &
   1 + (1+1)\ar@/^2ex/[dl]^{\idmap+\nabla} \\
1\ar@{ >->}[r]_-{\kappa_{1}} & 1+1
}$$

$$\xymatrix@R-.5pc@C+1pc{
1\ar@{=}[d]\ar@{ >->}[r]_-{\kappa_{1}\after\kappa_{2}}
   \ar@/^4ex/[rrr]^{\kappa_1}\pullback & 
   (1+1)+1\ar[d]^{[[\kappa_{2},\kappa_{1}], \kappa_{2}]}
      \ar[rr]^-{[[\kappa_{2}\after\kappa_{1},\kappa_{1}],
   \kappa_{2}\after\kappa_{2}]}_-{\cong} & &
   1 + (1+1)\ar@/^2ex/[dll]^{\idmap+\nabla} \\
1\ar@{ >->}[r]_-{\kappa_1} & 1+1
}$$

\noindent By Lemma~\ref{lem:extensive}~\eqref{lem:extensive:pb} we
then obtain pullbacks of the following form, both for $f$ and $g$,
$$\xymatrix@C+.5pc{
Y_{1}+Y'_{1}\ar[d]_{\bang+\bang}\ar@{ >->}[r]^-{[k_{1},k'_{1}]}\pullback & 
   Y\ar@<-.5ex>[d]_{f}\ar@<+.5ex>[d]^{g} 
\\
1+1\ar@{ >->}[r]_-{\kappa_{1}} & (1+1)+1
}$$

\noindent This shows that their `left' parts are equal. We now
concentrate on their `right' parts. Consider:
$$\xymatrix@C+.5pc{
Y\ar@<-.5ex>[d]_{f}\ar@<+.5ex>[d]^{g} &
   Y_{2}\ar@{..>}@<-.5ex>[d]_{f_{2}}\ar@{..>}@<+.5ex>[d]^{g_{2}}
   \ar@{ >->}[l]_-{k_2} 
& &
Y\ar@<-.5ex>[d]_{f}\ar@<+.5ex>[d]^{g} &
   Y'_{2}\ar@{..>}@<-.5ex>[d]_{f'_{2}}\ar@{..>}@<+.5ex>[d]^{g'_{2}}
   \ar@{ >->}[l]_-{k'_2} 
\\
(1+1)+1\ar[d]_{\IV} & 
   1+1\ar@{ >->}[l]^-{\kappa_{2}+\idmap}\ar[d]^{\nabla}\pullback[dl]
& &
(1+1)+1\ar[d]_{\XI} & 
   1+1\ar@{ >->}[l]^-{\kappa_{1}+\idmap}\ar[d]^{\nabla}\pullback[dl]
\\
1+1 & 1\ar@{ >->}[l]^-{\kappa_{2}}
& &
1+1 & 1\ar@{ >->}[l]^-{\kappa_{2}}
}$$

\noindent The lower rectangles are pullbacks, since they can be 
reorganised as:
$$\xymatrix@R-.5pc@C+1pc{
1+1\ar[d]_{\nabla}\ar@{ >->}[r]_-{\kappa_{2}+\idmap}\ar@/^4ex/[rr]^{\kappa_2}
   \pullback & 
   (1+1)+1\ar[d]^{[\idmap,\kappa_{2}]}\ar[r]^-{\alpha^{-1}}_-{\cong} &
   1 + (1+1)\ar@/^2ex/[dl]^{\idmap+\nabla} \\
1\ar@{ >->}[r]_-{\kappa_{2}} & 1+1
}$$

$$\xymatrix@R-.5pc@C+1pc{
1+1\ar[d]_{\nabla}\ar@{ >->}[r]_-{\kappa_{1}+\idmap}
   \ar@/^4ex/[rrr]^{\kappa_2}\pullback & 
   (1+1)+1\ar[d]^{[[\kappa_{2},\kappa_{1}], \kappa_{2}]}
      \ar[rr]^-{[[\kappa_{2}\after\kappa_{1},\kappa_{1}],
   \kappa_{2}\after\kappa_{2}]}_-{\cong} & &
   1 + (1+1)\ar@/^2ex/[dll]^{\idmap+\nabla} \\
1\ar@{ >->}[r]_-{\kappa_2} & 1+1
}$$

\noindent We now form the pullback:
$$\xymatrix{
Z\ar@{ >->}[d]_{\ell}\ar@{ >->}[r]^-{\ell'}\pullback & Y'_{2}\ar@{ >->}[d]^{k'_2}
\\
Y_{2}\ar@{ >->}[r]_-{k_2} & Y
}$$

\noindent We also use that we have pullback:
$$\xymatrix@R-.5pc@C+1pc{
1\ar@{ >->}[d]_{\kappa_2}\ar@{ >->}[r]^-{\kappa_2}\pullback &
   1+1\ar@{ >->}[d]^{\kappa_{2}+\idmap}\ar[dr]^{\kappa_2} \\
1+1\ar@{ >->}[r]^-{\kappa_{1}+\idmap}\ar@/_4ex/[rr]_{\idmap+\kappa_2} & 
   (1+1)+1\ar[r]^-{\alpha^{-1}}_-{\cong} &
   1 + (1+1)
}$$

\noindent Now we obtain:
$$\xymatrix{
Z\ar@/_2ex/[ddr]_{f_{2}\after\ell}\ar@/^2ex/[drr]^-{f'_{2}\after\ell'}
   \ar@{..>}[dr] & &
& 
Z\ar@/_2ex/[ddr]_{g_{2}\after\ell}\ar@/^2ex/[drr]^-{g'_{2}\after\ell'}
   \ar@{..>}[dr]
\\
& 1\ar[d]^{\kappa_2}\ar[r]_{\kappa_2}\pullback & 1+1\ar[d]^{\kappa_{1}+\idmap}
& 
& 1\ar[d]^{\kappa_2}\ar[r]_{\kappa_2}\pullback & 1+1\ar[d]^{\kappa_{1}+\idmap}
\\
& 1+1\ar[r]_-{\kappa_{2}+\idmap} & (1+1)+1
& 
& 1+1\ar[r]_-{\kappa_{2}+\idmap} & (1+1)+1
}$$

\noindent The outer diagram commutes since:
$$\begin{array}{rcccccl}
(\kappa_{2}+\idmap) \after f_{2} \after \ell
& = &
f \after k_{2} \after \ell
& = &
f \after k'_{2} \after \ell'
& = &
(\kappa_{2}+\idmap) \after f'_{2} \after \ell'.
\end{array}$$

\noindent And similarly for $g$. We now claim that we have pullbacks:
$$\xymatrix@C+.5pc{
Y\ar@<-.5ex>[d]_{f}\ar@<+.5ex>[d]^{g} & Z\ar[d]^{\bang}
   \ar@{ >->}[l]_-{k_{2}\after\ell = k'_{2}\after\ell'} 
\\
(1+1)+1 & 1\ar@{ >->}[l]^-{\kappa_{2}}
}$$

\noindent We shall prove this for $f$. First, the diagram commutes since:
$$\begin{array}{rcccccl}
\kappa_{2} \after \bang
& = &
(\kappa_{2}+\idmap) \after \kappa_{2} \after \bang
& = &
(\kappa_{2}+\idmap) \after f_{2} \after \ell
& = &
f \after k_{2} \after \ell.
\end{array}$$

\noindent Next, it is a pullback: if $a\colon W \rightarrow Y$
satisfies $f \after a = \kappa_{2} \after \bang$. Then:
$$\begin{array}{rcccccl}
h \after a
& = &
\IV \after f \after a
& = &
\IV \after \kappa_{2} \after \bang
& = &
\kappa_{2} \after \bang.
\end{array}$$

\noindent Hence there is a unique $a_{2} \colon W \rightarrow Y_{2}$
with $k_{2} \after a_{2} = a$. Similarly:
$$\begin{array}{rcccccl}
h' \after a
& = &
\XI \after f \after a
& = &
\XI \after \kappa_{2} \after \bang
& = &
\kappa_{2} \after \bang.
\end{array}$$

\noindent Hence there is a unique $a'_{2} \colon W \rightarrow Y'_{2}$
with $k_{2} \after a'_{2} = a$. But then there is a unique $a_{1} \colon
W \rightarrow Z$ with $\ell \after a_{1} = a_{2}$ and $\ell' \after a_{1}
= a'_{2}$. This $a_{1}$ is the (unique) map that we seek, since:
$$\begin{array}{rcccl}
k_{2} \after \ell \after a_{1}
& = &
k_{2} \after a_{2}
& = &
a.
\end{array}$$

We can thus conclude that the `right' parts of $f$ and $g$ are also
the same. Hence $f=g$.
}

We now know that the extensive category $\cat{A}$ that we started from
is an effectus. We turn to predicates and comprehension. For each
predicate $p\colon X \rightarrow 1+1$ in $\cat{A}$ we choose a
pullback:
\begin{equation}
\label{diag:extensivecmpr}
\vcenter{\xymatrix{
\cmpr{X}{p}\ar[d]\ar@{ >->}[r]^-{\pi_p}\pullback & X\ar[d]^{p}
\\
1\ar@{ >->}[r]_-{\kappa_1} & 1+1
}}
\end{equation}

\noindent We will need a number of facts about these maps 
$\pi_{p} \colon \cmpr{X}{p} \rightarrowtail X$.
{\renewcommand{\theenumi}{(\alph{enumi})}
\begin{enumerate}
\item \label{prop:extensiveeffectus:pb} The following diagrams are
  also pullbacks.
$$\xymatrix{
X\ar[d]_{p} & \cmpr{X}{p^\bot}\ar@{ >->}[l]_-{\pi_{p^\bot}}\ar[d]
& &
\cmpr{X}{p}+Y\ar[d]_{[\bang,\bang]=\bang}
   \ar@{ >->}[r]^-{\pi_{p}+\idmap} & X+Y\ar[d]^{[p,\one]}
\\
1+1 & 1\ar@{ >->}[l]^-{\kappa_2}
& &
1\ar@{ >->}[r]_-{\kappa_1} & 1+1
}$$

\noindent The combination of the left diagram and
Diagram~\eqref{diag:extensivecmpr} tells us that the cotuple
$[\pi_{p}, \pi_{p^\bot}] \colon \cmpr{X}{p}+\cmpr{X}{p^\bot}
\rightarrow X$ is an isomorphism. The consequence of the pullback on
the right is that $\cmpr{X\,}{\,[p,\one]} \cong \cmpr{X}{p}+Y$, so
that the effectus $\cat{A}$ has comprehension, see
Definition~\ref{def:cmpr}.

It is easy to see that the above diagram on the left is a pullback: if
$f\colon Y \rightarrow X$ satisfies $p \after f = \kappa_{2} \after
\bang$, then:
$$\begin{array}{rcccccccl}
p^{\bot} \after f
& = &
[\kappa_{2},\kappa_{1}] \after p \after f
& = &
[\kappa_{2},\kappa_{1}] \after \kappa_{2} \after \bang
& = &
\kappa_{1} \after \bang
& = &
\one.
\end{array}$$

\noindent Hence $f$ factors uniquely through $\pi_{p^\bot}$
in~\eqref{diag:extensivecmpr}.

The above rectangle on the right is a pullback diagram since one can
apply Lemma~\ref{lem:extensive}~\eqref{lem:extensive:pb} to
Diagram~\eqref{diag:extensivecmpr} and a trivial pullback.

\item \label{prop:extensiveeffectus:jm} The pair of comprehension maps
  $\pi_{p}, \pi_{p^\bot}$ is jointly epic.

Indeed, if $f \after \pi_{p} = g \after \pi_{p}$ and $f \after \pi_{p^\bot}
= g \after \pi_{p^\bot}$, then:
$$\begin{array}{rcccccl}
f \after [\pi_{p}, \pi_{p^\bot}]
& = &
[f \after \pi_{p}, f \after \pi_{p^\bot}]
& = &
[g \after \pi_{p}, g \after \pi_{p^\bot}]
& = &
g \after [\pi_{p}, \pi_{p^\bot}].
\end{array}$$

\noindent Hence $f=g$ since the cotuple $[\pi_{p}, \pi_{p^\bot}]$ is
an isomorphism.
\end{enumerate}}

We now define for each predicate $p\colon X \rightarrow 1+1$ an
`assert' map $X \rightarrow X+1$ as:
$$\xymatrix{
\asrt_{p} \;\defeq\; \Big(X\ar[rr]^-{[\pi_{p}, \pi_{p^\bot}]^{-1}} & &
   \cmpr{X}{p}+\cmpr{X}{p^{\bot}}\ar[r]^-{\pi_{p}+\bang} & X+1\Big)
}$$

\noindent Again we prove a number of basic facts.
{\renewcommand{\theenumi}{(\roman{enumi})}
\begin{enumerate}
\item \label{prop:extensiveeffectus:asrtproj} $\asrt_{p} \after \pi_{p} =
  \kappa_{1} \after \pi_{p}$ and $\asrt_{p} \after \pi_{p^\bot} = \zero$.

From $[\pi_{p}, \pi_{p^\bot}] \after \kappa_{1} = \pi_{p}$ we
obtain $[\pi_{p}, \pi_{p^\bot}]^{-1} \after \pi_{p} = \kappa_{1}$. Hence:
$$\begin{array}{rcccccl}
\asrt_{p} \after \pi_{p}
& = &
(\pi_{p}+\bang) \after [\pi_{p}, \pi_{p^\bot}]^{-1} \after \pi_{p} 
& = &
(\pi_{p}+\bang) \after \kappa_{1}
& = &
\kappa_{1} \after \pi_{p}.
\end{array}$$

\noindent Similarly:
$$\begin{array}{rcccccccl}
\asrt_{p} \after \pi_{p^\bot}
& = &
(\pi_{p}+\bang) \after [\pi_{p}, \pi_{p^\bot}]^{-1} \after \pi_{p^\bot} 
& = &
(\pi_{p}+\bang) \after \kappa_{2}
& = &
\kappa_{2} \after \bang
& = &
\zero.
\end{array}$$

\item \label{prop:extensiveeffectus:asrtbelow} $\asrt_{p} \leq \idmap$
  in the homset of partial maps $X\pto X$, since $\asrt_{p} \ovee
  \asrt_{p^{\bot}} = \idmap$, via the bound $b = \kappa_{1} \after
  (\pi_{p}+\pi_{p^{\bot}}) \after [\pi_{p}, \pi_{p^\bot}]^{-1} \colon
  X \rightarrow (X+X)+1$. Indeed:
$$\begin{array}{rcl}
\rhd_{1} \pafter b
& = &
(\idmap+\bang) \after (\pi_{p}+\pi_{p^{\bot}}) \after 
   [\pi_{p}, \pi_{p^\bot}]^{-1} \\
& = &
(\pi_{p}+\bang) \after [\pi_{p}, \pi_{p^\bot}]^{-1} \\
& = &
\asrt_{p}
\\
\rhd_{2} \pafter b
& = &
[\kappa_{2}\after\bang,\kappa_{1}] \after (\pi_{p}+\pi_{p^{\bot}}) \after 
   [\pi_{p}, \pi_{p^\bot}]^{-1} \\
& = &
[\kappa_{2}\after\bang,\kappa_{1}\after\pi_{p^\bot}] \after 
   [\pi_{p}, \pi_{p^\bot}]^{-1} \\
& = &
(\pi_{p^{\bot}}+\bang) \after [\kappa_{2},\kappa_{1}] \after 
   [\pi_{p}, \pi_{p^\bot}]^{-1} \\
& = &
(\pi_{p^{\bot}}+\bang) \after [\pi_{p^\bot}, \pi_{p}]^{-1} \\
& = &
\asrt_{p^\bot}
\\
\asrt_{p} \ovee \asrt_{p^\bot}
& = &
\nabla \pafter b \\
& = &
\kappa_{1} \after \nabla \after (\pi_{p}+\pi_{p^{\bot}}) \after 
   [\pi_{p}, \pi_{p^\bot}]^{-1} \\
& = &
\kappa_{1}
\end{array}$$

\item \label{prop:extensiveeffectus:asrtiso} This mapping $p\mapsto
  \asrt_{p}$ is the inverse of $\kerbot \colon \sEnd(X) \rightarrow
  \Pred(X)$. First, we have:
$$\begin{array}{rcl}
\kerbot(\asrt_{p})
& = &
(\bang+\idmap) \after (\pi_{p}+\bang) \after [\pi_{p}, \pi_{p^\bot}]^{-1} \\
& = &
(\bang+\bang) \after [\pi_{p}, \pi_{p^\bot}]^{-1} \\
& = &
p.
\end{array}$$

\noindent Conversely, let $f\colon X \pto X$ satisfy $f\leq \idmap$.
We need to prove $f = \asrt_{p}$, for $p = \kerbot(f) = \one \pafter f =
(\bang+\idmap) \after f$. Consider the diagram:
$$\xymatrix{
\cmpr{X}{p}\ar@{ >->}[r]^-{\pi_{p}}\ar@{..>}[d]_{h}\pullback & 
   X\ar[d]^-{f} & 
   \cmpr{X}{p^\bot}\ar@{ >->}[l]_-{\pi_{p^\bot}}\ar@{..>}[d]^{\bang}\pullback[dl]
\\
X\ar@{ >->}[r]_-{\kappa_{1}}\ar[d]_{\bang}\pullback & 
   X+1\ar[d]^{\bang+\idmap} & 
   1\ar@{ >->}[l]^-{\kappa_{2}}\ar@{=}[d]\pullback[dl]
\\
1\ar@{ >->}[r]_-{\kappa_{1}} & 1+1 & Y\ar@{ >->}[l]^-{\kappa_{2}}
}$$

\noindent We first show $h = \pi_{p}$. Since $f \leq \idmap$ we can
write $f \ovee g = \idmap$ for some map $g\colon X\pto X$. Then:
$$\begin{array}{rcccccccl}
\one
& = &
\one \pafter \idmap
& = &
\one \pafter (f\ovee g)
& = &
(\one \pafter f) \ovee (\one \pafter g)
& = &
p \ovee (\one \pafter g).
\end{array}$$

\noindent Hence $\one \pafter g = p^{\bot}$. Then: $\one \pafter (g
\after \pi_{p}) = p^{\bot} \after \pi_{p} = \zero$, so that $g\after
\pi_{p} = \zero$ by Lemma~\ref{lem:zero}. We now get $\pi_{p} = h$
from:
$$\begin{array}{rcccccccl}
\kappa_{1} \after \pi_{p}
& = &
(f \ovee g) \after \pi_{p}
& = &
(f \after \pi_{p}) \ovee (g \after \pi_{p})
& = &
(\kappa_{1} \after h) \ovee \zero
& = &
\kappa_{1} \after h.
\end{array}$$

\noindent Finally we obtain $f = \asrt_{p}$ from:
$$\begin{array}{rcccccl}
f \after [\pi_{p}, \pi_{p^\bot}]
& = &
[f \after \pi_{p}, f \after \pi_{p^\bot}]
& = &
[\kappa_{1} \after \pi_{p}, \kappa_{2} \after \bang]
& = &
\pi_{p}+\bang.
\end{array}$$

\item \label{prop:extensiveeffectus:asrtboolean} $\asrt_{p} \pafter
  \asrt_{p^\bot} = \zero$.

We simply calculate:
$$\begin{array}{rcl}
\asrt_{p} \pafter \asrt_{p^\bot}
& = &
[\asrt_{p}, \kappa_{2}] \after (\pi_{p^\bot}+\bang) \after
   [\pi_{p^\bot}, \pi_{p}]^{-1} \\
& = &
[\asrt_{p} \after \pi_{p^\bot}, \kappa_{2} \after \bang] \after 
   [\pi_{p^\bot}, \pi_{p}]^{-1} \\
& \smash{\stackrel{\ref{prop:extensiveeffectus:asrtproj}}{=}} &
[\kappa_{2} \after \bang, \kappa_{2} \after \bang] \after 
   [\pi_{p^\bot}, \pi_{p}]^{-1} \\
& = &
\zero \after [\pi_{p^\bot}, \pi_{p}]^{-1} \\
& = &
\zero.
\end{array}$$

\item \label{prop:extensiveeffectus:asrtcomm} Commutation of asserts
  --- that is, $\asrt_{p} \pafter \asrt_{q} = \asrt_{q} \pafter
  \asrt_{p}$ --- follows from the previous point, as described in
  Remark~\ref{rem:BooleanEffectusCommutation}.

\auxproof{
The proof of this equation requires some work. Since the mapping
$\kerbot$ is an isomorphism it suffices to prove
$\kerbot\big(\asrt_{p} \pafter \asrt_{q}\big) = \kerbot\big(\asrt_{q}
\pafter \asrt_{p}\big)$. This amounts to the equation:
$$\begin{array}{rcl}
[p,\kappa_{2}] \after \asrt_{q}
& = &
[q,\kappa_{2}] \after \asrt_{p}.
\end{array}$$

\noindent Since the pair $\pi_{p}, \pi_{p^\bot}$ is jointly epic ---
see point~\ref{prop:extensiveeffectus:jm} above --- we are done if
we can prove both:
$$\begin{array}{rcl}
[p,\kappa_{2}] \after \asrt_{q} \after \pi_{p}
& = &
[q,\kappa_{2}] \after \asrt_{p} \after \pi_{p} \\
& \smash{\stackrel{\ref{prop:extensiveeffectus:asrtproj}}{=}} &
[q,\kappa_{2}] \after \kappa_{1} \after \pi_{p}
\hspace*{\arraycolsep}=\hspace*{\arraycolsep}
q \after \pi_{p} 
\\
{[p,\kappa_{2}]} \after \asrt_{q} \after \pi_{p^\bot}
& = &
[q,\kappa_{2}] \after \asrt_{p} \after \pi_{p^\bot} \\
& \smash{\stackrel{\ref{prop:extensiveeffectus:asrtproj}}{=}} &
[q,\kappa_{2}] \after \kappa_{2} \after \bang
\hspace*{\arraycolsep}=\hspace*{\arraycolsep}
\kappa_{2} \after \bang
\hspace*{\arraycolsep}=\hspace*{\arraycolsep}
\zero.
\end{array}\eqno{(**)}$$

The predicates $p,q$ each decompose $X$ in two parts. Together they
decompose $X$ in four parts, in two ways, depending on the order of
decomposition (first $p$ or first $q$). We show how these two
4-decompositions are related.

Consider the following diagram, with the four isomorphisms $\varphi_i$
between pullbacks, at each corner.
$$\xymatrix@R-1pc{
\cmpr{\cmpr{X}{p}}{\pi_{p}^{*}(q)}
   \ar@/^1ex/@{ >->}[dr]^-{\pi_{\pi_{p}^{*}(q)}}\ar[d]_{\varphi_1}^{\cong}
& &
\cmpr{\cmpr{X}{p}}{\pi_{p}^{*}(q^{\bot})}
      \ar@/_1ex/@{ >->}[dl]_-{\pi_{p}^{*}(q^{\bot})}\ar[d]^{\varphi_2}_{\cong}
\\
\cmpr{\cmpr{X}{q}}{\pi_{q}^{*}(p)}
      \ar[r]\ar@{ >->}[dd]_{\pi_{\pi_{q}^{*}(p)}}\pullback & 
   \cmpr{X}{p}\ar[dd]^{\pi_{p}} & 
   \cmpr{\cmpr{X}{q^\bot}}{\pi_{q^{\bot}}^{*}(p)}
      \ar[l]\ar@{ >->}[dd]^{\pi_{\pi_{q^{\bot}}^{*}(p)}}\pullback[dl]
\\
\\
\cmpr{X}{q}\ar@{ >->}[r]^-{\pi_q} & X & \cmpr{X}{q^\bot}\ar@{ >->}[l]_{\pi_{q^\bot}}
\\
\\
\cmpr{\cmpr{X}{q}}{\pi_{q}^{*}(p^{\bot})}
      \ar[r]\ar@{ >->}[uu]^{\pi_{\pi_{q}^{*}(p^{\bot})}}\pullback[ur] & 
   \cmpr{X}{p^\bot}\ar[uu]^{\pi_{p^\bot}} & 
   \cmpr{\cmpr{X}{q^\bot}}{\pi_{q^\bot}^{*}(p^{\bot})}
      \ar[l]\ar@{ >->}[uu]_{\pi_{q^\bot}^{*}(p^{\bot})}\pullback[ul]
\\
\cmpr{\cmpr{X}{p^\bot}}{\pi_{p^\bot}^{*}(q)}
   \ar@/_1ex/@{ >->}[ur]_-{\pi_{p^\bot}^{*}(q)}\ar[u]^{\varphi_3}_{\cong}
& &
\cmpr{\cmpr{X}{p^\bot}}{\pi_{p^\bot}^{*}(q^{\bot})}
   \ar@/^1ex/@{ >->}[ul]^-{\pi_{p^\bot}^{*}(q^{\bot})}\ar[u]_{\varphi_4}^{\cong}
}$$

\noindent Again using that pairs of projections are jointly epic we
prove the equations~$(**)$. The first one is obtained via:
$$\begin{array}{rcl}
[p, \kappa_{2}] \after \asrt_{q} \after \pi_{p} \after \pi_{\pi_{p}^{*}(q)}
& = &
[p, \kappa_{2}] \after \asrt_{q} \after \pi_{q} \after \pi_{\pi_{q}^{*}(p)} 
   \after \varphi_{1} \\
& \smash{\stackrel{\ref{prop:extensiveeffectus:asrtproj}}{=}} &
[p, \kappa_{2}] \after \kappa_{1} \after \pi_{q} \after \pi_{\pi_{q}^{*}(p)} 
   \after \varphi_{1} \\
& = &
p \after \pi_{q} \after \pi_{\pi_{q}^{*}(p)} \after \varphi_{1} \\
& = &
p \after \pi_{p} \after \pi_{\pi_{p}^{*}(q)} \\
& = &
\one \\
& = &
q \after \pi_{q} \after \pi_{\pi_{q}^{*}(p)} \after \varphi_{1} \\
& = &
q \after \pi_{p} \after \pi_{\pi_{p}^{*}(q)} 
\\
{[p, \kappa_{2}]} \after \asrt_{q} \after \pi_{p} \after \pi_{\pi_{p}^{*}(q^{\bot})} 
& = &
[p, \kappa_{2}] \after \asrt_{q} \after \pi_{q^{\bot}} \after 
   \pi_{\pi_{q^{\bot}}^{*}(p)} \after \varphi_{2} \\
& \smash{\stackrel{\ref{prop:extensiveeffectus:asrtproj}}{=}} &
[p, \kappa_{2}] \after 0 \after \pi_{\pi_{q^{\bot}}^{*}(p)} \after \varphi_{2} \\
& = &
\zero \\
& = &
q \after \pi_{q^{\bot}} \after \pi_{\pi_{q^{\bot}}^{*}(p)} \after \varphi_{2} \\
& = &
q \after \pi_{p} \after \pi_{\pi_{p}^{*}(q^{\bot})}.
\end{array}$$

\noindent The second equation in~$(**)$ is obtained similarly.

\auxproof{
$$\begin{array}{rcl}
[p, \kappa_{2}] \after \asrt_{q} \after \pi_{p^\bot} \after \pi_{\pi_{p^\bot}^{*}(q)}
& = &
[p, \kappa_{2}] \after \asrt_{q} \after \pi_{q} \after \pi_{\pi_{q}^{*}(p^{\bot})} 
   \after \varphi_{3} \\
& \smash{\stackrel{\ref{prop:extensiveeffectus:asrtproj}}{=}} &
[p, \kappa_{2}] \after \kappa_{1} \after \pi_{q} \after \pi_{\pi_{q}^{*}(p^{\bot})} 
   \after \varphi_{3} \\
& = &
p \after \pi_{q} \after \pi_{\pi_{q}^{*}(p^{\bot})} \after \varphi_{3} \\
& = &
p \after \pi_{p^\bot} \after \pi_{\pi_{p^\bot}^{*}(q)} \\
& = &
0 \\
& = &
0 \after \pi_{\pi_{p^\bot}^{*}(q)}
\\
{[p, \kappa_{2}]} \after \asrt_{q} \after \pi_{p^\bot} \after 
   \pi_{\pi_{p^\bot}^{*}(q^{\bot})}
& = &
[p, \kappa_{2}] \after \asrt_{q} \after \pi_{q^{\bot}} \after 
   \pi_{\pi_{q^{\bot}}^{*}(p^{\bot})} \after \varphi_{4} \\
& \smash{\stackrel{\ref{prop:extensiveeffectus:asrtproj}}{=}} &
[p, \kappa_{2}] \after 0 \after 
   \pi_{\pi_{q^{\bot}}^{*}(p^{\bot})} \after \varphi_{4} \\
& = &
0 \\
& = &
0 \after \pi_{\pi_{p^\bot}^{*}(q^{\bot})}.
\end{array}$$
}
}
\end{enumerate}}

We may now conclude that the extensive category $\cat{A}$ is an
effectus with comprehension, by point~\ref{prop:extensiveeffectus:pb},
which is commutative, by points~\ref{prop:extensiveeffectus:asrtbelow}
-- \ref{prop:extensiveeffectus:asrtcomm}, and Boolean in particular,
by point~\ref{prop:extensiveeffectus:asrtboolean}. \QED
\end{proof}

We now arrive at the main result of this section.

\begin{theorem}
\label{thm:extensiveeffectus}
Let $\cat{B}$ be a category with finite coproducts and a final
object. Then: $\cat{B}$ is extensive if and only if it is a Boolean
effectus with comprehension.

In this situation quotients come for free.
\end{theorem}

\begin{proof}
By Proposition~\ref{prop:extensiveeffectus} we only have to prove that
a Boolean effectus $\cat{B}$ with comprehension is extensive. From
Proposition~\ref{prop:booleffectus} we know that all predicates are
sharp and form Boolean algebras.

For a predicate $p$ on $X$ we have the map $\asrt_{p} \colon X \pto X$
with $\kerbot(\asrt_{p}) = p$. In a Boolean effectus we have:
$$\begin{array}{rcccccccl}
\pbox{\asrt_{p}}(p)
& = &
\big(p^{\bot} \pafter \asrt_{p}\big)^{\bot}
& = &
\big(\one \pafter \asrt_{p^{\bot}} \pafter \asrt_{p}\big)^{\bot}
& = &
\big(\one \pafter \zero\big)^{\bot}
& = &
\one.
\end{array}$$

\noindent We then have a factorisation in $\Par(\cat{B})$ of the form:
$$\xymatrix@C+1pc@R-.5pc{
& \cmpr{X}{p}\ar@{ >->}[d]|-{\pafter}^{\klin{\pi_{p}}}
\\
X\ar@{..>}@/^2ex/[ur]|-{\pafter}^-{\xi_{p^\bot}}\ar[r]|-{\pafter}_-{\asrt_p} & X
}$$

\noindent The notation $\xi$ for these maps is deliberate, because they
yield quotients. But first we show that:
$$\begin{array}{rclcrcl}
\xi_{p^\bot} \pafter \klin{\pi_{p}}
& = &
\idmap
& \qquad\mbox{and}\qquad &
\xi_{p^\bot} \pafter \klin{\pi_{p^\bot}}
& = &
\zero.
\end{array}\eqno{(*)}$$ 

\noindent The proof uses that $\klin{\pi_p}$ is monic in $\Par(\cat{B})$:
$$\begin{array}{rcll}
\klin{\pi_p} \pafter \xi_{p^\bot} \pafter \klin{\pi_{p}}
& = &
\asrt_{p} \pafter \klin{\pi_{p}} \\
& = &
\kappa_{1} \after \pi_{1} & 
   \mbox{by Lemma~\ref{lem:commcmprasrt}} \\
& = &
\klin{\pi_{p}}
\\
\klin{\pi_{p}} \pafter \xi_{p^\bot} \pafter \klin{\pi_{p^\bot}}
& = &
\asrt_{p} \pafter \klin{\pi_{p^\bot}} \\
& = &
\zero & \mbox{by Lemma~\ref{lem:commcmprasrt}} \\
& = &
\klin{\pi_{p}} \pafter \zero.
\end{array}$$

We now show that $\xi_{p^\bot} \colon X \pto \cmpr{X}{p}$ is a
quotient map, where $\cmpr{X}{p}$ is the quotient object
$X/p^{\bot}$. First we have:
$$\begin{array}{rcccccccl}
\ker(\xi_{p^\bot})
& = &
\big(\one \pafter \xi_{p^\bot}\big)^{\bot}
& = &
\big(\one \pafter \klin{\pi_{p}} \pafter \xi_{p^\bot}\big)^{\bot}
& = &
\big(\one \pafter \asrt_{p}\big)^{\bot}
& = &
p^{\bot}.
\end{array}$$

\noindent Next, if $f\colon X \rightarrow Y$ satisfies $p^{\bot} \leq
\ker(f)$, then:
$$\begin{array}{rcl}
\one \pafter f \pafter \asrt_{p^\bot}
\hspace*{\arraycolsep}=\hspace*{\arraycolsep}
\kerbot(f) \pafter \asrt_{p^\bot}
& \leq &
p \pafter \asrt_{p^\bot} \\
& = &
\one \pafter \asrt_{p} \pafter \asrt_{p^\bot}
\hspace*{\arraycolsep}=\hspace*{\arraycolsep}
\one \pafter \zero
\hspace*{\arraycolsep}=\hspace*{\arraycolsep}
\zero.
\end{array}$$

\noindent This gives $f \pafter \asrt_{p^\bot} = \zero$, by
Lemma~\ref{lem:zero}, and thus:
$$\begin{array}{rcccccccl}
f
& = &
f \pafter \idmap
& = &
f \pafter (\asrt_{p} \ovee \asrt_{p^\bot})
& = &
(f \pafter \asrt_{p}) \ovee (f \pafter \asrt_{p^\bot})
& = &
f \pafter \asrt_{p}.
\end{array}$$

\noindent We obtain $\overline{f} = f \after \pi_{p} \colon 
\cmpr{X}{p} \pto Y$ as mediating map, since:
$$\begin{array}{rcccccl}
\overline{f} \pafter \xi_{p^\bot}
& = &
f \pafter \klin{\pi_{p}} \pafter \xi_{p^\bot}
& = &
f \pafter \asrt_{p}
& = &
f.
\end{array}$$

\noindent This map $\overline{f}$ is unique, since if $g\colon
\cmpr{X}{p} \pto Y$ also satisfies $g \pafter \xi_{p^\bot} = f$, then:
$$\begin{array}{rcccccccl}
\overline{f}
& = &
f \pafter \klin{\pi_{p}}
& = &
g \pafter \xi_{p^\bot} \pafter \klin{\pi_{p}}
& \smash{\stackrel{(*)}{=}} &
g \pafter \idmap
& = &
g.
\end{array}$$

Now that we have quotient maps we can form the total `decomposition'
map $\dc_{p} = \dtuple{\xi_{p^\bot}, \xi_{p}} \colon X \to
\cmpr{X}{p} + \cmpr{X}{p^\bot}$ from
Lemma~\ref{lem:quot}~\eqref{lem:quot:dc}. We claim that this map
$\dc_{p}$ is an isomorphism, with the cotuple $[\pi_{p},
  \pi_{p^{\bot}}]$ as inverse. We first prove $\dc_{p} \after \pi_{p}
= \kappa_{1}$ and $\dc_{p} \after \pi_{p^\bot} = \kappa_{2}$ by
using~\eqref{eqn:dtupleuniqueness}:
$$\begin{array}{rcl}
\dc_{p} \after \pi_{p}
\hspace*{\arraycolsep}=\hspace*{\arraycolsep}
\dtuple{\xi_{p^\bot}, \xi_{p}} \after \pi_{p}
& = &
\dtuple{\xi_{p^\bot} \after \pi_{p}, \xi_{p} \after \pi_{p}} \\
& \smash{\stackrel{(*)}{=}} &
\dtuple{\kappa_{1}, \zero}
\hspace*{\arraycolsep}=\hspace*{\arraycolsep}
\dtuple{\rhd_{1} \after \kappa_{1}, \rhd_{2} \after \kappa_{1}}
\hspace*{\arraycolsep}=\hspace*{\arraycolsep}
\kappa_{1}
\\
\dc_{p} \after \pi_{p^{\bot}}
\hspace*{\arraycolsep}=\hspace*{\arraycolsep}
\dtuple{\xi_{p^\bot}, \xi_{p}} \after \pi_{p^{\bot}}
& = &
\dtuple{\xi_{p^\bot} \after \pi_{p}, \xi_{p} \after \pi_{p^{\bot}}} \\
& \smash{\stackrel{(*)}{=}} &
\dtuple{\zero, \kappa_{1}}
\hspace*{\arraycolsep}=\hspace*{\arraycolsep}
\dtuple{\rhd_{1} \after \kappa_{2}, \rhd_{2} \after \kappa_{2}}
\hspace*{\arraycolsep}=\hspace*{\arraycolsep}
\kappa_{2}.
\end{array}$$

\noindent At this stage we can prove one part of the claim that the
decomposition map $\dc_{p}$ is an isomorphism:
$$\begin{array}{rcccccl}
\dc_{p} \after [\pi_{p}, \pi_{p^\bot}]
& = &
[\dc_{p} \after \pi_{p}, \dc_{p} \after \pi_{p^\bot}]
& = &
[\kappa_{1}, \kappa_{2}]
& = &
\idmap.
\end{array}$$

\noindent For the other direction, we take $b = \kappa_{1} \after
(\pi_{p}+\pi_{p^\bot}) \after \dc_{p} \colon X \pto X+X$ and claim
that $b$ is a bound for the pair of maps $\asrt_{p}, \asrt_{p^\bot}
\colon X \pto X$. This is the case, since:
$$\begin{array}{rcl}
\rhd_{1} \pafter b
\hspace*{\arraycolsep}=\hspace*{\arraycolsep}
(\idmap+\bang) \after (\pi_{p}+\pi_{p^\bot}) \after \dc_{p} 
& = &
(\pi_{p}+\idmap) \after (\idmap+\bang) \after \dc_{p} \\
& = &
(\pi_{p}+\idmap) \after \xi_{p^\bot} \\
& = &
\asrt_{p} 
\\
\rhd_{2} \pafter b
\hspace*{\arraycolsep}=\hspace*{\arraycolsep}
[\kappa_{2}\after\bang, \kappa_{1}] \after (\pi_{p}+\pi_{p^\bot}) \after \dc_{p} 
& = &
[\kappa_{2}, \kappa_{1} \after \pi_{p^\bot}] \after (\bang+\idmap) \after \dc_{p} \\
& = &
[\kappa_{2}, \kappa_{1} \after \pi_{p^\bot}] \after [\kappa_{2}, \kappa_{1}]
   \after \xi_{p} \\
& = &
(\pi_{p^\bot}+\idmap) \after \xi_{p} \\
& = &
\asrt_{p^\bot}.
\end{array}$$

\noindent Since the assert map is a homomorphism of effect algebras we have:
$$\begin{array}{rcl}
\kappa_{1}
\hspace*{\arraycolsep}=\hspace*{\arraycolsep}
\asrt_{p} \ovee \asrt_{p^\bot} 
\hspace*{\arraycolsep}=\hspace*{\arraycolsep}
(\nabla+\idmap) \after b 
& = &
\kappa_{1} \after \nabla \after (\pi_{p}+\pi_{p^\bot}) \after \dc_{p} \\
& = &
\kappa_{1} \after [\pi_{p}, \pi_{p^\bot}] \after \dc_{p}.
\end{array}$$

\noindent But then $[\pi_{p}, \pi_{p^\bot}] \after \dc_{p} = \idmap$, making
$\dc_p$ an isomorphism, as required.

We are finally in a position to show that the effectus $\cat{B}$ is an
extensive category, see Definition~\ref{def:extensive}. Let $f\colon Z
\rightarrow X+Y$ be an arbitrary map in $\cat{B}$.  Write $p =
(\bang+\bang) \after f \colon Z \rightarrow 1+1$, and consider the
following diagram.
$$\xymatrix{
\cmpr{Z}{p}\ar@{ >->}[r]^-{\pi_p}\ar@/_3ex/[dd]
   \ar@{..>}[d]\pullback & 
   Z\ar[d]^{f} & 
   \cmpr{Z}{p^\bot}\ar@{ >->}[l]_-{\pi_{p^\bot}}
      \ar@/^3ex/[dd]\ar@{..>}[d]
      \pullback[dl]
\\
Y\ar@{ >->}[r]^-{\kappa_{1}}\ar[d]\pullback & 
   Y+Z\ar[d]^-{\bang+\bang} & 
   Z\ar@{ >->}[l]_-{\kappa_{2}}\ar[d]\pullback[dl]
\\
1\ar@{ >->}[r]_-{\kappa_{1}} & 1+1 & 
   1\ar@{ >->}[l]^-{\kappa_{2}}
}$$

\noindent The two lower pullbacks exist because $\cat{B}$ is an
effectus. The two outer ones, with the curved arrows, exist because
$\cat{B}$ has comprehension. The two upper squares, with the dashed
arrows are then pulbacks by the Pullback Lemma. Thus, the pullbacks
along $f$ of the coprojections $Y\rightarrow Y+Z \leftarrow Z$
exist. Moreover, we have just seen that the cotuple $[\pi_{p},
  \pi_{p^\bot}]$ of the pulled-back maps is an isomorphism.  This is
one part of Definition~\ref{def:extensive}.

For the other part, consider a similar diagram where the top row is
already a coproduct diagram:
$$\xymatrix{
Z_{1}\ar[d]_{f_1}\ar@{ >->}[r]^-{\kappa_1} &
   Z_{1}+Z_{2}\ar[d]^{f} & 
   Z_{2}\ar[d]^{f_2}\ar@{ >->}[l]_-{\kappa_2}
\\
X\ar@{ >->}[r]^-{\kappa_1} & X+Y & 
   Y\ar@{ >->}[l]_-{\kappa_2}
}$$

\noindent Then: 
$$\begin{array}{rcccccccl}
f
& = &
f \after [\kappa_{1},\kappa_{2}]
& = &
[f \after \kappa_{1}, f \after \kappa_{2}]
& = &
[\kappa_{1} \after f_{1}, \kappa_{2} \after f_{2}]
& = &
f_{1}+f_{2}.
\end{array}$$

\noindent But the above two rectangles are then pullbacks
by Lemma~\ref{lem:effectuspb}~\eqref{lem:effectuspb:tot}. \QED
\end{proof}

We thus see that a Boolean effectus with comprehension corresponds to
the well-established notion of extensive category. It is an open
question whether there is a similar alternative notion for a
commutative effectus, capturing probabilistic computation.

Each topos is an extensive category, and thus an effectus. Notice that
in a topos-as-effectus the predicates are the maps $X \rightarrow
1+1$, and not the more general predicates $X\rightarrow\Omega$, where
$\Omega$ is the subobject classifier. There are many extensive
categories that form interesting examples of effectuses, such as: the
category $\Top$ of topological spaces, and it subcategory $\CH
\hookrightarrow \Top$ of compact Hausdorff spaces; the category
$\Meas$ of measurable spaces; the opposite $\op{\CRng}$ of the
category of commutative rings, see~\cite{ChoJWW15}.

\section{Combining comprehension and quotients}\label{sec:cmprquot}

In previous sections we have studied comprehension and quotients
separately. We now look at the combination, and derive some more
results. First we need the following observation.

\begin{lemma}
\label{lem:theta}
Let $\cat{C}$ be an effectus in partial form which has comprehension
and also has quotients. Write for a predicate $p$ on an object
$X\in\cat{C}$:
\begin{equation}
\label{diag:cmprquotcanmap}
\xymatrix@C+.5pc{
\cmpr{X}{p}\ar[rr]^-{\theta_{p} \;\defeq\; \xi_{p^{\bot}} \after \pi_{p}} 
   & & X/p^{\bot}
}\index{D}{$\theta_p = \xi_{p^{\bot}} \after \pi_{p}$}
\end{equation}

\noindent Then:
\begin{enumerate}
\item this map $\theta_p$ is total;

\item if $\theta_p$ is an isomorphism, then $p$ is sharp.
\end{enumerate}
\end{lemma}

\begin{proof}
The first point is easy: $\theta_p$ is total by Lemma~\ref{lem:partcmpr}~\eqref{lem:partcmpr:pred} and
Lemma~\ref{lem:quot}~\eqref{lem:quot:XiEqn}:
$$\begin{array}{rcccccccl}
\one \after \theta_{p}
& = &
\one \after \xi_{p^\bot} \after \pi_{p}
& = &
\kerbot(\xi_{p^\bot}) \after \pi_{p}
& = &
p \after \pi_{p}
& = &
\one.
\end{array}$$

For the second point let $\theta_p$ be an isomorphism, and let
$q\in\Pred(X)$ satisfy $q \leq p$ and $q \leq p^{\bot}$. In order to
prove that $p$ is sharp, we must show $q = \zero$. From $p^{\bot} \leq
q^{\bot} = \ker(q)$ we get a predicate $q/p$ in $\Pred(X/p^{\bot})$
with $q/p \after \xi_{p^\bot} = q$, see
Lemma~\ref{lem:quot}~\eqref{lem:quot:div}. Then:
$$\begin{array}{rcl}
q/p
\hspace*{\arraycolsep}=\hspace*{\arraycolsep}
q/p \after \theta_{p} \after \theta_{p}^{-1}
\hspace*{\arraycolsep}=\hspace*{\arraycolsep}
q/p \after \xi_{p^\bot} \after \pi_{p} \after \theta_{p}^{-1}
& = &
q \after \pi_{p} \after \theta_{p}^{-1} \\
& \leq &
p^{\bot} \after \pi_{p} \after \theta_{p}^{-1} \\
& = &
\zero \after \theta_{p}^{-1} \\
& = &
\zero.
\end{array}$$

\noindent Hence $q = q/p \after \xi_{p^{\bot}} = \zero \after
\xi_{p^\bot} = \zero$. \QED
\end{proof}

We now consider the condition which enforces an equivalence in the
above second point.

\begin{definition}
\label{def:cmprquot}
We say that an effectus in partial form has \emph{both quotients and
  comprehension}\index{S}{effectus!-- with both quotients and
  comprehension} if it has comprehension, like in
Definition~\ref{def:cmpr}~\eqref{def:cmpr:part}, and quotients, like
in Definition~\ref{def:quot}, such that for each \emph{sharp}
predicate $p\in\Pred(X)$ the total map $\theta_{p} \colon \cmpr{X}{p}
\rightarrow X/p^{\bot}$ in~\eqref{diag:cmprquotcanmap} is an
isomorphism.

We then define, for each such sharp $p\in\Pred(X)$ the map:
\begin{equation}
\label{diag:sharpassrt}
\xymatrix{
\asrt_{p} \;\defeq\; \Big(X\ar@{->>}[r]^-{\xi_{p^\bot}} & 
   X/p^{\bot}\ar[r]^-{\theta_{p}^{-1}}_-{\cong} & 
   \cmpr{X}{p}\ar@{ >->}[r]^-{\pi_p} & X\Big).
}\index{N}{$\asrt_p$, assert map for predicate $p$!when $p$ is sharp}
\end{equation}
\end{definition}

\begin{example}
\label{ex:cmprquot}
We briefly illustrate the map $\theta_{p} \colon \cmpr{X}{p}
\rightarrow X/p^{\bot}$ from~\eqref{diag:cmprquotcanmap} in our
running examples.

In the effectus $\Sets$ we have for a predicate $P\subseteq X$
$\cmpr{X}{P} = P$ and also $Q/P^{\bot} = \cmpr{X}{P^{\bot\bot}} = P$.
The map $\theta$ is the identity, and is thus always an isomorphism.
Indeed, in this Boolean effectus $\Sets$ every predicate is sharp.

More generally, consider an extensive category $\cat{B}$ with final
object, as a Boolean effectus with both comprehension and quotients,
see Theorem~\ref{thm:extensiveeffectus}. The quotient that is
constructed in the proof is of the form $\xi_{p^\bot} \colon X
\rightarrow \cmpr{X}{p}$. Hence we have $\theta_{p} = \idmap \colon
\cmpr{X}{p} \rightarrow X/p^{\bot}$.

In the effectus $\Kl(\Dst)$ we have for a predicate $p\in [0,1]^{X}$
an inclusion:
$$\xymatrix{
\cmpr{X}{p} = \set{x}{p(x) = 1}\;\ar@{^(->}[r]^-{\theta_p} &
   \set{x}{p(x) > 0} = X/p^{\bot}
}$$

\noindent When $p$ is sharp, we have $p(x) \in \{0,1\}$, so that
$\theta_{p}$ is the identity.

Next we consider the effectus $\op{\OUG}$ of order unit groups. For an
effect $e\in [0,1]_{G}$ in an order unit group $G$ we have a map
(in the opposite direction):
$$\xymatrix@R-2pc{
    G/e^{\bot} = \pideal{e}{G} \ar[r]^-{\theta_e} & 
   G/\pideal{e^\bot}{G} = \cmpr{G}{e}
\\
\qquad x\ar@{|->}[r] & x + \pideal{e^\bot}{G}
}$$

\noindent This map is unital since:
$$\begin{array}{rcccccccl}
\theta_{e}(\one_{G/e^\bot})
& = &
e + \pideal{e^\bot}{G}
& = &
e + e^{\bot} + \pideal{e^\bot}{G}
& = &
1 + \pideal{e^\bot}{G}
& = &
\one_{\cmpr{G}{e}}.
\end{array}$$

\noindent It can be shown that this map $\theta_e$ is an isomorphism
for sharp $e$, if one assumes that $G$ satisfies interpolation,
see~\cite{Goodearl86} for details.

Finally, in our effectus $\op{\vNA}$ of von Neumann algebras we also
have, for an effect $e\in[0,1]_{\mathscr{A}}$, a map of the form:
$$\xymatrix@R-2pc{
\mathscr{A}/e^{\bot} = \ceil{e}\mathscr{A}\ceil{e}\ar[r]^-{\theta_e} & 
   \floor{e}\mathscr{A}\floor{e} = \cmpr{\mathscr{A}}{e}
\\
\qquad x\ar@{|->}[r] & x \qquad
}$$

\noindent This map is well-defined since $\ceil{e}\cdot \floor{e} =
\ceil{e} = \floor{e}\cdot \ceil{e}$, using a more general fact: $a
\leq \ceil{e} \Rightarrow a\cdot \ceil{e} = a = \ceil{e}\cdot a$.

In case $e$ is sharp we get $\ceil{e} = e = \floor{e}$, so that
the above map $\theta_e$ is an isomorphism.
\end{example}

We collect some more results about the maps $\theta$.
Since these maps are total, we switch to the total perspective.

\begin{lemma}
\label{lem:cmprquotcanmap}
Let effectus in total form $\cat{B}$ have quotients and comprehension.
\begin{enumerate}
\item \label{lem:cmprquotcanmap:pb} For each predicate $p\in\Pred(X)$
  the two squares below are pullbacks in $\cat{B}$.
\begin{equation}
\label{diag:cmprquotcanmap:pb}
\vcenter{\xymatrix{
X/p^{\bot}\ar@{ >->}[d]_{\kappa_1} & 
   \cmpr{X}{p}\ar[l]_-{\theta_p}\ar[r]^-{\bang}\ar@{ >->}[d]_{\pi_p} 
   & 1\ar@{ >->}[d]^{\kappa_2}
\\
X/p^{\bot}+1 & X\ar[l]^-{\xi_{p^\bot}}\ar[r]_-{\xi_p} & X/p + 1
}}
\end{equation}

\item \label{lem:cmprquotcanmap:triangle} For each predicate $p$ on
  $X$, the following triangle commutes.
\begin{equation}
\label{diag:cmprquotcanmap:triangle}
\vcenter{\xymatrix@R-.5pc@C-1pc{
\cmpr{X}{p} + \cmpr{X}{p^\bot}\ar[dr]_{[\pi_{p}, \pi_{p^\bot}]\;}
   \ar[rr]^-{\theta_{p} + \theta_{p^\bot}} & & X/p^{\bot} + X/p
\\
& X\ar[ur]_{\dc_p} &
}}
\end{equation}

\noindent Thus, for \emph{sharp} $p$ the decomposition map $\dc_p$ is
a split epi, and the cotuple $[\pi_{p}, \pi_{p^\bot}]$ is a split
mono, in $\cat{B}$.

\item The maps $\theta$ commute with the distributivity isomorphisms
  from Lemma~\ref{lem:totcmpr}~\eqref{lem:totcmpr:sumiso} and
  Lemma~\ref{lem:quot}~\eqref{lem:quot:sumiso}, as in:
$$\xymatrix@R-1pc@C+1pc{
\cmpr{X}{p}+\cmpr{Y}{q}\ar@{=}[d]^{\wr}\ar[r]^-{\theta_{X}+\theta_{Y}} &
   X/p^{\bot}+Y/q^{\bot}\ar@{=}[d]^{\wr}
\\
\cmpr{X+Y\,}{\,[p,q]}\ar[r]_-{\theta_{[p,q]}} &
   (X+Y)/[p,q]^{\bot}\rlap{$\;=(X+Y)/[p^{\bot},q^{\bot}]$}
}\hspace*{6em}$$
\end{enumerate}
\end{lemma}

\begin{proof}
Recall that we are in an effectus in total form.
\begin{enumerate}
\item The left square in~\eqref{diag:cmprquotcanmap:pb} commutes 
  because $\theta_p$ is total. It forms a pullbacks since if $f\colon
  Y \rightarrow X/p^{\bot}$ and $g\colon Y \rightarrow X$ satisfy
  $\kappa_{1} \after f = \xi_{p^\bot} \after g$, then $g$ factors
  uniquely as $g = \pi_{p} \after \overline{g}$ since:
$$\begin{array}{rcl}
\tbox{g}(p)
\hspace*{\arraycolsep}=\hspace*{\arraycolsep}
p \after g
\hspace*{\arraycolsep}=\hspace*{\arraycolsep}
\kerbot(\xi_{p^\bot}) \after g
& = &
(\bang+\idmap) \after \xi_{p^\bot} \after g \\
& = &
(\bang+\idmap) \after \kappa_{1} \after \theta_{p} \\
& = &
\kappa_{1} \after \bang \after \theta_{p} \\
& = &
\kappa_{1} \after \bang \\
& = &
\one.
\end{array}$$

\noindent Then $\theta_{p} \after \overline{g} = f$ because the
coprojection $\kappa_{1}$ is monic:
$$\begin{array}{rcccccl}
\kappa_{1} \after \theta_{p} \after \overline{g}
& = &
\xi_{p^\bot} \after \pi_{p} \after \overline{g}
& = &
\xi_{p^\bot} \after g
& = &
\kappa_{1} \after f.
\end{array}$$

The rectangle on the right in~\eqref{diag:cmprquotcanmap:pb} commutes
by Lemma~\ref{lem:zero} since:
$$\begin{array}{rcccccl}
(\bang+\idmap) \after \xi_{p} \after \pi_{p}
& = &
\kerbot(\xi_{p}) \after \pi_{p}
& = &
p^{\bot} \after \pi_{p}
& = &
\zero.
\end{array}$$

\noindent It is a pullback: if $f\colon Y \rightarrow X$ satisfies
$\xi_{p} \after f = \kappa_{2} \after \bang = 0$, then $f$ factors
uniquely through $\pi_p$ since:
$$\begin{array}{rcl}
\tbox{f}(p)
\hspace*{\arraycolsep}=\hspace*{\arraycolsep}
p \after f
\hspace*{\arraycolsep}=\hspace*{\arraycolsep}
\ker(\xi_{p}) \after f
& = &
[\kappa_{2},\kappa_{1}] \after (\bang+\idmap) \after \xi_{p} \after f \\
& = &
[\kappa_{2},\kappa_{1}] \after (\bang+\idmap) \after \kappa_{2} \after \bang \\
& = &
\kappa_{1} \after \bang \\
& = &
\one.
\end{array}$$

\item We use the commuting rectangles~\eqref{diag:cmprquotcanmap:pb}
and the equations~\eqref{eqn:dtupleuniqueness} to get:
$$\begin{array}{rcl}
\dc_{p} \after \pi_{p}
\hspace*{\arraycolsep}=\hspace*{\arraycolsep}
\dtuple{\xi_{p^\bot}, \xi_{p}} \after \pi_{p}
& = &
\dtuple{\xi_{p^\bot} \after \pi_{p}, \xi_{p} \after \pi_{p}} \\
& = &
\dtuple{\kappa_{1} \after \theta_{p}, \zero} \\
& = &
\dtuple{\rhd_{1} \after \kappa_{1} \after \theta_{p}, 
   \rhd_{2} \after \kappa_{1} \after \theta_{p}}
\hspace*{\arraycolsep}=\hspace*{\arraycolsep}
\kappa_{1} \after \theta_{p}
\\
\dc_{p} \after \pi_{p^{\bot}}
\hspace*{\arraycolsep}=\hspace*{\arraycolsep}
\dtuple{\xi_{p}, \xi_{p^\bot}} \after \pi_{p^{\bot}}
& = &
\dtuple{\xi_{p} \after \pi_{p^{\bot}}, \xi_{p^\bot} \after \pi_{p^{\bot}}} \\
& = &
\dtuple{\zero, \kappa_{1} \after \theta_{p^{\bot}}} \\
& = &
\dtuple{\rhd_{1} \after \kappa_{2} \after \theta_{p^{\bot}}, 
   \rhd_{2} \after \kappa_{2} \after \theta_{p^{\bot}}}
\hspace*{\arraycolsep}=\hspace*{\arraycolsep}
\kappa_{2} \after \theta_{p^{\bot}}.
\end{array}$$

\noindent Hence we see that the
triangle~\eqref{diag:cmprquotcanmap:triangle} commutes:
$$\begin{array}{rcccccl}
\dc_{p} \after [\pi_{p}, \pi_{p^\bot}]
& = &
[\dc_{p} \after \pi_{p}, \dc_{p} \after \pi_{p^\bot}]
& = &
[\kappa_{1} \after \theta_{p}, \kappa_{2} \after \theta_{p^{\bot}}]
& = &
\theta_{p} + \theta_{p^\bot}.
\end{array}$$

\item We show that the composite:
$$\hspace*{-.5em}\xymatrix@C-1.2pc{
\cmpr{X+Y\,}{\,[p,q]}\ar[r]^-{\varphi}_-{\cong} &
   \cmpr{X}{p}+\cmpr{Y}{q}\ar[rr]^-{\theta_{X}+\theta_{Y}} & &
   X/p^{\bot} + Y/q^{\bot}\ar[r]^-{\psi}_-{\cong} & (X+Y)/[p,q]^{\bot}
}$$

\noindent satisfies the equation that defines $\theta_{[p,q]}$ as
$\kappa_{1} \after \theta_{[p,q]} = \xi_{[p,q]^{\bot}} \after
\pi_{[p,q]}$. 
$$\begin{array}[b]{rcl}
\lefteqn{\kappa_{1} \after \psi \after (\theta_{p}+\theta_{q}) \after 
   \varphi} \\
& = &
(\psi+\idmap) \after \kappa_{1} \after 
   [\kappa_{1} \after \theta_{p}, \kappa_{2} \after \theta_{q}]
   \after \varphi \\
& = &
(\psi+\idmap) \after 
   [(\kappa_{1}+\idmap) \after \kappa_{1} \after \theta_{p}, 
    (\kappa_{2}+\idmap) \after \kappa_{1} \after \theta_{q}]
   \after \varphi \\
& = &
(\psi+\idmap) \after 
   [(\kappa_{1}+\idmap) \after \xi_{p^\bot} \after \pi_{p}, 
    (\kappa_{2}+\idmap) \after \xi_{q^\bot} \after \pi_{q}]
   \after \varphi \\
& = &
(\psi+\idmap) \after 
   [(\kappa_{1}+\idmap) \after \xi_{p^\bot}, (\kappa_{2}+\idmap) \after \xi_{q^\bot}]
   \after (\pi_{p}+\pi_{q}) \after \varphi \\
& = &
(\psi \pafter (\xi_{p}\pplus\xi_{q})) \after  \pi_{[p,q]} \\
& = &
\xi_{[p,q]} \after \pi_{[p,q]}.
\end{array}\eqno{\QEDbox}$$
\end{enumerate}
\end{proof}

We continue with properties of the assert map~\eqref{diag:sharpassrt}
for sharp predicates $p$. We have seen `assert' maps in
Section~\ref{sec:commbool} as inverse of the kernel-supplement map
$\kerbot \colon \sEnd(X) \rightarrow \Pred(X)$. This isomorphism of
effect algebras is a key aspect of commutative effectuses (including
Boolean ones).  The assert map~\eqref{diag:cmprquotcanmap} in this
section is at the same time more general and also more restricted: it
is not necessarily side-effect free, that is, below the identity, nor
a map of effect algebras; and it is defined only for sharp predicates.

\begin{lemma}
\label{lem:sharpassert}
Let $(\cat{C}, I)$ be an effectus with quotients and comprehension, in
partial form. Let $p$ be a sharp predicate on $X\in\cat{C}$.
\begin{enumerate}
\item \label{lem:sharpassert:idemp} $\asrt_{p} \after \asrt_{p} = \asrt_{p}$
  and $\asrt_{p} \after \asrt_{p^\bot} = \zero$.

\item \label{lem:sharpassert:ker} $\ker(\asrt_{p}) = p^{\bot}$, and so
  $\kerbot(\asrt_{p}) = p$.

\item \label{lem:sharpassert:zeroone} $\asrt_{\zero} = \zero$ and
  $\asrt_{\one} = \idmap$.

\item \label{lem:sharpassert:equiv} For each $f\colon X \rightarrow Y$
  one has:
$$\begin{array}{rcl}
p^{\bot} \leq \ker(f)
& \Longleftrightarrow &
f \after \asrt_{p} = f.
\end{array}$$

\item \label{lem:sharpassert:eq} The following two diagrams are
  equalisers in $\cat{C}$.
$$\xymatrix{
\cmpr{X}{p^\bot}\ar@{ >->}[r]^-{\pi_{p^\bot}} &
   X\ar@/^1.5ex/[r]^-{\asrt_{p}}\ar@/_1.5ex/[r]_-{\zero} & X
& &
\cmpr{X}{p}\ar@{ >->}[r]^-{\pi_{p}} &
   X\ar@/^1.5ex/[r]^-{\asrt_p}\ar@/_1.5ex/[r]_-{\idmap} & X
}$$

\item \label{lem:sharpassert:coeq} Similarly, the next diagram is a
  coequaliser.
$$\xymatrix{
X\ar@/^1.5ex/[r]^-{\asrt_p}\ar@/_1.5ex/[r]_-{\idmap} & 
   X\ar@{->>}[r]^-{\xi_{p^\bot}} & X/p^{\bot}
}$$

\item \label{lem:sharpassert:sef} Let $f\colon X \rightarrow X$ be a
  side-effect free endomap, that is $f\leq \idmap[X]$ in the poset
  $\sEnd(X)$ of endomaps on $X$ below the identity. If the predicate
  $p = \kerbot(f) = \one \after f$ is sharp, then $f =
  \asrt_{p}$. 

Moreover, in that case the decomposition map $\dc_{p}$ from
Lemma~\ref{lem:quot}~\eqref{lem:quot:dc} is a total isomorphism:
$$\xymatrix@+1pc{
X\ar[r]^-{\dc_p}_-{\cong} & X/p^{\bot}+X/p
}$$
\end{enumerate}
\end{lemma}

\begin{proof}
Recall that $\asrt_{p} = \pi_{p} \after \theta_{p}^{-1} \after
\xi_{p^\bot} \colon X \rightarrow X$, see~\eqref{diag:sharpassrt}.
\begin{enumerate}
\item We compute:
$$\begin{array}{rcl}
\asrt_{p} \after \asrt_{p}
& = &
\pi_{p} \after \theta_{p}^{-1} \after \xi_{p^\bot} \after
   \pi_{p} \after \theta_{p}^{-1} \after \xi_{p^\bot} \\
& = &
\pi_{p} \after \theta_{p}^{-1} \after \theta_{p} \after
   \theta_{p}^{-1} \after \xi_{p^\bot} \\
& = &
\pi_{p} \after \theta_{p}^{-1} \after \xi_{p^\bot} \\
& = &
\asrt_{p}
\\
\asrt_{p} \after \asrt_{p^\bot}
& = &
\pi_{p} \after \theta_{p}^{-1} \after \xi_{p^\bot} \after
   \pi_{p^\bot} \after \theta_{p^\bot}^{-1} \after \xi_{p} \\
& \smash{\stackrel{\eqref{diag:cmprquotcanmap:pb}}{=}} &
\pi_{p} \after \theta_{p}^{-1} \after \zero \after 
   \theta_{p}^{-1} \after \xi_{p^\bot} \\
& = &
\zero.
\end{array}$$

\item Applying Lemma~\ref{lem:ker}~\eqref{lem:ker:tot} with the
total map $\pi_{p} \after \theta_{p}^{-1}$ yields:
$$\begin{array}{rcccccl}
\ker(\asrt_{p})
& = &
\ker\big(\pi_{p} \after \theta_{p}^{-1} \after \xi_{p^\bot}\big) 
& = &
\ker(\xi_{p^{\bot}}) 
& = &
p^{\bot}.
\end{array}$$

\item We have $\asrt_{\zero} = \zero$ since $\kerbot(\asrt_{\zero}) =
  \zero$ and $\kerbot$ reflects $\zero$, see
  Lemma~\ref{lem:kerbot}. Further, there are isomorphisms $X/\zero \cong X
  \cong \cmpr{X}{\one}$ by Lemma~\ref{lem:quot}~\eqref{lem:quot:zeroone}
  and Lemma~\ref{lem:totcmpr}~\eqref{lem:totcmpr:one}. Hence:
$$\xymatrix{
\asrt_{\one} \;=\; \Big(X\ar[r]^-{\xi_{\zero}}_-{\cong} & 
   X/\zero\ar[r]^-{\theta_{\one}^{-1}}_-{\cong} &
   \cmpr{X}{\one}\ar[r]^-{\pi_{\one}}_-{\cong} & X\Big) \;=\; \idmap,
}$$

\noindent since by definition $\theta_{\one} = \xi_{\zero} \after \pi_{\one}$,
and so $\theta_{\one}^{-1} = \pi_{\one}^{-1} \after \xi_{\zero}^{-1}$.


\item We need to prove the equivalence $p^{\bot} \leq \ker(f)
  \Longleftrightarrow f \after \asrt_{p} = f$. This is done in two
  steps.
\begin{itemize}
\item If $p^{\bot} \leq \ker(f)$, then we can write $f = \overline{f}
\after \xi_{p^{\bot}}$, so that:
$$\begin{array}{rcl}
f \after \asrt_{p}
\hspace*{\arraycolsep}=\hspace*{\arraycolsep}
\overline{f} \after \xi_{p^{\bot}} \after \pi_{p} \after \theta_{p}^{-1}
   \after \xi_{p^{\bot}} 
& = &
\overline{f} \after \theta_{p} \after \theta_{p}^{-1} \after \xi_{p^{\bot}} \\
& = &
\overline{f} \after \xi_{p^{\bot}} \\
& = &
f.
\end{array}$$

\item And if $f \after \asrt_{p} = f$, then by
Lemma~\ref{lem:ker}~\eqref{lem:ker:ord},
$$\begin{array}{rcccccl}
\ker(f)
& = &
\ker(f \after \asrt_{p})
& \geq &
\ker(\asrt_{p})
& = &
p^{\bot}.
\end{array}$$
\end{itemize}

\item The first kernel map equaliser is an instance of
  Lemma~\ref{lem:partcmpr}~\eqref{lem:partcmpr:kermap}, since
  $\ker(\asrt_{p}) = p^\bot$. For the other one, notice that:
$$\begin{array}{rcccccl}
\asrt_{p} \after \pi_{p}
& = &
\pi_{p} \after \theta_{p}^{-1} \after \xi_{p^{\bot}} \after \pi_{p}
& = &
\pi_{p} \after \theta_{p}^{-1} \after \theta_{p}
& = &
\pi_{p}.
\end{array}$$

\noindent And if $f\colon Y \rightarrow X$ satisfies $\asrt_{p} \after
f = f$, then $f$ factors uniquely through $\pi_{p}$ via the composite:
$\theta_{p}^{-1} \after \xi_{p^\bot} \after f \colon Y \rightarrow
\cmpr{X}{p}$.

\item First,
$$\begin{array}{rcccccl}
\xi_{p^\bot} \after \asrt_{p}
& = &
\xi_{p^\bot} \after \pi_{p} \after \theta_{p}^{-1} \after \xi_{p^{\bot}}
& \smash{\stackrel{\eqref{diag:cmprquotcanmap:pb}}{=}} &
\theta_{p} \after \theta_{p}^{-1} \after \xi_{p^{\bot}}
& = &
\xi_{p^\bot}.
\end{array}$$

\noindent Next, if $f\colon X\rightarrow Y$ satisfies $f \after
\asrt_{p} = f$, then $p^{\bot} \leq \ker(f)$ by
point~\eqref{lem:sharpassert:equiv}, so that $f = \overline{f} \after
\xi_{p^\bot}$, for a necessarily uniqe $\overline{f} \colon X/p^{\bot}
\rightarrow Y$. 

\item Let $f\colon X \rightarrow X$ satisfy $f\leq \idmap$, say via
  $f\ovee g = \idmap$. Further, the predicate $p = \kerbot(f) = \one
  \after f$ is sharp, by assumption. Then $\one \after g = p^{\bot}$
  by uniqueness of orthosupplements:
$$\begin{array}{rcccccl}
p \ovee (\one \after g)
& = &
\one \after (f \ovee g)
& = &
\one \after \idmap
& = &
\one.
\end{array}$$

\noindent As a result, $\one \after g \after \pi_{p} = p^{\bot} \after
\pi_{p} = \zero$, so that $g \after \pi_{p} = \zero$ by
Definition~\ref{def:FinPACwE}~\eqref{def:FinPACwE:zero}. Hence:
$$\begin{array}{rcccccccl}
f \after \pi_{p}
& = &
(f \after \pi_{p}) \ovee (g \after \pi_{p})
& = &
(f \ovee g) \after \pi_{p}
& = &
\idmap \after \pi_{p}
& = &
\pi_{p}.
\end{array}$$ 

\noindent By definition, $p^{\bot} = \ker(f)$, so that there is a
unique total map $\overline{f} \colon X/p^{\bot} \rightarrow X$ with
$f = \overline{f} \after \xi_{p^\bot}$. But then:
$$\begin{array}{rcccccl}
\pi_{p}
& = &
f \after \pi_{p}
& = &
\overline{f} \after \xi_{p^\bot} \after \pi_{p}
& = &
\overline{f} \after \theta_{p}.
\end{array}$$

\noindent Hence $\overline{f} = \pi_{p} \after \theta_{p}^{-1}$, using
that $p$ is sharp. Now we can conclude $f = \overline{f} \after
\xi_{p^\bot} = \pi_{p} \after \theta_{p}^{-1} \after \xi_{p^{\bot}} =
\asrt_{p}$. In a similar way one obtains $g = \asrt_{p^\bot}$.

In this situation the decomposition map $\dc_{p} =
\dtuple{\xi_{p^\bot}, \xi_{p}} \colon X \rightarrow X/p^{\bot} + X/p$
is an isomorphism.  We already know that it has a right inverse
$[\pi_{p}, \pi_{p^\bot}] \after
(\theta_{p}^{-1}+\theta_{p^\bot}^{-1})$ by
Lemma~\ref{lem:cmprquotcanmap}~\eqref{lem:cmprquotcanmap:triangle}.
We also have:
$$\begin{array}[b]{rcl}
\lefteqn{[\pi_{p}, \pi_{p^\bot}] \after (\theta_{p}^{-1}+\theta_{p^\bot}^{-1})
   \after \dc_{p}} \\
& = &
\nabla \after (\pi_{p}+\pi_{p^\bot}) \after (\theta_{p}^{-1}+\theta_{p^\bot}^{-1})
   \after \dtuple{\xi_{p^\bot}, \xi_{p}} \\
& \smash{\stackrel{\eqref{eqn:dtuplenat}}{=}} &
\nabla \after \dtuple{\pi_{p} \after \theta_{p}^{-1} \after \xi_{p^\bot}, 
   \pi_{p^\bot}  \after \theta_{p^\bot}^{-1} \after \xi_{p}} \\
& = &
\nabla \after \dtuple{\asrt_{p}, \asrt_{p^\bot}} \\
& = &
f \ovee g \\
& = &
\idmap.
\end{array}\eqno{\QEDbox}$$





\end{enumerate}
\end{proof}

If we add two more assumptions we can prove a lot more. At this
stage we start using images that are sharp.

\begin{proposition}
\label{prop:cmprquotimg}
Let $\cat{C}$ be an effectus with quotients and comprehension, in partial
form, which additionally:
\begin{itemize}
\item has sharp images

\item and satisfies, for all sharp predicates $p,q\in \Pred(X)$,
\begin{equation}
\label{equiv:projfulness}
\begin{array}{rcl}
p \leq q
& \qquad\Longleftrightarrow\quad &
\vcenter{\xymatrix@R-1pc@C-1.5pc{
\cmpr{X}{p}\ar@{ >->}[dr]_{\pi_p}\ar@{..>}[rr] & &
   \cmpr{X}{q}\ar@{ >->}[dl]^{\pi_q}
\\
& X &
}}
\end{array}
\end{equation}
\end{itemize}

\noindent Then we can prove the following series of results.
\begin{enumerate}
\item \label{prop:cmprquotimg:floorceildef} Take for an arbitrary
  predicate $p$,
$$\begin{array}{rclcrcl}
\floor{p} 
& \;\defeq\; &
\img(\pi_{p})
& \qquad\mbox{and}\qquad &
\ceil{p} 
& \;\defeq\; &
\floor{p^\bot}^{\bot}.
\end{array}$$

\noindent Then: $\floor{p}$\index{D}{$\floor{p}$, greatest sharp
  predicate below predicate $p$} is the greatest sharp predicate below
$p$, and $\ceil{p}$\index{D}{$\ceil{p}$, least sharp predicate above
  predicate $p$} is the least sharp predicate above $p$.

\item \label{prop:cmprquotimg:floorceileq} Let $p\in\Pred(X)$. For each
  \emph{total} map $f\colon Y\rightarrow X$ one has:
\begin{equation}
\label{eqn:partsubstfloorceil}
\begin{array}{rclcrcl}
\tbox{f}(p) = \one
& \Longleftrightarrow &
\tbox{f}(\floor{p}) = \one
& \quad\mbox{and} \quad &
\tbox{f}(p) = \zero
& \Longleftrightarrow &
\tbox{f}(\ceil{p}) = \zero.
\end{array}
\end{equation}

\noindent For an arbitrary map $g\colon Y \rightarrow X$ we have:
$$\begin{array}{rcl}
\pbox{g}(p) = \one
& \Longleftrightarrow &
\pbox{g}(\floor{p}) = \one.
\end{array}$$

\item \label{prop:cmprquotimg:cokernel} For a sharp predicate
  $p\in\Pred(X)$, the map $\asrt_{p}\colon X \rightarrow X$
  from~\eqref{diag:sharpassrt} satisfies $\img(\asrt_{p}) = p$. Hence
  by Lemma~\ref{lem:quot}~\eqref{lem:quot:coker} we have a cokernel
  map:
$$\xymatrix{
X\ar@/^1.5ex/[r]^-{\asrt_{p}}\ar@/_1.5ex/[r]_-{\zero} & 
   X\ar@{->>}[r]^-{\xi_{p}} & X/p
}$$

\item \label{prop:cmprquotimg:adj} For $f\colon X \rightarrow Y$ and
  $p\in\Pred(X)$ put:
$$\begin{array}{rcl}
\sum_{f}(p)
& \defeq &
\img\big(f \after \pi_{p}\big).
\end{array}$$

\noindent Then there is an adjunction:
\begin{equation}
\label{diag:shapredadj}
\vcenter{\xymatrix{
\ShaPred(X)\ar@/^2ex/[rr]^-{\sum_f} & \bot & 
   \ShaPred(Y)\ar@/^2ex/[ll]^-{\floor{\pbox{f}(-)}}
}}
\end{equation}

\noindent where $\ShaPred$\index{N}{$\ShaPred(X)$, collection of sharp
  predicates on an object $X$}\index{S}{sharp!-- predicate!collection
  of --} describes the subposet of sharp predicates.

\item \label{prop:cmprquotimg:projcomp} Comprehension maps
  $\pi$ are closed under composition, up-to-iso\-mor\-phism.

\item \label{prop:cmprquotimg:orthomodular} The posets $\ShaPred(X)
  \hookrightarrow \Pred(X)$ of sharp predicates are
  orthomodular\index{S}{orthomodular lattice} lattices.

\item \label{prop:cmprquotimg:factsyst} Comprehension maps $\pi$ and
  internal epis form a factorisation 
system\index{S}{factorisation system!-- via comprehension maps} in $\cat{C}$.
\end{enumerate}
\end{proposition}

\begin{proof}
The fact that images are assumed to be sharp plays an important role.
\begin{enumerate}
\item By definition the image $\floor{p} = \img(\pi_{p}) \colon X
  \rightarrow I$ is a sharp predicate. This $\floor{p}$ is below $p$
  by minimality of images, and $\pbox{\pi_{p}}(p) = (p^{\bot} \after
  \pi_{p})^{\bot} = \zero^{\bot} = \one$. If $s\in \Pred(X)$ is sharp
  with $s \leq p$, then $\pi_{s}$ factors through $\pi_{\floor{p}}$
  since:
\begin{equation}
\label{diag:projfloor}
\vcenter{\xymatrix@R-.5pc@C-.5pc{
\cmpr{X}{s}\ar@{ >->}[dr]_{\pi_{s}}\ar[r] & 
   \cmpr{X}{p}\ar@{ >->}[d]_{\pi_p}\ar@{..>}[r]^-{\cong} &
   \cmpr{X}{\floor{p}}\ar@{ >->}[dl]^{\pi_{\floor{p}}}
\\
& X &
}}
\end{equation}

\noindent The dashed arrow exists since by definition
$\pbox{\pi_{p}}(\floor{p}) = \pbox{\pi_{p}}(\img(\pi_{p})) = \one$.
Hence we get $s \leq \floor{p}$ by~\eqref{equiv:projfulness}.

We now have $\floor{p^\bot} \leq p^{\bot}$, and thus $p \leq
\floor{p^\bot}^{\bot} = \ceil{p}$. If $p \leq s$, where $s$ is sharp,
then $s^{\bot} \leq p^{\bot}$, so that $s^{\bot} \leq
\floor{p^{\bot}}$, and thus $\ceil{p} = \floor{p^\bot}^{\bot} \leq s$.

\item Let $f\colon Y \rightarrow X$ be a total map. If
  $\tbox{f}(\floor{p}) = \one$, then $\tbox{f}(p) = \one$ since
  $\floor{p} \leq p$. In the other direction, if $\tbox{f}(p) = \one$,
  then $f$ factors through $\pi_{p}$. But $f$ then also factors
  through $\pi_{\floor{p}}$, by the isomorphism in
  Diagram~\eqref{diag:projfloor}, so that $\tbox{f}(\floor{p}) =
  \one$.

We now get:
$$\begin{array}{rcl}
\tbox{f}(p) = \zero
& \Longleftrightarrow &
\tbox{f}(p^{\bot}) = \tbox{f}(p)^{\bot} = \one \\
& \Longleftrightarrow &
\tbox{f}(\floor{p^{\bot}}) = \one \qquad\qquad \mbox{as just shown} \\
& \Longleftrightarrow &
\tbox{f}(\ceil{p}) = \tbox{f}(\floor{p^\bot}^{\bot}) = 
   \tbox{f}(\floor{p^\bot})^{\bot} = \zero.
\end{array}$$

For an arbitrary map $g\colon Y \rightarrow X$, if
$\pbox{g}(\floor{p}) = \one$, then $\pbox{g}(p) = \one$ since
$\pbox{g}$ is monotone by Lemma~\ref{lem:partialmonotone}. In the
other direction, if $\pbox{g}(p) = \one$, then $\img(g) \leq p$ by
minimality of images, and thus $\img(g) \leq \floor{p}$ by
point~\eqref{prop:cmprquotimg:floorceildef}. But then $\one =
\pbox{g}(\img(g)) \leq \pbox{g}(\floor{p})$.

\item For a sharp predicate $p$ we have:
$$\begin{array}{rcll}
\img(\asrt_{p})
\hspace*{\arraycolsep}=\hspace*{\arraycolsep}
\img(\pi_{p} \after \theta_{p}^{-1} \after \xi_{p^\bot})
& = &
\img(\pi_{p}) \quad & 
   \mbox{by Lemma~\ref{lem:img}~\eqref{lem:img:compeq},~\eqref{lem:img:ext}} \\
& = &
\floor{p} \\
& = &
p   & \mbox{since $p$ is sharp.}
\end{array}$$

\item Let $f\colon X \rightarrow Y$ be an arbitrary map. We need to
  prove for sharp predicates $p\in\Pred(X)$ and $q\in\Pred(Y)$,
$$\begin{array}{rcl}
\sum_{f}(p) = \img(f \after \pi_{p}) \leq q
& \Longleftrightarrow &
p \leq \floor{\pbox{f}(q)}.
\end{array}$$

\noindent This works in the following way.
\begin{itemize}
\item Let $\sum_{f}(p) = \img(f \after \pi_{p}) \leq q$. We have:
$$\begin{array}{rcccccl}
\one
& = &
\pbox{(f \after \pi_{p})}(\img(f \after \pi_{p}))
& \leq &
\pbox{(f \after \pi_{p})}(q)
& = &
\pbox{\pi_{p}}\big(\pbox{f}(q)\big).
\end{array}$$

\noindent This implies that we have a map $\cmpr{X}{p} \rightarrow
\cmpr{X}{\pbox{f}(q)}$ commuting with the projections. We thus get $p
\leq \pbox{f}(q)$ by~\eqref{equiv:projfulness}, and thus $p \leq
\floor{\pbox{f}(q)}$.

\item In the other direction, let $p \leq \floor{\pbox{f}(q)} \leq
\pbox{f}(q)$. Then:
$$\begin{array}{rcccccl}
\one
& = &
\pbox{\pi_p}(p)
& \leq &
\pbox{\pi_{p}}\big(\pbox{f}(q)\big)
& = &
\pbox{(f \after \pi_{p})}(q).
\end{array}$$

\noindent But then $\sum_{f}(p) = \img(f \after \pi_{p}) \leq q$ by
minimality of images.
\end{itemize}

\item For two sharp predicates $p\in\Pred(X)$ and
  $q\in\Pred(\cmpr{X}{p})$ consider the sharp predicate on $X$ given
  by:
$$\begin{array}{rcl}
\sum_{\pi_{p}}(q)
& = &
\img(\pi_{p} \after \pi_{q}).
\end{array}$$

\noindent Notice that by Lemma~\ref{lem:img}~\eqref{lem:img:comp},
$$\begin{array}{rcccccccl}
\sum_{\pi_p}(q)
& = &
\img(\pi_{p} \after \pi_{q})
& \leq &
\img(\pi_{p})
& = &
\floor{p}
& = &
p.
\end{array}\eqno{(*)}$$

We are done if we can show that there are necessarily unique
maps $\varphi$, $\psi$ in a commuting square:
$$\xymatrix@C-1pc{
\quad\cmpr{\cmpr{X}{p}}{q}\;\ar@{ >->}[d]_{\pi_q}\ar@{..>}@/^2ex/[rr]^-{\varphi} 
   & \cong &
   \cmpr{X}{\sum_{\pi_p}(q)}\ar@{ >->}[d]^{\pi_{\sum_{\pi_p}(q)}}
      \ar@{..>}@/^2ex/[ll]^-{\psi}
\\
\cmpr{X}{p}\ar@{ >->}[rr]_-{\pi_p} & & X
}$$

\noindent The existence of the map $\varphi$ is easy since
comprehension maps are total and thus:
$$\begin{array}{rcccl}
\one
& = &
\pbox{(\pi_{p} \after \pi_{q})}\big(\img(\pi_{p} \after \pi_{q})\big)
& = &
\tbox{(\pi_{p} \after \pi_{q})}\big(\sum_{\pi_{p}}(q)\big).
\end{array}$$

\noindent In order to show the existence of the map $\psi$ we use
\emph{ad hoc} notation for:
$$\xymatrix{
p \sqcap q = \Big(X\ar[r]^-{\xi_{p^\bot}} & 
   X/p^{\bot}\ar[r]^-{\theta_{p}^{-1}}_-{\cong} &
   \cmpr{X}{p}\ar[r]^-{q} & I\Big)
}$$

\noindent Then:
$$\begin{array}{rcccccccl}
\tbox{\pi_{p}}(p \sqcap q)
& = &
(p \sqcap q) \after \pi_{p}
& = &
q \after \theta_{p}^{-1} \after \xi_{p^\bot} \after \pi_{p}
& = &
q \after \theta_{p}^{-1} \after \theta_{p}
& = &
q.
\end{array}$$

\noindent And thus:
$$\begin{array}{rcccccccl}
\pbox{(\pi_{p} \after \pi_{q})}(p \sqcap q)
& = &
\tbox{(\pi_{p} \after \pi_{q})}(p \sqcap q)
& = &
\tbox{\pi_q}\big(\tbox{\pi_p}(p \sqcap q)\big)
& = &
\tbox{\pi_{q}}(q)
& = &
\one.
\end{array}$$

\noindent The latter yields $\sum_{\pi_{p}}(q) = \img(\pi_{p} \after
\pi_{q}) \leq p \sqcap q$ by minimality of images. We now go back to
the inequality $\sum_{\pi_{p}}(q) \leq p$ from~$(*)$. It gives a total
map $f\colon \cmpr{X}{\sum_{\pi_p}(q)} \rightarrow \cmpr{X}{p}$ with
$\pi_{p} \after f = \pi_{\sum_{\pi_p}(q)}$. This $f$ factors through
$\pi_{q}$, and thus restricts to the required map $\psi$, since:
$$\begin{array}{rcl}
\tbox{f}(q)
\hspace*{\arraycolsep}=\hspace*{\arraycolsep}
q \after f
& = &
q \after \theta_{p}^{-1} \after \theta_{p} \after f \\
& = &
q \after \theta_{p}^{-1} \after \xi_{p^\bot} \after \pi_{p} \after f \\
& = &
(p\sqcap q) \after \pi_{\sum_{\pi_p}(q)} \\
& \geq &
\sum_{\pi_{p}}(q) \after \pi_{\sum_{\pi_p}(q)} \\
& = &
\one.
\end{array}$$

\item We first show how to obtain conjunction $\wedge$ for sharp 
predicates. For $p,q\in\ShaPred(X)$ define:
$$\begin{array}{rcccl}
p \wedge q
& = &
\sum_{\pi_p}\tbox{\pi_p}(q)
& = &
\img\big(\pi_{p} \after \tbox{\pi_{p}}(q)\big).
\end{array}$$

\noindent Since projections are closed under pullback --- see
Lemma~\ref{lem:totcmpr}~\eqref{lem:totcmpr:pb} --- we have a total
isomorphism $\varphi$ between two pullbacks in:
$$\vcenter{\xymatrix@R-1.5pc@C-1.5pc{
\cmpr{\cmpr{X}{p}}{\tbox{\pi_p}(q)}\ar@{..>}[dr]^(0.55){\varphi}_-{\cong}
   \ar@/^2ex/[drrrr]^-{\pi_{\tbox{\pi_p}(q)}}\ar@/_2ex/[dddr]
\\
& \cmpr{\cmpr{X}{q}}{\tbox{\pi_q}(p)}\ar[rrr]\ar@{ >->}[dd]^{\pi_{\tbox{\pi_q}(p)}}
   & & & \cmpr{X}{p}\ar@{ >->}[dd]^{\pi_p}
\\
\\
& \cmpr{X}{q}\ar@{ >->}[rrr]_-{\pi_q} & & & X
}}$$

\noindent Hence we can prove that $p\wedge q$ is a lower bound of
both $p$ and $q$ via Lemma~\ref{lem:img}~\eqref{lem:img:comp}:
$$\begin{array}{rcl}
p \wedge q
& = &
\img\big(\pi_{p} \after \tbox{\pi_{p}}(q)\big)
\hspace*{\arraycolsep}\leq\hspace*{\arraycolsep}
\img(\pi_{p})
\hspace*{\arraycolsep}=\hspace*{\arraycolsep}
\floor{p}
\hspace*{\arraycolsep}=\hspace*{\arraycolsep}
p
\\
p \wedge q
& = &
\img\big(\pi_{p} \after \tbox{\pi_{p}}(q)\big)
\hspace*{\arraycolsep}=\hspace*{\arraycolsep}
\img\big(\pi_{q} \after \tbox{\pi_{q}}(p) \after \varphi\big)
\hspace*{\arraycolsep}\leq\hspace*{\arraycolsep}
\img(\pi_{q})
\hspace*{\arraycolsep}=\hspace*{\arraycolsep}
\floor{q}
\hspace*{\arraycolsep}=\hspace*{\arraycolsep}
q.
\end{array}$$

\noindent We show that $p\wedge q$ is the greatest lower bound in
$\ShaPred(X)$: inequalities $r \leq p$ and $r\leq q$ yield maps
$f\colon \cmpr{X}{r} \rightarrow \cmpr{X}{p}$ and $g\colon \cmpr{X}{r}
\rightarrow \cmpr{X}{q}$ with $\pi_{p} \after f = \pi_{r} = \pi_{q}
\after g$.  The above pullback gives a mediating map $h\colon
\cmpr{X}{r} \rightarrow \cmpr{\cmpr{X}{p}}{\tbox{\pi_p}(q)}$ with
$\pi_{\tbox{\pi_p}(q)} \after h = f$. But then we are done:
$$\begin{array}{rcl}
r
\hspace*{\arraycolsep}=\hspace*{\arraycolsep}
\floor{r}
\hspace*{\arraycolsep}=\hspace*{\arraycolsep}
\img(\pi_{r})
\hspace*{\arraycolsep}=\hspace*{\arraycolsep}
\img(\pi_{p} \after f)
& = &
\img(\pi_{p} \after \pi_{\tbox{\pi_p}(q)} \after h) \\
& \leq &
\img(\pi_{p} \after \pi_{\tbox{\pi_p}(q)}) \\
& = &
p \wedge q.
\end{array}$$

\noindent We now also have joins $p \vee q$ in $\ShaPred(X)$ via De
Morgan: $p\vee q = (p^{\bot} \wedge q^{\bot})^{\bot}$. We prove the
orthomodularity law in the following way. Let $p\leq q$; we have to
prove $p \vee (p^{\bot} \wedge q) = q$. This is done essentially as
in~\cite[Prop.~1]{HeunenJ10}.
$$\begin{array}{rcll}
p \vee (p^{\bot} \wedge q)
& = &
(p\wedge q) \vee (p^{\bot} \wedge q) \\
& = &
\sum_{\pi_{q}}(p) \vee \sum_{\pi_{q}}(p^{\bot}) \\
& = &
\sum_{\pi_{q}}(p \vee p^{\bot}) 
   & \mbox{since $\sum_{\pi_{p}}$ is left adjoint, see~\eqref{diag:shapredadj}} \\
& = &
\sum_{\pi_{q}}(\one) & \mbox{since $p$ is sharp} \\
& = &
\sum_{\pi_{q}}(\tbox{\pi_q}(\one)) \\
& = &
q \wedge \one \\
& = &
q.
\end{array}$$

\item We finally have to show that comprehension maps form the
  `abstract monos' in a factorisation system, with internal epis as
  `abstract epis'. We recall the factorisation from
  Lemma~\ref{lem:partcmpr}~\eqref{lem:partcmpr:imgfact}:
$$\xymatrix@R-.5pc{
X\ar[rr]^-{f}\ar@{..>}@/_1ex/[dr]_-{\iep(f)} & & Y
\\
& \cmpr{Y}{\img(f)}\ar@{ >->}@/_1ex/[ur]_-{\pi_{\img(f)}}
}$$

\noindent The map $\iep(f)$ is the `internal epi' part of $f$. Let
$p\colon\cmpr{Y}{\img(f)} \rightarrow I$ be a predicate with
$\pbox{\iep(f)}(p) = \one$. If we can show that $p=\one$, then we know
that the map $\iep(f)$ is indeed internally epic. First we have:
$$\begin{array}{rcl}
\zero
\hspace*{\arraycolsep}=\hspace*{\arraycolsep}
\pbox{\iep(f)}(p)^{\bot}
\hspace*{\arraycolsep}=\hspace*{\arraycolsep}
p^{\bot} \after \iep(f)
& = &
p^{\bot} \after \theta_{\img(f)}^{-1} \after \theta_{\img(f)} \after \iep(f) \\
& = &
p^{\bot} \after \theta_{\img(f)}^{-1} \after \xi_{\imgbot(f)} \after 
   \pi_{\img(f)} \after \iep(f) \\
& = &
p^{\bot} \after \theta_{\img(f)}^{-1} \after \xi_{\imgbot(f)} \after f.
\end{array}$$

\noindent Lemma~\ref{lem:img}~\eqref{lem:img:post} then yields:
$$\begin{array}{rcl}
p^{\bot} \after \theta_{\img(f)}^{-1} \after \xi_{\imgbot(f)}
& \leq &
\imgbot(f).
\end{array}$$

\noindent But we also have, by
Lemma~\ref{lem:quot}~\eqref{lem:quot:XiComp},
$$\begin{array}{rcccl}
p^{\bot} \after \theta_{\img(f)}^{-1} \after \xi_{\imgbot(f)}
& \leq &
\imgbot(f)^{\bot}
& = &
\img(f).
\end{array}$$

\noindent Since images are sharp, by definition, we obtain:
$$\begin{array}{rcl}
p^{\bot} \after \theta_{\img(f)}^{-1} \after \xi_{\imgbot(f)}
& = &
\zero.
\end{array}$$

\noindent The fact that the map $\theta_{\img(f)}^{-1} \after
\xi_{\imgbot(f)}$ is epic --- see
Lemma~\ref{lem:quot}~\eqref{lem:quot:XiEpi} --- yields $p^{\bot} =
\zero$, and thus $p = \one$, as required.

There are a few more requirements of a factorisation system that we
need to check: the internal epis are closed under composition, by
Lemma~\ref{lem:img}~\eqref{lem:img:cat}; the comprehension maps are
closed under composition, by
point~\eqref{prop:cmprquotimg:projcomp}. We show that an internally
epic comprehension map is an isomorphism. Let $\pi_{p}$ be internally
epic, that is, $\one = \img(\pi_{p}) = \floor{p}$, so that
$\pi_{\floor{p}} = \pi_{\one}$ is an isomorphism by
Lemma~\ref{lem:totcmpr}~\eqref{lem:totcmpr:one}.  But $\pi_p$ is then
an isomorphism too, by the isomorphism in~\eqref{diag:projfloor}.

Finally we have to check that the diagonal-fill-in property holds:
in a commuting rectangle as below, we need a diagonal.
$$\xymatrix@C+2pc{
Y\ar[r]^-{f}_-{\text{internal epi }\;}\ar[d]_{g} & Z\ar[d]^{h}\ar@{..>}[dl]
\\
\cmpr{X}{p}\ar@{ >->}[r]_-{\pi_p} & X
}$$

\noindent We are done if we can show that $h$ factors through
$\pi_{p}$, that is, if $\pbox{h}(p) = \one$. By
Lemma~\ref{lem:img}~\eqref{lem:img:epichar} this is
equivalent to $\pbox{f}\big(\pbox{h}(p)\big) = \one$. But:
$$\begin{array}{rcccccl}
\pbox{f}\big(\pbox{h}(p)\big)
& = &
\pbox{g}\big(\pbox{\pi_p}(p)\big)
& = &
\pbox{g}(\one)
& = &
\one.
\end{array}\eqno{\QEDbox}$$
\end{enumerate}
\end{proof}

As a result of points~\eqref{prop:cmprquotimg:adj}
and~\eqref{prop:cmprquotimg:orthomodular} in
Proposition~\ref{prop:cmprquotimg} the assignment $X\mapsto
\ShaPred(X)$ gives a functor $\cat{C} \rightarrow \OMLatGal$, where
$\OMLatGal$ is the dagger kernel category of orthomodular lattices,
with Galois connections between them, as used in~\cite{Jacobs10b} (see
also~\cite{HeunenJ10}).

\begin{remark}
\label{rem:factsyst}
In an effectus as considered in Proposition~\ref{prop:cmprquotimg} we
have two factorisation systems, namely:
\begin{itemize}
\item the one with internal monos and quotient maps $\xi$ from
Proposition~\ref{prop:quotfactorisation};

\item the one with comprehension maps $\pi$ and internal epis from
  Propostion~\ref{prop:cmprquotimg}~\eqref{prop:cmprquotimg:factsyst}.
\end{itemize}

\noindent The question is how they are related. What is called the
`first isomorphism theorem' gives an answer. This `theorem' is the
familiar phenomenom, that can be expressed informally as $X/\ker(f)
\cong \img(f)$, for each map $f\colon X \rightarrow Y$. It holds for
many algebraic structures. However, this `theorem' does not hold in
effectuses, in general.

%

Abstractly this isomorphism is obtained via a canonical map from the
`coimage' to the `image', which is then an isomorphism. In the setting
of an effectus (in partial form), this canonical map is the dashed one
described in the following diagram --- which is the same as
diagram~(2) in~\cite{Grandis92} for the ideal of zero maps.
$$\xymatrix@C-.5pc{
\cmpr{X}{\ker(f)}\ar@{ >->}[r]^-{\pi_{\ker(f)}} & 
   X\ar@/^1.5ex/[rr]^-{f}\ar@/_1.5ex/[rr]_-{\zero}
      \ar@{->>}[d]_{\xi_{\floor{\ker(f)}}} & &
   Y\ar@{->>}[r]^-{\xi_{\img(f)}} & Y/\img(f)
\\
& X/\floor{\ker(f)}\ar@{..>}[rr] & & \cmpr{Y}{\img(f)}\ar@{ >->}[u]_{\pi_{\img(f)}}
}$$

\noindent We briefly explain this diagram, starting on the right.
\begin{itemize}
\item The map $\xi_{\img(f)}$ is the cokernel of $f$, see
  Lemma~\ref{lem:quot}~\eqref{lem:quot:coker}. The upgoing projection
  $\pi_{\img(f)}$ is its kernel map, because $\ker(\xi_{\img(f)}) =
  \img(f)$.

\item On the left, the projection $\pi_{\ker(f)}$ is the kernel map of
  $f$, see Lemma~\ref{lem:partcmpr}~\eqref{lem:partcmpr:kermap}. Its
  cokernel is the downgoing map $\xi_{\floor{\ker(f)}}$ since
  $\img(\pi_{\ker(f)}) = \floor{\ker(f)}$.
\end{itemize}

\noindent The dashed arrow exist, because the evident maps
$X/\floor{\ker(f)} \rightarrow Y$ and $X \rightarrow
\cmpr{Y}{\img(f)}$ restrict appropriately.

This canonical dashed map is in general not an isomorphism. For instance,
for a partial map $f\colon X \rightarrow \Dst(Y+1)$ in the Kleisli category
of the distribution monad $\Dst$ we have:
$$\begin{array}{rcl}
X/\floor{\ker(f)}
& = &
\setin{x}{X}{\floor{\ker(f)} < 1} \\
& = &
\setin{x}{X}{\ker(f)=0} \\
& = &
\setin{x}{X}{f(x)(*) = 0} \\
\cmpr{Y}{\img(f)}
& = &
\setin{y}{Y}{\img(f)(y) = 1} \\
& = &
\setin{y}{Y}{\exin{x}{X}{f(x)(y) > 0}}.
\end{array}$$

\noindent The canonical map $X/\floor{\ker(f)} \rightarrow
\cmpr{Y}{\img(f)}$ in the above diagram is given by $x \mapsto
\sum_{y, f(x)(y)>0} f(x)(y)\bigket{y}$, which is not an isomorphism.

More generally, a total internally epic map need not be an
isomorphism. The situation is reminiscent of Quillen model categories
as used in homotopy theory (see~\cite{Quillen67}), where one has two
factorisation systems, and mediating maps like the dashed one above
which are not necessarily isomorphisms.  However, in model categories
the relevant classes of maps satisfy closure properties which do not
hold here.

We have already hinted a few times that an effectus, in partial form,
is similar to, but more general, and less well-behaved, than an
Abelian category.\index{S}{Abelian!-- category} For such an Abelian
category $\cat{A}$ we do have a quotient-comprehension chain:
$$\vcenter{\xymatrix{
\Sub(\cat{A})\ar[d]_{\dashv\;}^{\;\dashv}
   \ar@/_8ex/[d]^{\;\dashv}_{\begin{array}{c}\scriptstyle\mathrm{Quotient} \\[-.7em]
                        \scriptstyle(U\rightarrowtail X)\mapsto X/U\end{array}}
   \ar@/^8ex/[d]_{\dashv\;}^{\begin{array}{c}\scriptstyle\mathrm{Comprehension} \\[-.7em]
                        \scriptstyle(U\rightarrowtail X) \mapsto U\end{array}} \\
\cat{A}\ar@/^4ex/[u]^(0.4){\zero\!}\ar@/_4ex/[u]_(0.4){\!\one}
}}$$

\noindent Quotients are obtained via cokernels: for a subobject $m
\colon U\rightarrowtail X$ one takes the codomain of the cokernel
$\coker(m)\colon X \twoheadrightarrow X/U$ as quotient unit.  In such
an Abelian category the first isomorphism theorem does hold.


The general theory that we are after will bear some resemblance to
recent work in (non-Abelian) homological algebra, see in
particular~\cite{Weighill14}, where adjunction chains like above are
studied, but also~\cite{Janelidze14,Grandis12}. There, part of the
motivation is axiomatising the category of (non-Abelian) groups,
following~\cite{MacLane50}. As a result, stronger properties are used
than occur in the current setting, such as the first isomorphism
theorem and left adjoints to substitution, corresponding to
bifibrations, which do not hold in general here.
\end{remark}

We still have to check that the assumptions in
Proposition~\ref{prop:cmprquotimg} make sense. This concentrates on
the equivalence in~\eqref{equiv:projfulness} saying that for sharp
predicates $p,q$ one has $p\leq q$ iff $\pi_{p} \leq \pi_{q}$, that
is, $\pi_{p}$ factors through $\pi_{q}$. We check this for the
effectus examples $\Sets$ --- in fact, more generally, for Boolean
effectuses --- for $\Kl(\Dst)$, and $\op{\vNA}$.

This equivalence~\eqref{equiv:projfulness} obviously holds in the
effectus $\Sets$, since the comprehension of a predicate $P\subseteq
X$ is given by $P$ itself.  It is less trivial that the
equivalence~\eqref{equiv:projfulness} also holds for Boolean
effectuses / extensive categories, see
Theorem~\ref{thm:extensiveeffectus}. Recall that all predicates are
automatically sharp in the Boolean case, see
Lemma~\ref{lem:booleffectus}~\eqref{lem:booleffectus:idemp}.

\begin{lemma}
Let $\cat{A}$ be an extensive category, understood as a Boolean
effectus with comprehension. Then $p \leq q$ iff $\pi_{p} \leq
\pi_{q}$ for all predicates $p,q$ on the same object.
\end{lemma}

\begin{proof}
We proceed in two steps, where we first prove that $\pi_{p} \leq
\pi_{q}$ implies $\pi_{q^\bot} \leq \pi_{p^{\bot}}$, and only then
that it also implies $p\leq q$.

So let $\pi_{p} \leq \pi_{q}$, say via a (necessarily unique) map
$\varphi \colon \cmpr{X}{p} \rightarrow \cmpr{X}{q}$ with $\pi_{q}
\after \varphi = \pi_{p}$. 
\begin{enumerate}
\item We first need to produce a map $\cmpr{X}{q^\bot} \rightarrow
  \cmpr{X}{p^\bot}$ commuting with the projections. Consider the
  following diagram in $\cat{A}$.
$$\xymatrix@R-.5pc{
A\ar@{ >->}[r]^-{g_1}\ar[dd]_{f_1}\pullback & 
   \cmpr{X}{q^\bot}\ar@{ >->}[d]^{\pi_{q^\bot}} & 
   B\ar@{ >->}[l]_{g_2}\ar[dd]^-{f_2}\pullback[dl]
\\
& X\ar[d]^-{[\pi_{p}, \pi_{p^\bot}]^{-1}}_-{\cong} &
\\
\cmpr{X}{p}\ar@{ >->}[r]_-{\kappa_1} & 
   \cmpr{X}{p}+\cmpr{X}{p^\bot} & \cmpr{X}{p^\bot}\ar@{ >->}[l]^-{\kappa_2}
}$$

\noindent We now have a situation:
$$\xymatrix@R-.5pc{
A\ar[d]_{f_1}\ar@/^2ex/[drr]^{g_1}\ar@{..>}[dr]^(0.6){\cong}
\\
\cmpr{X}{p}\ar@/_2ex/[dr]_{\varphi} & 
   0\ar[r]\ar[d]\pullback & 
   \cmpr{X}{q^\bot}\ar@{ >->}[d]^{\pi_{q^\bot}}
\\
& \cmpr{X}{q}\ar@{ >->}[r]_-{\pi_q} & Y+1
}$$

\noindent The outer diagram commutes since:
$$\begin{array}{rcl}
\pi_{q^\bot} \after g_{1}
& = &
[\pi_{p}, \pi_{p^\bot}] \after [\pi_{p}, \pi_{p^\bot}]^{-1} \after 
   \pi_{q^\bot} \after g_{1} \\
& = &
[\pi_{p}, \pi_{p^\bot}] \after \kappa_{1} \after f_{1} \\
& = &
\pi_{p} \after f_{1} \\
& = &
\pi_{q} \after \varphi \after f_{1}.
\end{array}$$

\noindent Thus, $A\cong 0$, since $0$ is strict, by
Lemma~\ref{lem:extensive}~\eqref{lem:extensive:coproj}. Since the
cotuple $[g_{1}, g_{2}] \colon A+B \rightarrow \cmpr{X}{q^\bot}$ is an
isomorphism, we obtain that $g_{2} \colon B \rightarrow
\cmpr{X}{q^\bot}$ is an isomorphism. But then $f_{2} \after g_{2}^{-1}
\colon \cmpr{X}{q^\bot} \rightarrow \cmpr{X}{p^\bot}$ is the required
map, since:
$$\begin{array}{rcl}
\pi_{p^\bot} \after f_{2} \after g_{2}^{-1}
& = &
[\pi_{p}, \pi_{p^\bot}] \after \kappa_{2} \after f_{2} \after g_{2}^{-1} \\
& = &
[\pi_{p}, \pi_{p^\bot}] \after [\pi_{p}, \pi_{p^\bot}]^{-1} \after \pi_{q^\bot} \\
& = &
\pi_{q^\bot}.
\end{array}$$

\item Our second aim is to prove $p\leq q$, assuming $\pi_{p} \leq
  \pi_{q}$ via the map $\varphi$. We define a predicate $r =
  \XI \after ((p\after\pi_{q})+\bang)  \after [\pi_{q},
    \pi_{q^\bot}]^{-1} \colon X \rightarrow 1+1$, and claim that
  $p\ovee r = q$. This proves $p\leq q$.

As bound $b \colon X \rightarrow (1+1)+1$ we take $b =
((p\after\pi_{q})+\bang) \after [\pi_{q}, \pi_{q^\bot}]^{-1}$. Then $\XI
\after b = r$ holds by construction. We have:
$$\begin{array}{rcl}
\IV \after b
& = &
[\idmap,\kappa_{2}] \after ((p\after\pi_{q})+\bang) \after 
   [\pi_{q}, \pi_{q^\bot}]^{-1} \\
& = &
[p\after\pi_{q}, \kappa_{2} \after\bang] \after [\pi_{q}, \pi_{q^\bot}]^{-1} \\
& = &
p.
\end{array}$$

\noindent The last equation follows from:
$$\begin{array}{rclcrcl}
[p\after\pi_{q}, \kappa_{2} \after\bang] 
& = &
p \after [\pi_{q}, \pi_{q^\bot}]
& \quad\mbox{and thus from}\quad &
p \after \pi_{q^{\bot}}
& = &
\zero.
\end{array}$$

\noindent As just shown we have $\pi_{q^\bot} \leq \pi_{p^\bot}$, say
via a map $\psi$. Then:
$$\begin{array}{rcccccl}
p \after \pi_{q^\bot}
& = &
p \after \pi_{p^\bot} \after \psi
& = &
\zero \after \psi
& = &
\zero.
\end{array}$$

\noindent We now obtain:
$$\begin{array}{rcl}
p \ovee r
\hspace*{\arraycolsep}=\hspace*{\arraycolsep}
(\nabla+\idmap) \after b
& = &
(\nabla+\idmap) \after ((p\after\pi_{q})+\bang) \after 
   [\pi_{q}, \pi_{q^\bot}]^{-1} \\
& = &
((\bang\after\pi_{q})+\bang) \after [\pi_{q}, \pi_{q^\bot}]^{-1} \\
& = &
(\bang+\bang) \after [\pi_{q}, \pi_{q^\bot}]^{-1} \\
& = &
q,
\end{array}$$

\noindent where the latter equation holds since:
$$\begin{array}{rcccccl}
q \after [\pi_{q}, \pi_{q^\bot}]
& = &
[q \after \pi_{q}, q \after \pi_{q^\bot}]
& = &
[\kappa_{1} \after \bang, \kappa_{2} \after \bang]
& = &
\bang+\bang.
\end{array}\eqno{\QEDbox}$$
\end{enumerate}
\end{proof}

\begin{example}
In the effectus $\Kl(\Dst)$ we have $p \leq q$ iff $\pi_{p} \leq
\pi_{q}$ for sharp $p,q$. Recall from
Example~\ref{ex:cmpr}~\eqref{ex:cmpr:KlD} that $\cmpr{X}{p} =
\set{x}{p(x) = 1}$, where $p(x),q(x)\in\{0,1\}$ because $p,q$ are
sharp. Let $\pi_{p} \leq \pi_{q}$, say via a function $\varphi \colon
\cmpr{X}{p} \rightarrow \Dst(\cmpr{X}{q})$ with $\pi_{q} \after
\varphi = \pi_{q}$. We have to prove $p(x) = 1 \Rightarrow q(x) =
1$. So assume $p(x)=1$, so that $x\in\cmpr{X}{p}$. Write $\varphi(x) =
\sum_{i}r_{i}\ket{x_i}$ with $x_{i}\in\cmpr{X}{q}$. The equation
$\pi_{p} = \pi_{q} \after \varphi$ yields:
$$\begin{array}{rcccccl}
1\ket{x}
& = &
\pi_{p}(x)
& = &
\big(\pi_{q} \after \varphi\big)(x)
& = &
\sum_{i}r_{i}\ket{x_i}.
\end{array}$$

\noindent This can only happen if $\varphi(x) = 1\ket{x}$, so that
$x\in\cmpr{X}{q}$. Hence $q(x)=1$.
\end{example}

Before showing the equivalence~\eqref{equiv:projfulness} for von
Neumann we give a general order result.

\begin{lemma}
\label{lem:vNdownsetproj}
Let~$\mathscr{A}$ be a von Neumann algebra with effects~$a, p \in
[0,1]_{\mathscr{A}}$.  If~$p$ is a projection (\textit{i.e.}~is
sharp), then the following are equivalent.
    \begin{multicols}{2}
    \begin{enumerate}
        \item $a \leq p$
        \item $ap = a$
        \item $pa = a$
        \item $pap = a$
    \end{enumerate}
    \end{multicols}
\noindent In particular, if~$a \leq p$, then~$a$ and~$p$ commute.
\end{lemma}

\begin{proof}
We will prove~(1) $\Rightarrow$ (2) $\Rightarrow$ (3) $\Rightarrow$
(4) $\Rightarrow$ (1).  Without loss of generality, we may
assume~$\mathscr{A}$ is a von Neumann algebra of operators on a
Hilbert space~$\mathscr{H}$.

(1) $\Rightarrow$ (2)\ Assume~$a \leq p$.  By definition,
$\inprod{av}{v} \leq \inprod{pv}{v}$ for all~$v \in \mathscr{H}$.
Thus $\| \sqrt{a} v \|^2 = \inprod{av}{v} \leq \inprod{pv}{v} =
\|pv\|^2$.  Hence~$pv=0$ implies~$av = 0$.  In particular, as~$p(1-p)v
= 0$ we have~$a (1-p)v=0$.  We compute~$av=apv+a(1-p)v=apv$.  Hence~$a
=ap$.

(2) $\Rightarrow$ (3)\ Assume~$ap = a$. Directly~$a=a^* = (ap)^* =
p^*a^*=pa$.

(3) $\Rightarrow$ (4)\ Assume~$pa = a$. Clearly~$pap=ap=(pa)^*=a^*=a$.

(4) $\Rightarrow$ (1)\ Assume~$pap = a$.  For any~$v \in \mathscr{H}$,
we have~$ \inprod{av}{v} = \inprod{papv}{v} = \inprod{apv}{pv} \leq
\inprod{pv}{pv} = \inprod{p^2v}{v} = \inprod{pv}{v}$.  Consequently~$a
\leq p$. \QED
\end{proof}

\begin{example}
The equivalence~\eqref{equiv:projfulness}, namely $p \leq q$ iff
$\pi_{p} \leq \pi_{q}$ for sharp $p,q$, also holds in the effectus
$\op{\vNA}$ of von Neumann algebra. Let $p,q\in[0,1]_{\mathscr{A}}$ be
sharp, that is, be projections, so that $p\cdot p = p$ and $q\cdot q =
q$. We first prove that $\img(\pi_{p}) = p$, and similarly for
$q$. Since $\pbox{\pi_p}(p) = \one$ we always have $\img(\pi_{p}) \leq
p$, so we only have to prove `$\geq$'.

Recall from Example~\ref{ex:cmpr}~\eqref{ex:cmpr:vNA} that the
projection $\pi_p$ is the map $\pi_{p} \colon \mathscr{A} \rightarrow
p\mathscr{A}p$ in $\vNA$ given by $p(x) = pxp$. Here we use that $p$
is sharp, and thus $\floor{p} = p$. According to~\eqref{eqn:imgvNA}
the image of $\pi_p$ is given by:
$$\begin{array}{rcl}
\img(\pi_{p}) 
& = & 
\displaystyle\bigwedge\setin{s}{\mathscr{A}}{s \mbox{ is
    a projection with } \pi_{p}(s) = psp = p = \pi_{p}(1)}.
\end{array}$$

\noindent We thus need to prove $s \geq p$ for a projection $s$ with
$p s p = p$. Write $t = s^{\bot} = 1-s$, so that $p t p = p p - p s p =
p - p = 0$. Hence:
$$\begin{array}{rcccccccl}
0
& = &
\|p t p \|
& = &
\|p t t p\|
& = &
\|(p t)\cdot (p t)^{*}\|
& = &
\|p t\|^{2}.
\end{array}$$

\noindent But then $0 = p t = p - ps$, so that $ps = s$. Hence $p \leq
s$ by Lemma~\ref{lem:vNdownsetproj}.

We can now reason abstractly: if $\pi_{p} \leq \pi_{q}$ via a map
$\varphi \colon \cmpr{X}{p} \rightarrow \cmpr{X}{q}$ with $\pi_{q}
\after \varphi = \pi_{q}$ in $\op{\vNA}$, then we are done by
Lemma~\ref{lem:img}~\eqref{lem:img:comp}:
$$\begin{array}{rcccccccl}
p
& = &
\img(\pi_{p})
& = &
\img(\pi_{q} \after \varphi)
& \leq &
\img(\pi_{q})
& = &
q.
\end{array}$$

\noindent This concludes the example.
\end{example}

\section{Towards non-commutative effectus theory}\label{sec:noncomm}

In the preceding sections we have presented the first steps of the
theory of effectuses. These sections provide a solid foundation.  At
this stage we have reached the boundaries of what we know for sure. We
add one more section that is of a more preliminary nature. It contains
a list of postulates for a class of effectuses that is intended to
capture the essential aspects of von Neumann algebras. Although this
`axiomatisation' is by no means fixed, we do already have a name for
this notion, namely \emph{telos}.\index{S}{telos} It is an effectus
satisfying the postulates~\ref{post:assertexist} --
\ref{post:asrtquot} below.  Thus, we are the first to admit that the
notion of telos is poorly defined. But by making our current thoughts
about this topic explicit, we hope to generate more research, leading
eventually to a properly defined concept.

In the previous sections we have frequently associated a partial
`assert' map $\asrt_{p} \colon X \rightarrow X+1$ with a predicate
$p\colon X \rightarrow 1+1$. This predicate-action correspondence
occurs for instance in:
\begin{itemize}
\item Definition~\ref{def:commbool} where the (unique) existence of
  assert maps as inverse of kernel-supplements is the essence of the
  definitions of `commutative' and `Boolean' effectus;

\item Definition~\ref{def:cmprquot} where the map $\asrt_p$ can be
  defined, but only for a sharp predicate $p$.
\end{itemize}

\noindent
(Proposition~\ref{lem:sharpassert}~\eqref{lem:sharpassert:sef} says
that when these two cases overlap, the relevant assert maps coincide.)

These assert maps $\asrt_{p} \colon X \rightarrow X+1$ incorporate the
side-effect map associated with a predicate $p$. This predicate-action
correspondence is similar to the situation in Hilbert spaces, where a
projection may be described either as a closed subset (a predicate),
or as a function that projects all elements into this subset (an
`assert' map). These assert maps express the dynamic character of
quantum computation. They incorporate the side-effect of an
observation, if any.  In the commutative (and Boolean) case the assert
maps have no side-effect --- technically, since such assert maps are
below the identity. But in the proper quantum case assert maps need
not be below the identity.  This is already clear for the assert maps
associated with sharp elements in Definition~\ref{def:cmprquot}.

We thus have:
\begin{center}
\begin{tabular}{|c|c|}
\hline
\textbf{Situation} & \textbf{Question} \\
\hline\hline
\begin{minipage}{13em}
Assert maps are fully determined in the
commutative (and Boolean) case, and also for
sharp elements.
\end{minipage}
&
\begin{minipage}{13em}
Can we also define assert maps in general, in the non-commutative
non-sharp case, via a certain property, or do we have to assume them
as separate structure?
\end{minipage} \vrule height4.5em depth4em width0em \\
\hline
\end{tabular}
\end{center}

\noindent This question arises in particular in the effectus $\op{\vNA}$
of von Neumann algebras.

Recall from Diagram~\eqref{diag:sharpassrt} in Definition~\ref{def:cmprquot} that for a sharp predicate
$p$ the associated assert map is defined via:
\begin{equation}
\label{diag:sharpassrtrepeat}
\xymatrix{
\asrt_{p} = \Big(X\ar@{->>}[r]^-{\xi_{p^\bot}} & 
   X/p^{\bot}\ar[r]^-{\theta_{p}^{-1}}_-{\cong} & 
   \cmpr{X}{p}\ar@{ >->}[r]^-{\pi_p} & X\Big).
}
\end{equation}

\noindent This definition relies on the assumption that the map
$\theta_{p} = \xi_{p^\bot} \after \pi_{p} \colon \cmpr{X}{p}
\rightarrow X/p^{\bot}$ is an isomorphism when $p$ is sharp.  That
assumption does not hold in general, for non-sharp $p$.

Still, in the commutative effectus $\Kl(\sDst)$, in partial form, we
have, for a fuzzy predicate $p\in[0,1]^{X}$, an assert map $\asrt_{p}
\colon X \rightarrow \sDst(X)$ given by $\asrt_{p}(x) =
p(x)\ket{x}$. Intruigingly, this assert map can also be described via
quotient and comprehension, as composite in $\Kl(\sDst)$ of the form:
$$\xymatrix@R-2pc{
X\ar[r]^-{\xi_{p^{\bot}}} & 
   X/p^{\bot} = \set{x}{p(x) > 0} = \cmpr{X}{\ceil{p}}\ar[r]^-{\pi_{\ceil{p}}}
   & X \\
x\ar@{|->}[r] & p(x)\ket{x}\ar@{|->}[r] & p(x)\ket{x}
}$$

\noindent Here we use that there is an equality (or isomorphism)
$\cmpr{X}{\ceil{p}} \rightarrow X/p^{\bot}$. The map $\theta_{p}$ used
above is a special case, since $\ceil{p} = p$ for a sharp predicate
$p$.

It turns out that we can do the same for von Neumann algebras. For an
effect $e\in [0,1]_{\mathscr{A}}$ in a von Neumann algebra
$\mathscr{A}$ we can define, formally in $\Par(\op{\vNA})$,
$$\xymatrix{
\asrt_{e} = \Big(\mathscr{A}\ar[r]^-{\xi_e} & \mathscr{A}/e^{\bot} =
   \ceil{e}\mathscr{A}\ceil{e} = 
   \floor{\ceil{e}}\mathscr{A}\floor{\ceil{e}} = 
   \cmpr{\mathscr{A}}{\ceil{e}}\ar[r]^-{\pi_{\ceil{e}}} & \mathscr{A}\Big)
}$$

\noindent The equality in the middle was already mentioned at the end
of Example~\ref{ex:quot}~\eqref{ex:quot:vNA}. As subunital map
$\asrt_{e} \colon \mathscr{A} \rightarrow \mathscr{A}$ this composite
becomes:
\begin{equation}
\label{eqn:asrtvNA}
\begin{array}{rcl}
\asrt_{e}(a)
\hspace*{\arraycolsep}=\hspace*{\arraycolsep}
\xi_{e^{\bot}}\big(\pi_{\ceil{e}}(a)\big)
\hspace*{\arraycolsep}=\hspace*{\arraycolsep}
\xi_{e^{\bot}}\big(\ceil{e}\, a\, \ceil{e}\big) 
& = &
\sqrt{e}\, \ceil{e}\, a\, \ceil{e}\, \sqrt{e} \\
& = &
\sqrt{e}\, a\, \sqrt{e}.
\end{array}
\end{equation}

\noindent In this way we obtain L{\"u}ders rule,
see~\cite[Eq.(1.3)]{BuschS98}. The problem of definining assert maps
via comprehension and quotient is thus reduced to having an
isomorphism $\cmpr{X}{\ceil{p}} \cong X/p^{\bot}$. But: \emph{which
  isomorphism}? 

We have not found, in general, a canonically defined isomorphism
$\cmpr{X}{\ceil{p}} \cong X/p^{\bot}$ in an effectus that gives rise
to the canonical description~\eqref{eqn:asrtvNA} in the effectus of
von Neumann algebras. This remains largely an unsolved problem,
although the situation is clearer in the special case of von Neumann
algebras, see Subsection~\ref{subsec:findimnoncomm} below. A related
question is: is this apparent lack of canonicity a `bug' or a
`feature'?

\medskip

Accepting for the time being that there is no such canonical
isomorphism $\cmpr{X}{\ceil{p}} \cong X/p^{\bot}$, there are two
possible ways forward.
\begin{enumerate}
\item Simply assume some isomorphism $\cmpr{X}{\ceil{p}} \cong
  X/p^{\bot}$ and use it to define an assert map as
  in~\eqref{diag:sharpassrtrepeat}, but with $\pi_{\ceil{p}}$ instead
  of $\pi_p$. Then we can see which requirements we need for this
  isomorphism in order to prove reasonable properties about assert
  maps.

\item Simply assume maps $\asrt_{p} \colon X \rightarrow X+1$
  satisfying some reasonable properties which induce an isomorphism
  $\cmpr{X}{\ceil{p}} \cong X/p^{\bot}$.
\end{enumerate}

\noindent In the sequel we shall follow the second approach. We have
already seen various properties of assert maps --- \textit{e.g.}~in
Lemmas~\ref{lem:commeffectus} and~\ref{lem:sharpassert}. Hence we can
check if they hold in the effectus $\op{\vNA}$, and if so, use them as
a basis for our preliminary axiomatisation of what we call a telos.

(This same approach is followed in~\cite{Jacobs15a}, where instrument
maps $\instr_{p} = \dtuple{\asrt_{p}, \asrt_{p^\bot}} \colon X
\rightarrow X+X$ are assumed, satisfying certain properties, instead
of assert maps. But the difference between using instruments or assert
maps is inessential.)

Below we describe a number of postulates that together give a
preliminary description of the notion of telos. Each postulate
contains a requirement, and a short discussion about its rationale and
consequences.

\begin{postulate}
\label{post:effectus}
Each telos is a monoidal effectus with sharp images. We shall describe
it in partial form (with special object $I$, as usual).

In the sequel the monoidal structure, see Section~\ref{sec:monoidal},
plays a modest role, but it should be included since it is important
for combining operations.  Similarly, images are a basic ingredient,
see Subsection~\ref{subsec:img}.
\end{postulate}

The next postulate introduces assert maps as actions that are
associated with predicates. Given such maps, we define an `andthen'
operation $\andthen$ on predicates, as before: $p \andthen q = q
\after \asrt_{p}$, for predicates $p,q$ on the same object. In the
current general situation $\andthen$ is not commutative, like in
Section~\ref{sec:commbool}.

\begin{postulate}
\label{post:assertexist}
For each predicate $p$ on an object $X$ in a telos there is an assert
map $\asrt_{p} \colon X \rightarrow X$\index{N}{$\asrt_p$, assert map for predicate $p$!in a telos} such that:
\begin{enumerate}
\item \label{post:assertexist:ker} $\ker(\asrt_{p}) = p^{\bot}$, or
  equivalently, $\kerbot(\asrt_{p}) = p$;

\item \label{post:assertexist:img} $\img(\asrt_{p}) = \ceil{p}$, where
  $\ceil{p}$ is the least sharp predicate above $p$;

\item \label{post:assertexist:inj} if $f \leq \idmap[X]$, then $f =
  \asrt_{p}$ for $p = \kerbot(f) = \one \after f$;


\item \label{post:assertexist:sum} $\asrt_{[p,q]} = \asrt_{p} +
  \asrt_{q} \colon X+Y \rightarrow X+Y$ for $p\in\Pred(X)$,
  $q\in\Pred(Y)$;

\item \label{post:assertexist:tensor} $\asrt_{p\otimes q} = \asrt_{p}
  \otimes \asrt_{q} \colon X\otimes Y \rightarrow X\otimes Y$, where
  $p\otimes q \colon X\otimes Y \rightarrow I\otimes I \cong I$.

\item \label{post:assertexist:andthen} $\asrt_{p} \after \asrt_{p} =
  \asrt_{p \andthen p}$;

\item \label{post:assertexist:sharp} $ \asrt_p \after f = f \after
  \asrt_{\pbox{f}(p)}$ for any predicate $p$ and any map~$f\colon Y \to
  X$ that preserves sharp elements: $\pbox{f}(q)$ is sharp if $q$ is
  sharp.
\end{enumerate}
\end{postulate}

We check that these properties hold in our leading example of a telos:
the opposite $\op{\vNA}$ of the category of von Neumann algebras.

\begin{example}
\label{ex:assertexist}
The effectus~$\op{\vNA}$ of von Neumann algebras, with assert maps
given by~$\asrt_p(x) = \sqrt{p}\, x\, \sqrt{p}$ as
in~\eqref{eqn:asrtvNA}, satisfies the previous postulate. Clearly our
chosen~$\asrt_p$ map is subunital, linear and positive.  It is also
completely positive~\cite[Thm.~1]{Stinespring55} and
normal~\cite[Lem.~1.7.4]{Sakai71}.  See also appendix
of~\cite{WesterbaanW15}.  We cover the different postulates one at a
time.
\begin{enumerate}
\item Obviously, $\kerbot(\asrt_p)= \asrt_{p}(1) = \sqrt{p}\,
  1\, \sqrt{p} = p$.

\item From $\sqrt{p}\, \ceil{p} = \sqrt{p}$ we obtain $\sqrt{p}
  \, \ceil{p}^{\bot} \, \sqrt{p} = 0$ and thus $\ceil{p}^{\bot}
  \after \asrt_{p} = \zero$. But then $\pbox{\asrt_{p}}(\ceil{p}) =
  (\ceil{p}^{\bot} \after \asrt_{p})^{\bot} = \one$, and thus
  $\img(\asrt_{p}) \leq \ceil{p}$ by minimality of images.

Now we prove the reverse inequality $\ceil{p} \leq \img(\asrt_{p})$.
Write~$b = \imgbot(\asrt_p)$.  We have $b \after \asrt_{p} = \zero$ by
Lemma~\ref{lem:img}~\eqref{lem:img:post}, so that $\asrt_{p}(b) = 0$
by~\eqref{eqn:compappOUG}.  That is: $\sqrt{p} \, b \, \sqrt{p}
= 0$.  By the~$C^*$-identity~$\| b\, \sqrt{p} \|^2 = \|\sqrt{p}
\, b \, \sqrt{p} \| = 0$.  Hence~$b\, \sqrt{p} = 0$.  But
then also~$\sqrt{p}\, b = (b\, \sqrt{p})^* = 0$.  Thus~$p$
and~$b$ commute.
By~\eqref{eqn:vNAceilMap} we obtain~$\ceil{p}\, b =0 = b\,
\ceil{p}$. Since both $\ceil{p}$ and $b$ are sharp we obtain that the
sum $\ceil{p} + b$ is sharp too, and thus an effect.  The latter
yields $\ceil{p} + b \leq 1$, and thus~$\imgbot(\asrt_p) = b \leq
\ceil{p}^\bot$.  Consequently~$\ceil{p} \leq \img(\asrt_p)$ as
desired.

\item In Example~\ref{ex:commbool}~\eqref{ex:commbool:vNA} we have
  already shown that a subunital map $f\colon \mathscr{A} \rightarrow
  \mathscr{A}$ with $f \leq \idmap$ is of the form $f(x) = f(1)\,
  x$, where the element $f(1)\in [0,1]_{\mathscr{A}}$ is central in
  $\mathscr{A}$. Hence $\sqrt{f(1)}$ is central too, so that the
  effect $p = \kerbot(f) = f(1)$ satiesfies:
$$\begin{array}{rcccccl}
\asrt_{p}(x)
& = &
\sqrt{f(1)}\, x \, \sqrt{f(1)}
& = &
f(1)\, x
& = &
f(x).
\end{array}$$

\item Simply:
$$\begin{array}{rcl}
\asrt_{[p,q]}(x,y)
& = &
\sqrt{(p,q)}\, (x,y) \, \sqrt{(p,q)} \\
& = &
(\sqrt{p}\, x\, \sqrt{p}, \sqrt{q} \, y \, \sqrt{q}) \\
& = &
(\asrt_p \oplus \asrt_q)(x,y).
\end{array}$$

\item
The linear span of predicates~$p\sotimes q$ is ultraweakly dense
in~$\mathscr{A}\otimes\mathscr{B}$, see
e.g.~\cite[Prop.~4.5.3]{Cho14}.  As our chosen~$\asrt$ map is
ultraweakly continuous and linear, it is sufficient to show the equality for
product-predicates:
$$\begin{array}{rcl}
\asrt_{p \sotimes q}(x \sotimes y) 
& = &
\sqrt{p \sotimes q} \, (x \sotimes y) \, \sqrt{p \sotimes q} \\
& = &
(\sqrt{p} \sotimes \sqrt{q}) \, (x \sotimes y) \,
   (\sqrt{p} \sotimes \sqrt{q}) \\
& = &
(\sqrt{p} \, x \, \sqrt{p}) \sotimes (\sqrt{q} \, y \, \sqrt{q}) \\
& = &
\asrt_{p}(x) \sotimes \asrt_{q}(y) \\
& = &
(\asrt_{p} \otimes \asrt_{q})(x \sotimes y).
\end{array}$$

\noindent For the last step, see e.g.~\cite[Prop.~4.5.5]{Cho14}.

\item Note that~$p \andthen p = \sqrt{p}\, p\, \sqrt{p} = p^2$
and so:
$$\begin{array}{rcl}
\asrt_p(\asrt_p(x))
& = &
\sqrt{p}\, \sqrt{p}\, x\, \sqrt{p}\, \sqrt{p} \\
& = &
p\, x\, p \\
& = &
\sqrt{p^2} \, x \, \sqrt{p^2}
\hspace*{\arraycolsep}=\hspace*{\arraycolsep}
\asrt_{p \andthen p}(x).
\end{array}$$

\item A completely positive map $f\colon \mathscr{A} \rightarrow
  \mathscr{B}$ preserves projections if and only if it is
  multiplicative.  Hence if $f$ is sharp, then it is a (unital) and
  preserves multiplication and thus square roots too. Hence:
$$\begin{array}{rcccccl}
\tbox{f}\big(\asrt_{p}(x)\big)
& = &
f\big(\sqrt{p} \, x \, \sqrt{p}\big)
& = &
\sqrt{f(p)} \, f(x) \, \sqrt{f(p)}
& = &
\asrt_{\tbox{f}(p)}\big(f(x)\big).
\end{array}$$
\end{enumerate}

\noindent This concludes the example.
\end{example}

Point~\eqref{post:assertexist:ker} allows us to define (total)
instrument maps like in
Lemma~\ref{lem:commeffectus}~\eqref{lem:commeffectus:instr}:
\begin{equation}
\label{diag:instr}
\xymatrix@C+2pc{
X\ar[rr]^-{\instr_{p} \defeq \dtuple{\asrt_{p}, \asrt_{p^\bot}}} & & X+X
}\index{N}{$\instr_p$, instrument map for predicate $p$}
\end{equation}

\noindent The side-effect associated with the predicate $p$ is the map
$\nabla \after \instr_{p} \colon X \rightarrow X$. We call $p$
side-effect free if this map $\nabla \after \instr_{p}$ is the
identity. Also, it allows us to define sequential composition
`andthen' on predicates (on the same object) as:
\begin{equation}
\label{eqn:andthen}
\begin{array}{rcl}
p \andthen q
& = &
q \after \asrt_{p},
\end{array}
\end{equation}

\noindent see
Lemma~\ref{lem:commeffectus}~\eqref{lem:commeffectus:order}.
Moreover, we can define conditional states $\omega|_{p}$ in a quantum
context as normalisation of $\asrt_{p} \after \omega$, like in
Example~\ref{ex:Bayes}, if we additionally assume normalisation in a
telos.

\begin{postulate}
\label{post:asrtcmpr}
Each assert map in a telos has a total kernel map --- that is, 
a kernel map which is total --- written as:
$$\xymatrix@C+1pc{
\cmpr{X}{p^\bot}\ar@{ >->}[r]^-{\pi_{p^\bot}}_-{\text{total}} &
   X\ar@/^1.5ex/[r]^-{\asrt_p}\ar@/_1.5ex/[r]_-{\zero} & X
}$$

\noindent The comprehension notation is deliberate, since this kernel
map $\pi_{p^\bot}$ is a comprehension map (for $p^\bot$), as in
Definition~\ref{def:cmpr}~\eqref{def:cmpr:part}: let $f\colon Y
\rightarrow X$ satisfy $\pbox{f}(p^{\bot}) = \one$. Then:
$$\begin{array}{rcccccccl}
\zero
& = &
\pbox{f}(p^{\bot})^{\bot}
& = &
p \after f
& = &
\kerbot(\asrt_{p}) \after f
& = &
\one \after \asrt_{p} \after f.
\end{array}$$

\noindent But then $\asrt_{p} \after f = \zero$, by
Lemma~\ref{lem:zero}, so that $f$ factors through the kernel map
$\pi_{p^\bot}$.

Notice that the assert maps are structure, but their kernel maps are
determined up to isomorphism, since they form comprehension maps and
thus a right adjoint to the truth functor.
\end{postulate}

\begin{postulate}
\label{post:asrtquot}
Postulate~\ref{post:assertexist}~\eqref{post:assertexist:img} tells
that $\img(\asrt_{p}) = \ceil{p}$. In particular,
$\pbox{\asrt_{p}}(\ceil{p}) = \one$, so that we obtain a
factorisation:
$$\xymatrix@C+1pc@R-.5pc{
& \cmpr{X}{\ceil{p}}\ar@{ >->}[d]^{\pi_{\ceil{p}}}
\\
X\ar@{..>}@/^2ex/[ur]^-{\xi_{p^\bot}}\ar[r]_-{\asrt_p} & X
}$$

\noindent using that the kernel maps $\pi$ are comprehension maps, see
Postulate~\ref{post:asrtcmpr}. Then, using
Lemma~\ref{lem:ker}~\eqref{lem:ker:tot},
$$\begin{array}{rcccccl}
\ker(\xi_{p^\bot})
& = &
\ker(\pi_{\ceil{p}} \after \xi_{p^\bot})
& = &
\ker(\asrt_{p})
& = &
p^{\bot}.
\end{array}$$

\noindent We postulate that in a telos these maps $\xi_p$ are
universal, forming quotients.

In this way the equation $X/p^{\bot} = \cmpr{X}{\ceil{p}}$ that we
discussed in the beginning of this section is built in. Moreover, if
$p$ is sharp, then $\ceil{p}=p$, so that we have an equality
$X/p^{\bot} = \cmpr{X}{p}$, as in
Definition~\ref{def:cmprquot}. Below, in
Proposition~\ref{prop:telos}~\eqref{prop:telos:sharptheta} it is shown
that this means that the canonical map $\theta_p$
from~\eqref{diag:cmprquotcanmap} is an isomorphism. Hence the
properties of Lemmas~\ref{lem:cmprquotcanmap}
and~\ref{lem:sharpassert} hold in a telos.

Since the $\xi$'s are quotient maps and $\img(\asrt_{p}) = \ceil{p}$
we have a coequaliser (cokernel map) diagram by
Lemma~\ref{lem:quot}~\eqref{lem:quot:coker}:
$$\xymatrix{
X\ar@/^1.5ex/[r]^-{\asrt_{p}}\ar@/_1.5ex/[r]_-{\zero} & 
   X\ar@{->>}[r]^-{\xi_{\ceil{p}}} & X/\ceil{p}
}$$
\end{postulate}

This concludes our description of the notion of telos. We continue
with some basic properties that hold in a telos about the assert maps
and the `andthen' operator $p \andthen q$, written as $p \after q$
in~\cite{GudderG02}, and as $[p?](q)$ in~\cite{Jacobs15a}.  Below we
prove the first three of the five requirements for andthen
in~\cite[\S\S3]{GudderG02}, see points~\eqref{prop:telos:andthenhom}
and~\eqref{prop:telos:andthenzero} below. We also prove that sharpness
is related to idempotency of $\andthen$, like in $C^*$-algebras.

\begin{proposition}
\label{prop:telos}
Let $\cat{C}$ be a telos, that is, $\cat{C}$ is an effectus in partial
form, satisfying the postulates~\ref{post:assertexist} --
\ref{post:asrtquot}. Then the following properties hold.
\begin{enumerate}
\item \label{prop:telos:zerone} The assert maps satisfy $\asrt_{\one}
  = \idmap \colon X \rightarrow X$ and $\asrt_{\zero} = \zero \colon X
  \rightarrow X$ for the truth and falsity predicates $\one, \zero$ on
  $X$; moreover, $\asrt_{s} = s \colon I \rightarrow I$ for each
  scalar $s$.

\item \label{prop:telos:andthenhom} For each predicate $p$ on
  $X\in\cat{C}$ we have a map of effect algebras:
$$\xymatrix{
\Pred(X)\ar[rr]^-{p\andthen(-)} & & \downset p
}$$

\noindent As a result, $p \andthen q \leq p$. Moreover, by
point~\eqref{prop:telos:zerone}:
$$\begin{array}{rccclcrcccl}
\one\andthen p 
& = &
p
& = &
p \andthen \one
& \qquad\mbox{and}\qquad &
\zero\andthen p 
& = &
\zero
& = &
p \andthen \zero.
\end{array}$$

\item \label{prop:telos:andthenzero} If $p\andthen q = \zero$, then
$p \andthen q = q \andthen p$.

\item \label{prop:telos:sef} For a predicate $p\in\Pred(X)$ the
  following side-effect freeness formulations are equivalent.
\begin{enumerate}
\item \label{prop:telos:sef:p} $\asrt_{p} \leq \idmap$;

\item \label{prop:telos:sef:pbot} $\asrt_{p^{\bot}} \leq \idmap$;

\item \label{prop:telos:sef:instr} $\nabla \after \instr_{p} = \idmap$.
\end{enumerate}

\item \label{prop:telos:sharpidemp} There are equivalences:
$$\begin{array}{rcl}
\mbox{$p$ is sharp}
& \Longleftrightarrow &
p \andthen p = p.
\end{array}$$

\item \label{prop:telos:sharptheta} For a sharp predicate $p$, the map
  $\theta_{p} = \xi_{p^\bot} \after \pi_{p} \colon \cmpr{X}{p}
  \rightarrow X/p^{\bot}$ from~\eqref{diag:cmprquotcanmap} is the
  identity. In particular, a telos has both comprehension and
  quotients as in Definition~\ref{def:cmprquot}.

\item \label{prop:telos:projimg} For each predicate $p$ the image of
  the comprehension map $\pi_{p}$ is given by $\img(\pi_{p}) =
  \floor{p}$. Further, for sharp predicates $p,q$ on the same object
  we have $p \leq q$ iff $\pi_{p} \leq \pi_{q}$, like
  in~\eqref{equiv:projfulness}. Hence
  Proposition~\ref{prop:cmprquotimg} applies in a telos.
\end{enumerate}
\end{proposition}

\begin{proof}
We must be careful not to assume more about the assert maps than is
postulated above.
\begin{enumerate}
\item We use
  Postulate~\ref{post:assertexist}~\eqref{post:assertexist:inj} each
  time. First, the identity map $\idmap \colon X \rightarrow X$
  evidently satisfies $\idmap \leq \idmap$, so that $\idmap =
  \asrt_{\one \after \idmap} = \asrt_{\one}$. Similarly, $\zero \leq
  \idmap$, so that $\zero = \asrt_{\one \after \zero} =
  \asrt_{\zero}$.  Further, we have $\idmap = \one \colon I
  \rightarrow I$, see
  Lemma~\ref{lem:FinPACwE}~\eqref{lem:FinPACwE:idI}. Hence every
  scalar $s\colon I \rightarrow I$ satisfies $s \leq \one = \idmap$,
  and thus $s = \asrt_{\one \after s} = \asrt_{\idmap \after s} =
  \asrt_{s}$.

\item By
  Proposition~\ref{prop:effectusPCM}~\eqref{prop:effectusPCM:pres} we
  have:
$$\begin{array}{rcl}
p \andthen (q_{1} \ovee q_{2})
& = &
(q_{1} \ovee q_{2}) \after \asrt_{p} \\
& = &
(q_{1} \after \asrt_{p}) \ovee (q_{2} \after \asrt_{p})
\hspace*{\arraycolsep}=\hspace*{\arraycolsep}
(p\andthen q_{1}) \ovee (p\andthen q_{2}).
\end{array}$$

\noindent As a result, $p \andthen (-)$ is monotone, and in particular
$p \andthen q \leq p$. Since $p\andthen \one = p$ we obtain that $p
\andthen (-)$ is a map of effect algebras $\Pred(X) \rightarrow
\downset p$.

\item Let $p\andthen q = \zero$, then $q \after \asrt_{p} = p \andthen
  q = \zero$, so that $q \leq \imgbot(\asrt_{p}) = \ceil{p}^{\bot}$ by
  Lemma~\ref{lem:img}~\eqref{lem:img:post}. But then $\ceil{q} \leq
  \ceil{p}^{\bot}$, since $\ceil{p}^{\bot}$ is sharp, and thus $p \leq
  \ceil{p} \leq \ceil{q}^{\bot} = \imgbot(\asrt_{q})$. Hence, again by
  Lemma~\ref{lem:img}~\eqref{lem:img:post}, $q \andthen p = p \after
  \asrt_{q} = \zero = p \andthen q$.

\item For the implication $\eqref{prop:telos:sef:p} \Rightarrow
  \eqref{prop:telos:sef:pbot}$, let $\asrt_{p} \leq \idmap$. Then
  there is a map $f\colon X \rightarrow X$ with $\asrt_{p} \ovee f =
  \idmap$. This $f$ then also satisfies $f \leq \idmap$, so that $f =
  \asrt_{q}$ for $q = \kerbot(f)$ by
  Postulate~\ref{post:assertexist}~\eqref{post:assertexist:inj}. We
  have:
$$\begin{array}{rcccccccl}
\one
& = &
\kerbot(\idmap)
& = &
\kerbot(\asrt_{p} \ovee f)
& = &
\kerbot(\asrt_{p}) \ovee \kerbot(f)
& = &
p \ovee \kerbot(f).
\end{array}$$

\noindent Hence $p^{\bot} = \kerbot(f) = q$. But then $\asrt_{p^\bot} =
\asrt_{q} = f \leq \idmap$. 

For the implication $\eqref{prop:telos:sef:pbot} \Rightarrow
\eqref{prop:telos:sef:instr}$ we assume $\asrt_{p^\bot} \leq
\idmap$. By reasoning as before, we get $\asrt_{p^{\bot}} \ovee f =
\idmap$ for $f = \asrt_{p} \leq \idmap$. Hence $\asrt_{p} \ovee
\asrt_{p^\bot} = \idmap$, so that we are done as in the proof of
Lemma~\ref{lem:commeffectus}~\eqref{lem:commeffectus:instr}.

Finally, for $\eqref{prop:telos:sef:instr} \Rightarrow
\eqref{prop:telos:sef:p}$, assume an equality of total maps $\nabla
\after \instr_{p} = \idmap \colon X \rightarrow X$. Then:
$$\begin{array}{rcl}
\asrt_{p} \ovee \asrt_{p^{\bot}}
& = &
(\nabla \after \kappa_{1} \after \asrt_{p}) \ovee
   (\nabla \after \kappa_{2} \after \asrt_{p^{\bot}}) \\
& = &
\nabla \after \big((\kappa_{1} \after \asrt_{p}) \ovee
   (\kappa_{2} \after \asrt_{p^{\bot}})\big) \\
& \smash{\stackrel{\eqref{eqn:dtupleovee}}{=}} &
\nabla \after \dtuple{\asrt_{p}, \asrt_{p^\bot}} \\
& = &
\nabla \after \instr_{p} \\
& = &
\idmap.
\end{array}$$

\noindent Hence $\asrt_{p} \leq \asrt_{p} \ovee \asrt_{p^\bot} =
\idmap$.

\item We use equivalences:
$$\begin{array}{rcll}
p \mbox{ is sharp}
& \Longleftrightarrow &
\img(\asrt_{p}) = \ceil{p} \leq p \\
& \Longleftrightarrow &
p^{\bot} \leq \imgbot(\asrt_{p}) \\
& \Longleftrightarrow &
p \andthen p^{\bot} = p^{\bot} \after \asrt_{p} = \zero \qquad
   & \mbox{by Lemma~\ref{lem:img}~\eqref{lem:img:post}} \\
& \smash{\stackrel{(*)}{\Longleftrightarrow}} &
(p\andthen p)^{\bot} = \zero \mbox{ in }\downset p \\
& \Longleftrightarrow &
p \andthen p = p.
\end{array}$$

\noindent The marked equivalence uses that $p \andthen (-)$ is a map
of effect algebras $\Pred(X) \rightarrow \downset p$, see
point~\eqref{prop:telos:andthenhom}. Hence it preserves
orthosupplements.

\item Let $p$ be a sharp predicate. Then $p\andthen p = p$ by
  point~\eqref{prop:telos:sharpidemp}, and thus $\asrt_{p} \after
  \asrt_{p} = \asrt_{p\andthen p} = \asrt_{p}$ by
  Postulate~\ref{post:assertexist}~\eqref{post:assertexist:andthen}.
Using that $p = \ceil{p}$, this last equation yields:
$$\begin{array}{rcccccl}
\pi_{p} \after \xi_{p^\bot} \after \pi_{p} \after \xi_{p^\bot} 
& = &
\asrt_{p} \after \asrt_{p}
& = &
\asrt_{p}
& = &
\pi_{p} \after \xi_{p^\bot}.
\end{array}$$

\noindent But then $\theta_{p} = \xi_{p^\bot} \after \pi_{p} =
\idmap$, since $\pi_p$ is monic, and $\xi_{p^\bot}$ is epic.

\item For an arbitrary predicate $p$ we have $\img(\pi_{\ceil{p}}) =
  \img(\pi_{\ceil{p}} \after \xi_{p^{\bot}}) = \img(\asrt_{p}) =
  \ceil{p}$, since $\xi_{p^\bot}$ is externally, and thus internally,
  epic. In particular, $\img(\pi_{q}) = q$ if $q$ is sharp.

We always have $\img(\pi_{p}) \leq p$ by minimality of images. We show
that $\img(\pi_{p})$ is the greatest sharp predicate below $p$. If $q$
is sharp, and $q \leq p$, then there is a (total) map $f\colon
\cmpr{X}{q} \rightarrow \cmpr{X}{p}$ with $\pi_{p} \after f =
\pi_{q}$. But then we are done: $q = \img(\pi_{q}) \leq \img(\pi_{p})$,
where the inequality follows from minimality of images, and:
$$\begin{array}{rcccccccl}
\pbox{\pi_{q}}\big(\img(\pi_{p})\big)
& = &
\tbox{\pi_{q}}\big(\img(\pi_{p})\big)
& = &
\tbox{f}\tbox{\pi_{p}}\big(\img(\pi_{p})\big)
& = &
\tbox{f}(\one)
& = &
\one.
\end{array}$$

Next we prove the equivalence $p \leq q \Longleftrightarrow \pi_{p}
\leq \pi_{q}$ for sharp predicates $p,q$ on the same object. The
direction $(\Rightarrow)$ always holds. For $(\Leftarrow)$ we use $p =
\img(\pi_{p}) \leq \img(\pi_{q}) = q$. \QED
\end{enumerate}
\end{proof}

\begin{remark}\label{remark:nonstandardandthen}
Let us try to find out in which sense assert maps satisfying the above
postulates are uniquely determined. To this end, assume we have two
sets of assert maps, written as $\asrt_{p}$ and $\asrt'_{p}$. 

They both have kernel maps like in Postulate~\ref{post:asrtcmpr}, 
written as:
$$\xymatrix@C+0pc{
\cmpr{X}{p^\bot}\ar@{ >->}[r]^-{\pi_{p^\bot}} &
   X\ar@/^1.5ex/[r]^-{\asrt_p}\ar@/_1.5ex/[r]_-{\zero} & X
& &
\cmpr{X}{p^\bot}'\ar@{ >->}[r]^-{\pi'_{p^\bot}} &
   X\ar@/^1.5ex/[r]^-{\asrt'_p}\ar@/_1.5ex/[r]_-{\zero} & X
}$$

\noindent In Postulate~\ref{post:asrtcmpr} we have seen that the
kernel maps form comprehension maps, and are thus determined
up-to-isomorphism. This means that for each predicate $p$ there is a
(total) isomorphism $\varphi_p$ in a commuting triangle:
$$\xymatrix@C-1pc@R-.5pc{
\cmpr{X}{p}\ar@{ >->}[dr]_{\pi_{p}}\ar@{..>}[rr]^-{\varphi_p}_-{\cong} & &
   \cmpr{X}{p}'\ar@{ >->}[dl]^{\pi'_{p}}
\\
& X &
}$$

\noindent By factoring like in Postulate~\ref{post:asrtquot}, we
obtain two maps $\xi_{p^\bot}$ and $\xi'_{p^\bot}$ in:
$$\xymatrix@C+1pc@R-.5pc{
& & & \cmpr{X}{\ceil{p}}\ar@{ >->}[dl]^{\pi_{\ceil{p}}}
   \ar[dd]_{\cong}^{\varphi_{\ceil{p}}}\ar@{..>}@/^8ex/[dd]_{\cong}^{\psi_{p}}
\\
X\ar@{..>>}@/^2ex/[urrr]^-{\xi_{p^\bot}}\ar@/^1ex/[rr]^(0.6){\asrt_p}
   \ar@{..>>}@/_2ex/[drrr]_-{\xi'_{p^\bot}}\ar@/_1ex/[rr]_(0.6){\asrt'_p} & & X
\\
& & & \cmpr{X}{\ceil{p}}'\ar@{ >->}[ul]_{\pi'_{\ceil{p}}}
}$$

\noindent Since both $\xi$ and $\xi'$ are universal quotient maps,
there is a second isomorphism, written as $\psi_{p}$, with $\psi_{p}
\after \xi_{p^\bot} = \xi'_{p^\bot}$. The endo map
$\varphi_{\ceil{p}}^{-1} \after \psi_{p} \colon \cmpr{X}{\ceil{p}}
\rightarrow \cmpr{X}{\ceil{p}}$ satisfies:
$$\begin{array}{rcccl}
\pi_{\ceil{p}} \after \big(\varphi_{\ceil{p}}^{-1} \after \psi_{p}\big)
   \after \xi_{p^\bot}
& = &
\pi'_{\ceil{p}} \after \xi'_{p^\bot}
& = &
\asrt'_{p}.
\end{array}$$

\end{remark}

\subsection{Uniqueness of assert maps, in von Neumann algebras}\label{subsec:findimnoncomm}

We have defined a telos
to be a special type of effectus
(Postulate~\ref{post:effectus})
endowed with a family of assert maps
(Postulate~\ref{post:assertexist})
that gives us comprehension (Postulate~\ref{post:asrtcmpr})
and quotients (Postulate~\ref{post:asrtquot}).
We have devoted much effort
to see whether
this list of postules
(or any extension of it)
uniquely determines the assert maps.
In this section,
we will show that
for the telos of von Neumann algebras,
$\op{\vNA}$,
the assert maps are
uniquely determined,
and are given by  $\asrt_p(x) = \sqrt{p}x\sqrt{p}$,
if we add
Postulate~\ref{pos:asrt-duality}
to the list.
Whether
the assert maps are uniquely determined 
by these postulates in general
remains an open problem.

We shall call a map $f\colon X \rightarrow Y$ a \emph{comprehension
  projection}, of simply a \emph{comprehension map} if there is a
sharp predicate $q$ on $Y$ with an isomorphism:
$$\xymatrix@R-1pc@C-2pc{
\;\;X\; \ar[dr]_{f}\ar@{=}[rr]^-{\sim} & & \cmpr{Y}{q}\ar@{ >->}[dl]^{\pi_q}
\\
& Y &
}$$

\noindent Such a comprehension is automatically monic.


\begin{postulate}
\label{pos:asrt-duality}
Let~$p$ be any predicate on~$X$ and $\pi\colon 1 \rightarrowtail X$ a
state that is also a comprehension map (in a telos). Then:
$$\begin{array}{rcl}
\img (p * \pi)
& = &
\ceil{p \andthen \img(\pi)},
\end{array}$$

\noindent where $p * \pi = \asrt_{p} \after \pi$ and $p \andthen
\img(\pi) = \img(\pi) \after \asrt_{p}$ as in~\eqref{eqn:andthen}.
\end{postulate}

The substate~$p * \pi$ is an unnormalised version of the conditional
state~$\pi|_p$ of Example~\ref{ex:Bayes}, which makes sense even if
the validity probability~$\pi \models p$ is zero.  If~$(\pi \models p)
\neq 0$, we have
$$\begin{array}{rcl}
\pi|_{p} \after (\pi\models p)
& = &
p * \pi.
\end{array}$$

Before we come to the main result, we show that this additional
postulate holds in the telos of von Neumann algebras. In doing so we
use the following two properties. For a non-zero $r\in [0,1]$,
\begin{equation}
\label{eqn:vNAimageceil}
\begin{array}{rclcrcl}
\img(r\cdot f)
& = &
\img(f)
& \qquad\mbox{and}\qquad &
\ceil{r\cdot p}
& = &
p,
\end{array}
\end{equation}

\noindent for a subunital map $f$ and a projection $p$.

\begin{example}
\label{ex:comprstates}
The postulate~\ref{pos:asrt-duality} is true in the
effectus~$\op{\vNA}$ with the standard~$\asrt$-maps $\asrt_p(x) =
\sqrt{p}\cdot x\cdot \sqrt{p}$ from Example~\ref{ex:assertexist}.  To
demonstrate this, we will first study states that are comprehension
maps.  Let~$\mathscr{A}$ be a von Neumann algebra.  We will state and
prove a number of claims.
\begin{enumerate}
\item \label{ex:comprstates:downset} For a sharp predicate~$s$
  on~$\mathscr{A}$, the mapping:
$$\xymatrix@R-2pc{
\downset s \ar[rr] & & \Pred(\cmpr{\mathscr{A}}{s})
   \rlap{$\;=[0,1]_{\cmpr{\mathscr{A}}{s}}$}
\\
a\ar@{|->}[rr] & & \tbox{\pi_{s}}(a) = \pi_{s}(a) = sas
}$$

\noindent is an order isomorphism --- where the downset $\downset s$
is a subset of $\Pred(\mathscr{A}) = [0,1]_{\mathscr{A}}$. Recall that
$\cmpr{\mathscr{A}}{s} = s\mathscr{A}s$, see
Example~\ref{ex:cmpr}~\eqref{ex:cmpr:vNA}.

To see this, note that~$sas \leq s$ for any~$sas \in
[0,1]_{\cmpr{\mathscr{A}}{s}}$ since $s$ is the unit element in
$\cmpr{\mathscr{A}}{s} = s\mathscr{A}s$.  Hence there is an
inclusion-map~$j \colon [0,1]_{s\mathscr{A}s} \to \downset
s$. It is the inverse to the above map $\tbox{\pi_s}$ since:
$$\begin{array}{rcll}
\tbox{\pi_{s}}\big(j(sas)\big)
& = &
s(sas)s
\hspace*{\arraycolsep}=\hspace*{\arraycolsep}
sas  \qquad & \mbox{since $s$ is a projection}
\\
j\big(\tbox{\pi_{s}}(a)\big)
& = &
sas
\hspace*{\arraycolsep}=\hspace*{\arraycolsep} a & \mbox{by
  Lemma~\ref{lem:vNdownsetproj} since $a \leq s$.}
\end{array}$$

\item \label{ex:comprstates:minimal} A comprehension map $\pi_{s}
  \colon \mathscr{A} \rightarrow \cmpr{\mathscr{A}}{s}$ for a sharp
  predicate~$s$ on~$\mathscr{A}$ is a state if and only if~$s$ is a
  minimal projection.

First, if $\pi_s$ is a state, then $\cmpr{\mathscr{A}}{s} \cong \C$,
so that by the previous point we have an order isomorphism:
$$\begin{array}{rcccccl}
\downset s
& \cong &
\Pred\big(\cmpr{\mathscr{A}}{s}\big)
& \cong &
\Pred(\C)
& = &
[0,1].
\end{array}$$

\noindent If $t \in \downset s$ is a projection, then it corresponds
to sharp element in $[0,1]$. But there are only two sharp elements in
$[0,1]$, namely $0$ and $1$, so that $t = 0$ or $t = s$. Hence $s$ is a
minimal projection.

Conversely, if $s$ is a minimal projection, then
$\Pred(\cmpr{\mathscr{A}}{s}) \cong \downset s$ has only two
projections. This means that the von Neumann algebra
$\cmpr{\mathscr{A}}{s}$ itself has two projections, and is thus
isomorphic two the unique von Neumann algebra $\C$ with two
projections.

\item Let~$\mathscr{H}$ be any Hilbert space with element $v\in
  \mathscr{H}$. The projection $\ket{v}\bra{v}\in \B(\mathscr{H})$ is
  minimal, and the corresponding state $\pi_{v} \colon \B(\mathcal{H})
  \rightarrow \C$ given by $\pi_{v}(t) = \inprod{tv}{v}$ is a
  comprehension map. One calls a state of this form $\pi_v$ a vector
  state. The image in $\B(\mathscr{H})$ of such a vector state $\pi_v$
  is given by the projection $\ket{v}\bra{v}$. Any state
  $\B(\mathscr{H}) \rightarrow \C$ which is a comprehension map is of
  this form.

\item \label{ex:comprstates:form} Any state that is a projection is
  of the form $\pi \colon \B(\mathscr{H}) \oplus \mathscr{B}
  \rightarrow \C$, where $\pi(a, b) = \inprod{av}{v}$ for some vector
  $v\in \mathscr{H}$. The image of $\pi$ is then the pair
  $(\ket{v}\bra{v}, 0)\in \B(\mathscr{B})\oplus \mathscr{B}$.

To show this, assume~$\pi\colon \mathscr{A} \to \C$ is a comprehension
map for a (consequently minimal) projection~$s$.  Let~$c(s)$ denote
the central carrier of~$s$ that is: $c(s)\in [0,1]_{\mathscr{A}}$ is
the least central projection above~$s\in [0,1]_{\mathscr{A}}$.  We
will show $\cmpr{\mathscr{A}}{c(s)} \cong c(s)\mathscr{A}$ is a (type
I) factor, in which the projection~$c(s)s$ is minimal.


Let~$z \leq c(s)$ be any central projection in~$c(s) \mathscr{A}$.  If
we can show~$z=0$ or~$z=c(s)$, we may conclude~$c(s) \mathscr{A}$ is a
factor.  Note~$z$ is central in~$\mathscr{A}$ as~$za = zc(s)a =
c(s)az=az$ for any~$a \in \mathscr{A}$.  Clearly~$zs$ is a projection
below~$s$.  Hence by minimality of~$s$, we have~$zs=s$ or~$zs=0$.  For
the first case, assume~$zs = s$.  Then~$s \leq z$ by
Lemma~\ref{lem:vNdownsetproj}.  Hence~$c(s) \leq z \leq c(s)$.
Thus~$z=c(s)$, as desired.  Now, we cover the other case~$zs = 0$.
That is: $zs^\bot = z$.  Hence~$z \leq s^\bot$ by
Lemma~\ref{lem:vNdownsetproj}.  So~$s \leq z^\bot$.  Hence~$c(s) \leq
z^\bot$.  Thus~$z \leq c(s)^\bot$ and~$z \leq c(s)$.
Consequently~$z=0$.
 
A fundamental result says that each such (type I) factor is given by
bounded operators on a Hilbert space, see
\textit{e.g.}~\cite[Corrolary 10]{Topping71}. Thus,
let~$c(s)\mathscr{A} \cong \B(\mathscr{H})$ for some Hilbert
space~$\mathscr{H}$. Consequently, there is an
isomorphism~$\vartheta\colon \mathscr{A} \to \B(\mathscr{H}) \oplus
c(s)^\bot \mathscr{A}$ such that~$\pi = \pi' \after \pi_1 \after
\vartheta$, where~$\pi'\colon \B(\mathscr{H}) \rightarrow \C$ is a
comprehension map for the minimal projection corresponding to~$c(s)s$
and hence a vector state.

\item We will now show that the additional
  postulate~\ref{pos:asrt-duality} holds in the telos~$\op\vNA$ for
  the canonical~$\asrt$-maps from
  Example~\ref{ex:assertexist}. Let~$\pi$ be a state that is also a
  comprehension map.  With the previous in mind, we may assume without
  loss of generality that~$\pi$ is of the form~$\pi\colon
  \B(\mathscr{H})\oplus \mathscr{B} \to \C$ with~$\pi(a,b) =
  \inprod{av}{v}$ for some~$v \in \mathscr{H}$.  Let~$e=(e_1,e_2)$ and
  $d = (d_1, d_2)$ be arbitrary effects on~$\B(\mathscr{H}) \oplus
  \mathscr{B}$.  Note that:
$$\begin{array}{rcl}
(e * \pi) (d)
\hspace*{\arraycolsep}=\hspace*{\arraycolsep}
\pi\big(\asrt_{e}(d)\big)
& = &
\pi\big(\sqrt{e_1}d_1\sqrt{e_1}, \sqrt{e_2}d_2\sqrt{e_2}\big) \\
& = &
\inprod{\sqrt{e_1}d_1\sqrt{e_1}\,v}{v} \\
& = &
\inprod{d_1\sqrt{e_1}\,v}{\sqrt{e_1}\,v}.
\end{array}$$

\noindent Suppose~$\sqrt{e_1}v=0$.  Then~$(e*\pi)(d) =
\inprod{d_1\sqrt{e_1}v}{\sqrt{e_1}v} = 0$, so that~$\img(e * \pi) =
0$.  Also $\ceil{e \andthen \img(\pi)} = 0$ since:
$$\begin{array}{rcll}
e \andthen \img(\pi)
& = &
e \andthen (\ket{v}\bra{v},0) & \mbox{see point~\eqref{ex:comprstates:form}} \\
& = &
\asrt_{e}(\ket{v}\bra{v},0) \\
& = &
(\ket{\sqrt{e_1} v}\bra{v \sqrt{e_1}}, \sqrt{e_2}0\sqrt{e_2}) \\
& = &
0.
\end{array}$$

\noindent For the other case, assume~$\sqrt{e_1}v\neq 0$.  Then, 
using what we have seen above:
$$\begin{array}{rcccl}
(e*\pi)(d)
& = &
\inprod{d_1 \sqrt{e_1}v}{\sqrt{e_1}v}
& = &
\|\sqrt{e_1}v\|^2 \inprod{d_1 \frac{\sqrt{e_1}v}{\|\sqrt{e_1}v\|}}{
                        \frac{\sqrt{e_1}v}{\|\sqrt{e_1}v\|}}.
\end{array}$$

\noindent This means~$e*\pi$ is a scaled vector state with:
$$\begin{array}{rcll}
\img (e * \pi)
& = &
(\ket{ \frac{\sqrt{e_1}v}{\|\sqrt{e_1}v\|}}\bra{
            \frac{\sqrt{e_1}v}{\|\sqrt{e_1}v\|}}, 0) &
   \mbox{by~\eqref{eqn:vNAimageceil} and point~\eqref{ex:comprstates:form}} \\
& = &
\frac{1}{\|\sqrt{e_1}v\|^2} (\ket{\sqrt{e_1}v}\bra{\sqrt{e_1}v}, 0) \quad \\
& = &
\ceil{\,(\ket{\sqrt{e_1}v}\bra{\sqrt{e_1}v}, 0)\,} &
   \mbox{by~\eqref{eqn:vNAimageceil}} \\
& = &
\ceil{e \andthen \img(\pi)}.
\end{array}$$

\noindent Hence, in both cases~$\img (e * \pi) = \ceil{e \andthen
  \img(\pi)}$, as desired.
\end{enumerate}
\end{example}

Now we are ready to show there is only one choice of assert maps
in~$\op\vNA$ that satisfies all the postulates ---
including~\ref{pos:asrt-duality}.  The result is a reformulation of a
result from~\cite{WesterbaanW15}, which in turn is inspired by the
characterization of the sequential product in Hilbert spaces by Gudder
and Latr\'emoli\`ere, see~\cite{gudder}.  Our
Postulate~\ref{pos:asrt-duality} should be compared with their
Condition~1.

\begin{theorem}
\label{thm:uniqueness-asrt}
For each von Neumann algebra~$\mathscr{A}$
and $p\in [0,1]_{\mathscr{A}}$,
let
\begin{equation*}
\asrt_p\colon \mathscr{A}\longrightarrow\mathscr{A}
\end{equation*}
be a completely positive normal subunital map.
Assume that these assert maps
on~$\op{\vNA}$
satisfy Postulate~\ref{post:effectus},
\ref{post:assertexist},
\ref{post:asrtcmpr},
\ref{post:asrtquot},
and~\ref{pos:asrt-duality}.

Then
for every von Neumann algebra~$\mathscr{A}$,
predicate $p\in[0,1]_\mathscr{A}$,
and $x\in\mathscr{A}$,
\begin{equation}
\label{eq:uniqueness-asrt}
\asrt_p(x) \ = \ \sqrt{p} \,x \,\sqrt{p}.
\end{equation}
\end{theorem}

\begin{proof}
Let~$\mathscr{H}$ be a Hilbert space.
We will first show that Equation~\eqref{eq:uniqueness-asrt}
holds for $\mathscr{A} = \B(\mathscr{H})$.
Let $p\in [0,1]_{\B(\mathscr{H})}$ be given.

By the discussion in Remark~\ref{remark:nonstandardandthen}
we already have the following
connection between the canonical assert map $a\mapsto \sqrt{p}a\sqrt{p}$
and the one, $\asrt_p$, we are given:
there is an automorphism $\vartheta$ on $\ceil{p}\B(\mathscr{H})\ceil{p}$
such that, for all~$a\in \B(\mathscr{H})$,
\begin{equation*}
\asrt_p(a)\ = \ \sqrt{p}
\ \vartheta(\,\ceil{p}a\ceil{p}\,)\,\sqrt{p}.
\end{equation*}

Since $\ceil{p}\B(\mathscr{H})\ceil{p}$
is a type~I factor (i.e.~isomorphic to a $\B(\mathscr{K})$),
it is known 
(see Theorem~3 of~\cite{kaplansky1952})
that~$\vartheta$ 
must be what is called an \emph{inner automorphism},
that is,
there is an unitary $u\in \ceil{p}\B(\mathscr{H})\ceil{p}$
such that,
$\vartheta(a) \ =\  u^* a u$
for all~$a\in \ceil{p}\B(\mathscr{H})\ceil{p}$.
Note that $\ceil{p}u=u$ since $u\in \ceil{p}\B(\mathscr{H})\ceil{p}$,
and
thus we have, for all~$a\in \B(\mathscr{H})$,
\begin{equation*}
\asrt_p(a)\ = \ \sqrt{p}
\,u^* \,a\,u\,\sqrt{p}.
\end{equation*}
Of course,
our ultimate goal should be to show that~$u=1$,
or at least that $u=\lambda\cdot 1$
for some $\lambda\in\mathbb{C}$ with $|\lambda|=1$.
Our first step is to prove $up=pu$.

To this end, we extract
some information about~$u$ from Postulate~\ref{pos:asrt-duality}.
Let~$x\in\mathscr{H}$
with~$\|x\|=1$ be given.
Let $\pi\colon\B(\mathscr{H})\to \mathbb{C}$
be given by~$\pi(a)=\inprod{x}{ax}$
for $a\in\B(\mathscr{H})$.
Then by Example~\ref{ex:comprstates} 
we know that~$\pi$ is a comprehension.
Thus, by Postulate~\ref{pos:asrt-duality},
we know that, in $\op{\vNA}$,
\begin{equation}
\label{eq:duality-vector-state}
\img(\asrt_p \circ \pi)
\ = \ \ceil{p\andthen \img(\pi)}.
\end{equation}
Before we continue,
observe that for $y \in \mathscr{H}$ with~$\|y\| \leq 1$
we have
\begin{equation*}
\ceil{\,\ket{y}\bra{y}\,} 
\ = \ 
\img(\,\inprod{y}{(-)\,y})\,)
\ = \ (\text{ projection onto }y\mathbb{C}\equiv
\{\lambda y\colon \lambda\in\mathbb{C}\}\ ).
\end{equation*}
Now, let us unfold Equation~\eqref{eq:duality-vector-state}.
\begin{alignat*}{3}
(\ \text{projection onto }(\sqrt{p}u^*x)\mathbb{C}\ ) \ 
=\ &\ceil{\ \ket{\,\sqrt{p}u^*x\,}\bra{\sqrt{p}u^*x\,}\ } \\
=\ &\ceil{\ \sqrt{p}u^* \,\ket{x}\bra{x}\, u\sqrt{p} \ }\\
=\ &\ceil{\ p\andthen \img(\pi) \  }\\
=\ &\img(\asrt_p\circ \pi) \\
=\ &\img(\ \inprod{x}{\,\sqrt{p}u^*(-)u\sqrt{p}x\,} \ ) \\
=\ &\img(\ \inprod{\,u\sqrt{p}x\,}{\,(-)\,u\sqrt{p}x\,} \ ) \\
=\ &(\ \text{projection onto }(u\sqrt{p}x)\mathbb{C}\ )
\end{alignat*}
Hence, 
for every $x\in \mathscr{H}$ with~$\|x\|=1$
there is~$\alpha\in \mathbb{C}$
with $\alpha\neq 0$ such that
\begin{equation*}
\sqrt{p}u^* x \ =\  \alpha \cdot u\sqrt{p} x.
\end{equation*}
By scaling it is clear that 
this statement is also true
for all~$x\in \mathscr{H}$
(and not just for $x\in \mathscr{H}$
with $\|x\|=1$).
While a priori~$\alpha$ might depend on~$x$,
we will show that there is~$\alpha\in \mathbb{C}\backslash\{0\}$
such that $\sqrt{p}u^* = \alpha\cdot u\sqrt{p}$.

First note that  $\sqrt{p}u^*x = 0$
iff $u\sqrt{p}x=0$ for all~$x\in \mathscr{H}$.
Thus we may factor~$\sqrt{p}u^*$
and~$u\sqrt{p}$ through
the quotient map  $q\colon \mathscr{H}\to\mathscr{H}/{K}$,
where
\begin{equation*}
K
\ =\ \{x\in\mathscr{H}\colon \sqrt{p}u^*x=0\}
\ =\ \{x\in\mathscr{H}\colon u\sqrt{p}x=0\}.
\end{equation*}
(Here $\mathscr{H}/K$
is just the quotient of~$\mathscr{H}$ as vector
space.)
Let~$t,s\colon \mathscr{H}/K\to \mathscr{H}$
be given by $s\circ q = \sqrt{p}u^* $
and $t\circ q = u\sqrt{p}$.
Then~$s$ and~$t$ are injective,
and writing $V=\mathscr{H}/K$,
it is not hard to see that
for every $x\in V$
there is $\alpha\in \mathbb{C}\backslash\{0\}$
with $s(x)=\alpha \cdot t(x)$. 

We will show that there is $\alpha\in \mathbb{C}\backslash\{0\}$
with $s = \alpha\cdot t$.
If~$V=\{0\}$ then this is clear, so assume that $V\neq \{0\}$.
Pick $x\in V$ with~$x\neq 0$, and let~$\alpha\in\mathbb{C}\backslash\{0\}$
be such that $s(x) = \alpha\cdot t(x)$.
Let~$y\in V$ be given;
we must show that $s(y)=\alpha\cdot t(y)$.
Now,  either~$t(x)$ and~$t(y)$ are linearly independent
or not.

Suppose that $t(x)$ and $t(y)$ are linearly independent.
Let $\beta,\gamma\in \mathbb{C}\backslash\{0\}$
be such that
$s(y)=\beta \cdot t(y)$
and $s(x+y) = \gamma\cdot t(x+y)$.
Then 
\begin{equation*}
(\gamma-\alpha)\cdot t(x)\,+\,(\gamma-\beta)\cdot t(y) \ = \ 0.
\end{equation*}
Thus, as $t(x)$ and~$t(y)$
are linearly independent,
$\gamma-\alpha=0$ and $\gamma-\beta=0$,
and so $\alpha=\beta$.
Hence $s(y) = \alpha\cdot  t(y)$.

Suppose that $t(x)$ and $t(y)$ are linearly dependent.
Since~$x\neq 0$, we have $t(x)\neq 0$---as~$t$ is injective---,
and thus $t(y)=\varrho t(x)$ for some~$\varrho\in\mathbb{C}$.
Then $t(y-\varrho x)=0$, and so $y=\varrho x $
since~$t$ is injective.
Then
\begin{equation*}
 s(y)
\ = \ \varrho s(x)
\ = \ \varrho \alpha t(x)
\ = \ \alpha t(y).
\end{equation*}
Thus, in any case, $ s(y) = \alpha t(y)$.
Hence $s=\alpha\cdot t$.
Thus $\sqrt{p}u^* = \alpha\cdot u\sqrt{p} $.

It follows that~$p = \sqrt{p} u^* u\sqrt{p}
= \alpha\cdot u\sqrt{p} u\sqrt{p}
= u\sqrt{p} \sqrt{p} u^* = upu^*$,
and so $pu=up$.
Then also $\sqrt{p}u = u\sqrt{p}$.
Thus~$\sqrt{p}u^* = \alpha u\sqrt{p} = \alpha \sqrt{p} u$.

Now, note that $(\sqrt{p}u^*)^*=u\sqrt{p}$,
and so $u\sqrt{p}=\alpha^*\sqrt{p}u^*=\alpha^*\alpha u\sqrt{p}$.
Then if~$u\sqrt{p}\neq 0$ we get $\alpha^*\alpha =1$,
and if~$u\sqrt{p}$ then we can put~$\alpha=1$ 
and still have both $\sqrt{p}u^* = \alpha u\sqrt{p}$
and $\alpha^*\alpha=1$.

It follows that, for all~$b\in\B(\mathscr{H})$,
\begin{equation*}
\sqrt{p} u^* \,b\, u\sqrt{p}
\ = \ 
\sqrt{p} u\, b\, u^* \sqrt{p}.
\end{equation*}
By the universal
property of the quotient, 
we conclude that~$u^*  (-)  u = u (-) u^*$,
and thus $u^2 b = b u^2$
for all~$b\in\B(\mathscr{H})$.
Hence~$u^2$ is central in~$\B(\mathscr{H})$.

Since~$\B(\mathscr{H})$ is a factor,
we get $u^2 = \lambda\cdot 1$
for some~$\lambda\in\mathbb{C}$.
Then by Postulate~\ref{post:assertexist}~\eqref{post:assertexist:andthen}
and using $\sqrt{p}u = u\sqrt{p}$,
we get, for all~$b\in\B(\mathscr{H})$,
\begin{equation*}
p\,b\,p
\ = \ 
\sqrt{p}u^*\sqrt{p} u^* b u \sqrt{p} u \sqrt{p}
\ = \ 
(\asrt_p \circ \asrt_p)(b) 
\ = \ 
\asrt_{p\andthen p}(b).
\end{equation*}
Note that $p\andthen p 
= \asrt_p(p) = \sqrt{p}u^* p u \sqrt{p} = p^2$.
Thus, for all~$b\in \B(\mathscr{H})$,
\begin{equation*}
\asrt_{p^2}(b)\ =\  p\,b\,p.
\end{equation*}
Since every element of~$[0,1]_{\B(\mathscr{H})}$
is a square,
we have proven Equation~\eqref{eq:uniqueness-asrt}
when~$\mathscr{A}\equiv\B(\mathscr{H})$.

Now, let
us consider the  general case (so $\mathscr{A}$ 
need not be of the form  $\mathscr{B}(\mathscr{H})$).
Let~$\omega\colon \mathscr{A} \to
\C$ be any normal state on~$\mathscr{A}$.  
Let~$a\in \mathscr{A}$ be given.
As normal states are
separating, it suffices to prove that~$\omega(\asrt_p(a)) =
\omega(\sqrt{p} a \sqrt{p})$.  
Let $\varrho\colon \mathscr{A} \to \B(\mathscr{K})$
be the GNS-representation 
of~$\mathscr{A}$ for the state~$\omega$ with cyclic
vector~$x\in\mathscr{K}$. 
Let~$\pi\colon \B(\mathscr{K}) \to \C$
be given by $\pi(b) = \inprod{x }{bx}$
for all~$b\in\mathscr{A}$. 
Then 
we have~$\omega = \pi
\after \varrho$ (in~$\vNA$).
Thus:
$$\begin{array}{rcll}
\omega(\asrt_p(a))
& = &
\pi ( \varrho (\asrt_p( a) )) \\
& = &
\pi ( \asrt_{\varrho(p)}( \varrho(a ))) & 
   \mbox{by Postulate~\ref{post:assertexist}~\eqref{post:assertexist:sharp}} \\
& = &
\pi(\sqrt{\varrho(p)} \varrho(a) \sqrt{\varrho(p)}) \qquad &
   \mbox{by uniqueness for~$\B(\mathscr{K})$} \\
& = &
\pi(\varrho(\sqrt{p} a \sqrt{p})) \\
& = &
\omega(\sqrt{p} a \sqrt{p})
\end{array}$$
Hence $\asrt_p(a) = \sqrt{p}a\sqrt{p}$. 
We have proven Equation~\eqref{eq:uniqueness-asrt}.
\QED
\end{proof}

\section{Conclusions and future work}\label{sec:conclusion}

This text collects definitions and results about the new notion of
effectus in categorical logic. Already at this early stage it is clear
that the theory of effectuses includes many examples that are of
interest in quantum (and probability) theory. But much remains to be
done. We list a few directions for further research.
\begin{enumerate}
\item Which constructions exist to obtain new effectuses from old,
  such as products, slices, (co)algebras of a (co)monad,
  \textit{etc.}?  A related matter is the definition of an appropriate
  notion of morphism of effectuses: one can take a functor that
  preserves finite coproducts and the final object; alternatively, one
  can take adjoints as morphisms, like in geometric morphisms between
  toposes.

\item Tensors in effectuses have been discussed in
  Section~\ref{sec:monoidal}, but only in a very superficial way. They
  deserve more attention, leading to a closer connection with the work
  done in the Oxford school (see the introduction). For instance, the
  combination of tensors $\otimes$ and coproducts $+$ could lead to a
  3-dimensional graphical calculus that combines (parallel)
  composition and effect logic.

\item The approach in this text is very much
  logic-oriented. Connections with quantum theory are touched upon,
  but should be elaborated further. In particular, the formulation
  (and correctness!) of concrete quantum protocols in the present
  setting is missing.

\item An internal language for effectuses, along the lines
  of~\cite{Adams14,AdamsJ15}, may be useful for the verification of
  probabilistic and/or quantum protocols.

\item The possibility of doing homological algebra (in abstract form,
  see~\cite{Grandis92,Grandis12} also deserves attention.
\end{enumerate}

\subsubsection*{Acknowledgements} This document benefitted from discussion
with and/or feedback from: Robin Adams, Tobias Fritz, Robert Furber, Aleks
Kissinger, Mathys Rennela, Sam Staton, Sean Tull, Sander Uijlen, and
Fabio Zanasi. We like to thank them all.

\bibliographystyle{alpha} 

\clearpage
\addcontentsline{toc}{section}{References}
\bibliography{common}

\printindex{N}{Notation Index}

\printindex{S}{Subject Index}

\end{document}